\RequirePackage{fix-cm} % fix some latex issues see: http://texdoc.net/texmf-dist/doc/latex/base/fixltx2e.pdf
\documentclass[ twoside,openright,titlepage,numbers=noenddot,headinclude,%1headlines,% letterpaper a4paper
                footinclude=true,cleardoublepage=empty,abstractoff, % <--- obsolete, remove (todo)
                BCOR=5mm,paper=a4,fontsize=11pt,%11pt,a4paper,%
                ngerman,american,%
                ]{scrreprt}

\PassOptionsToPackage{utf8}{inputenc}
	\usepackage{inputenc}

\PassOptionsToPackage{eulerchapternumbers,listings,%drafting,%
					 pdfspacing,%floatperchapter,%linedheaders,%
					 beramono,eulermath,parts}{classicthesis}                                        %subfig,
\newcommand{\myTitle}{Renormalization and Coarse-graining of Loop Quantum Gravity\xspace}
\newcommand{\mySubtitle}{~ \xspace}

\newcommand{\myName}{Christoph Charles\xspace}

\newcommand{\myFaculty}{\xspace}

\newcommand{\myUni}{École Normale Supérieure de Lyon\xspace}
\newcommand{\myLocation}{Lyon\xspace}
\newcommand{\myTime}{November 2016\xspace}
\newcommand{\myVersion}{version 0.1\xspace}

\newcounter{dummy} % necessary for correct hyperlinks (to index, bib, etc.)
\providecommand{\mLyX}{L\kern-.1667em\lower.25em\hbox{Y}\kern-.125emX\@}

	\usepackage{babel}                  

\usepackage{csquotes}
\PassOptionsToPackage{%
    backend=biber, %instead of bibtex
        bibencoding=utf8,%
	language=auto,%
	style=numeric-comp,%
    sorting=nyt, % name, year, title
    maxbibnames=10, % default: 3, et al.
    natbib=true % natbib compatibility mode (\citep and \citet still work)
}{biblatex}
    \usepackage{biblatex}

\PassOptionsToPackage{fleqn}{amsmath}       % math environments and more by the AMS 
    \usepackage{amsmath}
    \usepackage{amssymb}

\PassOptionsToPackage{T1}{fontenc} % T2A for cyrillics
    \usepackage{fontenc}     
\usepackage{textcomp} % fix warning with missing font shapes
\usepackage{scrhack} % fix warnings when using KOMA with listings package          
\usepackage{xspace} % to get the spacing after macros right  
\usepackage{mparhack} % get marginpar right
\usepackage{fixltx2e} % fixes some LaTeX stuff --> since 2015 in the LaTeX kernel (see below)
\PassOptionsToPackage{printonlyused,smaller}{acronym} 
    \usepackage{acronym} % nice macros for handling all acronyms in the thesis
\usepackage{tabularx} % better tables
\usepackage{caption}
\usepackage{listings} 
\PassOptionsToPackage{pdftex,hyperfootnotes=false,pdfpagelabels}{hyperref}
    \usepackage{hyperref}  % backref linktocpage pagebackref
\PassOptionsToPackage{pdftex}{graphicx}
    \usepackage{graphicx}

\hypersetup{%
    colorlinks=true, linktocpage=true, pdfstartpage=3, pdfstartview=FitV,%
    breaklinks=true, pdfpagemode=UseNone, pageanchor=true, pdfpagemode=UseOutlines,%
    plainpages=false, bookmarksnumbered, bookmarksopen=true, bookmarksopenlevel=1,%
    hypertexnames=true, pdfhighlight=/O,%nesting=true,%frenchlinks,%
    urlcolor=webbrown, linkcolor=RoyalBlue, citecolor=webgreen, %pagecolor=RoyalBlue,%
    pdftitle={\myTitle},%
    pdfauthor={\textcopyright\ \myName, \myUni, \myFaculty},%
    pdfsubject={},%
    pdfkeywords={},%
    pdfcreator={pdfLaTeX},%
    pdfproducer={LaTeX with hyperref and classicthesis}%
}   

\makeatletter
\@ifpackageloaded{babel}%
    {%
       \addto\extrasamerican{%
                }%
       \addto\extrasngerman{% 
                }%  
    }{\relax}
\makeatother

\usepackage{classicthesis} 
\usepackage{subcaption}
\usepackage{tikz}
\usetikzlibrary{calc}
\usetikzlibrary{decorations.pathmorphing}

\usepackage{pgfplots}

\newlength\fullmarginwidth
\fullmarginwidth=\marginparwidth
\advance\fullmarginwidth by \marginparsep

\newlength\fullwidth
\fullwidth=\textwidth
\advance\fullwidth by \fullmarginwidth

\usepackage{changepage}

\newcommand{\largebox}[1]{
\checkoddpage
\edef\side{\ifoddpage l\else r\fi}
\makebox[\textwidth][\side]{
\begin{minipage}{\fullwidth}
#1
\end{minipage}
}
}

\newcommand{\inspiquote}[2]{
\largebox{
\begin{flushright}
\begin{minipage}{0.5\textwidth}
\hfill {\slshape #1} \quad {--- #2}
\vspace{1em}
\hrule
\vspace{1em}
\end{minipage}
\end{flushright}
}
}

\usepackage{pdfpages}

\begin{document}

\includepdf[pages=1-last]{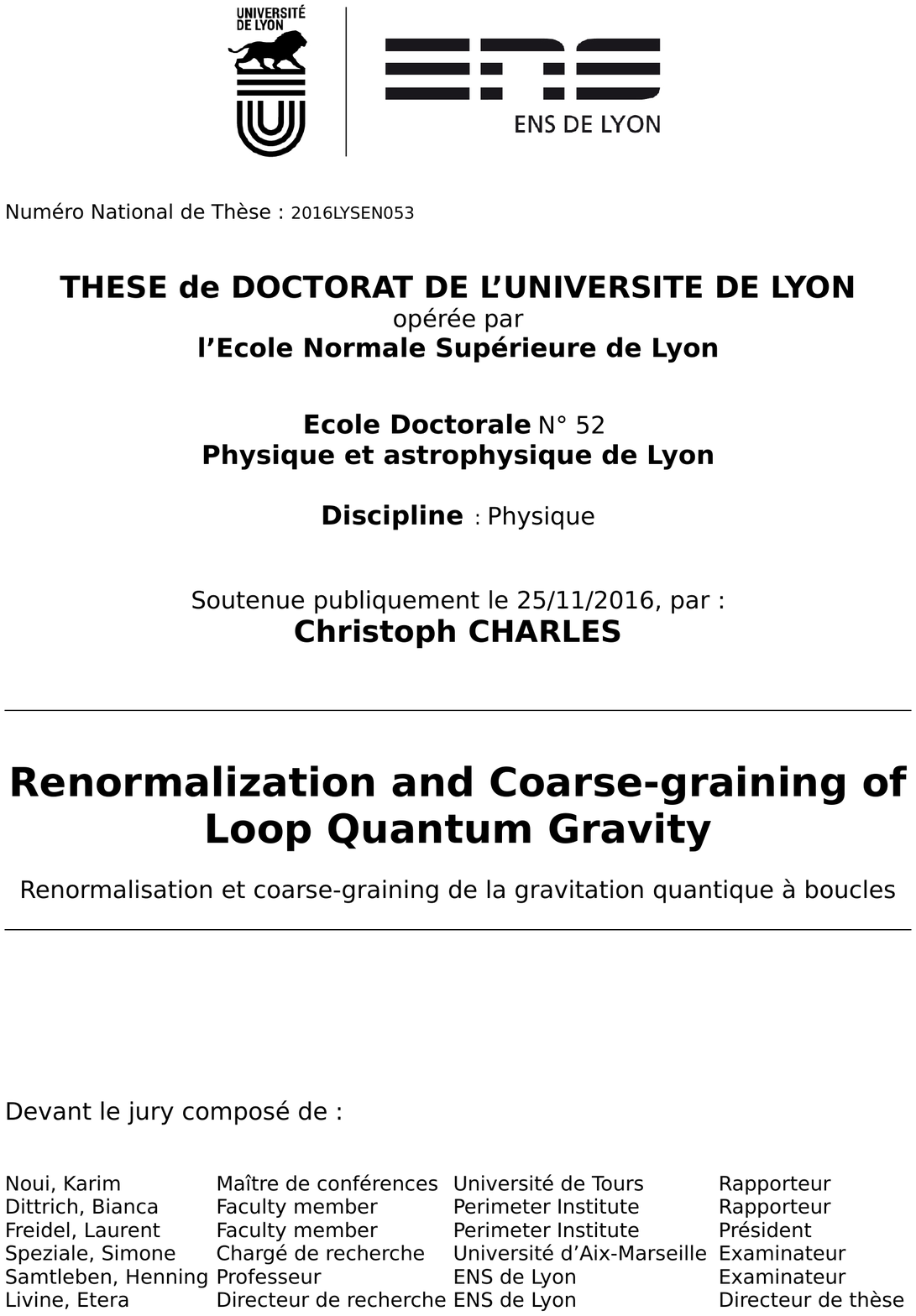}
\includepdf[pages=1-last]{blank.pdf}

\frenchspacing
\raggedbottom
\selectlanguage{american} % american ngerman
%\renewcommand*{\bibname}{new name}
%\setbibpreamble{}
\pagenumbering{roman}
\pagestyle{plain}
%********************************************************************
% Frontmatter
%*******************************************************
%*******************************************************
% Little Dirty Titlepage
%*******************************************************
\thispagestyle{empty}
%\pdfbookmark[1]{Titel}{title}
%*******************************************************
\begin{center}
    \spacedlowsmallcaps{\myName} \\ \medskip                        

    \begingroup
        \color{Maroon}\spacedallcaps{\myTitle}
    \endgroup
\end{center}        

%*******************************************************
% Titlepage
%*******************************************************
\begin{titlepage}
    % if you want the titlepage to be centered, uncomment and fine-tune the line below (KOMA classes environment)
    \begin{addmargin}[-1cm]{-3cm}
    \begin{center}
        \large  

        \hfill

        \vfill

        \begingroup
            \color{Maroon}\spacedallcaps{\myTitle} \\ \bigskip
        \endgroup

        \spacedlowsmallcaps{\myName}

        \vfill

        \mySubtitle \\ \medskip   
        %\myDegree \\
        %\myDepartment \\                            
        %\myFaculty \\
        %\myUni \\ \bigskip

        \myTime\ % -- \myVersion

        \vfill                      

    \end{center}  
  \end{addmargin}       
\end{titlepage}   

\thispagestyle{empty}

\hfill

\vfill

\noindent\myName: \textit{\myTitle,} \mySubtitle, %\myDegree, 
\textcopyright\ \myTime

%\bigskip
%
%\noindent\spacedlowsmallcaps{Supervisors}: \\
%\myProf \\
%\myOtherProf \\ 
%\mySupervisor
%
%\medskip
%
%\noindent\spacedlowsmallcaps{Location}: \\
%\myLocation
%
%\medskip
%
%\noindent\spacedlowsmallcaps{Time Frame}: \\
%\myTime

\cleardoublepage%*******************************************************
% Dedication
%*******************************************************
\thispagestyle{empty}
%\phantomsection 
\refstepcounter{dummy}
\pdfbookmark[1]{Dedication}{Dedication}

\vspace*{3cm}

\begin{center}
    ``Together, or not at all'' \\ \medskip
    --- Amy Pond
%    \emph{Ohana} means family. \\
%    Family means nobody gets left behind, or forgotten. \\ \medskip
%    --- Lilo \& Stitch    
\end{center}

\medskip

\begin{center}
To my parents to whom I owe nearly everything. \\ \smallskip
To my wife who has always been supportive.% \\ \smallskip
%To my God who makes all things possible.
    %To my parents who tought me to love knowledge
%    Dedicated to the loving memory of Rudolf Miede. \\ 
%    1939\,--\,2005
\end{center}

%\cleardoublepage\include{FrontBackmatter/Foreword}
\cleardoublepage%*******************************************************
% Abstract
%*******************************************************
%\renewcommand{\abstractname}{Abstract}
\pdfbookmark[1]{Abstract}{Abstract}
\begingroup
\let\clearpage\relax
\let\cleardoublepage\relax
\let\cleardoublepage\relax

\chapter*{Abstract}
The continuum limit of loop quantum gravity is still an open problem. Indeed, no proper dynamics in known to start with and we still lack the  mathematical tools to study its would-be continuum limit. In the present PhD dissertation, we will investigate some coarse-graining methods that should become helpful in this enterprise. We concentrate on two aspects of the theory's coarse-graining: finding natural large scale observables on one hand and studying how the dynamics of varying graphs could be cast onto fixed graphs on the other hand.

To determine large scale observables, we study the case of hyperbolic tetrahedra and their natural description in a language close to loop quantum gravity. The surface holonomies in particular play an important role. This highlights the structure of double spin networks, which consist in a graph and its dual, which seems to also appear in works from Freidel \textit{et al}. To solve the problem of varying graphs, we consider and define loopy spin networks. They encode the local curvature with loops around an effective vertex and allow to describe different graphs by hidding them in a coarse-graining process. Moreover, their definition gives a natural procedure for coarse-graining allowing to relate different scales.

Together, these two results constitute the foundation of a coarse-graining programme for diffeomorphism invariant theories.

%Short summary of the contents in English\dots a great guide by 
%Kent Beck how to write good abstracts can be found here:  
%\begin{center}
%\url{https://plg.uwaterloo.ca/~migod/research/beckOOPSLA.html}
%\end{center}

%\vfill
\newpage

\pdfbookmark[1]{Résumé}{Résumé}
\chapter*{Résumé}
Le problème de la limite continue de la gravitation quantique à boucle est encore ouvert. En effet, la dynamique précise n’est pas connue et nous ne disposons pas des outils nécessaires à l’étude de cette limite le cas échéant. Dans cette thèse, nous étudions quelques méthodes de \textit{coarse-graining} (étude à gros grains) qui devraient contribuer à cette entreprise. Nous nous concentrons sur deux aspects du flot: la détermination d’observables naturelles à grandes échelles d’un côté et la manière de s’abstraire du problème de la dynamique à graphe variable en la projetant sur des graphes fixes de l'autre.

Pour déterminer les observables aux grandes distances, nous étudions le cas des tétraèdres hyperboliques et leur description naturelle dans un langage proche de celui de la gravitation quantique à boucle. Les holonomies de surface en particulier jouent un rôle important. Cela dégage la structure des \textit{double spin networks} constitués d'un graphe et de son dual, structure qui semble aussi apparaître dans les travaux de Freidel \textit{et al}. Pour résoudre le problème des graphes variables, nous considérons et définissons les \textit{loopy spin networks}. Ils encodent par des boucles la courbure locale d'un vertex effectif et permettent ainsi de décrire différents graphes en les masquant via le processus de \textit{coarse-graining}. De plus, leur définition donne un procédé naturel systématique de \textit{coarse-graining} pour passer d'une échelle à une autre.

Ensemble, ces deux principaux résultats posent le fondement d'un programme de \textit{coarse-graining} pour les théories invariantes sous difféomorphismes.

%Kurze Zusammenfassung des Inhaltes in deutscher Sprache\dots 

\endgroup			

\vfill

\cleardoublepage%*******************************************************
% Publications
%*******************************************************
\pdfbookmark[1]{Publications}{publications}
\chapter*{Publications}%\graffito{This is just an early --~and currently ugly~-- test!}
%This might come in handy for PhD theses: some ideas and figures have appeared previously in the following publications:

%\noindent Put your publications from the thesis here. The packages \texttt{multibib} or \texttt{bibtopic} etc. can be used to handle multiple different bibliographies in your document.

\begin{refsection}[ownpubs]
    \small
    \nocite{*} % is local to to the enclosing refsection
    \printbibliography[heading=none]
\end{refsection}

%\emph{Attention}: This requires a separate run of \texttt{bibtex} for your \texttt{refsection}, \eg, \texttt{ClassicThesis1-blx} for this file. You might also use \texttt{biber} as the backend for \texttt{biblatex}. See also \url{http://tex.stackexchange.com/questions/128196/problem-with-refsection}.

\cleardoublepage%*******************************************************
% Acknowledgments
%*******************************************************
\pdfbookmark[1]{Acknowledgments}{acknowledgments}

%\begin{flushright}{\slshape    
%    We have seen that computer programming is an art, \\ 
%    because it applies accumulated knowledge to the world, \\ 
%    because it requires skill and ingenuity, and especially \\
%    because it produces objects of beauty.} \\ \medskip
%    --- \defcitealias{knuth:1974}{Donald E. Knuth}\citetalias{knuth:1974} \citep{knuth:1974}
%\end{flushright}

\bigskip

\begingroup
\let\clearpage\relax
\let\cleardoublepage\relax
\let\cleardoublepage\relax
\chapter*{Acknowledgments}

I was told this is the hardest part to write. It is certainly most difficult not to forget someone as so many have helped in one way or another. The order of mention should not be taken to represent anything and certainly does not reflect the amount of work these people provided.

\bigskip

First and foremost, I want to give a very special thanks to my advisor, doctor Etera Livine, who has mentored me, taught me and guided me during the whole thesis. Etera, it took some time to adapt to each other but in the end, I'm grateful. I'm grateful for learning to work with people with different approaches but I'm more especially grateful for all the sacrifices you made and the help you provided as a mental support, most notably in the end. You have introduced the world of research to me and taught me to love it enough to continue. I discovered quantum gravity and more specifically loop quantum gravity thanks to you. For all this, I want to say \textit{thank you}.

Thank you also to all of the defense committee members, doctor Laurent Freidel, doctor Henning Samtleben, doctor Bianca Dittrich, doctor Karim Noui and doctor Simone Speziale. All of you have been kind enough to answer every e-mail, to wait when I had some trouble getting all the information in time and have been flexible enough to accommodate all the difficulties. Thank you also for the help you gave when I had administrative issues.

\bigskip

I also want to thank the administrative team of the ENS of Lyon, which is one of the most effective I have ever met. But most importantly, I am always impressed by their willingness to help, their kindness and the support they provide. I must thank most particularly Myriam Friat who has borne with me for at least two weeks in the last run for the Ph.D. defense and helped me at every step of the process.

My labmates have been of tremendous help too and it is difficult not to omit someone. I want to thank all of them for welcoming me in the first place, for all the discussions about nearly everything, but also for all the support and the encouragement they provided. Thank you to Michel Fruchart, David Lopes Cardozo, to Daniele Malpetti and to Baptiste Pezelier. A particular thank you must be addressed to Michel who have taken the trouble of listening to me talking for hours about quantum gravity just to help me clarify my ideas. His help cannot be overstated. I also want to thank Christophe Göller whose internship I had the chance to oversee. He encouraged me a lot, allowed some very interesting conversations and made the summer work a bit more lifeful. I want to thank him more specifically for the help he provided on some quite lengthy computations.

\bigskip

I want to thank all my close friends who have been an incredible and inexhaustible source of courage and joy, namely, Chérif Adouama, Ivan Bannwarth, Dimitri Cobb, Siméon Giménez and Cyril Philippe. Thank you for all the discussions. Thank you Cyril for our talks on epistemology guiding my research. Thank you Dimitri for your help in math, where your knowledge is quickly exceeding mine. Thank you Ivan and Chérif for all the discussions about the Ph.D. life and your attentive ear. Thank you Siméon for the numerous conversations, the way too many jokes and the galaxy savings.

I am also thinking about my flatmates though this group intersects with the previous one. Thank you Nathanaël Dobé and Marjorie Dobé, thank you Cyril Philippe, thank you Siméon Giménez, thank you Dimitri Cobb, thank you Maxime Kristanek and thank your Alexandre Drouin. Thank you for all the time playing, sharing meals and talking. Thank you for the good food and wine we shared and thank you for the less good burgers and coke we shared. I also want to thank people from church who have always been keen to ask for news and have been of great support: Gilles and Yvette Campoy, Denis and Maïlys Blum, Jonah and Amy Haddad, Alexandre and Suzanne Sarran. I want to thank you all for the help you gave: this dissertation would not exist without you.

\bigskip

I also have to thank some specific people for more particular and tedious work that has been done. Thank you to Alexandre Feller for the courtesy of giving a picture for the conclusion. I want to thank Dimitri Cobb and Jonah Haddad for their help in spell checking and in correcting the English. They have seen a lot of mistakes, way more than I could ever hope to find alone. If something seems particularly well expressed, most chances are one of them suggested it. If there is a mistake however, most chances are it is on me. Thank you also to my PhD advisor, Etera Livine, who provided insightful comments pertaining to the scientific content of this work and helped clarify a few points.

I want to thank my family who has always been supportive. I want to thank my brothers who have shown interest in my work even though it is not their subject at all. I want to thank them for their presence which was all the more palpable during the wedding. I want to thank my in-laws who helped us with the difficulties of newly-wed life and therefore alleviated much difficulties that would have impacted my work. Thank you!

\bigskip

I want to finish with particularly important people, who are very dear to me and to whom this dissertation is dedicated. First, to my mom and dad. When you could not help me directly for the work, you helped all the same in anyway you could: helping us with installation, with financial issues, with administrative problems or even dealing with stress. You were always of great advice. You strengthened me spiritually. And you helped with the Ph.D. directly too! By listening to me for hours reflecting on quantum gravity, you helped me a lot. But more fundamentally, you taught me to love knowledge and to seek it. You taught me how to work and how to work hard. I owe you everything. I have toward you a debt I cannot possibly repay. And for this, I thank you!

Finally, I want to thank my wife, Aurélie. You ran the house and it is thanks to you that I kept my head up during the hard times of work. You had worried for me when I could not anymore. You allowed work to come into the house when it was needed. You listened to me talking about research when you were tired. You reminded what was important and what was not, what must be done before all else and what matters less. You listened to me when everything was falling apart and you helped me back into prayer. You helped me go through the dark times in faith. We had our fair share of troubles since we got married and yet, for all my problems, you were there. To \textit{you} this work is dedicated. This is to remind me what is most important: \textit{together or not at all}.

\bigskip
\begin{center}
Thank you!
\end{center}

%Put your acknowledgments here.

%Many thanks to everybody who already sent me a postcard!

%Regarding the typography and other help, many thanks go to Marco 
%Kuhlmann, Philipp Lehman, Lothar Schlesier, Jim Young, Lorenzo 
%Pantieri and Enrico Gregorio\footnote{Members of GuIT (Gruppo 
%Italiano Utilizzatori di \TeX\ e \LaTeX )}, J\"org Sommer, 
%Joachim K\"ostler, Daniel Gottschlag, Denis Aydin, Paride 
%Legovini, Steffen Prochnow, Nicolas Repp, Hinrich Harms, 
% Roland Winkler, Jörg Weber, Henri Menke, Claus Lahiri, 
% Clemens Niederberger, Stefano Bragaglia, Jörn Hees, 
% and the whole \LaTeX-community for support, ideas and 
% some great software.

%\bigskip

%\noindent\emph{Regarding \mLyX}: The \mLyX\ port was intially done by 
%\emph{Nicholas Mariette} in March 2009 and continued by 
%\emph{Ivo Pletikosi\'c} in 2011. Thank you very much for your 
%work and for the contributions to the original style.

\endgroup

\pagestyle{scrheadings}
\cleardoublepage%*******************************************************
% Table of Contents
%*******************************************************
%\phantomsection
\refstepcounter{dummy}
\pdfbookmark[1]{\contentsname}{tableofcontents}
\setcounter{tocdepth}{2} % <-- 2 includes up to subsections in the ToC
\setcounter{secnumdepth}{3} % <-- 3 numbers up to subsubsections
\manualmark
\markboth{\spacedlowsmallcaps{\contentsname}}{\spacedlowsmallcaps{\contentsname}}
\tableofcontents 
\automark[section]{chapter}
\renewcommand{\chaptermark}[1]{\markboth{\spacedlowsmallcaps{#1}}{\spacedlowsmallcaps{#1}}}
\renewcommand{\sectionmark}[1]{\markright{\thesection\enspace\spacedlowsmallcaps{#1}}}
%*******************************************************
% List of Figures and of the Tables
%*******************************************************
\clearpage

\begingroup 
    \let\clearpage\relax
    \let\cleardoublepage\relax
    \let\cleardoublepage\relax
    %*******************************************************
    % List of Figures
    %*******************************************************    
    %\phantomsection 
    \refstepcounter{dummy}
    %\addcontentsline{toc}{chapter}{\listfigurename}
    \pdfbookmark[1]{\listfigurename}{lof}
    \listoffigures

    \vspace{8ex}

    %*******************************************************
    % List of Tables
    %*******************************************************
    %\phantomsection 
    %\refstepcounter{dummy}
    %\addcontentsline{toc}{chapter}{\listtablename}
    %\pdfbookmark[1]{\listtablename}{lot}
    %\listoftables
        
    %\vspace{8ex}
%   \newpage
    
    %*******************************************************
    % List of Listings
    %*******************************************************      
      %\phantomsection 
    %\refstepcounter{dummy}
    %\addcontentsline{toc}{chapter}{\lstlistlistingname}
    %\pdfbookmark[1]{\lstlistlistingname}{lol}
    %\lstlistoflistings 

    %\vspace{8ex}
       
    %*******************************************************
    % Acronyms
    %*******************************************************
    %\phantomsection
    \newpage

    \refstepcounter{dummy}
    \pdfbookmark[1]{Acronyms}{acronyms}
    \markboth{\spacedlowsmallcaps{Acronyms}}{\spacedlowsmallcaps{Acronyms}}
    \chapter*{Acronyms}
    \begin{acronym}[UMLX]
        \acro{LQG}{Loop Quantum Gravity}
        \acro{LQC}{Loop Quantum Cosmology}
        \acro{GR}{General Relativity}
        \acro{GFT}{Group Field Theory}
        \acro{QFT}{Quantum Field Theory}
        \acro{QED}{Quantum ElectroDynamics}
        \acro{QCD}{Quantum ChromoDynamics}
        \acro{AS}{Asymptotic Safety}
    \end{acronym}                     
\endgroup

%********************************************************************
% Mainmatter
%*******************************************************
\cleardoublepage\pagenumbering{arabic}
%\setcounter{page}{90}
% use \cleardoublepage here to avoid problems with pdfbookmark
\cleardoublepage
%*****************************************
\chapter{Introduction} \label{ch:Introduction}
%*****************************************

\inspiquote{Allons-y!}{The Doctor}

\section{The problem}

\ac{LQG}\graffito{For textbooks and reviews on Loop Quantum Gravity, see \cite{Gambini:2011zz,rovelli2007quantum,Thiemann:2007zz,rovelli2014quantum,Ashtekar2004,Rovelli2008,Mercuri2010,Dona2010}.} provides the framework for all the work done in this thesis. This theory is, as its names suggests, a quantum theory of gravity, or at least a proposal thereof. It is now well-known that the two masterpieces of theoretical physics of the $20^\textrm{th}$ century, namely \ac{QFT} and \ac{GR}
%\cite{Einstein1916}
%(for lectures, see \cite{Carroll1997,Feynman2003,misner1973gravitation})
, are incredibly successful to a degree that is astonishing. \ac{QFT}, which successfully gives a quantum mechanical description of relativistic phenomena, now has experimental support to $13$ digits \cite{PhysRevLett.97.030801,PhysRevLett.97.030802}. The Standard Model of particle which is expressed in the \ac{QFT} framework was confirmed quite recently when the final piece of the particles it needed was discovered in 2012: the Higgs boson \cite{Chatrchyan:2012xdj,Aad:2012tfa}. \ac{GR} also has its own impressive track record, with precision tests in the solar system \cite{2009ApJ...699.1395F,PhysRevLett.45.2081}, the development of cosmology or the very recent direct detection of gravitationnal waves \cite{Abbott:2016blz}. The successes of these theories suggest that they get \textit{something right}. This entails, with increasing evidence, that a complete description of physics must be able to fathom events that are quantum mechanical, relativistic and gravitationnal at the same time.

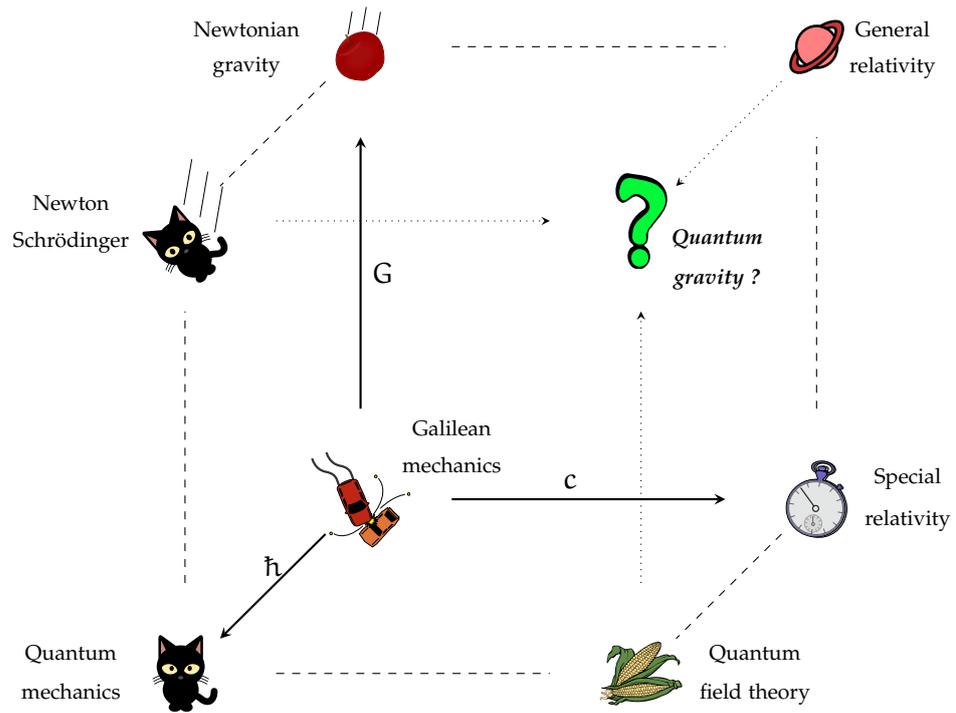
\begin{figure}[h!]
  \centering

  \begin{tikzpicture}[scale=1]
    \def \d {6};
    \def \p {0.2};
    \def \decal {1};
    
    \coordinate (GM) at (0,0,0); %Galilean Mechanics
    \coordinate (NG) at (0,\d,0); % Newton Gravity
    \coordinate (QM) at (0,0,\d); % Quantum Mechanics
    \coordinate (SR) at (\d,0,0); % Special Relativity

    \coordinate (GR) at ($(SR) + (NG)$); % General Relativity
    \coordinate (QFT) at ($(SR) + (QM)$); % Quantum Field Theory
    \coordinate (NS) at ($(NG) + (QM)$); % Newton Schroedinger

    \coordinate (QG) at ($(GR) + (QM)$); % Quantum Gravity

    \draw[->,>=stealth,thick] ($(GM) + \p*(NG)$) -- node[midway,right]{$G$} ($(NG) - \p*(NG)$);
    \draw[->,>=stealth,thick] ($(GM) + \p*(QM)$) -- node[midway,above]{$\hbar$} ($(QM) - \p*(QM)$);
    \draw[->,>=stealth,thick] ($(GM) + \p*(SR)$) -- node[midway,above left]{$c$} ($(SR) - \p*(SR)$);

    \draw[dashed] ($(NG) +\p*(SR)$) -- ($(GR) - \p*(SR)$);
    \draw[dashed] ($(NG) +\p*(QM)$) -- ($(NS) - \p*(QM)$);

    \draw[dashed] ($(QM) +\p*(SR)$) -- ($(QFT) - \p*(SR)$);
    \draw[dashed] ($(QM) +\p*(NG)$) -- ($(NS) - \p*(NG)$);

    \draw[dashed] ($(SR) + \p*(QM)$) -- ($(QFT) - \p*(QM)$);
    \draw[dashed] ($(SR) + \p*(NG)$) -- ($(GR) - \p*(NG)$);

    \draw (GM) node{\includegraphics[scale=0.15]{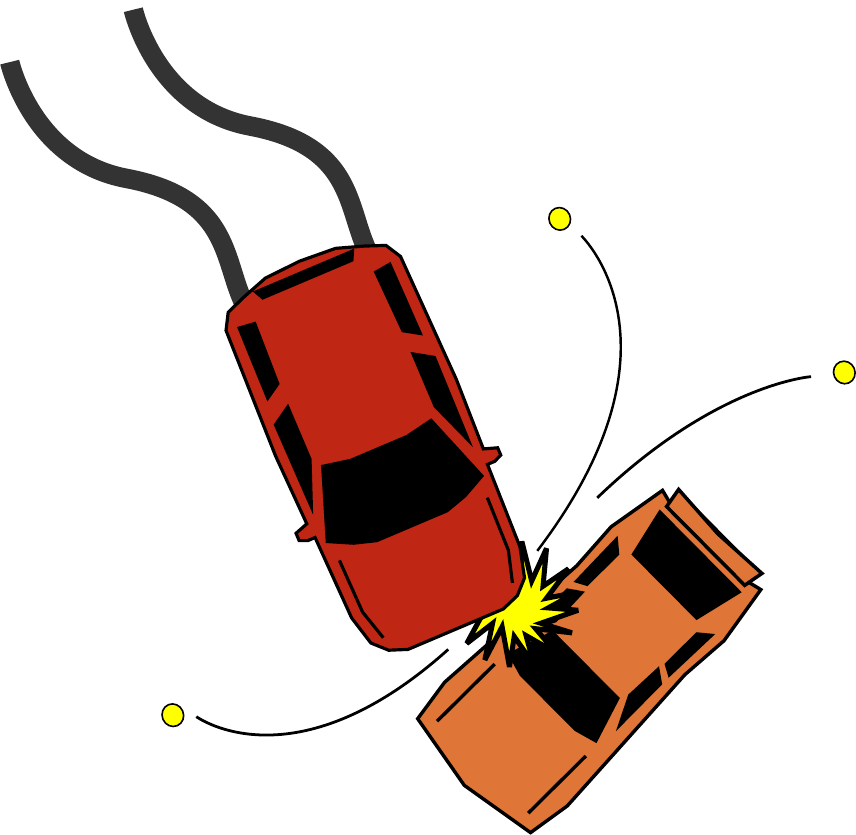}} ++(1.2*\decal,0.7*\decal) node[above]{\scriptsize Galilean} node[below]{\scriptsize mechanics};
    \draw (NG) node{\includegraphics[scale=0.04]{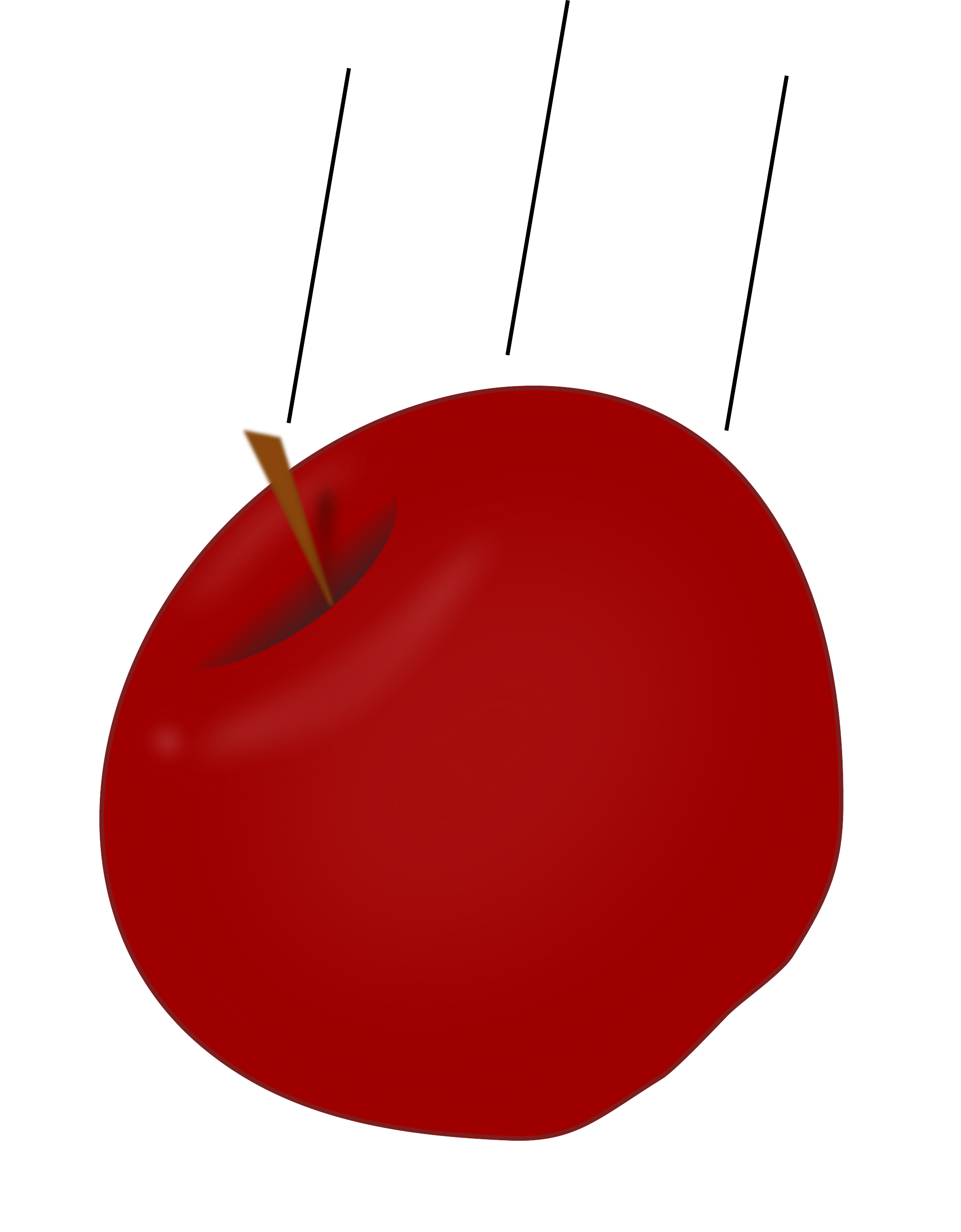}} ++(-1.5*\decal,0) node[above]{\scriptsize Newtonian} node[below]{\scriptsize gravity};
    \draw (SR) node{\includegraphics[scale=0.08]{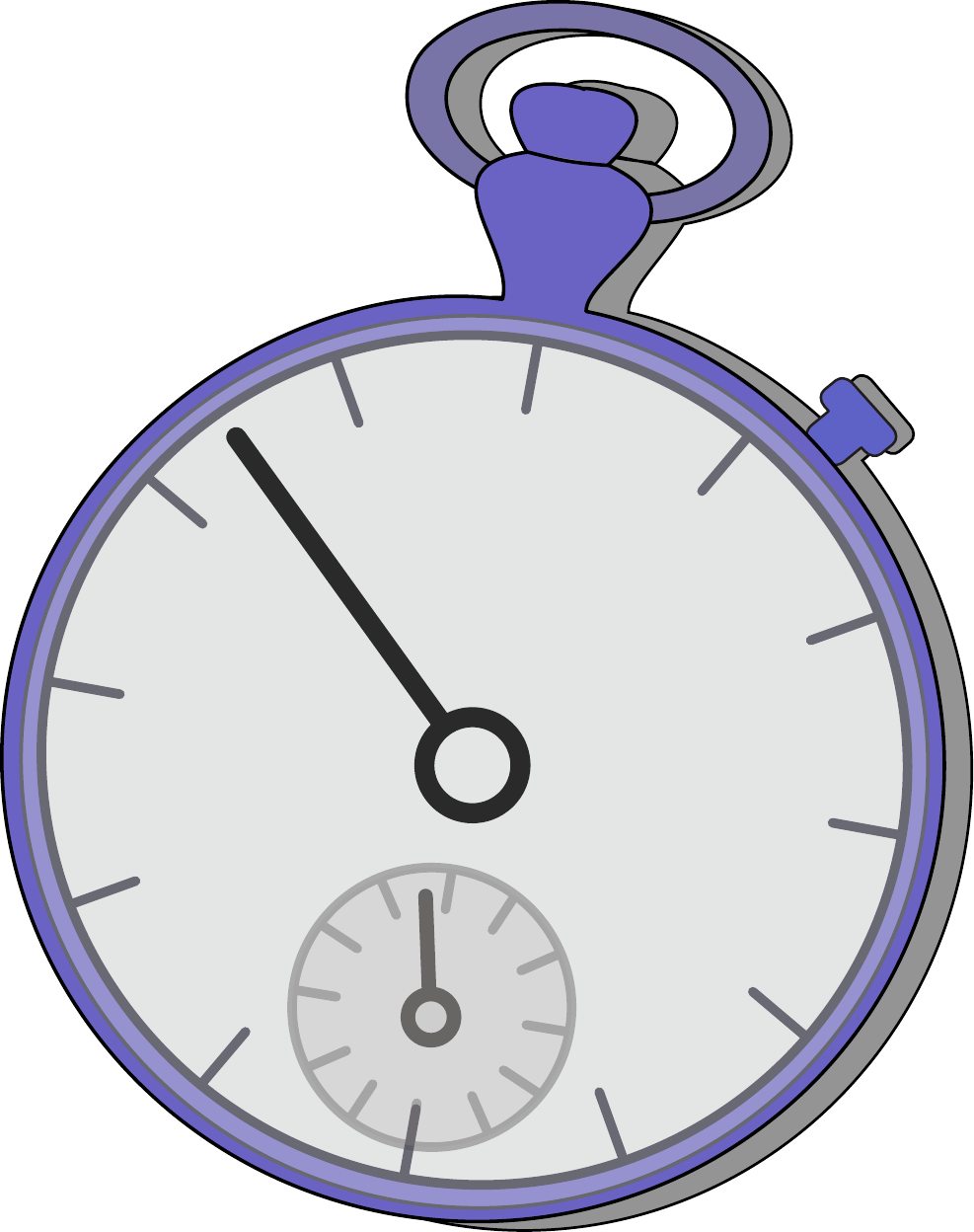}} ++(1.2*\decal,0) node[above]{\scriptsize Special} node[below]{\scriptsize relativity};
    \draw (GR) node{\includegraphics[scale=0.2]{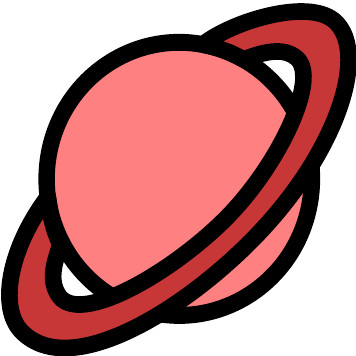}} ++(\decal,0) node[above]{\scriptsize General} node[below]{\scriptsize relativity};
    
    \draw[->,>=stealth,dotted] ($(GR) + \p*(QM)$) -- ($(QG) - \p*(QM)$);
    \draw[->,>=stealth,dotted] ($(QFT) + \p*(NG)$) -- ($(QG) - \p*(NG)$);
    \draw[->,>=stealth,dotted] ($(NS) + \p*(SR)$) -- ($(QG) - \p*(SR)$);

    \draw (QFT) node{\includegraphics[scale=0.1]{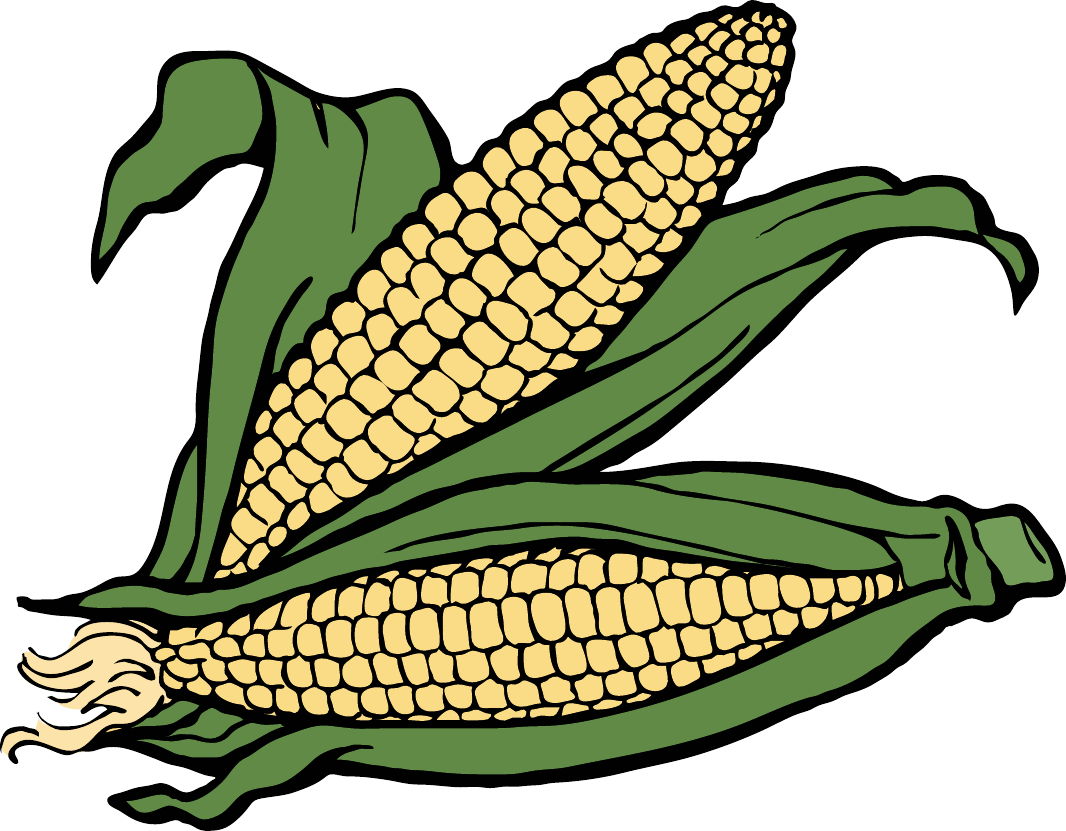}} ++(1.5*\decal,0) node[above]{\scriptsize Quantum} node[below]{\scriptsize field theory};
    \draw (QM) node{\includegraphics[scale=0.4]{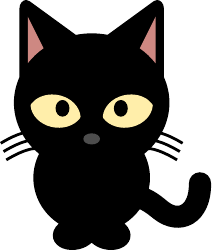}} ++(-1.5*\decal,0) node[above]{\scriptsize Quantum} node[below]{\scriptsize mechanics};
    \draw (NS) node{\includegraphics[scale=0.4]{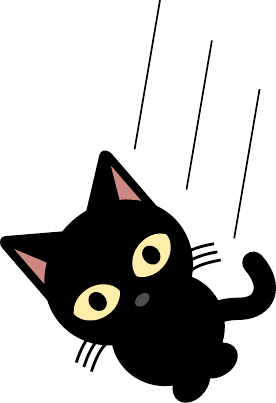}} ++(-1.5*\decal,0) node[above]{\scriptsize Newton} node[below]{\scriptsize Schrödinger};
    
    \draw (QG) node{\includegraphics[scale=0.07]{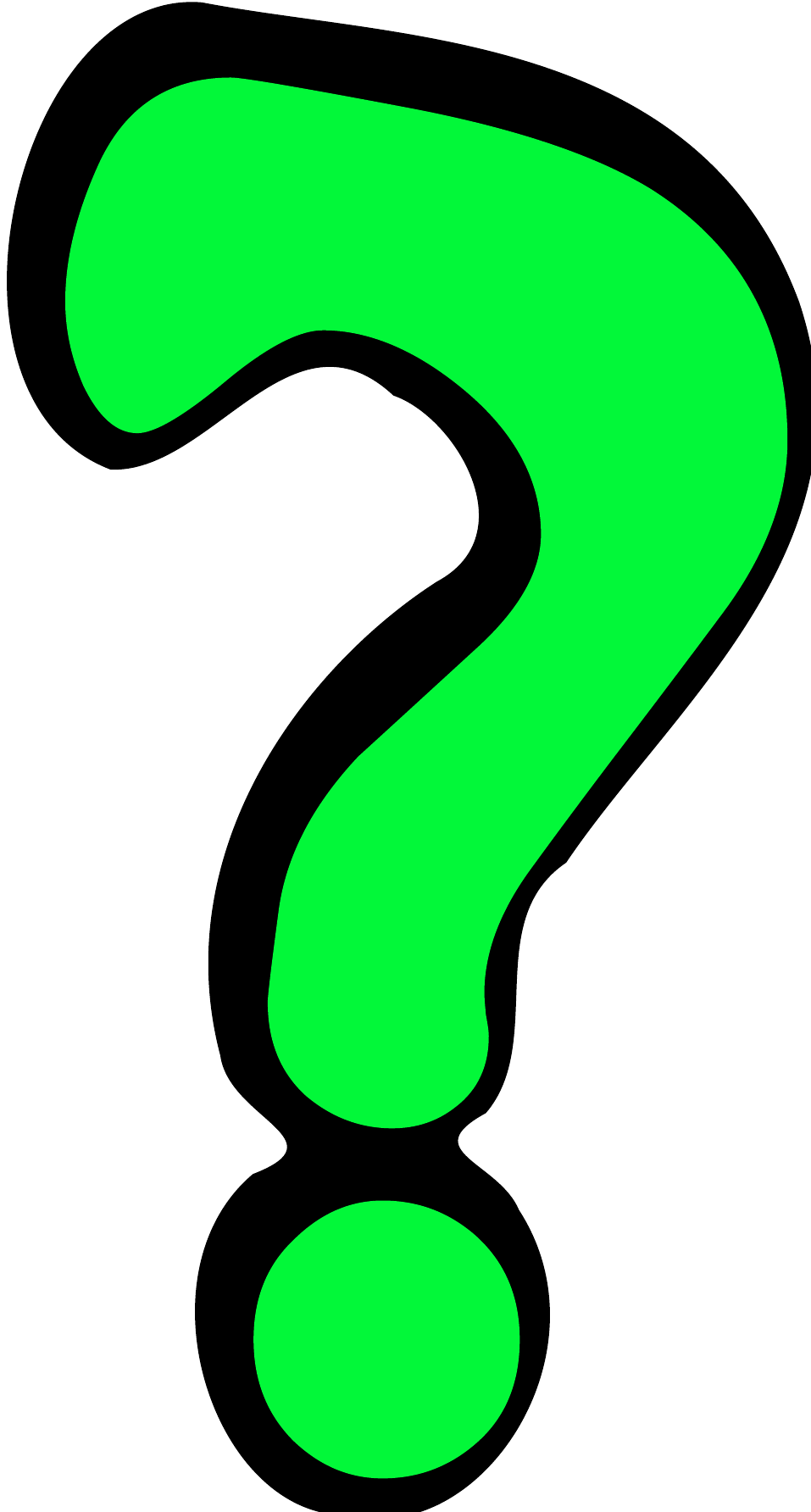}} ++(\decal,-0.5*\decal) node[above]{\scriptsize \textit{\textbf{Quantum}}} node[below]{\scriptsize \textit{\textbf{gravity ?}}};
  \end{tikzpicture}  
  \caption{The cube of fundamental constants (cliparts taken or adapted from \href{http://openclipart.org}{Open Clipart} -- free of use).}
  \label{fig:fundamentalCube}
\end{figure}

A quantum theory of gravity (already called quantum gravity) would give a precise and complete description of gravity that is quantum mechanical. It should reduce to \ac{GR} in some large-distance, small energy density limit, and lead to \ac{QFT} when the gravitationnal field is weak. This can be recast in a diagramatic form as seen on fig.\ref{fig:fundamentalCube}\graffito{The details of the cube might be up for discussion. Indeed, the newtonian limit is also usually considered in the regime of low curvature, not only $c=\infty$. But it is a handy tool to present the scope of possible theories.} as \textit{the cube of fundamental constants}. Each dimension of the cube corresponds to some fundamental physical constant that can be considered or discarded. $c$ of course encodes (special) relativity, $\hbar$ quantum mechanics and $G$ gravity. If all are discarded, we fall back to Newtonian mechanics without gravity. We know how to consistantly implement every couple of constants, with general relativity ($c$ and $G$) and quantum field theory ($c$ and $\hbar$) and even, surprisingly the couple $\hbar$ and $G$ \cite{Christian:1997wj}. Note that it is possible to define other limits by changing which constant is considered more fundamental. Nothing prevents us from replacing $G$ by a string tension, a Planck length or a mass for instance. These choices will underline different approaches. This does not change however the fact that the last corner of the cube of quantum gravity is still to be conquered. So far, it is only possible to consider quantum field theory on a given classical background \cite{Baez1996,WALD1994} or limits in a non-relativistic setting \cite{Nesvizhevsky2003}, but no full quantum theory of general relativity, not even with pure gravity is perfectly known. Building such a full theory turns out to be difficult \cite{Carlip2001}. At first sight, because of the huge success of gauge theories, it seems natural to try and follow the same path for quantum gravity as was used for the other interactions: start with the canonical action and quantize it. This method does not really work straightforwardly even in the case of QCD or any Yang-Mills theory, since infinities appear. Fortunately these infinities are absorbable by renormalization. This is not the case however with quantum gravity. And it was to be expected as gravity is conceptualy quite different from gauge theories. Indeed, the symmetries of gauge theories are \textit{internal} symmetries, which means that they act at a given point in spacetime. The symmetries of gravity, however, include the diffeomorphism group. \graffito{Some approaches to quantum gravity even try to start from notions like causality. See for instance causal sets and causel histories and \cite{Surya:2011yh} for a recent overview.}This makes it difficult to define a precise notion of causality and locality in the usual manner. Indeed, we loose the global symmetry group of spacetime (namely the Poincare group) which was essential in the construction of quantum field theories. Though causality and locality are still defined in principle, they cannot be implemented as usual through a lightcone structure on a fixed background.

There is a simple enough way to alleviate the problem though: considering the linearization of gravity around a flat Minkowskian metric and developing the theory perturbatively. With such a procedure, spacetime is now equipped with a natural (Minkowskian) metric and such notions are well-defined on it. It might even be possible to develop around other backgrounds, the expansion corresponding to small perturbation coming from the quantum corrections. This however can not work in a simple way. The coupling constant of gravity is $G$ which is homogeneous to the inverse square of a mass in dimension $3+1$ (in units where $\hbar = c = 1$). In particular, from a power counting argument, it appears that the theory is non-renormalizable, meaning that each order of the theory will bring new terms to the action, with their corresponding coefficients. In fact, there is a miracle for the one-loop amplitude of general relativity, implying that it is renormalizable to one-loop\graffito{In principle, renormalizability is a property of the whole perturbation process, not of one order. So, \textit{renormalizability at one-loop} means that all infinities at one-loop can be absorbed by field redefinition.} \cite{tHooft:1974bx} with the appropriate redefinitions of the terms. This property however disappears at next orders or when we include matter fields.

From this point on, there are essentially two possible directions, two possible mindsets: either the technique applied to gauge theories is not generalizable and something peculiar is happening with gravity, or there is something wrong with general relativity. Considering that something is wrong with general relativity does not mean, of course, that everything and anything is possible. It means rather that general relativity, though a good theory, must be completed at high energies, as was the Fermi interaction by the weak force. In this regard, \ac{GR} must be considered as an effective field theory, valid only at low energies. This leads to the search for either some new symmetry (like supersymmetry) or new physics, which for example leads to string theory. The first avenue or mindset is also possible and is usually taken to mean that the main problem comes from the perturbative expansion around flat space (or any other space). Therefore, the natural way to solve the problem, in such a perspective, would be to look for a non-perturbative quantization or treatment of quantum gravity.

Some scales naturally appear in quantum gravity on a dimensionnal ground. \graffito{It should be noted that the dimension of $G$ depends on the dimension of spacetime. In particular, for $2$-dimensionnal spacetimes, $G$ is dimensionless and for $3$-dimensionnal spacetimes, $G$ is directly the inverse of a mass even at the classical level ($\hbar$ does not intervene).}Let us fix units where $\hbar = c = 1$ as is now common in high energy physics. In this kind of units, only one unit is necessary and can be taken to be a length or a mass for instance. In particular, $G$ is dimensionful and can be written as:
\begin{equation}
G=\frac{1}{m_P^2}
\end{equation}
where $m_P$ is the Planck mass. It represents the scale at which quantum effects arise in gravity or, conversly, at what scale gravity effects enter quantum theory. In quantum theory, this mass scale also gives a length scale: the Planck length $\ell_P$ which is inversely proportionnal to the Planck mass. We expect quantum gravity effects to be important at that scale. Indeed, there is a very simple physical interpretation of the Planck mass and of the Planck length. For a particle of a given mass $m$, the Scharzschild radius $r_S$ of the black hole it would create runs linearly with the mass $m$ following:
\begin{equation}
  r_S = 2Gm.
\end{equation}
But its Compton length $l_C$ caracterizing its dispersion due to quantum effects increases with the inverse of the mass as:
\begin{equation}
  l_C = \frac{1}{m}.
\end{equation}
For low energies, the Compton length is much larger than the Schwarzschild radius and therefore gravity effects can be neglected compared to quantum mechanical ones. At some point however, the lengths cross:
\begin{equation}
  r_S = l_C \Leftrightarrow 2Gm = \frac{1}{m} \Leftrightarrow m = \frac{1}{\sqrt{2G}}
\end{equation}
\graffito{Discussion of factors, like $2$ or $8\pi$ have been completely omitted in this dimensionnal analysis.}and this corresponds to the point where gravity and quantum mechanics effects have the same order of magnitude. The tipping mass is the Planck mass. The Planck length is then the radius of the corresponding black hole. Now, as a scale naturally enters quantum gravity, we do not expect spacetime to simply be Minkowskian. There should be quantum fluctuation down at the Planck scale which might solve the problem of perturbative expansion.\graffito{It is quite interesting to see how non-renormalizable theories tend to predict the scale at which they are broken, as was the case with the Fermi interaction.} It can be noted that this would presumably also appear in other theories of quantum gravity even if we consider general relativity to be some low energy approximation. The new physics solving the pertubative expansion issue would appear at that scale. The new physics is then expected to solve some of the singularity issues in general relativity which are awaited in black holes \cite{Modesto2004,G.:2015sda} and at the Big Bang for instance \cite{Ashtekar2009}.

An interesting argument, which can be found in Thiemann's book \cite{Thiemann:2007zz} and was originally introduced in \cite{1988npcg.book.....A}, seems to support this idea that non-pertubative effect \textit{can} solve the problem of renormalization even at the classical level. Let us consider the self-energy of a point particle of mass $m_0$ and electric charge $q$. \graffito{Including the energy from the fields is natural from an ``inertial'' as well as ``gravitionnal'' mass point of view. For the gravitationnal mass, energy gravitates not mass. For the inertial mass, we should consider that the fields themselves will have more energy if the particle moves.}Keeping the units where $c=1$ (and $\hbar = 1$, but this won't appear in the calculation), the self-energy can be written as a physical mass $m$ rather than an energy. This physical mass corresponds to the total mass of the particle plus the mass induced by all the energy coming from the fields it generates. All this contributes to the rest energy of the system which is the definition of mass. It can be written:
\begin{equation}
m = m_0 + \frac{3 q^2}{5 \epsilon} - \frac{3 G m_0^2}{5 \epsilon}
\end{equation}
where the vacuum permittivity $\epsilon_0$ was put to $1$ (canonical units) and $\epsilon$ acts as a cut-off for the computation as the integral defining the energy diverges near the particle. The cut-off is necessary to make the computation finite as the fields diverge. This amounts to consider that the particle is a small ball of uniform mass and radius $\epsilon$. Still, the limit $\epsilon \rightarrow 0$ should be taken at some point eventhough it is not possible here. The idea of renormalization is to use the bare mass $m$ as the good parameter. That means that instead of parametrizing the theory by bare parameters (which might be remote or unaccessible anyway, not to mention ill-defined), the theory will be expressed in term of measurable quantities, in this case the full measurable mass. This idea works if the limit can be taken after having reexpressed in such a way. For this, the bare masses will be made to include counter-terms that are infinite or, at finite cut-off, which includes terms of the form $\frac{1}{\epsilon}$. But because of the square of $m_0$ which appears in the coupling with gravity, no polynomial expression of $m_0$ in terms of $\frac{1}{\epsilon}$ can ever alleviate all the powers of $\epsilon$. Therefore, perturbatively, renormalization fails.

Now, general relativity comes into play and tells us that gravity actually couples to the physical mass and not the bare mass (as the equivalence principle implies). Therefore, we should have written:
\begin{equation}
m = m_0 + \frac{3 q^2}{5 \epsilon} - \frac{3 G m^2}{5 \epsilon}
\end{equation}
This equation can be solved exactly as:
\begin{equation}
m = \frac{5\epsilon}{6G}\left(\sqrt{1 + \frac{12 G}{5\epsilon}\left(m_0 + \frac{q^2}{\epsilon}\right)} -1\right)
\end{equation}
Now, the physical mass exists for $\epsilon = 0$ (the limit can be taken), without reabsorbing any infinities. It reads:\graffito{The physical mass we find do not even depend on the bare mass $m_0$. This is suspiscious of course but we should remember that this is a hand-waving argument, mixing arguments from classical physics, newtonian gravity and relativity, which should not be trusted quantitatively.}
\begin{equation}
m = \sqrt{\frac{5q^2}{3G}}
\end{equation}
which is finite. This surprising result is obviously non-analytic, explaining the failure of the perturbative process, and making the self-energy finite even in presence of an electric field. Therefore, we expect gravity phenomena could cure the divergences of field theory, even of other coupled fields as the electromagnetic field. This should be done at a non-pertuabtive level. Moreover, the central argument here comes from diffeomorphism invariance even though it is quite hidden at first sight. Indeed, diffeomorphism invariance is the symmetry that allows the principle of equivalence and, as a consequence, the coupling to the physical mass, rather than the bare mass. Some interplay is then to be expected between diffeomorphism invariance and non-perturbative definition of the theory.

There are several routes to non-pertubative quantum gravity. A possible, and very interesting possibility, is the fixed-point approach to the renormalization group or \ac{AS} scenario \cite{Niedermaier2006}. In this thesis however, we will rather consider the loop quantum gravity approach. \graffito{The fact that asymptotic safety and loop quantum gravity are somewhat linked can be argued most persuably by the presence of \ac{AS} gravity researcher at nearly every \ac{LQG} conference.}Though it could arguably be linked to the first project \cite{Perez:2005fn}, this programme has its own history. This approach makes an attempt at quantifying precisely general relativity from the Hamiltonian perspective and by so doing has already developed a huge framework that appears promising. In the next section, we will briefly review the historical development of the theory, leaving for now the technical sides as we will get back to those in the body of the thesis.

%\textbf{TODO:} present both theories quickly (gravity, quantum), success of standard model ? quantization of gr (seminal work), problems with renormalization, dimensionnal argument, appraoches : effective theory or non perturbative problem, scale, idea of gravity induced regularization, mass small argument

%*****************************************

\section{Non-pertubative quantization programme}

\ac{LQG} already has some history and its main lines of development can be found in \cite{Rovelli2001}. We focus here on the points relevant to the rest of the thesis. We can argue that the \ac{LQG} programme relies on four major points:
\begin{itemize}
\item the first development of quantization of general relativity and its failure,
\item the discovery of the new variables,
\item the loop quantization,
\item the development of the dynamics.
\end{itemize}
In Loop Quantum Gravity, we concentrate on the Hamiltonian perspective\graffito{From the path integral approach, a lot of infinities tend to appear and are discared as irrelevant factors. Though this can be justified, it is easier to treat the problem from an Hamiltonian perspective.}. The canonical quantization programme was launched by Dirac and Bergmann already in late $40$s \cite{PhysRev.75.680,RevModPhys.21.480,Bergmann1956,10.2307/100496} and matured in the early $50$s. This allowed the canonical treatment of general relativity (still at the classical level) by the end of the $50$s \cite{RevModPhys.29.497,Dirac1958}. What we now call the ADM variables, with their nice geometrical interpretation, was developed in $1961$ by Arnowit, Deser and Misner and also led to the first clue that nonperturbative quantum gravity might be finite \cite{Arnowitt1960,Arnowitt2004}. By the end of the $60$s, the Wheller-De Witt equation was introduced \cite{PhysRev.160.1113,PhysRev.162.1195,PhysRev.162.1239} and canonical \ac{GR} was defined in a formal sense. And this opened the road for the first failures in the early $70$s when it became clear that perturbative quantization would not work at least naively \cite{Deser:1974cy,PhysRevD.10.401}.

Loop Quantum Gravity is founded on two main principles: the idea that it is the perturbative treatment that is problematic, and the idea of keeping the canonical perspective. The problem was that the Wheeler-DeWitt equation was far too hard to define properly, let alone to solve. A major advance came with the discovery of Ashtekar connection \cite{PhysRevLett.57.2244} and then Ashtekar-Barbero variables \cite{BarberoG.1995,Immirzi:1996di} allowing the theory to be expressed in terms of polynomial constraints (except for the Hamiltonian constraint when using non-self-dual variables). Though already interesting at the classical level, this is of course a great simplification for quantization since it removes a lot of ambiguities which plagued the previous framework. Loop Quantum Gravity was born, arguably with the seminal work of Jacobson and Smolin, finding loop solutions to the Wheeler-DeWitt equations with the new variables \cite{JACOBSON1988295}. This finally led to the loop representation \cite{PhysRevLett.61.1155}, and finally polymer quantization \cite{Ashtekar1995,Ashtekar1996,Ashtekar1997,Ashtekar:1998ak} giving a rigorous basis for the quantum kinematical phase space of general relativity. Uniqueness theorems were even proved making the diffeomorphism assumptions essential in the construction \cite{Lewandowski:2005jk}.

This led to the present day problems of fixing the dynamics. Thiemann's original proposal \cite{Thiemann1996a}, though not satisfying \cite{Lewandowski:1997ba,Gambini:1997bc,Smolin:1996fz}, is a milestone. Since then, the development in the canonical framework, for the full-theory, is rather shy. We have the master constraint programme \cite{Thiemann:2003zv} and some variations \cite{Alesci:2015wla}. However, the work on toy-models or simplified versions of the full-theory has shed great light on the subject and allowed a better comprehension of the hurdles we have to take care of \cite{Kuchar:1989wz,Laddha:2008em,Laddha:2010hp,Laddha:2010wp,Tomlin:2012qz,Henderson:2012ie,Henderson:2012cu} (see also \cite{Bonzom:2011jv} for a review). But most of the developments in the dynamics have been on the covariant sides, where we consider evolution of the geometry states along spacetime and consider amplitudes associated to them. The first notable model is the Barrett-Crane model \cite{Barrett:1997gw,Barrett:1999qw}, which had a flatness issue \cite{Kaminski:2013yca}. This was solved with the most recent EPRL \cite{Engle:2007wy} and FK \cite{Freidel:2007py} vertex which included an Immirzi parameter. Since then, other variations have been proposed (for instance \cite{Dupuis:2011fz, Engle:2011un, Engle:2015zqa}) and the renormalization of these vertices have been investigated (though not solved) as in \cite{Banburski:2014cwa, Banburski:2015kmc}. It should be noted that the work is way more promising and accomplished on the (euclidean) 3d quantum gravity side, first investigated with the so-called combinatorial quantization \cite{Witten1} and with models in the canonical and covariant approaches \cite{Noui2005,Bonzom:2011hm,Bonzom:2014bua,KauffmanBracketLQG} including or not a cosmological constant \cite{Mizoguchi:1991hk,Bonzom:2014wva,Dupuis:2013haa,pranzetti}. This covariant approach is developed in two different (though close) directions: spinfoams \cite{Perez:2012wv,Livine:2010zx,Nicolai2007,Bianchi:2012nk,Reisenberger:1996pu,Engle2007,Engle:2007wy,Baez1998} and \ac{GFT} \cite{Oriti:2014uga,Carrozza:2014rya}.

\medskip

\begin{figure}[h!]
  \centering
  \begin{minipage}[b]{.45\linewidth}
    \centering
    \begin{tikzpicture}[scale=1]
      \coordinate(A) at (0,0);
      \coordinate(B) at (2,0);
      \coordinate(C) at (2,-2);
      \coordinate(D) at (0,-2);

      \draw (A) to[bend left] node[midway,sloped]{$>$} node[midway,above]{$\ell_1$} (B);
      \draw (B) to[bend left] node[midway,sloped]{$>$} node[midway,right]{$\ell_2$} (C);
      \draw (C) to[bend left] node[midway,sloped]{$>$} node[midway,above]{$\ell_3$} (D);
      \draw (D) to[bend right] node[midway,sloped]{$>$} node[midway,left]{$\ell_4$} (A);

      \draw (A) to[bend left] node[midway,sloped]{$>$} ++(-1,0.5) node[above]{$j_1$};
      \draw (B) to[bend right] node[midway,sloped]{$>$} ++(1,0.5) node[above]{$j_2$};
      \draw (C) to[bend right] node[midway,sloped]{$>$} ++(1,-0.5) node[right]{$j_3$};
      \draw (D) to[bend left] node[midway,sloped]{$>$} ++(-1,-0.5) node[left]{$j_4$};

      \draw (A) node {$\bullet$} node[above]{$i_1$};
      \draw (B) node {$\bullet$} node[above]{$i_2$};
      \draw (C) node {$\bullet$} node[right]{$i_3$};
      \draw (D) node {$\bullet$} node[left]{$i_4$};
    \end{tikzpicture}
    \subcaption{Spin networks are graphs coloured with spins at the edges and intertwiners at the vertices.}
    \label{fig:spinnetwork_a}
  \end{minipage} \qquad
  \begin{minipage}[b]{.45\linewidth}
    \centering
    \begin{tikzpicture}[scale=1]
      \coordinate (O) at (0,0,0);

      \coordinate (A) at (0,1.061,0);
      \coordinate (B) at (0,-0.354,1);
      \coordinate (C) at (-0.866,-0.354,-0.5);
      \coordinate (D) at (0.866,-0.354,-0.5);

      \draw[blue] (A) -- (B);
      \draw[blue] (A) -- (C);
      \draw[blue] (A) -- (D);
      \draw[blue] (B) -- (C);
      \draw[dashed,blue] (C) -- (D);
      \draw[blue] (D) -- (B);

      \draw[dotted] (O) -- ++(0,-0.531,0);
      \draw (0,-0.531,0) -- ++(0,-0.531,0);

      \draw[dotted] (O) -- ++(0,0.177,-0.5);
      \draw[dashed] (0,0.177,-0.5) -- ++(0,0.177,-0.5);

      \draw[dotted] (O) -- ++(0.433,0.177,0.25);
      \draw (0.433,0.177,0.25) -- ++(0.433,0.177,0.25);

      \draw[dotted] (O) -- ++(-0.433,0.177,0.25);
      \draw (-0.433,0.177,0.25) -- ++(-0.433,0.177,0.25);

      \draw (O) node{$\bullet$};
      \draw[blue] (A) node{$\bullet$};
      \draw[blue] (B) node{$\bullet$};
      \draw[blue] (C) node{$\bullet$};
      \draw[blue] (D) node{$\bullet$};
    \end{tikzpicture}
    \subcaption{Each node of a spin network can naturally be interpretated as a polyhedron (in blue on the figure), each side corresponding to an edge of the graph.}
    \label{fig:spinnetwork_b}
  \end{minipage} \\
  \begin{minipage}[b]{.45\linewidth}
    \centering
    \begin{tikzpicture}[scale=1]
      \coordinate (A1) at (0,0,-2);
      \coordinate (B1) at (0,1,0);
      \coordinate (C1) at (-0.866,-0.5,0);
      \coordinate (D1) at (0.866,-0.5,0);

      \fill[blue!50] (B1) -- (C1) -- (D1) -- cycle;
      \fill[blue!50] (B1) -- (D1) -- (A1) -- cycle;

      \draw[blue] (A1) -- (B1);
      \draw[blue,dashed] (A1) -- (C1);
      \draw[blue] (A1) -- (D1);
      \draw[blue] (B1) -- (C1);
      \draw[blue] (C1) -- (D1);
      \draw[blue] (D1) -- (B1);

      \draw[blue] (A1) node{$\bullet$};
      \draw[blue] (B1) node{$\bullet$};
      \draw[blue] (C1) node{$\bullet$};
      \draw[blue] (D1) node{$\bullet$};

      \coordinate (A2) at (0,0,0.5);
      \coordinate (B2) at (0,-1,0);
      \coordinate (C2) at (-0.866,0.5,0);
      \coordinate (D2) at (0.866,0.5,0);

      \fill[red!50] (A2) -- (B2) -- (D2) -- cycle;
      \fill[red!50] (A2) -- (C2) -- (D2) -- cycle;
      \fill[red!50] (A2) -- (B2) -- (C2) -- cycle;

      \draw[red] (A2) -- (B2);
      \draw[red] (A2) -- (C2);
      \draw[red] (A2) -- (D2);
      \draw[red] (B2) -- (C2);
      \draw[red] (C2) -- (D2);
      \draw[red] (D2) -- (B2);

      \draw[red] (A2) node{$\bullet$};
      \draw[red] (B2) node{$\bullet$};
      \draw[red] (C2) node{$\bullet$};
      \draw[red] (D2) node{$\bullet$};
    \end{tikzpicture}
    \subcaption{Spin network geometry is not Regge geometry since the faces of the glued polyhedra do not have to match.}
    \label{fig:spinnetwork_c}
  \end{minipage}
  \caption{The geometrical interpretation of spin networks}
  \label{fig:spinnetwork}
\end{figure}
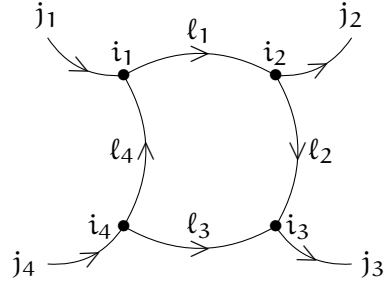
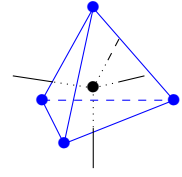
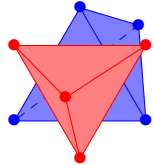
Let us now sum up the modern understanding of the theory in a slightly more technical manner. \ac{LQG} relies on the reformulation of \ac{GR} as a gauge field theory using Ashtekar-Barbero variables and provides a canonical quantization of it. The kinematics of the theory is well-understood at a rigorous mathematical level \cite{Lewandowski:2005jk,Ashtekar1995}: the states of the theory are wavefunctions on the space of generalized connections and are defined using cylindrical functions over the connection (that is functions depending only on a finite number of holonomies). A basis of this state space is labeled by spin networks (see for a review \cite{Baez1995}): they are graphs with extra spins and intertwiners  as colouring on the edges and the vertices respectively \cite{Rovelli1995,Rovelli:1995ac} (see fig.\ref{fig:spinnetwork_a}). They were already introduced by Penrose for his twistor approach \cite{Penrose1971,Penrose1973,Penrose1984}. These states possess a natural geometrical interpretation as quantum discrete geometries. Each node of the graph represents a polyhedron and the links of the graph indicate how to glue them together \cite{Rovelli2010}. The extra spins are related to the area of the faces of the polyhedra and the intertwiners of $\mathrm{SU}(2)$ at the nodes specify the remaining degrees of freedom for the polyhedra \cite{Bianchi:2010gc,Freidel:2009nu,Barbieri1998} (see fig.\ref{fig:spinnetwork_b}). Therefore, they generalize usual Regge geometries since the shape of the faces of the polyhedra are not required to match when glued together (see fig.\ref{fig:spinnetwork_c}). They can be better understood in the twisted geometry framework \cite{Freidel:2010aq} or in the more recent spinning geometries \cite{Freidel2014} which are both generalizations of Regge geometries that take into account metric discontinuity and torsion respectively \cite{Haggard:2012pm,Freidel:2011ue}. Curvature is allowed on spin networks but is supported by loops in the graph: the graph links carry elementary excitations of the connection and thus correspond to non-trivial parallel transport accross the faces of the polyhedra. Curvature is encoded in holonomies around loops which correspond to the edges of the polyhedra when the dual triangulation exists (that is when the matching conditions of Regge geometries are verified). Then, the usual continuous curved spacetime is to be recovered in a large-scale limit. A Hilbert space can be defined for any single graph, but the kinematical space of the full theory contains states in superpositions from different graphs. As any states can be enlarged to a bigger graph by allowing trivial dependancy of the wavefunction on the extra edges, it is possible to identify two wavefunctions on two graphs if they only depends on the shared parts of the graphs and if they do coincide on them. This consistency condition (called the cylindrical consistency condition) implements the fact that the states should not depend on the particular graph, finer or coarser, as long as the graph is large enough to encode the relevant degrees of freedom. This defines equivalence classes of wavefunctions and allows for the writing of a theory on varying graphs. It can also act as a guideline for coarse-graining as was explored by Dittrich \cite{Dittrich:2012jq}.

\vspace{1em}

The physics of \ac{LQG}, however, is contained in its dynamics. Several proposal have been carried out, in the canonical framework and in the covariant setup. One of the other major challenges which is of concern for this paper, is to show that the theory, at least in one of these approaches, does indeed reproduce \ac{GR} in the continuum semi-classical limit. The problem of the dynamics is tentatively difficult notably because the dynamics have two main aspects: the dynamics at \textit{fixed} graph and the dynamics \textit{changing} the graph. This makes direct calculations or even numerical simulations difficult in practice because the two aspects come simultaneously. The usual stategy for discrete systems on fixed graphs, as in condensed matter theory and statistical physics, is to coarse-grain the theory, that is, to integrate the microscopic degrees of freedom inside bounded regions, thus assimilated to points, and to write effective theories for the relevant macroscopic degrees of freedom. This process of coarse-graining naturally makes possible the study of a continuum limit as we will argue below. In the \ac{LQG} context, though the theory works with varying graphs, we could hope to extend this process and further use it to map the varying graphs dynamics onto a fixed graph dynamics (following developments along the KPZ conjecture \cite{Knizhnik:1988ak}). The rationale behind this is the following: starting from a base graph, each node will correspond to a varying coarse-grained region. This means that the internal degrees of freedom of the vertices in the effective theory should reproduce a varying structure. This method should mimic a development around this base graph considered as a skeleton graph for the excitations. The problem may be treated in other ways and other viable physical and geometrical pictures will be described below. Still, this image can be used as a guide in the search of a good coarse-graining scheme.

Though the physics is not fully elucidated yet, the work is indeed promising. Models exists and some first checks have been possible. The \ac{LQC} field \cite{Bojowald2008,Ashtekar2009}, which we will study further, has interesting developments and simplified settings are increasingly studied. Let us mention here the results \cite{Rovelli2006}, which are a good indication that the theory we are studying is indeed gravity. The link between spinfoams, \ac{GFT} and canonical quantization is also getting clearer \cite{Oriti2001,Oriti:2013aqa,Livine:2002ak,Dupuis:2010jn} and some computations can be made \cite{Alesci2013,Alesci2007}. As a final note, a nice review of the current status of the programme can be found in \cite{Rovelli2011}.

%*****************************************

\section{Coarse-graining}

As argued in the previous section, coarse-graining enters the study of loop quantum gravity in at least one respect: the study of the continuum limit. There is a second aspect that is also more or less automatically taken care of: the study of large scales. Indeed, the study of the renormalization and coarse-graining of a theory usually allows for using large approximate discetizations still given accurate result compared to a naive discretization of the Hamiltonian. It should be noted here, that coarse-graining might not be the only way to study large scales since symmetry reduction can also be used (as is done in cosmology). And some approaches have been devised along this line \cite{Alesci:2013xd}. However, with the current state of the dynamics, coarse-graining is mostly needed for two reasons: either for defining a continuum limit (as we are still lacking a good Hamiltonian for instance) or for renormalizing a given dynamics (as in group field theory of spin foam approaches).

How is the coarse-graining to help in the definition of a continuum limit? As we saw, we still lack a correct Hamiltonian, which would satisfy the  (four dimensionnal) diffeomorphism algebra and which reduces to general relativity in some limit. It is easier to suggest a discretization of such an Hamiltonian (as in \cite{Bonzom:2013tna}), especially since the Hilbert space of states carries a natural discrete structure. Then, the process of coarse-graining can begin its work: at each step of coarse-graining, the dynamics captures finer and finer details. In particular, because the scale of the study gets large compared to the fundamental blocks, the dynamics approaches a continuum hamiltonian. At the fixed point, the hamiltonian should be perfectly continuous, though expressed on a truncation. The take-home idea is that the correct algebra should be regained for fixed-points of the coarse-graining flow \cite{Dittrich:2012jq}. This means that if a correct coarse-graining process is devised, we could, rather than studying a given dynamics, search for fixed points which would encode the continuum directly. There are of course subtilities in the interpretation, as the discretization do not rely on some fixed lattice, but the lattice itself is, in some sense, dynamical. Still, coarse-graining can be understood as an approximation scheme, where the number of blocks is not related to some scale but to the number of independant degrees of freedom we consider. This idea is well-supported on simple models like reparametrization invariant free particles \cite{Rovelli:2011fk}.

\medskip

Still, there are obvious problems with such an approach. The first one comes with the representation of the fixed point: the couplings might be highly non-local (as is the case when renormalizing the Ising model for instance). This can be seen in simple cases. Let us consider for instance a simple free (scalar) field theory in $1+1d$ spacetime, so that there is only one direction of propagation. In that case, it is totally possible to consider a discrete theory with a field with values on discrete points (the space will be isomorphic to $\mathbb{Z}$) with some spacing $a$ between them (see figure \ref{fig:discrete_scalar}).

\begin{figure}[h!]
  \centering

  \begin{tikzpicture}[scale=0.8]
    \coordinate (BL) at (-1,0);
    \coordinate (A1) at (0,0);
    \coordinate (A2) at (2,0);
    \coordinate (A3) at (4,0);
    \coordinate (A4) at (6,0);
    \coordinate (A5) at (8,0);
    \coordinate (BR) at (9,0);

    \draw (A1) node{$\bullet$} node[above]{$\phi_{-2}$};
    \draw (A2) node{$\bullet$} node[above]{$\phi_{-1}$};
    \draw (A3) node{$\bullet$} node[above]{$\phi_{0}$};
    \draw (A4) node{$\bullet$} node[above]{$\phi_{1}$};
    \draw (A5) node{$\bullet$} node[above]{$\phi_{2}$};

    \draw[<->,>=stealth,thick] (A3) ++(0,-0.3) -- node[midway,below]{$a$} ++(2,0);

    \draw (BL) node{$\cdots$};
    \draw (BR) node{$\cdots$};

    \draw[->,>=stealth,thick] (-2,-1) -- (10,-1) node[right]{$x$};
    \draw (0,-1) node{$|$};
    \draw (2,-1) node{$|$};
    \draw (4,-1) node{$|$} (4,-1.5) node{$0$};
    \draw (6,-1) node{$|$};
    \draw (8,-1) node{$|$};
  \end{tikzpicture}
  
  \caption{We consider a scalar field on a discrete line, which is represented by a function $\phi: \mathbb{Z} \rightarrow \mathbb{C}$. We set the lattice spacing to $a$. This field can be thought of as a truncation of a continuous field $\phi(x)$.}
  \label{fig:discrete_scalar}
\end{figure}
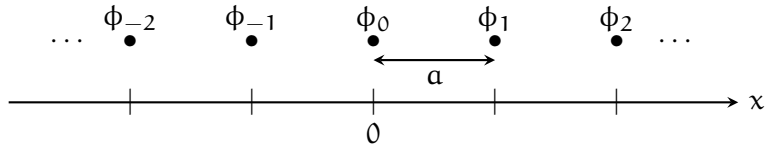

This discrete theory can encode naturally the full continuum theory if we consider a truncation thereof. Indeed, the discretization does not allow to decipher between waves of momentum $k$ and waves of momentum $k + \frac{2\pi}{a}$. But if we restrict to waves of momentum between $-\frac{\pi}{a}$ and $\frac{\pi}{a}$, the ambiguity is solved and we can adjust the dynamics so that the continuum theory is reproduced. But the natural variables are now in the momentum space and involves highly non-local evaluation to compute the hamiltonian. Indeed, the evalutation involves a Fourier transform which takes into account the whole line. This can be seen quite simply, because our cut-off is a window on momenta which can be applied quite simply in Fourier space by multiplication. Back to real space, this will involve convolution with a $\mathop{sinc}$ function which decreases quite slowly at infinity.

This underlines a few problems we have to deal with: how to choose a good cut-off and how to choose the right variables. Indeed, we cannot expect, as for a free field theory, to find an exact cut-off. In the absence of a metric, it will also be quite difficult to find some meaningful notion of Fourier transform or something equivalent to solve the problem of non-locality. We might also not be able to find the fixed-points exactly. But an approximation scheme could be devised where the correct dynamics is approximated as some non-exact solution to the fixed-point equation. For all this to be possible, good-choices of variables and cut-off are needed and, in the case of general relativity, should presumably be motivated by geometrical aspects. An interesting programme was started in this regard by Dittrich \textit{et al.} \cite{Dittrich:2012jq,Dittrich:2014wpa,Bahr:2015bra,Dittrich:2013xwa}. The idea is to find what is called a physical vacuum. This vacuum would be the solution of the Hamiltonian constraint with minimal excitation (for example by choosing some homogeneous space). Then the (discrete) Hamiltonian could be written for a finite number of excitations over this vacuum. The coarse-graining should be especially nice in this language as the vacuum and the way the excitations are implemented and selected by the dynamics. Our approach is rather different in spirit but should made contact somewhere. Indeed, if we manage to define a good choice of variables, presummably, the large blocks of the discretization would be in a state similar to this physical vacuum.

Still, this idea of searching for fixed points leads us to a research programme, which involves various steps:
\begin{itemize}
\item First, we need to identify good variables for the description of the macroscopic world. This is done at the kinematical level (of course, since we do not have a dynamics yet). But it can be motivated on dynamical grounds. A good choice should lead us to an exact discrete theory, or at least to controllable approximations.
\item Then, we need to relate these variables to the actual phase space of loop quantum gravity as expressed for the full theory, since the previous variables would presummably be found in classical \ac{GR}. Ideally, we would find that the complete phase space can be described with these variables (maybe with additionnal constraints) meaning that they do correspond to a simple change of variables. Then, the coarse-graining would still occur but the variables would be particularly well-suited for truncations leading to a nice coarse-graining scheme relating the variables at different truncations.
\item Finally, we could write down the flow equations for the Hamiltonian and try to identify fixed points.
\end{itemize}

In this thesis, we will concentrate on the first two points. For the first point, we will follow the (original) work done in \cite{Charles:2015lva,Charles:2016xzi}. The idea is to identify natural coarse-grained variables, in the context of hyperbolically curved spaces, thanks to some algebraic conditions related to the symmetries. For the second point, we will start the programme by considering coarse-graining by gauge-fixing, as we have done in \cite{Charles:2016xwc}. The goal is to find a natural way to cast numereous degrees of freedom at some scale into excitations of the larger scale. The third point is left for further investigation though we will suggest some possible routes based on the work done in \cite{Charles:2015rda}. We will also concentrate on the canonical approach, though other roads are possible most notably in the \ac{GFT} formalism \cite{Carrozza:2013mna}.

%It should be noted here that some other possibilities have been developed in the context of coarse-graining and should be mentionned. \textbf{TODO} (should not redo end of last section)

%\textbf{TODO:} probelm of defining hamiltonian, enters coarse-graining, anyway: useful for physical predictions (even if we had an hamiltonian, we will need coarse-graining), so work part of a programme for solving both problems at once: continuum limit and physical predictions, part on difficulties : scale independance, what is coarse-graining (view as approximation scheme), two roads: finding the right variables and building from the theory, explanation of the full programme, explain what has been done

%*****************************************

\newpage
\section{Outline of the thesis}

This thesis is organized as follows. It contains four different parts. The first two are the state of the art and the necessary grounds for our own developments. In the first part, we concentrate on the kinematical aspects of loop quantum gravity. In the first chapter, we first review the usual classical hamiltonian formulation of general relativity using the ADM variables. The quantification problems are quickly reviewed to motivate the study of the new variables. The new variables, namely the Ashtekar-Barbero variables are introduced in the second chapter. We discuss their constructions and the role of the various Immirzi parameters. We also discuss the possible physical relevance of using such variables. In the third and last chapter of the first part, we finally consider loop quantization (also called polymer quantization). We explain the procedure and present the underlying programme. We explain briefly how the various kinematical constraints, that is the gauge constraint and the (spatial) diffeomorphism constraints, are implemented. We finish this chapter with a discussion on the geometrical operators and the relevance of their spectrum to the full theory. In the second part, we concentrate on the dynamics of the theory at the Hamiltonian level and the various problems that appear in its study. The first chapter of this second part is a short review of the Hamiltonian development of the dynamics, what are the most common approaches and the recent results in that direction. In the second chapter, we review a more successful theory with regard to the dynamics: loop quantum cosmology. This is a quick review as it will serve as a guideline for our own work. And in the third chapter, we introduce the coarse-graining technique (coarse-graining by gauge fixing) that we used in this work.

\medskip

The last two parts constitute the original work of this thesis. In the first of these, we concentrate on the first point of coarse-graining mentionned in this introduction: the search of good coarse-grained variables. This part then exposes the work done in the following papers \cite{Charles:2015lva,Charles:2016xzi}. The first chapter introduces the main idea of taking the closure condition as a help for building good variables and a new interpretation of it that is compatible with coarse-graining. The second chapter introduces variables in the hyperbolic setting that make use of the previous construction and, therefore, should enter a coarse-grained description. A third and last chapter will comment on the role of the Immirzi parameter or Immirzi-like parameter as they were introduced in the previous chapters. We will also suggest some possible (and speculative) avenues for coarse-graining based on the Immirzi parameter. This is based on the work done in \cite{Charles:2015rda}. In the last part, we consider the other problem of coarse-graining: starting from the usual phase space and applying the coarse-graining by gauge-fixing directly. This presents the work done in the paper \cite{Charles:2016xwc}. This is done in three chapters: the first one introduces the concepts and the relevant spaces and operators. The second one discusses various cut-offs and their relevance to coarse-graining. Finally, the third and last chapter considers a real life example by taking the example of BF theory and writing it in this language.

%\textbf{TODO:} obvious

%*****************************************
%*****************************************
%*****************************************
%*****************************************
%*****************************************

\cleardoublepage
\ctparttext{}%You can put some informational part preamble text here. 
%Illo principalmente su nos. Non message \emph{occidental} angloromanic
%da. Debitas effortio simplificate sia se, auxiliar summarios da que,
%se avantiate publicationes via. Pan in terra summarios, capital
%interlingua se que. Al via multo esser specimen, campo responder que
%da. Le usate medical addresses pro, europa origine sanctificate nos se.}
\part{Loop quantum gravity kinematics}
%*****************************************
\chapter{Hamiltonian formulation of classical general relativity} \label{ch:HamiltonianGR}
%*****************************************

\inspiquote{[...] Alonso!}{The Doctor}

The final goal is therefore to quantize general relativity. The notion of quantizing a theory might be a bit undefined in general\graffito{It is true that quantization is a bit weird in principle, as nature is quantum mechanical and not classical. Therefore, it might seem more natural to consider the reverse process of classicalizing. But we are doing research and therefore, this is an epistemological order not a logical one.} but can be summed up as follows: finding a quantum theory, with a Hilbert space, an algebra of operators, constraints and hamiltonian, so that its classical limit, usually devised through coherent states method, is a given classical theory, in our case general relativity. Such a theory might not be unique, and actually we don't expect it to be unique at all, given our experience in quantum mechanics. So it might even be good to find all the possible theories giving the same classical limit. In the case of GR though, finding just one to start with would already be a success.

The process might look formidable at first as there is a huge space of possible quantum theories and studying them all is a lost cause. That been said, from the point of view of the quantum theory, the structures of the classical theory do not come from nowhere. In particular, the Poisson bracket should correspond to a $\hbar \rightarrow 0$ limit of the quantum commutator. So a particularly productive method would be to start with expressions and data from the classical formulation, specifically the Hamiltonian formulation which is closer to the quantum language, and find equivalent in the quantum realm. It turns out a quite effective process is to start with the a choice of variables in the hamiltonian formalism, and then substitute each expression by an operator equivalent. In particular, the brackets are replaced by commutators, so that the first orders in  $\hbar$ match. Special care must be taken when reproducing the algebra of symmetry groups (either gauge or global) in order to reproduce them non-anomalously.

This process of \textit{canonical quantization} is the one we will briefly sketch in this first part as it was done in the early days of canonical quantum gravity \cite{Rovelli2001}. In this first chapter more specifically, we will concentrate on the classical theory and its expression as a hamiltonian theory. Indeed, as we underlined, this formulation is closer to the quantum theory and is not usually developed for general relativity as the covariance is somewhat hidden. Our goal is this chapter is therefore to present the usual hamiltonian formulation, called the ADM formulation \cite{Arnowitt1960} after their discoverers. We will also underline the geometrical interpretation of the different variables and quickly survey how the covariance is preserved at the hamiltonian level. And because this formulation is not the one used in loop quantum gravity, we will sketch its limits with respect to quantization.

\section{Original formulation}

General relativity is a theory describing a metric on a manifold (for lecture books, see \cite{weinberg1972gravitation,wald1984general,misner1973gravitation}). In our world, which seems to have $3$ spatial dimensions and $1$ dimension of time, the manifold is taken to be $4$ dimensional and the metric to have a $(- + + +)$ signature\graffito{Different signatures exist in the literature. Some are more popular in high energy physic but as long as it has a relative sign for time, the rest is pure convention.}, others in pure general relativity . The metric encodes information distance between points on the manifold. This distance should not always be interpreted as a spatial distance but as a relativistic distance encoding either proper time between events if they are separated by a timelike path or minimal distance between events if the separation is spacelike. For two infinitesimally separated points, the square distance $\mathrm{d}s^2$ between them is \cite{Einstein1916}:
\begin{equation}
  \mathrm{d}s^2 = g_{\mu \nu} \mathrm{d}x^\mu \mathrm{d}x^\nu
\end{equation}
where we used the Einstein summation convention\graffito{The Einstein summation convention, of implicit sums over repeated indices (or exponents), will be used throughout this work.}. The Greek letters ($\mu$, $\nu$, ...) denote here coordinates numbering in spacetime. By integrating $\mathrm{d}s$, the length of a path is naturally defined.

The metric also gives a natural notion of angle through the scalar product. In particular, for two vector $v$ and $w$ with the same base point, their scalar product is defined to be $g_{\mu\nu} v^\mu w^\nu$. This notion of angle helps define a canonical notion of parallel transport on the manifold. Indeed, we can concentrate on the parallel transports (given by a connection) that preserve the angle between two vectors being transported simultaneously. This compatibility condition reads:
\begin{equation}
  \forall \mu,\nu,\sigma,\ \nabla_\mu g_{\nu \sigma} = 0
\end{equation}
where $\nabla$ is the covariant derivative associated to the connection as in: $\nabla_\mu v_\nu = \partial_\mu v_\nu - \Gamma^\sigma_{~\nu\mu} v_\sigma$. The affine connection $\Gamma^\mu_{~\nu\sigma}$ is the unique connection satisfying the compatibility condition along with the torsionless condition $\Gamma^\mu_{~[\nu\sigma]} = 0$. This new condition can be seen as a zero-curvature condition on scalar fields (the commutator of the covariant derivative is zero on scalars). As a function of the metric, it is given by:
\begin{equation}
  \Gamma^\mu_{~\nu\sigma} = \frac{1}{2} g^{\mu \rho} \left( \partial_\sigma g_{\nu \rho} + \partial_\nu g_{\rho \sigma} - \partial_\rho g_{\sigma \nu} \right)
\end{equation}
This natural definition of parallel transport gives us a natural definition of curvature. The Riemann tensor, or curvature tensor, is defined as the curvature for the transport of vectors: 
\begin{equation}
  [\nabla_\mu, \nabla_\nu] v_\sigma = R^\tau_{~\sigma \mu \nu} v_\tau
\end{equation}
Here, we use some loose notations for indices but as this is this is not exceedingly confusing, we'll stick to the loose notation as long as there is no ambiguity. Now, in terms of the affine connection, we have:
\begin{equation}
  R^\tau_{~\sigma \mu \nu} = \partial_\nu \Gamma^\tau_{\sigma\mu} - \partial_\mu \Gamma^\tau_{\sigma \tau} + \Gamma^\rho_{\sigma \mu} \Gamma^\tau_{\rho \nu} - \Gamma^\rho_ {\sigma \nu} \Gamma^\tau_{\rho \mu}
\end{equation}
This last quantity transforms as a rank $4$ tensor under diffeomorphism. \graffito{It is sadly quite standard to use the same letter $R$ to denote these $3$ different quantities. This is why we keep indices everywhere we can to avoid confusion.}By contracting it, we obtain a rank $2$ tensor, the Ricci tensor, usually denoted $R_{\mu \nu}$:
\begin{equation}
R_{\mu \nu} = R^\tau R_{\mu \nu \tau}
\end{equation}
and by further contracting it with the inverse metric, we find a scalar quantity, the scalar curvature $R$:
\begin{equation}
R = g^{\mu \nu} R_{\mu \nu}
\end{equation}

General Relativity is then the study of the metrics satisfying the Einstein equation given by:
\begin{equation}
R_{\mu \nu} - \frac{1}{2} g_{\mu \nu} R = \frac{8 \pi G}{c^4} T_{\mu \nu}
\end{equation}
where $G$ is Newton's gravitational constant and $c$ is the speed of light in vacuum. In what follows, we will use the Planck units where $c=1$ and $\hbar=1$ in order to simplify the writings. The constants can be recovered by dimensional analysis. Here, $T_{\mu \nu}$ is the stress-energy tensor of matter. It is equal to $0$ in the vacuum. Let us point out here that what is called the vacuum still contains the gravitational field. It is not like the quantum field theory vacuum which is the lowest excitation state. In the context of general relativity, \textit{vacuum} means \textit{no matter}, and \textit{matter} means anything but the gravitational field. Therefore, the electromagnetic field would be considered matter. That been said, for most of this thesis, we will consider the pure gravity case, that is the vacuum case, or $T_{\mu\nu}=0$.

Interestingly, this equation can be derived from a variational principle, that is, it can be derived from an action, namely the Einstein-Hilbert action \cite{Hilbert1924}. It is:\graffito{Actually, this action principle is not well defined without the Gibbons-Hawking-York boundary term \cite{York:1972sj,Gibbons:1976ue} which makes the action differentiable. It is not a problem of course when there is no boundary.}
\begin{equation}
S_{EH}[g_{\mu\nu}] = - \frac{1}{16\pi G} \int_\mathcal{M} \sqrt{|g|} R d^4 x
\end{equation}
where $g$ is the determinant of $g_{\mu \nu}$ and $EH$ stands for Einstein-Hilbert. Looking for the extrema of this action, we will find the previous equation of motion. That such a writing exists is good news for a quantization programme: this is a pretty good clue that a quantum theory should exists as a variational formalism is usually a hint to a semi-classical expansion of a quantum theory.

As was said in the beginning of this chapter, for the canonical quantization of general relativity, it is more natural to start with the Hamiltonian formalism \cite{1980grg1.conf..227B,RevModPhys.33.510,Komar:1971rg,Bergmann1958} rather than the action of the lagrangian formalism. Of course, both formalism are linked and can be derived from each other by Legendre transform. This is what will be done in the next section. However, there are some subtleties as general relativity is a \textit{totally constrained theory}. It means that, because of the symmetries of the theory, there are only constraints. In particular, the Hamiltonian of the theory is always $0$ and all the dynamics itself is encoded in constraints \cite{Dirac1958}. Indeed, the time parameter must not have any physical consequences as it is arbitrary. Therefore, the time dynamics is trivial. Of course, the usual non-relativistic dynamics must be hidden somewhere, and it is in the correlations of observables. These correlations are enforced by constraints, which also encode the symmetry of the problem. This means, that what we want to call dynamics is actually encoded in the way the time diffeomorphisms act which is itself encoded in the constraint enforced by the lapse. Because, this constraint gives the natural dynamics, it is loosely called the \textit{Hamiltonian constraint}. We will therefore have to deal with constraints at the quantum level. Let us then recap the quantization programme when dealing with constraints:
\begin{itemize}
\item First, develop the classical hamiltonian formulation. In particular, identify the canonical variables and their conjugates as in the ADM formalism.
\item Then, write down all the constraints of the theory. The Hamiltonian should be a linear combination of them \cite{Arnowitt2004}.
\item We will then turn to the quantization. Formally, this amounts to find a representation of the phase space variables where the classical Poisson bracket $\{~,~\}$ is replaced by $-\frac{\mathrm{i}}{\hbar}[~,~]$. Ambiguities can arise here with the choice of variables. This might lead to different quantum theories.
\item Finally, promote the constraints as operators and find their kernel. At this level new ambiguities arise as operator ordering.
\end{itemize}
In the case of ADM variables, the first two steps (as usual in the classical theory) can be entirely conducted. We will see however that the quantization is rather more complicated. This will be our main motivation to switch to the \ac{LQG} formalism.

%*****************************************

\section{ADM variables}

The most natural way to tackle the hamiltonian formulation is to make a $3+1$ splitting of spacetime, that is, to choose a coordinate system such as one (preferred) coordinate corresponds to time and the $3$ others to space. This might seem to break diffeomorphism invariance at first, but as long as we consider only covariant quantities, this won't be a problem (as was underlined in \cite{Thiemann:1996aw}). In other terms, we might choose a coordinate system but as long as the computation result does not depend on it, we're safe. We can regard this problem in yet another way: the phase space can be defined as the space of solutions of the equation. In that way, it is a covariant notion. But, in order to put coordinate on it, we are brought to break the diffeomorphism invariance. This is not a problem though and just reveals our freedom in choosing the coordinates, exactly like we are free to pick a reference frame in usual Galilean relativity. Still, to write it down as a splitting is natural and gives a natural definition for canonical momenta for instance. Moreover, the Hamiltonian theory still preserves a notion of covariance in the Dirac algebra, on which we will come back later on.

However, it is not always true that such a $3+1$ splitting exists. We will assume (still following \cite{Thiemann:1996aw}) that the spacetime manifold $\mathcal{M}$ is diffeomorphic to $\mathbb{R}\times\Sigma$ where $\Sigma$ is a three dimensional surface representing space. In that case, a lot of possible solutions are excluded from the onset. First, we are here assuming that the topology of space is fixed. This might be reasonable in some cases, but it is hard to see why it would be the case in general. Then, some solutions do not have such a nice splitting because the chronology is not well-defined, for instance for rotating universe solutions as Gödel's \cite{RevModPhys.21.447}. Still, these are highly exotic solutions. And as long as spacetime is globally hyperbolic, a coordinate system to our convenience exists. As global hyperbolicity more or less corresponds to our intuitive notion of the existence of causality, we will, for now, sweep all that kind of problems under the rug, hoping that we can solve them later. For a first exploration, this is a very reasonable hypothesis.

Let us define more properly our spacetime coordinates. We have the four dimensional manifold $\mathcal{M}$. We consider a one-parameter family of three dimensional hypersurfaces $\Sigma_t$ \cite{wald1984general}. The variable $t$ is our time\graffito{Remember here that calling $t$ ``time'' is rather arbitrary as we could have chosen another coordinate system.}. We will parametrize the surface $\Sigma_t$ (for a given $t$) by a set of coordinates $(x^a)_{a=1,2,3}$. Here, we have followed the usual convention of using latin indices ($a$, $b$, $c$, ...) to denote spatial coordinates. Spacetime is therefore an infinite collection of three dimensional surfaces and can be parametrized by $(t,x^1,x^2,x^3)$. In order for our foliation to correspond to some notion of space and time, we will ask for the surfaces $\Sigma_t$ to be spacelike, that is for all their tangent vectors to be spacelike or, equivalently, for their normal vectors to be timelike.

Let us now consider a timelike vector field $t^\mu$ satisfying:
\begin{equation}
  t^\mu \nabla_\mu t = 1
\end{equation}
This vector field gives us at each point of spacetime a notion of \textit{going forward in time}. Its direction is fixed toward growing $t${\scriptsize s}, but there is otherwise a large choice of such vector fields. Indeed, any timelike vector field can be rescaled to satisfy such a condition. In particular, we have such a vector for all good choices of time variables on the spacetime manifold. So the liberty of choice of $t^\mu$ reflects the general covariance of the theory. But as we have chosen a foliation of spacetime, selecting a generic field $t$ is a way to restore covariance. By decomposing our field onto coordinates associated to the foliation, we are using a generic frame of reference but we map it in language usable in a Hamiltonian formalism. Therefore, let's decompose $t^\mu$ into its normal and tangential parts on $\Sigma_t$:
\begin{equation}
  t^\mu = N^\mu + Nn^\mu
\end{equation}
where $n^\mu$ is the unit normal vector field to the $\Sigma_t$ given by \graffito{Once again, loose notations are used for the inverse metric, but this is standard. $g^{\mu \nu}$ is the inverse of $g_{\mu \nu}$ and could be more properly written $(g^{-1})^{\mu\nu}$.}$n^\mu = g^{\mu \nu} n_\nu$ and $N^\mu$ is a purely tangential vector, that is it is orthogonal to $n^\mu$ ($N^\mu n_\mu = 0$) (see figure \ref{fig:t_splitting}). Note here that the indices can be raised or lowered by the metric. Upper indices corresponds to contravariant quantities (usually vectors) and lower indices to covariant quantities (like forms). The metric induces a bijection between these by Riesz theorem, and therefore raises and lowers indices and exponents accordingly. The expression for $N^\mu$ and $N$ are:
\begin{equation}
  \left\{
  \begin{array}{rcl}
    N &=& -t^\mu n_\mu \\
    N_\mu &=& q_{\mu \nu} t^\nu
  \end{array}
  \right.
\end{equation}
where $q_{\mu\nu} = g_{\mu\nu} + n_\mu n_\nu$ is the induced (spatial) metric on $\Sigma_t$. $N$ is called the lapse function and $N^\mu$ the shift vector. The spatial metric is spatial in the sense that $q_{\mu \nu} n^\mu = 0$, so that only tangential vectors have non-zero norms. It means in particular that we can write this tensor in a good coordinate system, all the relevant information will be in the spatial components. By good coordinate system, it is understood a coordinate system for which a fixed value of the time parameter corresponds to the spatial surface of interests ($\Sigma_t$). Such a good coordinate system is the system $(t,x^1,x^2, x^3)$ aforementioned. In particular, it means that the normal to the surface corresponds to the increasing (or decreasing) time direction. Expressed in such a system, because the \textit{spatial} coordinates do indeed parametrize the spatial slice, pull-backs onto the space slice will will simply be given by the spatial components of the objects. For instance, the induced metric is simply $q_{ab}$ where $a$ and $b$ denotes spatial components only. Indeed, given a coordinate on $\Sigma_t$ (rather than the whole spacetime), considering the induced metric amounts to considering all the scalar products of all the vectors of the coordinate basis. Of course, only the distances on one surface $\Sigma_t$ for a given $t$ can be reconstructed with this information.

\begin{figure}[h!]
  \centering

  \begin{tikzpicture}
    \def \l {1.5};
    
    \coordinate (A) at (0,0,0);
    \coordinate (B) at (0,0,4);
    \coordinate (C) at (6,0,4);
    \coordinate (D) at (6,0,0);

    \draw (A) to[bend right] coordinate[pos=0.33](C1) coordinate[pos=0.67](C2) (B);
    \draw (B) -- coordinate[pos=0.25](B1) coordinate[pos=0.5](B2) coordinate[pos=0.75](B3) (C);
    \draw (C) to[bend left] coordinate[pos=0.33](D2) coordinate[pos=0.67](D1) (D);
    \draw (D) -- coordinate[pos=0.25](A3) coordinate[pos=0.5](A2) coordinate[pos=0.75](A1) (A);

    \draw (D) node[below right]{$\Sigma_t$};

    \draw[dashed] (A1) to[bend right] (B1);
    \draw[dashed] (A2) to[bend right] coordinate[pos=0.33](P) (B2);
    \draw[dashed] (A3) to[bend right] (B3);
    
    \draw[dashed] (C1) -- (D1);
    \draw[dashed] (C2) -- (D2);

    \draw[->,>=stealth] (P) -- ++(0,\l,0) node[above]{$\vec{t}$};
    \draw[->,>=stealth,dashed] (P) -- ++(0,0.9*\l,0.436*\l) node[left]{$N\vec{n}$};
    \draw[->,>=stealth,dashed] (P) ++(0,0.9*\l,0.436*\l) -- node[midway,above left]{$\vec{N}$} ($(P)+(0,\l,0)$);
    
%    \draw (A1) to[bend right=20] (B1) to [bend right=20] (C1) to [bend right=20] (D1);
  \end{tikzpicture}
  
  \caption{Decomposition of $\vec{t}$ along the normal $\vec{n}$ of $\Sigma_t$ and a tangent vector $\vec{N}$.}
  \label{fig:t_splitting}
\end{figure}
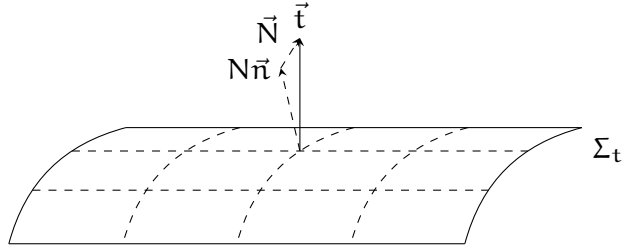

%\textbf{TODO: SCHEMA} ($t$ and splitting)

These quantities, that is the lapse $N$, the shift $N^\mu$ and the induced metric $q_{ab}$ are the natural variables for the ADM (Richard Arnowitt, Stanley Deser and Charles W. Misner) formalism \cite{Arnowitt1960}. Let us interpret these quantities geometrically:
\begin{itemize}
\item The induced metric has a very simple interpretation: it gives the notion of distance on the surface $\Sigma_t$ at fixed $t$. So, for two infinitesimally close points on $\Sigma_t$, the distance will be:
  \begin{equation}
    \mathrm{d}s^2 = g_{\mu\nu} \mathrm{d}x^\mu \mathrm{d}x^\nu = q_{ab} \mathrm{d}x^a \mathrm{d}x^b
  \end{equation}
  since the time component is $0$.
\item The shift $N^\mu$ is the displacement between two instants of the observer corresponding to the frame of reference given by $t^\mu$.
\item Finally, the lapse $N$ represents the proper time between two events of same \textit{spatial} coordinate but on different times. Alternatively, it can be understood as the time dilation factor between the time induced by the vector field $t^\mu$ and the time coordinate $t$.
\end{itemize}
This means we can reexpress the generic distance $\mathrm{d}s^2$ between two infinitesimally close points using only the lapse, the shift and the induced metric.

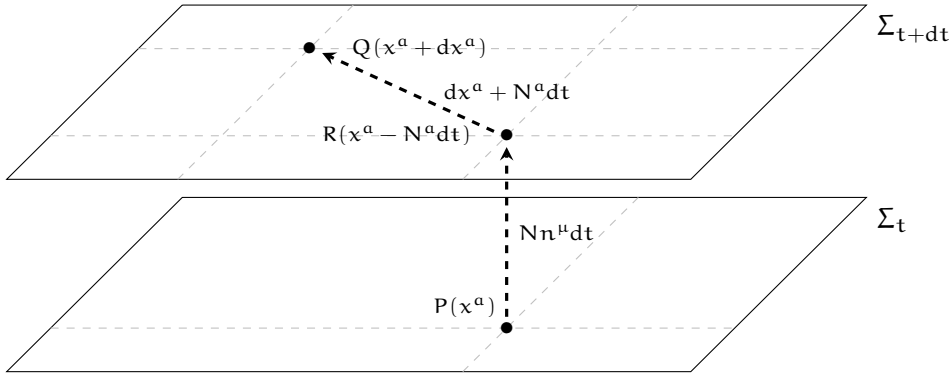
\begin{figure}[h!]
  \centering

  \begin{tikzpicture}[scale=1.5]
    \coordinate (A1) at (0,0,0);
    \coordinate (B1) at (0,0,4);
    \coordinate (C1) at (6,0,4);
    \coordinate (D1) at (6,0,0);

    \def \h{1.7}
    
    \coordinate (A2) at (0,\h,0);
    \coordinate (B2) at (0,\h,4);
    \coordinate (C2) at (6,\h,4);
    \coordinate (D2) at (6,\h,0);

    \draw (A1) -- (B1) -- (C1) -- (D1) node[below right]{$\Sigma_{t}$} -- (A1);
    \draw (A2) -- (B2) -- (C2) -- (D2) node[below right]{$\Sigma_{t+\mathrm{d}t}$} -- (A2);

    \coordinate (P) at (4,0,3);
    \coordinate (P1) at (0,0,3);
    \coordinate (P2) at (6,0,3);
    \coordinate (P3) at (4,0,0);
    \coordinate (P4) at (4,0,4);
    \coordinate (Q) at (1.5,\h,1);
    \coordinate (Q1) at (0,\h,1);
    \coordinate (Q2) at (6,\h,1);
    \coordinate (Q3) at (1.5,\h,0);
    \coordinate (Q4) at (1.5,\h,4);
    \coordinate (R) at (4,\h,3);
    \coordinate (R1) at (0,\h,3);
    \coordinate (R2) at (6,\h,3);
    \coordinate (R3) at (4,\h,0);
    \coordinate (R4) at (4,\h,4);

    \draw[dashed,lightgray] (P1) -- (P2);
    \draw[dashed,lightgray] (P3) -- (P4);

    \draw[dashed,lightgray] (Q1) -- (Q2);
    \draw[dashed,lightgray] (Q3) -- (Q4);

    \draw[dashed,lightgray] (R1) -- (R2);
    \draw[dashed,lightgray] (R3) -- (R4);
    
    \draw (P) node{$\bullet$} node[above left]{\scriptsize $P(x^a)$};
    \draw (Q) node{$\bullet$} node[right]{\scriptsize ~~~~~$Q(x^a + \mathrm{d}x^a)$};
    \draw (R) node{$\bullet$} node[left]{\scriptsize $R(x^a - N^a\mathrm{d}t)$~~~~~};

    \draw[dashed,->,>=stealth,very thick] (P) ++(0,0.1,0) -- node[midway,right]{\scriptsize$Nn^\mu \mathrm{d}t$} ($(R)+(0,-0.1,0)$);
    \draw[dashed,->,>=stealth,very thick] (R) ++(-0.15,0,-0.1) -- node[midway,right]{\scriptsize ~~~~$\mathrm{d}x^a + N^a \mathrm{d}t $} ($(Q)+(0.15,0,0.1)$);
  \end{tikzpicture}

  \caption{The geometrical interpretation of the splitting of the metric into lapse, shift and induced (or spatial) metric.}
  \label{fig:decomp_variables}
\end{figure}

%\textbf{TODO: SCHEMA} (lapse, shift, induced metric + P and Q)

At the infinitesimal level, we can use a Pythagorean development. Indeed, let's consider two points $P$ and $Q$ anywhere in spacetime, as long as they are infinitesimally close to each other. In general, the point $P$ will be on a first hypersurface $\Sigma_t$ and the point $Q$ will be on a second and different $\Sigma_{t+\mathrm{d}t}$. The coordinates of $P$ on $\Sigma_t$ could be written $(x^a)$ and the coordinates of $Q$ on $\Sigma_{t+\mathrm{d}t}$ would be infinitesimally close\graffito{Granted that we are supposing the the coordinate mapping is smooth.} and could be written $(x^a + \mathrm{d}x^a)$. Now, let's define a new point $R$ on the surface $\Sigma_{t+\mathrm{d}t}$ which is also the time-slice to which $Q$ belongs, but with the coordinates $(x^a - N^a dt)$ (see figure \ref{fig:decomp_variables}). That way, the point $R$ is the point we obtain when we start from the point $P$ and then move along the direction $n^\mu$ toward $\Sigma_{t+\mathrm{d}t}$ until we reach the second sheet. In particular, this means that we have split the path from $P$ to $Q$ into two parts with a square angle for a turn. Therefore, we can write the distance between $P$ and $Q$ as:
\begin{equation}
  \mathrm{d}s^2 = - (N\mathrm{d}t)^2 + q_{ab} (\mathrm{d}x^a + N^a\mathrm{d}t)(\mathrm{d}x^b + N^b\mathrm{d}t) 
\end{equation}
We therefore can write $g_{\mu \nu}$ in the coordinate system $(t,x^1, x^2, x^3)$. We get:
\begin{equation}
(g_{\mu\nu}) = \begin{pmatrix}
g_{00} & g_{0b} \\
g_{a0} & g_{ab}
\end{pmatrix} = \begin{pmatrix}
(N_a N^a - N^2) & N_b \\
N_a & q_{ab}
\end{pmatrix}
\end{equation}
All this allows us to write down the action in terms of the lapse, shift and induced metric. We get \cite{Arnowitt2004}:
\begin{equation}
  \begin{array}{rcl}
    S_{EH} &=& \frac{1}{16\pi G}\int \mathcal{L}  d^4 x \\
    \mathcal{L} &=& - q_{ab} \partial_t \pi^{ab} - NR^0 - N_a R^a -2\partial_a\left(\pi^{ab}N_b - \frac{1}{2} \pi N^i + \nabla^i N \sqrt{q}\right)
  \end{array}
\end{equation}
where:
\begin{equation}
  \begin{array}{rcl}
    R^0 &\equiv& - \sqrt{q}\left(^3 R + q^{-1} (\frac{1}{2} \pi^2 - \pi^{ab} \pi_{ab}\right) \\
    R^i &\equiv& -2\nabla_b \pi^{ab}
  \end{array}
\end{equation}
$\pi^{ab}$ was used to designate the conjugate momentum to the induced metric as can be seen from the first term of the lagrangian. Their expression can be derived as:
\begin{equation}
  \pi^{ab} = \sqrt{|g|} (\Gamma_{c d}^{~0~} - q_{cd} \Gamma_{e f}^{~0~}q^{ef}) q^{ac} q^{bd}
\end{equation}
where the summation is only over spatial indices and the $0$ denotes the time component. $q$ is the determinant of the induced metric, $\pi$ is the trace of the momenta ($\pi = \pi^a_a$), the quantity $^3 R$ is the curvature of the spatial metric and the indices are raised and lowered using the spatial metric $q_{ab}$ and its inverse. We see here that the time-derivatives of the lapse and the shift do not intervene in the action. Therefore, they do not have conjugate momenta and are simply Lagrange multipliers. We already foresee here the meaning of general relativity as being a totally constrained theory (for an introduction to constrained theory see \cite{dirac2001lectures,Henneaux1994}): all the terms appear with a factor $N$ or $N^a$. By their Lagrange multipliers nature, they will impose these terms to be zero on the solutions of the equations of motion. In general, the Lagrangian is constituted of two terms: one encoding the Poisson structure with time-derivatives and momenta and a second term which is, up to a sign, the Hamiltonian. This means in particular, for the solution of the equation of motions, that the Hamiltonian will also be zero. Therefore, all the dynamical aspects must be stored somehow in the constraints enforced by the Lagrange multipliers.

Let us now deal with a Legendre transform and the proper writing of the Hamiltonian. We defined the canonical momenta associated to $q_{ab}$ and wrote them $\pi^{ab}$. The phase space is naturally equipped with the canonical Poisson brackets \cite{Dirac1958}:
\begin{equation}
  \{\pi^{ab}(t,x),q_{cd}(t,y)\} = 2\kappa \delta^a_{(c}\delta^b_{d)}\delta(x-y)
\end{equation}
all other brackets being $0$. The Hamiltonian of the theory then reads:
\begin{equation}
  H  = -\frac{1}{16\pi G} \int \left(NR^0 + N_a R^a + 2\partial_a\left(\pi^{ab}N_b - \frac{1}{2} \pi N^i + \nabla^i N \sqrt{q}\right)\right) d^3 x
\end{equation}
We will comment on this form in the next section. In the mean-time, let's study the geometrical interpretation of $\pi^{ab}$.

Indeed, $\pi^{ab}$ can be rewritten as:
\begin{equation}
  \pi^{ab} = q^{-\frac{1}{2}} \left(K^{ab} - K q^{ab}\right)
\end{equation}
where $q$ is the determinant of the induced metric $q_{ab}$, $K = K_{ab} q^{ab}$ is the trace of $K_{ab}$ and $K_{ab}$ is the extrinsic curvature. It is the pull-back of $K_{\mu\nu} = q^\sigma_\mu \nabla_\sigma n_\nu$. Geometrically, it represents the projection onto $\Sigma_t$ of the derivative of the normal, or in more informal terms: how the normal change from one point of $\Sigma_t$ to the other. But the change in the normal do correspond intuitively to the extrinsic curvature of $\Sigma_t$, that is the curvature due to its embedding.

So, we have a new point of view on general relativity. We usually see it as a metric theory of spacetime, with equations of motion governing the intrinsic curvature of spacetime. But when we define the notion of space and time, it can be reexpressed as a theory relating the intrinsic curvature of space and its extrinsic curvature when seen as an embedded manifold into spacetime. The intrinsic curvature is encoded in $q_{ab}$, the induced metric (as it is encoded in $g_{\mu \nu}$ for spacetime), and the extrinsic curvature is encoded in the canonical momenta $\pi^{ab}$. They are related by four constraints imposed through the Lagrange multipliers $N$ and $N^a$. The last three are grouped in a (three dimensional) vector form and is globally called the vector constraint. The first constraint, enforced by the lapse, is called the scalar constraint (for reasons that will appear clear soon). They read:
\begin{equation}
  \begin{array}{rcl}
    -V^c(q_{ab},\pi^{ab}) & \equiv & 2\nabla^{(3)}_d \left(q^{-\frac{1}{2}} \pi^{dc}\right) = 0 \\
    -S(q_{ab},\pi^{ab}) & \equiv & \left(q^{\frac{1}{2}}\left[R^{(3)} - q^{-1}\pi_{cd}\pi^{cd} + \frac{1}{2}q^{-\frac{1}{2}} \pi^2\right]\right) = 0
  \end{array}
\end{equation}

%*****************************************

\section{Constrained theories}

How are we to deal with such constraints at the Hamiltonian level? As we said, the theory of general relativity is a totally constrained theory. In that case, it means that the dynamics is trivial (the Hamiltonian is identically zero on the constraint surface) or, from a different point of view, it is totally contained in the constraints. But constraints can of course appear in a wider context of (not necessarily totally) constrained theories \cite{Dirac1958}. Let us review some instances and develop the Hamiltonian perspective on these.

From a Hamiltonian perspective, a constrained theory is defined by:
\begin{itemize}
\item A phase space as usual, equipped with Poisson brackets
\item In the generic case, a Hamiltonian whose flow gives the time evolution
\item A set of constraints
\end{itemize}
The constraints are functions of the canonical variables and momenta. We look for solutions on the surface constraints, that is on the surface where all the constraints vanish. We should not forget here that in order to be complete, the set of constraints should be preserved in some sense by time evolution. Let us look at a simple example: a free non-relativistic particle in two dimensional space. The Hamiltonian is:
\begin{equation}
  H = \frac{1}{2m}\left(p_x^2 + p_y^2\right)
\end{equation}
where $m$ is the mass of the particle and $\vec{p}$ is the canonical momentum in 2d. Let us consider the constraint $C$:
\begin{equation}
  C \equiv y = 0
\end{equation}
which constrains the motion to be one dimensional in the $x$ direction. It is quite obvious that this constraint cannot be satisfied at all time unless $\{C, H\} = \frac{p_y}{m} = 0$ on the initial time-slice. Constraints that come from this consistency over time requirement can of course themselves require new constraints to be preserved. We will only consider here complete sets of constraints, where the Poisson bracket with the Hamiltonian is guaranteed to vanish as soon as all the constraints are satisfied.

In this perspective, the constraints select a subsurface of the phase space, called the constraint surface, surface which is preserved by time evolution. We now have a collection of constraints, let's label them $C_i$. When all the constraints are satisfied, the system is said to be \textit{on-shell} and when not, it is \textit{off-shell}. Now, most equalities must be verified \textit{weakly} that is only when the system is on-shell. Such an equality will be written with $\approx$. We can then consider the matrix $M_{ij} = \{C_i,C_j\}$. It characterize how constraints commute with each other in the Poisson bracket sense. We can consider two kinds of constraints: constraints which do commute (weakly that is when the constraints are satisfied) with all the other constraints and constraints that do not. The first type are called (quite unimaginatively) \textit{first class constraints} and the second type \textit{second class constraints}\graffito{Once again, the standard names are not really satisfying, but we are now stuck with it.}. The example we gave of the free particle with constraints only has second class constraints. But, in the case of gravity, we will be very much interested by first class constraints, as they are linked to gauge invariance.

Indeed, if a constraint, say $C_1$, commutes (on-shell) with all the others, then the flow generated by this constraint under the Poisson bracket preserves all constraints. That means, we can generate new solutions of the equations of motion and of the constraints by starting from a known solutions and consider its flow under $C_1$ action. If a collection of constraints commute with all the other constraints (on-shell) and their Poisson brackets with each other is linear in them, they naturally generate a Lie algebra (or a generalization in the infinite dimensional case), where the Poisson bracket has the role of a Lie bracket. In that case, the collection of constraints can be considered as the generators of a group, which can itself be considered a symmetry group. This highlights the role of constraints in gauge theory: they are linked to the gauge group and generate it.

How does this work in the case of general relativity? As we said earlier, general relativity is a totally constrained theory. That means the Hamiltonian is trivial and all the dynamics is specified by the constraint surface which, by its choice, induces correlations between observables. The constraints of general relativity are \cite{Arnowitt1960}:
\begin{equation}
  \begin{array}{rcl}
    -V^c(q_{ab},\pi^{ab}) & \equiv & 2\nabla^{(3)}_d \left(q^{-\frac{1}{2}} \pi^{dc}\right) = 0 \\
    -S(q_{ab},\pi^{ab}) & \equiv & \left(q^{\frac{1}{2}}\left[R^{(3)} - q^{-1}\pi_{cd}\pi^{cd} + \frac{1}{2}q^{-\frac{1}{2}} \pi^2\right]\right) = 0
  \end{array}
\end{equation}
These constraints are actually an infinite collection of constraints: one for each point of the three-dimensional manifold $\Sigma_t$. It is usual therefore to represent them in an integrated fashion. We define the diffeomorphism constraint (the name gives away their action) as:\graffito{We could think that the name, \textit{Hamiltonian} constraint, is once again a misnomer. But for once, it does carry the meaning of its action: defining the dynamics.}
\begin{equation}
  D[N^b] = \int N_b V^b d^3 x
\end{equation}
The new constraints depend on a test field, which we named after the Lagrange multipliers, the shift vector. In the very same way, we define the Hamiltonian constraint:
\begin{equation}
  H[N] = \int N S d^3 x
\end{equation}
The Lagrange multiplier is now the lapse as it appears in the Hamiltonian. These constraints satisfy the following algebra \cite{Arnowitt1960}:
\begin{equation}
  \begin{array}{rcl}
    \{D[N^b],D[M^c]\} &=& D[\mathcal{L}_{N^b}M^c] \\
    \{D[N^b],H[N]\} &=& H[\mathcal{L}_{N^b}N] \\
    \{H[N],H[M]\} &=& -D[q^{ab}(N\partial_b M - M \partial_b N)]
  \end{array}
\end{equation}
where $\mathcal{L}$ is the Lie derivative. This algebra is known as the \textit{Dirac algebra} and encodes the symmetry of general relativity. The algebra is actually the algebra of four-dimensional diffeomorphisms. It is split into two parts: the vector constraints, that transform as vectors under \textit{spatial} diffeomorphism and indeed encode the spatial diffeomorphism themselves, and the scalar constraint that transform as a scalar under spatial diffeomorphism and encodes the dynamics of the theory as it links how different time slices are related.

How are we to promote this at the quantum level? For now, we are lucky, we just have to deal with first class constraints. The procedure is then quite simple, at least in principle:
\begin{itemize}
\item We define a representation of the variables and canonical momenta as is usual in quantization. It might be done for instance through the multiplication and derivative action.
\item We quantize the constraints\graffito{If the algebra is not satisfied at the quantum level, the quantum theory has an anomaly. It can arise of course, but we usually try to avoid for gauge theory as the symmetry itself guarantees the good definition of the theory.} by looking at equivalent expressed in terms of operators on the previous space. Ambiguities due to ordering might arise. But they might also be solved by the request that the algebra of constraints is realized non-anomalously.
\item We solve the equation $C_i|\psi\rangle = 0$, that is, we look for the kernel of the constraints. This is the equivalent of looking for the surface constraints in the classical setting.
\end{itemize}
It is standard to solve all the constraints but the Hamiltonian constraint first. Indeed, the Hamiltonian formalism quite naturally brings a three-dimensional scene. It is therefore usually simpler to solve constraints that only concern the same time slice. This corresponds to building a \textit{kinematical phase space} of allowed configuration at a given time. The Hamiltonian constraint because of its link with the dynamics will be solved in another pass.

But in the case of the ADM variables, no representation is known as of today which would include a nice scalar product and an action of the diffeomorphism group. It does not really mean that such a representation cannot exist but that, as its stands the programme cannot successfully be conducted. Still, it is possible to sketch a formal quantization process and see where this leads.

%*****************************************

\section{A first approach to quantization}

Let us now sketch such a would-be quantization as suggested in \cite{Arnowitt2004,Arnowitt1960} and effectively sketched in \cite{rovelli2007quantum} for instance. This is formal at best as the quantization of the ADM variables is not well-defined (at least for now) but this approach will give us a feeling of what's going on and will guide us in the loop quantization to come.

\graffito{It is customary to call a function of a function, a functional. Though the term is unnecessary in principle, it is certainly explicit and usually carries some implicit notion of differentiability and continuity.}In principle, we would like to promote the variables and momenta to operators. The Hilbert space would be the space of wavefunctions over the metric (for instance). That means the wavefunctions would actually be functionals. The precise definition for the scalar product is left imprecise. That would be a problem on the long run of course, but as we will not dwell into this quantization effort, let's not weep over this. The variables could then be represented by the multiplication action:
\begin{equation}
  \left(\widehat{g_{ab}(x)} \psi\right)(\{g_{ab}\}) = g_{ab}(x) \psi(\{g_{ab}\})
\end{equation}
And, in order to get the correct commutators, the momenta would be represented by the (functional) derivative. To simplify the writing, we have assumed $8 \pi G = 1$ to go to the full Planck units. For our introductory purpose here, this will make the equations easier to read:
\begin{equation}
  \left(\widehat{\pi^{ab}(x)} \psi\right)(\{g_{ab}\}) = -\mathrm{i}\frac{\delta \psi}{\delta g_{ab}(x)}(\{g_{ab}\})
\end{equation}
\graffito{It is now customary to call Wheeler De Witt equation any Hamiltonian constraint in the quantum theory.} We could now write down the Hamiltonian constraint (and all the other constraints) in operator form. Up to ordering ambiguities, we will then get the main equation of quantum general relativity, the \textit{Wheeler De-Witt} equation:
\begin{equation}
  \hat{H} |\psi\rangle = 0
\end{equation}
where we used the Dirac notation for the wavefunction. The equation must interpreted as an equation over $|\psi\rangle$ as it is a constraint equation. \graffito{Most of the difficulty of the quantization comes from the non-linear, even non-analytical, behavior of the Wheeler De-Witt equation.}Of course, there is a slight problem here: nobody  knows how to give a precise sense to the above equation in the quantum framework. What should we do then? Well, as usual in research, the right strategy is to start with a simpler problem and work our way up from there. We will therefore consider symmetry reduced version of this equation.

The idea is to classically reduce the problem we are considering and then quantize. This approach might not be motivated on a mathematical level for example, but we must recognize it is effective. This is precisely what we do for the quantization of the hydrogen atom for instance. As Ashtekar notes \cite{Ashtekar2011}, historically, we started from a classical electron in a potential and quantized this theory. But the true theory, as we now know, is \ac{QED}. From the classical perspective, an electron in a potential is a reduced version of electrodynamics where the modes of electromagnetism have been frozen (and the non-relativistic limit is taken). Still, quantizing the simple theory works surprisingly well. So, this is what we can do here for quantum general relativity.

What kind of simplified theory could we study? The simplest idea is to study homogeneous and isotropic universes, that is cosmology\graffito{It is indeed quite surprising to see how homogeneous the universe is, with relative temperature fluctuations, in the early universe, as low as $\frac{\Delta T}{T} \simeq 10^{-5}$ \cite{Bennett:1996ce}.}. Let us reduce the problem to a simpler problem of cosmology and then quantize the equations. In order to have a non-trivial theory (a homogeneous universe with pure gravity is kind of boring), we will introduce matter. But we'll concentrate, of course, on matter of the simplest kind: a scalar field with no mass. The action corresponding to such a scalar field (without any assumptions on the homogeneity) is:
\begin{equation}
S_\varphi = - \int_\mathcal{M} \frac{1}{2}g^{\mu\nu} \partial_\mu \varphi \partial_\nu \varphi \sqrt{|g|} d^4 x
\end{equation}
where $\varphi$ is the scalar field. So we get the total action:
\begin{equation}
  S_\textrm{tot} = - \int_\mathcal{M} \sqrt{|g|} \left(\frac{1}{16\pi G} R +  \frac{1}{2}g^{\mu\nu} \partial_\mu \varphi \partial_\nu \varphi \right) d^4 x
 \end{equation}
Let us go to the symmetry reduced version of this action. We choose a splitting of spacetime such that $\varphi$ is just a function of the time variable $t$. We assume space to be homogeneous too, that means that the line element can be written as \cite{Ashtekar2009}:
\begin{equation}
  \mathrm{d}s^2 = -N^2(t) \mathrm{d}t^2 + a(t)^2 \mathrm{d}\Sigma^2
\end{equation}
where $\mathrm{d}\Sigma$ is the volume element and $a(t)$ represents the scale factor of the universe. This scale factor solely depends on time, giving a precise notion of homogeneity. The factor $N(t)$ is the lapse. Fixing it fixes our time parametrization. But because we want to stay somewhat close to constrained theories (like general relativity), we will keep it unfixed. Therefore the dynamics will still be encoded in a Hamiltonian constraint. Let us assume spherical coordinates, then we can write:
\begin{equation}
   \mathrm{d}\Sigma^2 = \frac{\mathrm{d}r^2}{1 -kr^2} + r^2\left(\mathrm{d}\theta^2 + \sin^2 \theta \mathrm{d}\phi^2\right)
\end{equation}
$k$ here represents the curvature of the universe. If it is $0$, then the universe is flat. It has positive curvature (like a sphere)  when $k = +1$ and negative curvature (like a hyperboloid) when $k = -1$. $r$ should not be taken to be the physical distance, though it is linked to it. The reduced action now reads:
\begin{equation}
  S_\textrm{tot,reduced} = - \int_\mathcal{M} N(t)\left( 3a(k-\dot{a}^2) + \frac{a^3}{2} \dot{\varphi}^2 \right)d^4x
\end{equation}
where the dots represent time derivatives and the dimensionfull factors have been removed for clarity. The variables are now the scale factor $a$ and its canonical conjugate momentum, let's call it $\pi_a$, and the field $\varphi$ along with its momentum $\pi_\varphi$. Precisely, we have:
\begin{equation}
  \left\{
  \begin{array}{rcl}
    \pi_a &=& 6Na\dot{a} \\
    \pi_\varphi &=& Na^3 \dot{\varphi}
  \end{array}
  \right.
\end{equation}
The time derivative of $N$ does not appear and therefore $N$ is a Lagrange multiplier. From the Hamiltonian point of view, it will enforce the Hamiltonian constraint associated to reparametrization of time. All in all, we find the following Hamiltonian constraint:
\begin{equation}
  H = a^2(ka^2-\pi_a^2) + 6 \left(\pi_\varphi\right)^2
\end{equation}
where appropriate multiplications have been done to simplify the denominator. If we consider the case of flat space $k=0$, the hamiltonian simplifies further and simply gives:
\begin{equation}
  H = -a^2\pi_a^2 + 6 \left(\pi_\varphi\right)^2
\end{equation}
This is interesting since, we want to impose $H=0$. We can very well see that this can simplify (up to a sign) to:
\begin{equation}
  \pi_\varphi = \frac{a\pi_a}{\sqrt{6}}
\end{equation}
which, apart from some ordering ambiguities should be quantizable.

Let us do the quantization part then. We consider wavefunction of $a$ and $\varphi$. The operators are represented as follows:
\begin{equation}
  \begin{array}{rcl}
    \hat{a}\psi(a,\varphi) &=& a\psi(a,\varphi) \\
    \widehat{\pi_a}\psi(a,\varphi) &=& -\mathrm{i} \frac{\partial \psi}{\partial a}(a,\varphi) \\
    \hat{\varphi}\psi(a,\varphi) &=& \varphi\psi(a,\varphi) \\
    \widehat{\pi_\varphi}\psi(a,\varphi) &=& -\mathrm{i} \frac{\partial \psi}{\partial \varphi}(a,\varphi)
  \end{array}
\end{equation}
The Hamiltonian constraint now reads, up to ordering ambiguities:
\begin{equation}
  \frac{\partial \psi}{\partial \varphi} = \frac{1}{2\sqrt{6}}\left(a \frac{\partial \psi}{\partial a} + \frac{\partial (a \psi)}{\partial a}\right) = \frac{1}{2\sqrt{6}} \left(2a \frac{\partial \psi}{\partial a} + \psi\right)
\end{equation}
which can be solved! The general solutions can be written as follows:
\begin{equation}
  \psi(\varphi,a) = \int e^{-\mathrm{i}\alpha \varphi} a^{-\mathrm{i}\sqrt{6}(\alpha -2)} \psi_\alpha \mathrm{d}\alpha
\end{equation}

The problem appears in the limit $\varphi\rightarrow-\infty$. Indeed, here we are using the matter field as a \textit{relational clock}. That means we are using the value of the field as a way to measure time. A good relational clock is a quantity that takes one different value for every different time. Using such a quantity, we can reexpress any measurement with respect to time without even mentioning a special reference frame but in relations of different physical quantities. The matter field here has this special role. It goes to infinity as time passes and it goes to minus infinity toward the initial singularity. But this is the important point: eventhough the singularity is reported to minus infinity with respect to the scalar field, it is still obtained in finite time for any observer, at least at the classical level. As one of our challenge when considering quantum general relativity is to resolve singularities, we should certainly hope that it is solved in the simple case of quantum cosmology. And for this, we must check that the probability for the quantum solution to hit $a = 0$ is zero, even when $\varphi$ goes to $-\infty$, and since must be the case in a wide variety of solutions. It surely isn't when we look at the wavefunction we gave above as $a = 0$ cannot be avoided \textit{generically}.

The problem is worse that it seems at first sight. First, ordering ambiguities do not lift the problem and no clever choice of ordering changes the result (as explained in \cite{Ashtekar2011}). So, we could hope that the representation of the variables is the problem. But it turns out that there are some constraints in the possible representation of the variables. Indeed, in this simple case, we are just considering the Heisenberg algebra which is the standard algebra of quantum mechanics operators of position and momentum. In this framework, the Von Neumann theorem \cite{v.Neumann1931} implies that the representation is unique given a few reasonable hypotheses, including in particular the weak continuity of the representation. This seems to doom our enterprise of quantizing general relativity. Indeed, it seems that any theory must give in the homogeneous case the previous development which cannot solve the singularity problem.

One hope may still exist. it is that we, somehow, escape the hypothesis of the Von Neumann theorem. This might be unreasonable at first. But as it was argued recently \cite{Dittrich:2015vfa}, this might be needed to handle diffeomorphism invariant theories. Indeed, in diffeomorphism invariant theories, Dirac observables are generally not continuous leading to difficulties in the definition of a Poisson structure. Resolving this problem seems to lead to another quantization scheme, which is incidentally very similar to loop quantization. That this method works and solves our problem with the usual representation is the hope we will nourish from now on.

\vspace{1em}

In this chapter, we started from the usual lagrangian formulation of general relativity and explored its hamiltonian formulation. Some subtleties where underlined in the interpretation of the Hamiltonian framework, specifically how the covariance was maintained though it seemed it was broken by the choice of spacetime splitting. The Hamiltonian formulation was laid out in the ADM variables which are now standard and their geometrical interpretation was given. Finally, we discussed the system of constrained and explored the quantization of general relativity using the simpler system described by the Friedman equation. The resolution of the singularity (which might be a good test for the good definition of the theory) was not found, at least, not generically. It might be hoped at this stage that the problem comes from a bad choice of quantization. In particular, it would be good to circumvent one the hypothesis of the Von Neumann theorem which states that the representation of quantum mechanics is unique. This is what we are going to explore in the rest of this part. The next chapter will concentrate on a new set of variables which might induce or at least suggest a new representation more suited to our needs.

%*****************************************
%*****************************************
%*****************************************
%*****************************************
%*****************************************
 % Hamiltonian General relativity
%\addtocontents{toc}{\protect\clearpage} % <--- just debug stuff, ignore
%*****************************************
\chapter{The Ashtekar-Barbero variables} \label{ch:AshBarbVariables}
%*****************************************

\inspiquote{One may tolerate a world of demons for the sake of an angel.}{Reinette}

As was developed in the previous chapter, the straightforward path to quantum gravity, starting with an Hamiltonian form and quantizing it canonically, seems a dead end. But a possible solution was also brought to light: the possible existence of nonequivalent quantization schemes. And indeed, \ac{LQG} offers a different quantization process, using new representations and new variables. These new variables will be the point of interest of this chapter.

We should illustrate this in a simpler example. Indeed, already in the cosmological setting, we have the problem of inequivalent representation. So, let's consider a non-relativistic particle in a one-dimensional space. Forgetting about the dynamics and concentrating on the kinematics, the particle is described by two variables which are its position $x$ and its momentum $p$. They satisfy a Poisson bracket which is, up to sign conventions:
\begin{equation}
  \{x,p\} = 1
\end{equation}
There is a natural representation of this at the quantum level by the multiplicative operator $\hat{x}$ on $L^2$ wavefunctions and the derivative $\hat{p} = -\mathrm{i}\hbar\frac{\mathrm{d}}{\mathrm{d}x}$. According to the Von Neumann theorem, this representation is unique with a very few hypothesis. Interestingly, the Von Neumann theorem does not start with the variables $x$ and $p$ but with their exponentials (which are well-defined operators on the whole $L^2$ space). But, in order to recover the existence of the position and momentum \textit{per se}, we need a continuity hypothesis. There are (natural) cases where this is not true and where the right variables are the exponentials which turn out not to be (weakly) continuous. A simple example of this is a particle moving on a circle, where the angle variable $\theta$ is not well-defined but its exponential is.

Back to gravity, this leads to a natural question: what is the exponentiated version of our variables? Is it natural to consider that some operator, presumably $\hat{q_{ab}}$ or $\hat{\pi^{ab}}$, does not exist in the full quantum gravity theory but only its exponential? And more importantly, does this solve our problems? It turns out that all these problems can, for the major part, be solved together along with another problem: the non-polynomial writing of the constraints. Indeed, because the theory is highly non-linear, it is very difficult to find a coherent quantization of the formulas. Therefore, the quantization might actually need a rewriting of the theory in polynomial form. In particular, might a clever change of variable bring the theory to such a writing? Thought this might seem implausible at first, there is actually a way to do this, and this is what we are going to study in this chapter.

In this first section, we will not solve this whole problem, but we will explore other formulation of general relativity. The end goal is to formulate general relativity as a gauge theory which will bring it closer to Yang-Mills theory and allow us to use technology from the treatment of these theories. It will also allow for a rewriting of general relativity in \textit{first-order} language in the second section, which is an important point in the development of this chapter. This will unveil a hidden symmetry of general relativity namely, the local Lorentz invariance. In the third and fourth section, we will introduce the now paramount variables for the quantization programme of \ac{LQG}: the Ashtekar-Barbero variables. And we will close this chapter with a discussion on the role and possible physical implication of this choice of variables.

\section{Tetrad variables}

General relativity is a theory of a dynamical metric $g_{\mu\nu}$. For all intents and purposes, it can be understood as a rank two symmetric tensor. The symmetry guaranties that the metric is diagonalizable as a real matrix. In particular, we can write:
\begin{equation}
g_{\mu\nu} = e^I_\mu e^J_\nu \eta_{IJ}
\end{equation}
where $\eta$ is Minkowski's metric with signature $(- + + +)$. Indeed, this can be shown by diagonalizing $g$ then by rescaling the transfer matrices. Only the signs remain so we can't change the signs of the signature.

What is the geometrical meaning of this operation? At each point, spacetime is locally flat. That means there exists a frame of reference, and a corresponding choice of coordinates, such that the metric is Minkowskian. The condition is actually more stringent as it can be shown that the affine connection also vanishes (locally of course). Locally, we can represent this choice of coordinates as a set of four orthogonal vectors: they correspond to the basis  vectors associated to the local Cartesian coordinates (corresponding to the local trivialization) projected onto the original coordinate system. As this Cartesian coordinate system correspond to a frame of reference, the vectors have a physical interpretation. Indeed, in such a set, one vector must be timelike and the other three spacelike. The timelike vector represents the proper time direction of the observer. The other three selects three spatial directions and so distinguish observers rotated with respect to each other.

Technically, the tetrad is actually the corresponding base for forms and the vectors $e^\mu_I$ are defined as follows:\graffito{Once again the notation is a mess. Inverse tetrads $e^\mu_I$ and tetrad $e^I_\mu$ are denoted by the same symbol $e$ and only context or indices help distinguish between them.}
\begin{equation}
  e^I_\mu e^\mu_J = \delta^I_J
\end{equation}
These new vectors form the inverse tetrad. They are precisely the set of four vectors corresponding to a local observer. Computing the dot product will allow to check that the base is actually orthonormal and therefore diagonalize the metric and can be interpreted as a local choice of observer:
\begin{equation}
e^\mu_I g_{\mu\nu} e^\nu_J = \eta_{IJ}
\end{equation}
Therefore diagonalizing the metric actually corresponds to finding a local falling observer (since the metric is flat in its coordinate system at least locally). There is an $\mathrm{SO}(3,1)$ freedom in choosing this set of vector corresponding to the symmetries of the Minkowski metric. From the physical point of view, this is the group sending one falling observer onto the other. Since, this symmetry acts trivially on the metric, it is respected by general relativity even if reexpressed in terms of the tetrad. We have therefore uncovered a hidden symmetry of general relativity. The presence of such a hidden symmetry should not be a surprise. For instance, in electromagnetism, the $\mathrm{U}(1)$ symmetry cannot be seen with the electric and magnetic fields only: we have to use the potentials to uncover it. It works quite similarly here: the metric corresponds to a Lorentz invariant construction from more fundamental variables, the tetrad.

Having tangential indices means that we can have a new connection, acting in the tangent space. Indeed, the tetrad acts as a map between the usual coordinates and coordinates in tangent space. Therefore, we can define a connection, the spin connection, which is such that acting with the tetrad before or after covariant derivation obtains the same result. This \textit{compatibility requirement} is equivalent to;
\begin{equation}
D_\mu e^I_\nu = \partial_\mu e^I_\nu + \Gamma^\sigma_{~\nu\mu} e^I_\sigma + \epsilon^I_{~JK} \Gamma^J_\mu e^K_\nu = 0
\end{equation}
where the usual affine connection is used for spacetime indices and the new spin connection (denoted by the same $\Gamma$ letter again) is used for tangential indices. This formalism might seem an unnecessary addition at this point except maybe to point out the Lorentzian symmetry. It is however necessary to use it in order to write down general relativity coupled to fermions. Indeed, fermions are naturally written in the tangent space (or in a spin fiber) and the spin connection naturally appears as the connection to use for them, as the one they couple to. In particular, this means that fermions are the only species that couple to the spin connection as will become important latter on.

\section{Palatini's formalism}

Now that we have uncover the hidden Lorentz symmetry, let's dig in the connection formalism. Let us start from the Einstein-Hilbert action in the pure gravity case. It reads:
\begin{equation}
  S_{EH}[g_{\mu\nu}] = -\frac{1}{16\pi G}\int \sqrt{|g|} R d^4 x
\end{equation}
Let us forget about all the conventional wisdom and just play a bit with the formulas. First, we know that the scalar curvature $R$ is obtained by taking the trace of the Ricci tensor $R_{\mu\nu}$. This tensor is itself obtained by contraction of the Riemann tensor which is entirely expressed as a function of the affine connection $\Gamma^\mu_{~\nu\sigma}$. The action can therefore be written:
\begin{equation}
  S_{EH}[g_{\mu\nu}] = -\frac{1}{16\pi G}\int \sqrt{|g|} g^{\mu\nu}R_{\mu\nu}(\Gamma(g_{\mu\nu})) d^4 x
\end{equation}
The idea of Palatini's formalism is to look at what happens if we consider the affine connection to be independent of the metric. That is, we look at a new theory \textit{a priori} different from Einstein's and governed by the following action:
\begin{equation}
  S_{Palatini}[g_{\mu\nu},\Gamma^\mu_{~\nu\sigma}] = -\frac{1}{16\pi G}\int \sqrt{|g|} g^{\mu\nu}R_{\mu\nu}(\Gamma) d^4 x
\end{equation}
We can derive the equation of motion to find the following result. Varying with respect to the metric still gives Einstein's equation but as a function of $\Gamma$, that is we get:
\begin{equation}
  R_{\mu\nu}(\Gamma) - \frac{1}{2}g_{\mu\nu} g^{\sigma\tau}R_{\sigma\tau}(\Gamma) = \frac{8\pi G}{c^4}T_{\mu\nu}
  \end{equation}
where $T_{\mu\nu}$ was added for clarity but is of course $0$ when no matter is present. Varying the action with respect to $\Gamma$ however gives the following (algebraic) equation of motion:
\begin{equation}
\Gamma = \Gamma(g_{\mu\nu})
\end{equation}
Where the left-hand side represents the $\Gamma$ variables, but the right-hand side correspond to the unique torsion-free metric compatible connection. \graffito{It seems to me that this must have some deep geometrical significance, but which still eludes me. Still, this is probably a good clue of the physical relevance of the first order theory.}This means that the equation of motion gives the link between the connection and the metric, it is not anymore imposed on the theory. This is a remarkable fact that allows general relativity to be expressed in a manner more closely related to Yang-Mills theory.

Of course, the same thing can be done in the tetrad formalism. Rather than using the metric, we use the tetrad and rather than using the affine connection, we will use the spin connection (the unique torsionless compatible with metric and the tetrad connection). As stated earlier, this formalism allows for fermions in the theory. We should make a side note here: when we include fermions, they do couple with the spin connection so that, if we treat the connection and the tetrad separately, the equation of motion might not induce a torsionfree connection. The fermions act as a source of torsion and therefore the theory, though close, is not exactly general relativity in the second order formulation. Finding the right formulation of general relativity, how to include fermions and what is the most natural action is a very interesting problem. But as for this thesis, as we are only interested in pure gravity for now, all these problems do not arise. Therefore let's not dwell into them.

Let us unravel a bit the tetrad formulation. The action of (pure) general relativity can be written using the tetrad and in the first order formulation as follows:
\begin{equation}
  S_\textrm{tetrad} = -\frac{1}{16\pi G} \int_\mathcal{M} \epsilon_{IJKL} e^I \wedge e^J \wedge F^{KL}(A)
\end{equation}
From the Hamiltonian point of view, the conjugate momentum of the Lorentz connection $A^{IJ}$ will be the bivector $\epsilon_{IJKL} e^K \wedge e^L$. This shows the close relationship between general relativity and BF theory.\graffito{It is said that BF theory is a misnomer, since we should rather call it EF theory as the $B$ field has a role similar to the electric field in electromagnetism.} BF theory is the theory defined by the following action:
\begin{equation}
  S_{BF} = \int_\mathcal{M} B_{IJ} \wedge F^{IJ}(A)
\end{equation}
Bf theory is therefore the simplest possible theory we can imagine on a connection: the $B$ field acts as a Lagrange multiplier and impose flat curvature everywhere on the manifold. It might be surprising to say that general relativity is anywhere close to such a theory, but several points should be noted in this direction.

First and foremost, it should be said that in $3$ dimensions of spacetime, general relativity is exactly BF theory. Indeed, a simple counting argument shows that general relativity in $3$ dimensions has no local degrees of freedom and further study explicitly demonstrates that the equation of motion impose flatness of the connection. But even in $4$ spacetime dimensions, the theory is not that different. The previous writing of $S_\textrm{tetrad}$ underlines this. It is precisely BF, where the $B$ field is constrained to have the form:
\begin{equation}
B_{IJ} = \epsilon_{IJKL} e^K \wedge e^L
\end{equation}
that is $B_{IJ}$ is constrained to be a \textit{simple} bivector. As a consequence, general relativity is sometimes called a constrained BF theory. The advantage of looking at general relativity this way is that we progressively express it in a theory of a connection. And theories of connection are much more well understood than metric theories. In particular, Yang-Mills theory, the archetype of a connection theory, is very well studied. This means, we can hope that expressing general relativity as a connection theory will give us access to all the technology developed on these theories. Still, we need to rewrite still a few points to arrive at a full theory of connection. As a side point, let's note here that general relativity can be written without the $B$ field as a theory of connection \textit{alone} \cite{Krasnov:2011pp}. This might play some role in the renormalization process. Though interesting, this is not the road considered in this thesis.

%*****************************************

\section{Ashtekar self-dual variables}

Let us write the theory in a slightly different manner. We define the following connection, known as the Ashtekar connection:
\begin{equation}
A^i_a = \Gamma^{0i}_a + \mathrm{i}\frac{\epsilon^i_{~jk}}{2}\Gamma^{jk}_a
\end{equation}
where $\Gamma$ is the spin connection introduced earlier for the fermions, which encode the same information as the affine connection. The indices in small Latin letters of the middle of the alphabet ($i$, $j$, $k$) are spatial only but still in the tangent space. The indices from the beginning of the alphabet still refers to coordinates in space as was defined in the Hamiltonian point of view. This new connection is an $\mathrm{SL}(2,\mathbb{C})$ connection and thanks to the $\mathrm{i}$ parameter that makes it complex, it still has the same information as the pull-back of the Lorentz connection on the spatial manifold. This connection is called the self-dual connection as it is self-dual with respect to the natural Hodge star on $\mathrm{SL}(2,\mathbb{C})$.

These variables are exactly the variables we need to reexpress general relativity in the Hamiltonian framework. As we hinted earlier, a new gauge symmetry is revealed and corresponds to local Lorentz transform. In our case, the Lorentz connection is contained in an $\mathrm{SL}(2,\mathbb{C})$ self-dual complex connection. The corresponding generators (since we are allowed complex parameters) are therefore the generators of $\mathrm{SU}(2)$ and read:
\begin{equation}
G_I = D_a E_I^a
\end{equation}
where $E_I^a$ is the momentum associated to the Ashtekar variables and sometimes called the gravitational \textit{electric field} as it plays a similar role to the electric fields or more generally gauge fields in Yang-Mills theory. The constraints themselves look exactly like the Gauß constraint of electromagnetism. The parallel is to be expected as both express local gauge freedom. The spatial diffeomorphism constraints can also be reexpressed using the new variables and have the nice following polynomial form:
\begin{equation}
C_a = E^b_I F^I_{ab}
\end{equation}
where $F$ is the curvature of $A$. And finally the Hamiltonian constraint which encodes the dynamics can be written as:
\begin{equation}
H = \frac{1}{2} \epsilon^{IJ}_{~~K} \frac{E^a_I E^b_J F^K_{ab}}{\sqrt{\det E}}
\end{equation}
The simplicity of this form is astonishing. We are of course still discussing at the classical level and the quantum level will bring its level of difficulties. But apart from the denominator (which guaranties that $H$ is a density one scalar), all the terms are polynomial in the variables or their conjugate momenta. Moreover, since we want to impose $H=0$, apart from some specific choices of $E$ (\textit{i.e.} when $E$ is degenerate), the condition is equivalent to $\epsilon^{IJ}_{~~K} E^a_I E^b_J F^K_{ab} = 0$ which is completely polynomial. This means we have achieved the formulation of general relativity in polynomial terms only. This is a tremendous achievement that should help us in the pursuit of the quantum theory.

So what is the catch?\graffito{It is quite normal to expect a catch: except when a really novel idea which is in itself a solution, there seems to be a law of \textit{conservation of difficulty} at work in research.} What have we traded to have such a simple form? The main problem comes from the non-reality of the variables. Indeed, written as this, the theory will actually be complex general relativity where the metric is complex. The connection $A$ is self-dual by construction from the previous connection but when we express only the theory in terms of $A$ and $E$, we lost the specificity of the construction and therefore the self-duality of the connection. We can, of course, retrieve it by imposing new constraints, corresponding to self-duality. They are suggestively called \textit{reality conditions} \cite{Immirzi:1992ar} (see also \cite{Wieland2012} for interesting developments). They simply read:
\begin{equation}
  \begin{array}{rcl}
    E^i_a E^{ja} - \overline{E^i_a E^{ja}} &=& 0 \\
    \mathrm{i} \epsilon^{IJK} E^c_K (E^a_I D_c E^b_J - E^b_J D_c E^a_I) - \textrm{c.c.} &=& 0
  \end{array}
\end{equation}
The first condition just expresses the reality of the (spatial) metric. The second condition comes from time evolution or, seen in another way, is necessary to guaranty the cancellation of all Poisson brackets with the constraints. From a spacetime point of view, we can also see this second condition as a reality condition on the time components of the spacetime metric.

Such conditions are hard to solve, especially at the quantum level. We will not dwell into those problems and restrict ourselves to real variables in just about a moment. But still, let's rapidly sketch we would implement such conditions. There are naturally two possible avenues:
\begin{itemize}
\item The first possibility is to treat these conditions as they are: second class constraints. This can be done for instance with the Dirac procedure and by defining new brackets which will solve these conditions. When we do this, the connection becomes however non-commutative and when we don't know (for now) how to quantize such a theory. A way out is to find a commutative sub-algebra, which precisely leads to the use of self-dual variables and to the issue of reality conditions.
\item Therefore, a second possibility is to try and implement reality conditions on commuting variables (the self-dual variables). In practice, we look for a Hilbert space such that, with respect to the scalar product, the conditions are automatically satisfied. Indeed, the reality conditions are really conditions between complex variables and their conjugates. When promoted to quantum conditions, they become conditions on operators and their adjoints. But adjoints are defined by the scalar product and so, we can try and tweak the scalar product so that these conditions are automatically satisfied.
\end{itemize}
So far, however, this problem of solving the reality condition remains open though it has been solved in reduced problems as in \ac{LQC} \cite{Wilson-Ewing:2015lia}. This problem with using complex variables has motivated people to move to yet another form of variables: the Ashtekar-Barbero variables. Indeed, these variables are for the most part as convenient as the original Ashtekar variables but they are real. Of course, they come with their own caveats which we will study in the coming sections.

%*****************************************

\section{Ashtekar-Barbero variables}

So, we are lead to search for real variables only. Barbero and Immirzi found a trick just to do so. The idea is to generalize the Ashtekar variables a bit and change the $\mathrm{i}$ factor into a free parameter, called the Immirzi parameter, which we will write $\beta$. Therefore, let's consider the new variables, called the Ashtekar-Barbero variables \cite{Immirzi:1996di,BarberoG.1995} defined as:
\begin{equation}
A^{i(\beta)}_a = \Gamma^{0i}_a + \beta\frac{\epsilon^i_{~jk}}{2}\Gamma^{jk}_a
\end{equation}
This is a plain generalization of the Ashtekar variables and quite notably, corresponds to a canonical transform \cite{BarberoG.1995,Rovelli:1997na}. The interesting point is that if we choose $\beta$ to be real, the connection is a pure $\mathrm{SU}(2)$ connection and as such, is real. But of course, we are not allowed to take whatever we please as variables in a theory. Let us see how these variables can be good variables for general relativity.

A rather interesting fact is that the Ashtekar-Barbero variables emerge quite naturally when considering a slight modification of the action \cite{Holst1996}. Let us consider:
\begin{equation}
  S_\textrm{AB} = -\frac{1}{16\pi G}\int_\mathcal{M} \left(\epsilon_{IJKL} e^I \wedge e^J \wedge F^{IJ} + \frac{1}{\beta} e_I \wedge e_J \wedge F^{IJ}\right)
\end{equation}
As usual, the $e^I$ are the tetrad, and the $F$ is the curvature of the Lorentzian connection. As can be seen, we added a new term, called the Holst term. When the torsion is zero, this term is equal to the Nieh-Yan term which is topological. Therefore, as long as torsion vanishes, the Holst term does not change the equations of motion.

How does the Ashtekar-Barbero variables appear from here? We must start by fixing the gauge a little. Indeed, apart from the diffeomorphism invariance, Palatini's formalism reveals a local Lorentz gauge freedom. We will fix (partially) this gauge freedom by using the time-gauge. Indeed, in the tetrad formalism, one of the tetrad vector is timelike. Since we chose a coordinate system on the spacetime manifold with a preferred time direction, it is natural to ask for the time vector of the tetrad to point in this direction. Technically, we will require that the other three tetrad vectors have no time component. Mathematically, we ask for:
\begin{equation}
e^I_0 = (1~0~0~0)
\end{equation}
This fixes the gauge freedom only partially are we are left with a local $\mathrm{SO}(3)$ or $\mathrm{SU}(2)$ invariance corresponding to local rotation in the space manifold. The advantage of this though, is that we can now separate the relevant degrees of freedom following Holst. The Ashtekar-Barbero naturally appears as the canonical variables of Palatini's formulation of general relativity with the addition of the Holst term and the time-gauge. We can from here derive the Hamiltonian formulation which is strikingly similar to the previous formulation in terms of the Ashtekar connection:
\begin{equation}
  \begin{array}{rcl}
    G_I &=& D_a E_I^a \\
    C_a &=& E^b_I F^I_{ab} \\
    H &=& \frac{1}{2} \epsilon^{IJ}_{~~K} \frac{E^a_I E^b_J F^K_{ab}}{\sqrt{\det E}} +\frac{1+\beta^2}{\beta^2}\frac{(E^a_I E^b_J - E^a_J E^b_I)K^I_a K^J_b}{\sqrt{\det E}}
  \end{array}
\end{equation}
where $K$ is the extrinsic curvature and can be expressed in terms of the variables as: $\beta K^I_a = A^I_a - \Gamma^I_a$ where $\Gamma$ is the torsionless spin connection on the three dimensional space manifold associated to the triad. The added bonus of this formulation is of course the reality of the variables which imply that we do not need reality conditions. The obvious disadvantage is the non polynomial character of the Hamiltonian. We can still rejoice in the fairly simple writing of the local $\mathrm{SU}(2)$ gauge constraint and the spatial diffeomorphism constraint that we know can hope to solve. Let us note also that the gauge group is now compact and this will be one major point in the development of the Hilbert space and one which we help us evade the no-go theorem of Von Neumann. These variables can also be derived in another way: as we mentioned, they are a canonical transform of the variables of general relativity but only at the classical level \cite{Rovelli:1997na}.

There are other problems that plague this choice of variables. Most notably, the fact that we used the time gauge means that the spacetime interpretation of the Ashtekar-Barbero variables is not trivial. And indeed, we have been careful to call it variables and not connection because it does not transform as a spacetime connection \cite{Samuel:2000ue}. This situation might make it difficult to interpret the theory in a covariant way afterwards \cite{Geiller:2012dd}. Also, fixing the time-gauge before quantization might lead to an anomalous theory. This has lead some people to rethink the possibility of using complex variables and to tackle the problem of a covariant connection \cite{Livine2009,Geiller2011,Alexandrov2006,Alexandrov2002a,Geiller2011a} which usually implies the treatment of the reality conditions. The problem was also analyzed in simpler contexts as in 3d \cite{Achour:2013gga} or for the entropy of black holes \cite{Achour:2015xga}. But the Ashtekar-Barbero variables should not be dismissed to easily either. For instance, some work has been done in the spacetime interpretation of the Ashtekar-Barbero variables. It gradually seems that a meaningful notion of Lorentz symmetry can be restored \cite{Rovelli2011a}.

To us, it seems that Wieland has done some very promising work in this direction \cite{Wieland2012}. Indeed, he offered the idea of separating between the Holst-Immirzi parameter which appears before the Holst term in the action and between the Barbero-Immirzi parameter which appears in the choice of variable. In particular, if we do not use the time-gauge, it is quite natural to define the variables to be Ashtekar original connection, whatever the Immirzi-Holst parameter is. As it turns out, it seems that the obtained theory is quite close to the theory using Ashtekar-Barbero variables. From this perspective then, using Ashtekar-Barbero variables is just a clever choice revealing the physical role of the Immirzi-Holst parameter. So for what follows, we will keep the Ashtekar-Barbero variables, with real Immirzi parameter. And we will hope that we are indeed lucky and that we can fix the time gauge before quantizing.

%*****************************************

\section{Physical consequences of the Immirzi parameter}

Let us side back from the Ashtekar-Barbero variables discussion and let's concentrate on the Holst-Immirzi parameter that appears in the action. If it does not appear in the equations of motion, what is its physical consequences? Let us start by mentioning an important point: the Immirzi parameter can have consequences on the equations of motion. If we had fermions, as we have seen, torsion does not vanish on-shell. This means that the Immirzi parameter then does change the equation of motion. The influence can be shown to lead to a four fermions interaction \cite{Perez:2005pm,Freidel:2005sn}. This is however linked to the way the Immirzi parameter is introduced. We can for instance introduce non-minimal couplings for the fermions which compensate the four-point interaction \cite{Mercuri:2006um}. This is equivalent to using the Nieh-Yan topological invariant rather than the Holst term (as the Nieh-Yan term contains a contribution from torsion that compensate the contribution from the fermions) \cite{Mercuri:2006wb}. This is linked to the precise implementation of the Immirzi parameter. It might also be seen as the implementation of the dynamics of torsion.

Our concern, of course, is more general. Apart from this contribution from torsion (which might disappear depending on the use of the Nieh-Yan or Holst term), what are the physical consequences of the Immirzi parameter? In particular, let's concentrate on the cases where the equations of motion are untouched, the two canonical examples of that being by the use of the Nieh-Yan invariant or even the simpler case of pure general relativity with no matter and therefore with no torsion. Then the role of the Immirzi parameter seems to be purely quantum. Indeed, classically the Immirzi parameter has no influence because of the cancellation of torsion. But quantum mechanically, the torsion might vanish only in average. The remaining quantum fluctuations might then have a role and in that case, we also expect the Immirzi parameter to enter the physical theory. Still, this role must be, in some sense, tiny and should be negligible in the large scale semi-classical limit.

Sadly, a precise and total understanding of the role of the Immirzi parameter is still missing. If we have confidence in the results obtained in the time-gauge, then the Immirzi-parameter appears in the spectra of geometrical operators and determines the scale of quantum gravity. It was argued however that the Immirzi dependence might drop when considering Dirac observables and therefore disappear from physical predictions \cite{Dittrich:2007th}. The role of the Immirzi parameter also has been compared with the role of the $\theta$ parameter in \ac{QCD} \cite{Mercuri:2010yj}. Indeed, the Nieh-Yan term has a strikingly similar form as the correspond topological term of \ac{QCD}. They are both topological and are $CP$-odd. Therefore, the Immirzi parameter might also be linked to $CP$ violation. We should not haste in the direction though since no clear topological interpretation as transition amplitudes between vacua for example has been provided for the Immirzi parameter as there is for the \ac{QCD} $\theta$ term. We have ourselves suggested that the Immirzi parameter might correspond to a truncation of the phase space and therefore corresponds to a cut-off \cite{Charles:2015rda}. This would support the view of the Immirzi parameter as defining the scale of quantum gravity and might be relevant in the renormalization process \cite{Benedetti:2011yb,Benedetti:2011nd}.

In any case, though the role of the Immirzi parameter is still a bit mysterious, it is necessary to define the Ashtekar-Barbero variables and it naturally selects these variables when using the time-gauge. Therefore, for our concern of coarse-graining loop quantum gravity, we will gladly use it and leaves these questions for further inquiries.

\vspace{1em}

In this chapter, therefore, we have studied various formulation of general relativity. It turned out that using first-order formulations, a local symmetry is revealed, namely local Lorentz symmetry. This new formulation also allows for the use of a connection independent from the metric and even to see general relativity as a theory of connection rather than a metric theory. In particular, a nice choice of variables, the Ashtekar connection, leads to a polynomial theory which might be easier to quantize. Sadly, reality conditions need to be solved which are rather non-trivial at the quantum level. A new choice of variables was introduced then , allowing real variables and rather simple constraints, if we allow for a partial gauge fix. These new variables, the Ashtekar-Barbero variables, might have physical consequences but are retained thanks to their quantization friendly structure. It might even be that their physical consequences is what saves quantum gravity. In what follows, we will therefore see how we can quantize the kinematics of general relativity expressed as a theory of the Ashtekar-Barbero variables.

%*****************************************
%*****************************************
%*****************************************
%*****************************************
%*****************************************
 % Ashtekar-Barbero variables
%*****************************************
\chapter{Loop Quantization} \label{ch:LoopQuantize}
%*****************************************

\inspiquote{Come on, Rory! It isn't rocket science, it's just quantum physics!}{The Doctor}

Equipped with the new Ashtekar-Barbero variables, we can now turn to the quantization of the theory. Let us detail a bit the enterprise. We want indeed to quantize a constrained theory. More, we want to quantize a totally constrained theory. In essence, we want to find a representation of the basic variables and then solve all the constraints. We won't do this in one pass though. Indeed, there is first a natural splitting of the constraints into two categories: the gauge constraints and the (spacetime) diffeomorphisms constraints. These constraints naturally separate because the groups they generate commute with the other. There is also a natural second separation, as we will see, between the Hamiltonian constraint (temporal diffeomorphism) and the spatial diffeomorphisms constraints. This is natural because it corresponds to the splitting between the kinematics and the dynamics. Also, since the Hamiltonian constraint is so much more involved, it is natural to treat it separately.

So, the first step will be to define a Hilbert space $\mathcal{H}$. From there, we will impose the constraints step by step following this sequence of the Dirac canonical quantization programme \cite{dirac2001lectures,Henneaux1994,Matschull1996}:
\begin{equation*}
\mathcal{H} \xrightarrow{~Gauß~} \mathcal{H}_{gauge} \xrightarrow{~Diffeo~} \mathcal{H}_{diff} \xrightarrow{~Hamiltonian~} \mathcal{H}_{phys}
\end{equation*}
In this sequence, we mean that we start from the representation $\mathcal{H}$ of the variables. This representation is the representation of the loop algebra \cite{Gambini1996,Ashtekar:1993wf} which will be our basic algebra of observables for quantum gravity. Then we look for the kernel of the gauge constraints $G_a$. The Hilbert space that pops out is $\mathcal{H}_{gauge}$.\graffito{Of course, it is more involved mathematically. More often than not, the kernel of the constraints does not exists as a subspace and an appropriate use of distributions will be needed.} Once the gauge constraints are imposed, we can look for the kernel of the spatial diffeomorphism constraints $D_a$ in the gauge invariant space and call this new space $\mathcal{H}_{diff}$ sometimes called the \textit{kinematical} Hilbert space, as it corresponds to the space describing the kinematics of the theory. Therefore, once all this is done, we would have more or less solved the kinematics up to finding interesting observables on the last Hilbert space. The final step of the procedure corresponds to the imposition of the dynamics and once again corresponds to the search of a kernel. This last step is a bit more technical, and we will postpone its discussion to the next chapter. Indeed, we will even have a whole part dedicated to the dynamics of \ac{LQG}.

As a first step in this chapter, we will try to explore how we could possibly hope to solve the Gauß constraints. This will motivate the definition of our spaces. Let us note here that we won't exactly follow the mentioned sequence as it is not yet natural to do. But, after having found some natural gauge invariant wavefunctions, generalize them to some discrete version and defined a full continuum theory, we will be able to look back and see how the full (kinematical) sequence can be solved.

%*****************************************

\section{The loop excitation}

In this section, we will work at a very formal level.\graffito{Note that \textit{formal} in physicists language usually means non-rigorous. It means the same thing here.} We want to define a quantum theory of the Ashtekar-Barbero variables. For now, we will ignore the diffeomorphism constraints, either the spatial of temporal ones and concentrate on the definition of a space that would help us solve the Gauß constraints. The variables we have are a spatial connection and its conjugate momentum, the densitized triad. This should mean that we can consider the configuration space to be the space of $\mathrm{SU}(2)$ spatial connection. So at the quantum level, we are led to consider the space of wavefunctions over the spatial connection, that is functional of the form $\psi[A^i_a]$. We do not care for now about how to define the scalar product on this space. We will come back to this question later.

How then should we represent our variables? Once again, let's not dwell in mathematical subtleties for now\graffito{One should not think that mathematical subtleties are useless or unimportant. Indeed, they usually help a lot in understanding of how some problem arises or is resolved. But we should not \textit{start} with them in research.} (we will have plenty of time for that by the end of the chapter). We will simply represent the connection variable by multiplication, that is:
\begin{equation}
\left(\widehat{A^i_a(x)} \psi \right)[A] = A^i_a(x)\psi[A]
\end{equation}
and the conjugate momentum by (functional) derivation:
\begin{equation}
\left(\widehat{E^a_i(x)} \psi \right)[A] = -\mathrm{i} \frac{\delta \psi}{\delta A^i_a(x)}[A]
\end{equation}
We can now try to quantize the Gauß constraints. Let us restate the classical constraints:
\begin{equation}
G_I = D_a E^a_I = \partial_a E^a_I + \epsilon_{IJ}^{~~K} A^J_a E^a_K
\end{equation}
The simple course of action is of course to promote everything to operators and look at the result. There is however a more clever way to implement this constraint. Let us remark that we want these constraints to be implemented without anomalies, that is, we want the quantum bracket to reproduce exactly the classical brackets. Note, that the quantum brackets usually carry higher order terms compared to the classical Poisson brackets. This is where the choice of fundamental variables is important because their Poisson brackets will be the one reproduced exactly (with an $\hbar$ factor). But, once a given choice is made and once we try to implement this quantum mechanically, we want the brackets to be reproduced. This would mean that the operators do indeed act as the generators of a group. More precisely, they would act as the generators of the local $\mathrm{SU}(2)$ gauge transform. But we know how the connection transforms under gauge transform. So if the quantum theory reproduces in any plausible sense the symmetry, if we consider the gauge transformation $g$ and its action $\triangleright$ on the connection, we should have a unitary transform $U(g)$ such as:
\begin{equation}
\left(U(g)\psi\right)[A] = \psi[g\triangleright A]
\end{equation}
Then being in the kernel of the constraints would mean to be invariant under the action of $g$. This means, we are looking for wavefunctions $\psi$ that check the following condition:
\begin{equation}
\forall g,~\psi = U(g)\psi \Leftrightarrow \forall g,~\forall A,~\psi[A] = \psi[g\triangleright A]
\end{equation}
Or to put in more verbal terms, we are looking for wavefunctions that depend only on the orbit under gauge transformation of the connection. Do we know such functions? That is, do we know functions of the connection that are gauge invariant?

Indeed we do, and they come from classical theory. Any classical function of the connection that is gauge invariant will make a perfect candidate for a gauge invariant wavefunction. The simplest of them all is the holonomy. Let us consider a closed path $\mathcal{C}$ in space. It is natural to integrate the infinitesimal transformations encoded in the connection along this path. Mathematically, let's parametrize the path via a real number $s$ between $0$ and $1$ using the coordinates $x^a(s)$. Then, we can define a function $g$ of this real number $s$ valued in the group $\mathrm{SU}(2)$ satisfying the following equations:
\begin{equation}
  \left\{
  \begin{array}{rcl}
    g(0) &=& \mathbb{1} \\
    \frac{\mathrm{d}g}{\mathrm{d}s}(s) &=& \mathrm{i}\frac{\mathrm{d}x^a}{\mathrm{d}s}(s)A_a^i(x(s)) \frac{\sigma_i}{2} g(s)
  \end{array}
  \right.
\end{equation}
where the $\sigma$s are the Pauli matrices and therefore the $\frac{\sigma_i}{2}$ are the hermitian generators of the $\mathrm{SU}(2)$ group. This definition formalizes the following intuitive notion: when we start at the origin point $s=0$, the transformation is still the identity. Indeed, we haven't moved already. Then, for any infinitesimal displacement, the derivative of the coordinates gives the displacement vector. Contracted with the connection, this gives the infinitesimal transformation along the vector. Contracted with the generators, we get a lie algebra element which can be ``added'' to the transformation with the natural composition law on the group. This way, $g(t)$ represents the transformation in $\mathrm{SU}(2)$ when following the curve $\mathcal{C}$ from the point at $s=0$ to the point $s=t$.

The interesting point is that $g(1)$ is not necessarily equal to $\mathbb{1}$. Let us say this in another fashion: the infinitesimal transform along a closed path is not necessarily trivial. This is precisely a mark of curvature: a space is curved if one closed holonomy (that is the transformation that results from the integration around a closed loop) is non trivial. So, we have a non-trivial quantity $g(1)$ which depends on the connection. And this quantity is interesting from several aspects. If we backpedal a bit, and consider a generic $g(s)$ (so not necessarily on a closed loop), we see that the element $g(s)$ has nice transformation properties. Let us consider a (finite) gauge transform $h(y)$. At each point $y$, we rotate via the $\mathrm{SU}(2)$ element $h(y)$. Then, $g(s)$ transforms as follows:
\begin{equation}
g(s) \rightarrow h(x(s))^{-1} g(s) h(x(0)) 
\end{equation}
That means that the transformation of $g(s)$ only depends on the initial and final point of integration. In particular, if we consider $g(1)$ which is defined on a closed loop the transformation reads:
\begin{equation}
g(1) \rightarrow h(x(1))^{-1} g(1) h(x(0)) = h(x(0))^{-1} g(1) h(x(0)) 
\end{equation}
So that $g(1)$ is just conjugated when applying a local gauge transform.

Back to our quest of gauge invariant functions: if we can find a function on the group that depends only on the conjugation class, by composition we will have a gauge invariant function of the connection. And there is a very simple such function: the trace. We now have a recipe for constructing a certain class of gauge invariant functions over the space of connections: let's choose a closed loop $\mathcal{C}$. The trace of the holonomy of the connection around this loop is a non-trivial gauge invariant function of the connection. With our previous notation, it reads:
\begin{equation}
f[A] = \mathop{Tr} g(1)
\end{equation}
where it is understood that $g$ depends on $A$ and the path $\mathcal{C}$. Note here that we did not specify the starting point on the loop, thought it is needed in the definition of $g$. This is because changing the starting point will just amount to a conjugation and therefore the dependence will drop off thanks to the gauge invariance condition.

We should pause for a second and dwell a bit about the various meaning of \textit{holonomy}. As we introduced it, it is a natural observable of the \textit{classical} theory. Disregarding the problem of the scalar product, how does this function come to play any role in the quantum theory? How is it that this function is not promoted to an operator for instance? Let us disentangle all this now. The function we just defined, that is the trace of a holonomy, is a mathematically well-defined function over the space of connection on a given manifold. Mathematically, it is just one object. It appears, however, in two different physical contexts. It is used first in a classical context, and in that context it is a function of the configuration space but is also a genuine observable. It can also be used in a quantum context where it is to be interpreted as a wavefunction.

\begin{figure}[h!]
  \centering
  \begin{tikzpicture}[scale=1.5]
    \draw[dashed] (0,0) -- (30:1);
    \draw[dashed] (0,0) -- (0,1);
    \draw (0,0) circle (1);
    \draw[red] (30:1) node {$\bullet$};
    \draw (30:1.2) ++(0.4,0) node {particle};
    \draw[->,>=stealth,blue] (0,0.5) arc (90:30:0.5);
    \draw[blue] (60:0.7) node {$\theta$};
  \end{tikzpicture}
  \label{fig:particle_on_circle}
  \caption{We consider the simple example of a particle (in red) moving on a circle. The position is represented by the (oriented) angle $\theta$.}
\end{figure}
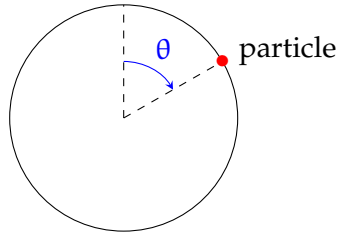

We have a similar duality in simpler context. Let us consider, for instance, a particle moving on a circle (see figure \ref{fig:particle_on_circle}). Let us now consider the function:
\begin{equation}
f(\theta) = e^{\mathrm{i}\theta}
\end{equation}
where $\theta$ represents the position of the particle. It can be interpreted as a function of the variables in the classical theory. Or it can be understood as a wavefunction and as such would represent an eigenstate of the momentum observable. Incidentally, the trace of the holonomy has a very similar role as a wavefunction: it is an eigenstate of the momenta, which are the densitized triad, and therefore represents eigenstates of the induced metric. Now of course, for the observable, there should also be a quantum equivalent which would be a quantum operator. This is for instance the case in the case of a particle on a circle. $f(\theta)$ can be understood as a classical observable but can also be understood as the corresponding quantum observable $\widehat{f(\theta)}$. It acts as expected by multiplication on the wavefunctions:
\begin{equation}
\left(\widehat{f(\theta)} \psi\right) (\theta) = f(\theta) \psi(\theta)
\end{equation}
In a very similar manner, there is a quantum operator for the trace of an holonomy. So, we have three versions of the holonomy. The first two are mathematically identical and differ only in their physical interpretations. The former is a classical observable, depending only on the configuration variables, but the latter is interpreted as a wavefunction. The third holonomy is the quantum operator that corresponds to the classical observable. But this one is mathematically different, though its physical interpretation is similar to the classical observable. These three versions of holonomy should not, of course, be confused.

%*****************************************

\section{Spin networks}

Now, we have a set of interesting candidates for wavefunctions. We would probably want to add some interesting properties to our space. First, we should of course allow for superpositions of states, therefore making it a vector space. Also, as the eluded holonomy operator will act by multiplication, we should want the space to be stable under (pointwise) multiplication. This would actually be a good strategy and would lead to \textit{loop} quantization \textit{per se}.\graffito{The name \textit{loop} comes from historic reasons as the quantization would rely on Wilson loops. \textit{Polymer} quantization or \textit{spin network} quantization might better describe the modern process.} There is however a more modern way of constructing the space, and for this, we must first generalize the loop states. Indeed, to generate the whole space, we will need products and sums of loops. But it can become quite cumbersome to keep track of all possible multiplications. We will therefore look for a generalization of loop states that appear after finite multiplication and additions.

The relevant property of the loop states is their one-dimensional aspect. Indeed, if we stare at the definition above long enough, we will see that the wavefunction depends only on the value of the connection along a one-dimensional manifold, namely the closed path $\mathcal{C}$. If we multiply or add loop states, we will never increase the dimensionality of the manifold. To be precise, if there are crossing, the support line will not be a manifold anymore, but will become a graph. This is however all that can happen. So, we are now looking for wavefunctions representing one-dimensional excitations of the connection. These excitations will therefore live on \textit{graphs}. To be precise mathematically, a graph is a collection of points and a collection of edges, with two additional functions from the edges to the points called \textit{source} and \textit{target} describing the connectivity of the graph (see figure \ref{fig:graph_def})). This formal definition will be handy in a moment, but for now, we will rather consider \textit{embedded} graphs. They are still collections of points and edges, but the points are points in a manifold and the edges are paths in the manifold linking these points. With our previous construction, they are what naturally appears when considering a set of loops.

\begin{figure}[h!]
  \centering
  \begin{tikzpicture}[scale=1]
    \coordinate(A) at (0,0);
    \coordinate(B) at (2,0);
    \coordinate(C) at (2,-2);
    \coordinate(D) at (0,-2);

    \draw[red] (A) to[bend left] node[midway,sloped]{$>$} node[midway,above]{$e$} (B);
    \draw (B) to[bend left] node[midway,sloped]{$>$} (C);
    \draw (C) to[bend left] node[midway,sloped]{$>$} (D);
    \draw (D) to[bend right] node[midway,sloped]{$>$} (A);

    \draw (A) to[bend left] node[midway,sloped]{$>$} ++(-1,0.5);
    \draw (B) to[bend right] node[midway,sloped]{$>$} ++(1,0.5);
    \draw (C) to[bend right] node[midway,sloped]{$>$} ++(1,-0.5);
    \draw (D) to[bend left] node[midway,sloped]{$>$} ++(-1,-0.5);

    \draw (A) node {$\bullet$} node[above]{$s(e)$};
    \draw (B) node {$\bullet$} node[above]{$t(e)$};
    \draw (C) node {$\bullet$};
    \draw (D) node {$\bullet$};
  \end{tikzpicture}
  \caption{A graph is a collection of links and vertices. Each link is oriented and has a source and a target vertex. For the edge $e$ for instance (in red), the source vertex is $s(e)$ and the target vertex $t(e)$.}
  \label{fig:graph_def}
\end{figure}
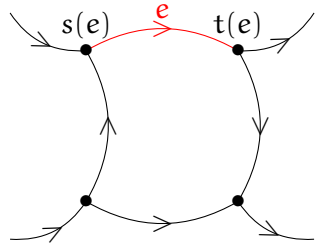

So if the loop states are wavefunctions representing excitations of the connection on a given loop, we will now consider graph states, wavefunctions representing excitations of the connection on a given graph (embedded in the space manifold). Said in yet another way, we will consider wavefunctions with support on a graph. So, let's take a graph $\Gamma$ embedded in the space manifold $\Sigma$ and let's see how we can construct wavefunctions around it. The concept of (open) holonomies will be central: for each edge, we can set an orientation and compute the open holonomy\graffito{Normally, the term \textit{holonomy} should only be used for closed paths. But because it is rather natural to consider open paths as well, we usually allow ourselves, in the quantum gravity literature, to talk about open holonomies.} along the path as advocated above. This will give a group element associated to each edge corresponding to the finite parallel transport between its source and its target. Of course, if we were to assign the reverse orientation on an edge, we would get the inverse group element. So, our graph is now colored with one group element per edge representing parallel transport (see figure \ref{fig:graph_coloring}). We should note here that this group element is not fixed but is actually a function of the connection, as was the group element for the closed holonomy. We are therefore following the exact same recipe: we start by defining group elements that are function of the connection and then we will look for functions over these group elements that are gauge invariant.

\begin{figure}[h!]
  \centering
  \begin{tikzpicture}[scale=1]
    \coordinate(A) at (0,0);
    \coordinate(B) at (2,0);
    \coordinate(C) at (2,-2);
    \coordinate(D) at (0,-2);

    \draw (A) to[bend left] node[midway,sloped]{$>$} node[midway,above]{$g_1$} (B);
    \draw (B) to[bend left] node[midway,sloped]{$>$} node[midway,right]{$g_2$} (C);
    \draw (C) to[bend left] node[midway,sloped]{$>$} node[midway,above]{$g_3$} (D);
    \draw (D) to[bend right] node[midway,sloped]{$>$} node[midway,left]{$g_4$} (A);

    \draw (A) to[bend left] node[midway,sloped]{$>$} node[midway,below]{$g_5$} ++(-1,0.5);
    \draw (B) to[bend right] node[midway,sloped]{$>$} node[midway,above]{$g_6$} ++(1,0.5);
    \draw (C) to[bend right] node[midway,sloped]{$>$} node[midway,above right]{$g_7$} ++(1,-0.5);
    \draw (D) to[bend left] node[midway,sloped]{$>$} node[midway,above]{$g_8$} ++(-1,-0.5);

    \draw (A) node {$\bullet$};
    \draw (B) node {$\bullet$};
    \draw (C) node {$\bullet$};
    \draw (D) node {$\bullet$};
  \end{tikzpicture}
  \caption{We color the edges of the graph with group elements (noted $g_i$ on the figure). The group elements come from the (open) holonomies of the connection along the edge of the graph.}
  \label{fig:graph_coloring}
\end{figure}
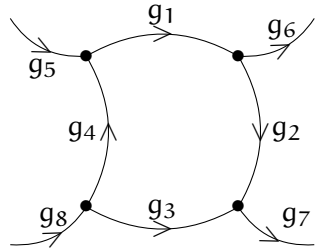

Let us number each edge of our graph $\Gamma$, therefore naming them $e_i$, and then also number the group elements $g_i$ accordingly. If the graph has $k$ edges, then we have $k$ group elements from $g_1$ to $g_k$. We are looking for a function over the group elements that end up in $\mathbb{C}$ that is a function of the form:
\begin{equation}
  \begin{array}{rcl}
    f : \mathrm{SU}(2)^k &\xrightarrow{L_2}& \mathbb{C} \\
    (g_1,...,g_k) &\mapsto& f(g_1,...,g_k)
  \end{array}
\end{equation}
Note that the functions are required to be square integrable with respect to some scalar product that we will make more precise in a moment. The goal is of course to have a respectable Hilbert space structure for our quantum theory. We will come back to this point later on. If we do not add any further requirement on the function, this might be a very well defined wavefunction. Indeed, let's define $\psi$ as:
\begin{equation}
  \psi[A] = f(g_1[A],...,g_k[A])
\end{equation}
This is indeed a function of the connection that depends only on its value along a given graph, specifically $\Gamma$. Of course, we have no reason to believe that it is gauge invariant (we can easily find counter examples in the case $k=1$). But this is still a honest wavefunction of the connection.

The interesting point at this stage is that, because the space of square integrable function over $\mathrm{SU}(2)^k$ is a natural Hilbert space, if we restrict to functions that depend on the connection only through the group elements $g_1$ to $g_k$, this set will inherit the Hilbert space structure. Let us describe the original Hilbert space. For any compact group $G$\graffito{The compactness of the group $G$ is mathematically very important. The LOST theorem relies on it \cite{Lewandowski:2005jk} and no coherent quantization schemes is known for non-compact group for now.}, we can define a measure on the group, that is invariant under left and right composition. More precisely, it means that we can define a measure, let's call it $\mathrm{d}g$ such that:
\begin{equation}
  \forall f,~\forall h\in G,~ \int f(g)\mathrm{d}g = \int f(hg) \mathrm{d}g = \int f(gh) \mathrm{d}g
\end{equation}
With the additional requirement that this measure is a measure on the Borel subsets, it is unique up to a global factor, which can be set by requiring $\int \mathrm{d}g = 1$, which means that the volume of the group is unity. This measure is called the (left and right) Haar measure. It exists in some wider context than compact groups. The simplest example would be the Lebesgue measure on the group $\mathbb{R}$, which is indeed invariant under translation (left and right can not be distinguished for an abelian group). Equipped with this measure, we can define a scalar on the functions over the group $G$:
\begin{equation}
  \forall (f_1,f_2),~\langle f_1 | f_2 \rangle = \int \overline{f_1}(g)f_2(g) \mathrm{d}g 
\end{equation}
The space of square integrable functions (which are identified when equal almost everywhere) on the group $G$ is then a Hilbert space. We can use this scalar product for our wavefunctions defined over graph.

Let us now look for a basis of this Hilbert space. How are we to find one? If the group was $\mathrm{U}(1)$ (which it isn't), the functions would be functions over multiple copies of $\mathrm{U}(1)$ and could therefore be interpreted as periodic functions on $\mathbb{R}^k$. The natural way to go would be Fourier analysis, which would exactly provide an orthogonal basis given by the exponentials. But here, the group is not $\mathrm{U}(1)$. In the case of quantum gravity, the group is $\mathrm{SU}(2)$. We should look for an equivalent of the Fourier decomposition on the $\mathrm{SU}(2)$ group (or more generally a compact group). And there is: it is given by the Peter-Weyl decomposition from the Peter-Weyl theorem. It says that the elements of matrices of the irreducible representations of the group are an orthogonal basis of the functions over the group. So, for $\mathrm{SU}(2)$, the irreducible representations are labelled by half-integers $j\in\frac{\mathbb{N}}{2}$. An element $g\in\mathrm{SU}(2)$ is represented by a matrix $D^j(g)$ in the representation $j$. This matrix has elements labelled by row and column $D^j(g)_{mn}$ which form the basis, so that any function over $\mathrm{SU}(2)$ valued in $\mathbb{C}$ can be written:
\begin{equation}
  f(g) = \sum_{j,m,n} f_j^{mn} D^j(g)_{mn}
\end{equation}
If we think about it, this is exactly Fourier transform when the group is abelian. For $\mathrm{U}(1)$ for instance, the irreducible representation are one-dimensional and are labelled by integers $n\in \mathbb{Z}$ corresponding to the power of the element. A matrix in such a representation is just a complex number and turns out to be an exponential since the representation is unitary. This means that function over $\mathrm{U}(1)$ can be expanded as:
\begin{equation}
  f(\theta) = \sum_{n} f_n \mathrm{e}^{\mathrm{i}n\theta}
\end{equation}
which is precisely Fourier transform. So now that we have an equivalent of the Fourier transform, and by that, an orthogonal basis (which we will be able to normalize), we can study a bit more the space. The group is no longer $\mathrm{SU}(2)$ but $\mathrm{SU}(2)^k$, this won't stop us of course, and the development is very similar. We write our functions over $g_1$ to $g_k$ as:
\begin{equation}
  f(g_1,...,g_k) = \sum_{\{j_i,m_i,n_i\}} f_{j_1,...,j_k}^{m_1,n_1,...,m_k,n_k} D^{j_1}(g_1)_{m_1 n_1} ... D^{j_k}(g_k)_{m_k n_k}
\end{equation}
The basis will therefore be the functions labelled by the spins $j$ (labelling the representations) and the indices $m$ and $n$ (labelling the matrix elements). They can be represented by \textit{colored} graphs: each edge of the graph carries three colors, the spin $j$ and the indices $m$ and $n$. These colors describe a function from the Peter-Weyl basis as follows: it defines completely the dependence over the holonomy along the corresponding edge (see figure \ref{fig:PeterWeylDecomp}).

\begin{figure}[h!]
  \centering
  \begin{tikzpicture}[scale=1]
    \coordinate(A) at (0,0);
    \coordinate(B) at (2,0);
    \coordinate(C) at (2,-2);
    \coordinate(D) at (0,-2);

    \draw (A) to[bend left] node[midway,sloped]{$>$} node[very near start,above]{\tiny $m_1$} node[midway,above]{$j_1$} node[very near end,above]{\tiny $n_1$}  (B);
    \draw (B) to[bend left] node[midway,sloped]{$>$} node[very near start,right]{\tiny $m_2$} node[midway,right]{$j_2$} node[very near end,right]{\tiny $n_2$}  (C);
    \draw (C) to[bend left] node[midway,sloped]{$>$} node[very near start,above]{\tiny $n_3$} node[midway,above]{$j_3$} node[very near end,above]{\tiny $m_3$} (D);
    \draw (D) to[bend right] node[midway,sloped]{$>$} node[very near start,left]{\tiny $m_4$} node[midway,left]{$j_4$} node[very near end,left]{\tiny $n_4$} (A);

    \draw (A) to[bend left] node[midway,sloped]{$>$} node[very near start, above]{\tiny $n_5$} node[midway,above]{$j_5$} node [very near end, above]{\tiny $m_5$} ++(-1,0.5);
    \draw (B) to[bend right] node[midway,sloped]{$>$} node[very near start, above]{\tiny $m_6$} node[midway,above]{$j_6$} node [very near end, above right]{\tiny $n_6$} ++(1,0.5);
    \draw (C) to[bend right] node[midway,sloped]{$>$} node[very near start, below]{\tiny $m_7~~$} node[midway,below]{$j_7$} node[very near end, below right]{\tiny $n_7$} ++(1,-0.5);
    \draw (D) to[bend left] node[midway,sloped]{$>$} node[very near start, below]{\tiny $~~n_8$} node[midway,below right]{$j_8$} node[very near end, below]{\tiny $m_8$} ++(-1,-0.5);

    \draw (A) node {$\bullet$};
    \draw (B) node {$\bullet$};
    \draw (C) node {$\bullet$};
    \draw (D) node {$\bullet$};

    \draw (6,-1) node {$=\qquad\prod_{i=1}^8 D_{j_i}(g_i)_{m_i n_i}$};
  \end{tikzpicture} 
  \caption{The Peter-Weyl basis is labelled by colored graph. Each edge of the graph is labelled by a spin and each half-edge is labelled by an index. Here the $8$ edges are colored by the spins for $j_1$ to $j_8$. The half-edge connected to the source vertex of the complete edge are colored by $m_1$ to $m_8$. Similarly, the half-edges connected to the target vertex are colored by $n_1$ to $n_8$. The graph given in the figure corresponds to the product of Wigner matrix elements given on the right.}
  \label{fig:PeterWeylDecomp}
\end{figure}
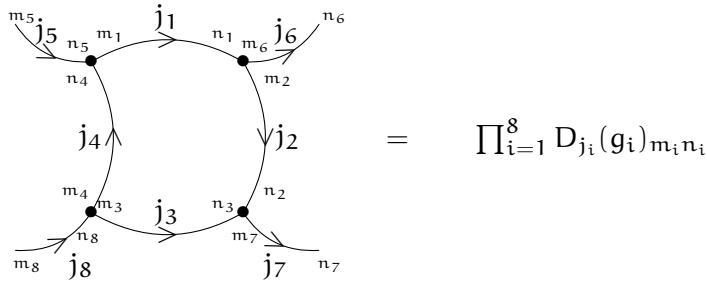

What is the physical or geometrical interpretation of such a wavefunction? The intuition should be alerted by the words of \textit{Fourier transform}. Indeed, usually when we start from a position representation and use the Fourier basis, we get eigenvectors of the momenta (indeed, the Fourier basis can be defined in this way). Here we start from the configuration space of the connection and so, we should expect the Peter-Weyl basis to diagonalize in some sense the operators corresponding to the densitized triad. And indeed they do. Now, the precise mathematical sense will be explored by the end of the chapter. For now, we just have the problem of defining operators that correspond to the triad. Let us just work informally and see where it goes. The densitized triad operators act by derivation. From the Peter-Weyl decomposition, they act be inserting a $\sigma$ (the generator of the group) in the holonomy, at the point on the link where the operator acts. In particular, we insert a $\sigma$ at the left or at the right of the holonomy if we consider the densitized triad at the beginning or at the end of the edge. Then we see that the densitized triad becomes a quantum vector with its behavior dictated by the representation theory of $\mathrm{SU}(2)$ as an angular momentum vector. The length of the densitized triad (whether at the source or the target of a link) is given by the spin of the representation and the indices give the projection of the vector on the $z$ axis, the first index $m$ corresponds to the source while the second index $n$ corresponds to the target. This finally means that these states correspond to eigenstates of the length and the $z$ component of the densitized triad at the source and targets of the links. By construction, the length of both these vectors are always the same (this comes from the fact that they are image from one another by parallel transport).

The interest of the loop excitations was their gauge invariance. This is a property we totally overlooked so far in the case of graphs excitations. How are we to implement this? In the case of loops, this was done by a clever choice of a function from the group element into $\mathbb{C}$ that had to be compatible with gauge transform. We can do exactly the same here. When we do a gauge transform, it acts only at the starting and end points of the holonomies (as mentioned above). Let us write the gauge transform $h(x)$ where $x$ is a point in $\Sigma$. Then the holonomies transform as:
\begin{equation}
  g_i \rightarrow h(t_i)^{-1} g_i h(s_i)
\end{equation}
where $s_i$ is the source of the link $e_i$ and $t_i$ is the target of the same link. Then, to have a gauge invariant function is to have a function which is invariant under such transformation. So we are looking for a function $f$ over $\mathrm{SU}(2)^k$ such that:
\begin{equation}
  f(\{g_i\}) = f(\{h(t_i)^{-1} g_i h(s_i)\})
\end{equation}
From the Peter-Weyl decomposition point of view, this means that the indices of the Wigner matrices must be contracted through a gauge-invariant tensor, also called an intertwiner. The new natural basis then still has spins on the links of the graphs but no longer indices as we only need to know how to glue them together. This is given by the intertwiner at each node of the graph. If we choose a basis for the intertwiners (which depend on the spins on the edges), the gauge invariant basis is labelled by the coloring of the graphs: a spin at each edge and a basis vector of the intertwiner space at each node, the intertwiner space being the one of the corresponding representations of the edges meeting at the node (see figure \ref{fig:spinnetwork_def}). Of course, all this development, while a bit more complicated, could be carried for loops. Indeed, a loop is a special kind of graph with only one point (the root point) and only one link with the source and target being the root point. The holonomy along this link then does transform by conjugation. The trace is just a special kind of gauge invariant function: it is the contraction between the fundamental representation and the only intertwiner between the two $\frac{1}{2}$ representations. It can easily be seen that a natural generalization exists to any representations for the loops and are given by traces in other representations. They would be obtained by successive multiplication in the previous approach.

\begin{figure}[h!]
  \centering

  \begin{tikzpicture}[scale=1]
    \coordinate(A) at (0,0);
    \coordinate(B) at (2,0);
    \coordinate(C) at (2,-2);
    \coordinate(D) at (0,-2);

    \draw (A) to[bend left] node[midway,sloped]{$>$} node[very near start,above,gray]{\tiny $(m_1)$} node[midway,above]{$j_1$} (B);
    \draw (B) to[bend left] node[midway,sloped]{$>$} node[midway,right]{$j_2$} (C);
    \draw (C) to[bend left] node[midway,sloped]{$>$} node[midway,above]{$j_3$} (D);
    \draw (D) to[bend right] node[midway,sloped]{$>$} node[midway,left]{$j_4$} node[very near end,left,gray]{\tiny $(n_4)$} (A);

    \draw (A) to[bend left] node[midway,sloped]{$>$} node[very near start, above,gray]{\tiny $(n_5)$} node[midway,above]{$j_5$} node [very near end, above]{\tiny $m_5$} ++(-1,0.5);
    \draw (B) to[bend right] node[midway,sloped]{$>$} node[midway,above]{$j_6$} node [very near end, above right]{\tiny $n_6$} ++(1,0.5);
    \draw (C) to[bend right] node[midway,sloped]{$>$} node[midway,below]{$j_7$} node[very near end, below right]{\tiny $n_7$} ++(1,-0.5);
    \draw (D) to[bend left] node[midway,sloped]{$>$} node[midway,below right]{$j_8$} node[very near end, below]{\tiny $m_8$} ++(-1,-0.5);

    \draw (A) node {$\bullet$} ++(-30:0.5) node{$i_1$};
    \draw (B) node {$\bullet$} node[above]{$i_2$};
    \draw (C) node {$\bullet$} node[right]{$i_3$};
    \draw (D) node {$\bullet$} node[left]{$i_4$};
  \end{tikzpicture} 
  
  \caption{The spin network basis of gauge invariant states is labelled by graphs colored by spins and intertwiners. On this graph, the labels $m_i$ and $n_i$ have been omitted because they are implicitly contracted over at each vertex with the intertwiner. For instance, the intertwiner $i_1$ carries $3$ indices which are contracted with $n_5$, $m_1$ and $n_4$ (in parenthesis and in gray). The indices $m_5$, $n_6$, $n_7$ and $m_8$ have been kept precisely because they should be summed over with the rest of the graph which does not appear on the figure.}
  \label{fig:spinnetwork_def}
\end{figure}
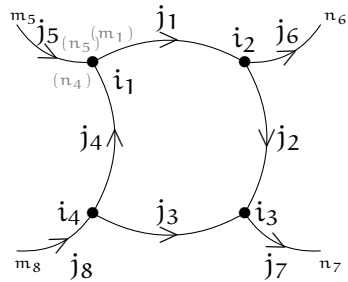

The space just described corresponds only to the excitations along a given graph. So, we will need to go further and develop a \textit{continuum} theory. This is what we will do in the next section. Until then, the theory so far resembles a discrete theory and in some ways is very similar to lattice Yang-Mills theory \cite{Baez1996,Thiemann2000} (see figure \ref{fig:lattice_YM}). Let us compare the two for a moment. It should be noted that lattice Yang-Mills has exactly the same Hilbert space as the one just described, except maybe that the graph might be infinite. Still, the connection variable is usually replaced in the discrete theory by group elements and the generalized electric field has value at the end of the links, with the constraint that it should the same norm at both ends. These similarities come from the similarities of the variables. It goes even further when we study the gauge invariant functions since the basis we just developed is also a natural basis for Yang-Mills theory since it diagonalizes the electric field. This similarity comes from the similarity in the constraint algebra, since in both cases, we are dealing with Gauß constraints. What are the differences then? The first point, on the mathematical side, is that in the quantum gravity, we will develop a continuum theory. That might be difficult for a Yang-Mills theory because of the lack of diffeomorphism invariance. But the second point comes from the interpretation. In a Yang-Mills theory, the graph carries some information about space and distances. The graph represents a sampling of space and the fields leaves on space. In the quantum gravity theory, the graph represents connectivity information (as will become apparent when we'll have dealt with diffeomorphism invariance). The graph does not by itself carry distances or geometrical information. The distances are carried by the states themselves. This was of course obvious from the beginning: if they have a somewhat similar aspect, $\mathrm{SU}(2)$ Yang-Mills and quantum gravity must be different. They are not with respect to the variables algebra and from the Gauß constraint perspective. Therefore the difference between the two theories will be encoded in the precise dynamics, but also in the constraints.

\begin{figure}[h!]
  \centering

  \begin{tikzpicture}[scale=1]
    \coordinate(A) at (0,0);
    \coordinate(B) at (2,0);
    \coordinate(C) at (4,0);
    \coordinate(D) at (0,-2);
    \coordinate(E) at (2,-2);
    \coordinate(F) at (4,-2);

    \draw (A) ++(-1,0) -- (A) -- node[midway,above]{$j_1$} (B) -- node[midway,above]{$j_2$} (C) -- ++(1,0);
    \draw (D) ++(-1,0) -- (D) -- node[midway,above]{$j_3$} (E) -- node[midway,above]{$j_4$} (F) -- ++(1,0);
    \draw (A) ++ (0,1) -- (A) -- node[midway,left]{$j_5$} (D) -- ++(0,-1);
    \draw (B) ++ (0,1) -- (B) -- node[midway,left]{$j_6$} (E) -- ++(0,-1);
    \draw (C) ++ (0,1) -- (C) -- node[midway,left]{$j_7$} (F) -- ++(0,-1);
    
    \draw (A) node {$\bullet$} node[below right] {$i_1$};
    \draw (B) node {$\bullet$} node[below right] {$i_2$};
    \draw (C) node {$\bullet$} node[below right] {$i_3$};
    \draw (D) node {$\bullet$} node[below right] {$i_4$};
    \draw (E) node {$\bullet$} node[below right] {$i_5$};
    \draw (F) node {$\bullet$} node[below right] {$i_6$};
  \end{tikzpicture} 
  
  \caption{The spin network basis of gauge invariant states is also natural for lattice Yang-Mills. Edges are colored by irreducible representations of the gauge group, which are spins in the case of $\mathrm{SU}(2)$, and the vertices by intertwiners of these representations. Of course, there is a huge difference in the interpretation of those states, compared to quantum gravity, as these spins and intertwiners do not encode physical distance, which is in fact represented by the graph itself in the case of lattice Yang-Mills.}
  \label{fig:lattice_YM}
\end{figure}
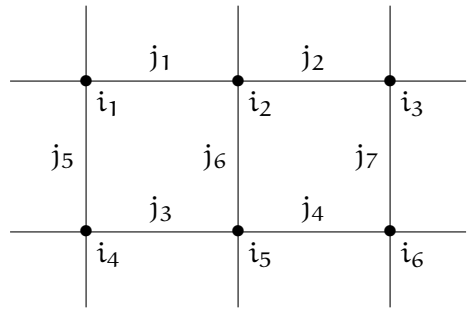

%*****************************************

\section{The continuum phase space}

Until then, we forgot the elephant in the room: what about a Hilbert space for the continuum space? Indeed, the Hilbert space we just defined is for a given support graph and therefore corresponds to some intuitive notion of discrete excitations. How do we start from here in order to go to the continuum Hilbert space. We will do this in a manner similar to the usual one in discrete theories: we will consider refinement and somehow take a limit. In this section, we will describe how to do such a process.

The main idea is the following: instead of capturing a continuum theory directly, we will rather capture an arbitrarily refined theory. Indeed, each state in a given discrete state space can be interpreted as a state of more refined graph. So, if we call the previous Hilbert space on the graph $\Gamma$ $\mathcal{H}_\Gamma$, our goal is to define a Hilbert space like:
\begin{equation}
\mathcal{H} = \left(\bigcup_{\Gamma} \mathcal{H}_\Gamma \right)\Bigg{/}\sim
\end{equation}
where the union is over all possible graphs and the quotient is done with respect to an equivalence relation which will give a precise notion to the inclusion of Hilbert space. In essence therefore, we will identify states in coarse graphs with their corresponding versions in the refined graphs. Intuitively at least, this should work, since if we are allowed any refinement we could consider any excitations as long as it is finite in the sense that it has support on some graph. Of course, this will only work because the support is not fixed, meaning that it depends on the precise excitation but also that it can be as refined as we want.

So as a first step, let's first see how the previous construction naturally gives a way to include coarse states into refined graphs. Let us consider some graph $\Gamma$ and a wavefunction $\psi$ with support on this graph. Let us now consider, as an example, a link $\ell$ not present in the graph, with the only requirement that it is not in the graph $\Gamma$ but its starting and ending point to land in $\Gamma$ (not necessarily on nodes of $\Gamma$). So $\ell$ is a genuine new link with respect to $\Gamma$ (see figure \ref{fig:adding_link}). \graffito{The existence of constant functions is what makes the compactness of the gauge group needed in the construction. Indeed, constant functions are not square integrable for non-compact groups at least for the Haar measure.}We could consider the new graph $\tilde{\Gamma} = \Gamma \cup \ell$ with the added link which is, arguably, more refined than $\Gamma$. Can we understand $\psi$ as a function with support on $\tilde{\Gamma}$? Yes indeed! Indeed, nothing prevents a function over $\tilde{\Gamma}$ to not depend on the particular holonomy along $\ell$. Constant functions are still allowed. And this is the basic idea on how to interpret a wavefunction as a wavefunction on a finer graphs than its original support: we only need to consider the function as constants on added elements. There is a subtlety we ignored here: if the link $\ell$ starts or ends on a link of $\Gamma$ rather than a node, that means that $\tilde{\Gamma}$ has an additional node and one of the link of $\Gamma$ must be split. \graffito{It should be noted here that graph excitations, in light of this refining procedure, are gauge-invariant almost everywhere except at the nodes of the graph.}Let us call the split link $m$ and the two parts $m_1$ and $m_2$. This means that the wavefunction no longer depends on a single holonomy along $m$. Therefore, gauge invariance will be automatically satisfied at the added node. It must rather depend on two holonomies along $m_1$ and $m_2$. How is this not a problem? We can define the dependence on $m$ as being a dependence of the composition $m_2 m_1$, following the definition of the holonomy along $m$. This also deals with another way to refine a graph: we can simply add a node in the middle of a link, without adding any link. This is dealt by the previous comment on composition.

\begin{figure}[h!]
  \centering

  \begin{tikzpicture}
    \coordinate(A) at (0,0);

    \draw (A) -- ++(2,0);
    \draw (A) -- ++(-2,0);
    \draw (A) -- ++(0,2);
    \draw (A) -- ++(0,-2);

    \draw (A) node{$\bullet$} ++(45:1) node{$\Gamma$};

    \draw[dashed,red] (-1,0) node{$\bullet$} -- node[midway,below left]{$\ell$} (0,-1) node{$\bullet$};
  \end{tikzpicture}
  
  \caption{To define the continuum limit, we need to define the \textit{cylindrical consistency}. In essence, we need to define an inclusion of Hilbert spaces at a finite number of excitations. Here, we show one of the elementary step to define the inclusion: the adjunction of a single link to a graph. The new link $\ell$ must start and finish on $\Gamma$ introducing new vertices if necessary.}
  \label{fig:adding_link}
\end{figure}
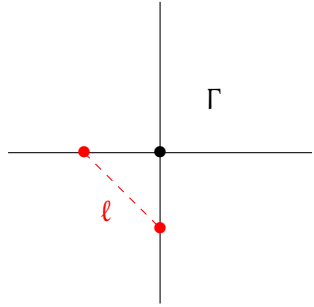

From there, we will use the mathematical construct of the projective limit. The projective limit is precisely the construction corresponding to the construction of the limit of ascending unions. For this, we must construct the equivalence relation giving a sense to the ascending part of the previously mentioned union. In order to have the equivalence relation, we will first define the inclusion structure of the graphs. Let us explain the idea in a simpler setting. Let us imagine we wanted to define a function over an infinite number of (real) variables. So we want to give a precise sense to $f(x_1,x_2,...)$ where there is an infinite (let's say countable) number of variables. Of course, the problem is not to define such a function \textit{per se} but to define it with a reasonable sense of differentiability, of scalar products and properties alike. The idea of the projective limit is to define such notions on a finite number of variables and then by an equivalence relation uphold this to an infinite number of them. In a very similar way to the construction presented above, a function over $x_1$ can be interpreted as a function over $x_1$ and $x_2$ but with no dependence on $x_2$. Let us note here that we \textit{must} distinguish between functions of one-variable but different variables. For instance, the Hilbert space of functions over $x_1$ is not the Hilbert space of functions over $x_{42}$, though they are isomorphic. It is of course similar for functions over several variables since, as functions over the infinite set of variables, these functions are different. And because we must distinguish between them, we must take care of the inclusion of the different Hilbert space among themselves.

In this simple case, the inclusion structure is quite simply the natural inclusion of sets. We will identify functions over different subsets of the integer in the following way. Given two functions, we look for a refinement of \textit{both} support sets. Then for any variables labelled by this set, we will define the pull-back of these variables onto smaller subsets of the integer (which is simply given by dropping additional variables). We can then compare the two functions over any value labelled by the bigger set. If they match, we consider them to be equivalent. The idea is quite intuitive: considering the pull-back means we drop the dependence on the added variables. So if we take two functions, we are actually asking that they depend only on the variables of the intersection of their support sets and if they match on this subset.

This can recast in a language closer to the language of connections and holonomies we used. Indeed, even if, for our functions over the reals, the support is not graphs but finite subsets of the (non-negative) integers, the parallel could be made with zero-dimensional graphs. The integers could be understood as dots in some space, and so the subsets could be understood as zero-dimensional graphs with nodes only. A set of dots is a coarse version of another if it is included in it. A function over the reals can also be understood as a function over the colors of these dots, where the colors are real numbers put on these dots. The variables then acts as coloring of the dots.

We can do exactly the same thing, though it is of course a bit more involved, for our wavefunctions of the connection. We define inclusions of graph quite simply: a graph $\Gamma_1$ is finer than another one $\Gamma_2$ if all the nodes (as points on the manifold) of $\Gamma_2$ are in $\Gamma_1$ and if all the links (considered as sets of points on the manifolds) of $\Gamma_2$ are in $\Gamma_1$. Simply put: the inclusion of graphs is defined as the inclusion of sets when the graphs are seen as the collections of their points (coming from nodes or links) (see figure \ref{fig:cylindrical}). The pullback must be defined a bit more carefully to take into account the composition problem: links might be split in the finer graphs. Apart from this, the definition works quite as fine. Once the pull-back is defined, we can define the equivalence relation and therefore define our Hilbert space. We give more detail for this procedure in appendix \ref{app:proj}.

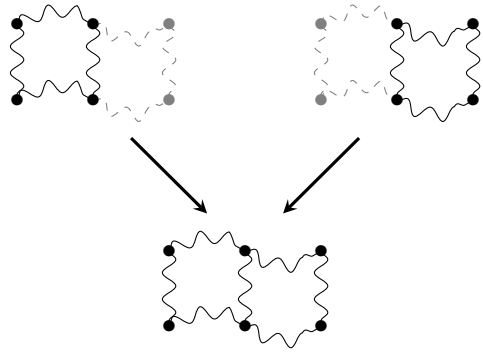
\begin{figure}[t!]
  \centering

  \begin{tikzpicture}[scale=0.5]
    \coordinate(A) at (0,0);
    \coordinate(B) at (2,0);
    \coordinate(C) at (4,0);
    \coordinate(D) at (0,-2);
    \coordinate(E) at (2,-2);
    \coordinate(F) at (4,-2);

    \draw[decorate, decoration=snake] (A) to[bend left] (B);
    \draw[decorate, decoration=snake] (B) to[bend right] (C);
    \draw[decorate, decoration=snake] (D) to[bend left] (E);
    \draw[decorate, decoration=snake] (E) to[bend right] (F);

    \draw[decorate, decoration=snake] (A) -- (D);
    \draw[decorate, decoration=snake] (B) -- (E);
    \draw[decorate, decoration=snake] (C) -- (F);

    \draw (A) node {$\bullet$};
    \draw (B) node {$\bullet$};
    \draw (C) node {$\bullet$};
    \draw (D) node {$\bullet$};
    \draw (E) node {$\bullet$};
    \draw (F) node {$\bullet$};

    \coordinate(O1) at (-4,6);
    \coordinate(A1) at ($(O1)+(0,0)$);
    \coordinate(B1) at ($(O1)+(2,0)$);
    \coordinate(C1) at ($(O1)+(4,0)$);
    \coordinate(D1) at ($(O1)+(0,-2)$);
    \coordinate(E1) at ($(O1)+(2,-2)$);
    \coordinate(F1) at ($(O1)+(4,-2)$);

    \draw[<-,>=stealth,very thick] (1,1) -- (-1,3);

    \draw[decorate, decoration=snake] (A1) to[bend left] (B1);
    \draw[gray,dashed,decorate, decoration=snake] (B1) to[bend right] (C1);
    \draw[decorate, decoration=snake] (D1) to[bend left] (E1);
    \draw[gray,dashed,decorate, decoration=snake] (E1) to[bend right] (F1);

    \draw[decorate, decoration=snake] (A1) -- (D1);
    \draw[decorate, decoration=snake] (B1) -- (E1);
    \draw[gray,dashed,decorate, decoration=snake] (C1) -- (F1);

    \draw (A1) node {$\bullet$};
    \draw (B1) node {$\bullet$};
    \draw[gray] (C1) node {$\bullet$};
    \draw (D1) node {$\bullet$};
    \draw (E1) node {$\bullet$};
    \draw[gray] (F1) node {$\bullet$};

    \coordinate(O2) at (4,6);
    \coordinate(A2) at ($(O2)+(0,0)$);
    \coordinate(B2) at ($(O2)+(2,0)$);
    \coordinate(C2) at ($(O2)+(4,0)$);
    \coordinate(D2) at ($(O2)+(0,-2)$);
    \coordinate(E2) at ($(O2)+(2,-2)$);
    \coordinate(F2) at ($(O2)+(4,-2)$);

    \draw[<-,>=stealth,very thick] (3,1) -- (5,3);

    \draw[gray,dashed,decorate, decoration=snake] (A2) to[bend left] (B2);
    \draw[decorate, decoration=snake] (B2) to[bend right] (C2);
    \draw[gray,dashed,decorate, decoration=snake] (D2) to[bend left] (E2);
    \draw[decorate, decoration=snake] (E2) to[bend right] (F2);

    \draw[gray,dashed,decorate, decoration=snake] (A2) -- (D2);
    \draw[decorate, decoration=snake] (B2) -- (E2);
    \draw[decorate, decoration=snake] (C2) -- (F2);

    \draw[gray] (A2) node {$\bullet$};
    \draw (B2) node {$\bullet$};
    \draw (C2) node {$\bullet$};
    \draw[gray] (D2) node {$\bullet$};
    \draw (E2) node {$\bullet$};
    \draw (F2) node {$\bullet$};
  \end{tikzpicture}

  \caption{To define the scalar product in the continuum, we use cylindrical consistency: wavefunctions with trivial dependency on some edges are identified with functions on coarser graphs (with the gray dashed edges removed from the graph). As a consequence, two coarse graphs can always be considered as being embedded in another finer graph on which the scalar product is well-defined.}
  \label{fig:cylindrical}
\end{figure}

We do such a complicated thing for several reasons: first it is difficult to define infinitesimal links. Starting from the discrete case and extending to the continuum is therefore natural. But second, this allows us to indeed define notions as differentiability and scalar products. Indeed, the natural scalar products and differentiability structures on the wavefunctions over finite graphs are compatible with the equivalence relation. This means that if we take two equivalence classes and define their scalar product (for instance) by the scalar product of their representatives, then the result does not depend on the choice of representatives. Therefore the scalar product is well defined for the equivalence classes themselves. The same goes for differentiability. And it can even be showed in a very precise sense, that this construction initially done by Ashtekar and Lewandowski, corresponds to the natural structure of square integrable functions over the space of generalized connections, that is distribution-valued connections.

Let us see how this works for the Peter-Weyl basis. Given a basis vector on a graph $\Gamma$, does it correspond to a basis vector of a finer graph and if so, which one? The answer to this question is pretty straightforward: yes and it corresponds to basis vector with spin $0$ on the additional edge and contraction of indices of the split edges. Indeed, the spin $0$ corresponds to the trivial representation which is constant and therefore gives the decomposition of constant functions. The contraction comes from the fact that by using the composition of holonomies, a function on the coarse-graph is gauge invariant at the splitting of edges. Therefore, we must use the only intertwiner available to glue the two dependencies. This means, from the basis perspective that we identify the functions with $0$ spin and the function without the corresponding edge.\graffito{The cylindrical consistency condition implies that functions with support on different graphs may not be orthogonal. It might not be the case in particular if the dependence along an edge is trivial.}

In particular, any continuum observable on the phase space must be invariant when removing an edge with $0$ spin. To be compatible with this means that the function does not depend on a choice of representative of the equivalence class. This property of an observable is called cylindrical consistency. All honest continuum observables must be cylindrically consistent. This is the essence of our continuum construction: a continuum function can be defined on discrete states and there will be some relations between finer and coarser states. But conversely when starting from discrete states, functions satisfying cylindrical consistency conditions can be extended to the continuum states.

Now, as advertised in the beginning of this chapter, we did not really follow the quantization process. Instead, we first got a grasps of what invariant states should look like and then constructed backward in order to have something rather consistent. We should now ponder a bit and see if the various steps of this quantization sequence of states:
\begin{equation}
  \mathcal{H} \xrightarrow{~Gauß~} \mathcal{H}_{gauge} \xrightarrow{~Diffeo~} \mathcal{H}_{diff}
\end{equation}
can be defined properly. We omitted the dynamical constraint as we argued before and certainly, at this point, we won't consider the diffeomorphism constraints (they are the object of the next section). Still, what should $\mathcal{H}$ be and what is $\mathcal{H}_{gauge}$. More importantly, of what algebra of observables is $\mathcal{H}$ a representation of? We alluded to the fact that the conventional variables are not very well suited. Our discrete construction should highlight that. So, what kind of variables should we consider? Integrated ones. We can consider first the (commutative) algebra of cylindrical functions. These are functions over the holonomies of links in a graph. As for the holonomy excitations, we should distinguish here between two things, though they seem similar at first sight. There ares the cylindrical functions as observables and there are the cylindrical functions as wavefunctions. For now, we introduced the wavefunctions. But as we are considering the algebra of observables, we are turning to the second kind. The cylindrical functions of the wavefunctions define observables and we can define their quantum counter-part by a multiplicative action. We can then amend this algebra with the momenta. The momenta will be integrated versions of the densitized triad. The densitized triads are secretly bivectors and therefore are integrated on surfaces rather than lines. Historically, the fluxes were understood as derivation of the classical configuration space and would be defined as follows:
\begin{equation}
E_{S,f} = \iint_S f^i E_i^a n_a\mathrm{d}S
\end{equation}
where $S$ is a compact surface, $f^i$ is a vector function with support on $S$ and $n$ the (unit) normal vector to the surface. $f^i$ would serve as a test function but its main purpose is to track the problems link to parallel transport when integrating. The Poisson algebra however turns out to be non-intuitive and the fluxes do not commute among themselves. This was understood as some difficulty in the limit procedure for taking surface integrals rather than volume integrals \cite{Ashtekar:1998ak}. A more modern treatment \cite{Dittrich:2008ar, Freidel:2010aq, Bonzom:2009wm} (explicit calculations in \cite{Thiemann2000}) consider that this non-commutativity comes from the parallel transport which is done through the connection which does not commute with the fluxes:
\begin{equation}
E_{i,S} = \iint_S g \triangleright E_i^a n_a\mathrm{d}S
\end{equation}
where $g$ is the parallel transport by $A$ through an appropriate path in the surface $S$. It is a more \textit{simplicial} (meaning discrete) point of view of the same problem. In any case, it turns out the integrated triad act as derivation on the cylindrical functions. On the previously defined Hilbert space, they therefore also act by derivation. \graffito{This ``simplicial'' point of view highlights though yet another problem with changing the Immirzi parameter: because the connection enters the integrated fluxes, these fluxes will also change when the Immirzi parameter is changed. This makes implementing a \textit{discrete} Immirzi transformation really difficult to implement.}This defines the algebra of observables called the holonomy-flux algebra and the Hilbert space $\mathcal{H}$. The imposition of the Gauß constraints is not exactly straightforward but it can be carried as the Hilbert space carries an action of the local $\mathrm{SU}(2)$ group. The results is, quite as expected, the projective limits of gauge-invariant wavefunctions with support on (finite) graphs as the same procedure can be carried and gauge-invariance is a cylindrically consistent property (if a wavefunction is gauge invariant all its restrictions are as long as the support graph captures all the degrees of freedom).

Let us close this section with a remark comparing this representation with the usual Fock space. Indeed nothing, so far, needs diffeomorphism invariance. Our construction would seem natural for any gauge theory with compact gauge group. So, is this representation the same as the usual Fock space (are they isomorphic) or are there differences? If so, could we write Yang-Mills theory for example on such a representation or what would be the difficulties? First, no, the representations are not equivalent. Our new representation evades the weak-continuity hypothesis of the von Neumann theorem. In particular, there is no connection operator and this forbids the equivalence of the representations. So, could we write Yang-Mills theory in this new representation? Well, no. The problem is the Hilbert space thus defined is not really a Hilbert space as it is not separable \cite{Fairbairn2004}. This will be solved in the case of quantum gravity thanks to diffeomorphism invariance. But in a generic Yang-Mills theory, the Ashtekar-Lewandowski space is not suited for a good definition of the theory. There is also a difference in the vacuum state. In the usual Fock space representation, the vacuum state is annihilated by one of the ladder operators. It corresponds to a state where the configuration variable and the conjugate momentum are zero in their expectation value but spread statistically. Therefore, the Fock vacuum state is not an eigenvector of any of the canonical observables. This state is selected by its property of symmetry under Poincare transform. In contrast, we do not have such symmetries in quantum gravity and therefore, the space we've built so far do not have a canonical vacuum which is Poincaré invariant. We will see however, that it is possible to select a preferred vacuum through diffeomorphism invariance. In that case, it turns out to be the trivial state with no dependence at all. It is annihilated by every flux observable and therefore is, of course, an eigenvector of those. As a consequence, it is maximally distributed in the holonomy observables. The structure of the kinematical space is therefore quite different in a diffeomorphism invariant theory. We must notice something however: we lost the weak continuity hypothesis and the Poincare invariance. But it seems that diffeomorphism helps us in several regards: with respect to the separability and also in the definition of a vacuum. So, in a sense, we must choose our hypothesis. In the Yang-Mills case, it is natural to consider Poincare invariant vacuum and try to simplify things as much as we can by searching for weak continuity. But from the quantum gravity perspective, diffeomorphism invariance is the holy grail. Therefore it is natural to turn to a representation that can represent it faithfully. And as we will soon see, diffeomorphism invariance will force us to forget weak continuity. So, let's turn to our next section and consider how to implement it concretely in the theory.

%*****************************************

\section{Diffeomorphism invariance}

How are we to impose the diffeomorphism constraints? We would like a similar trick that the one used for gauge invariance: we would like to find a natural action of the diffeomorphism group and then look for diffeomorphism invariant states. The first step is relatively natural: given a function $\psi$ over a connection $A$ and a diffeomorphism $\phi$, we define the diffeomorphism action as follows:
\begin{equation}
(\phi\triangleright\psi)[A] = \psi(\phi^\star A)
\end{equation}
In other words, because the diffeomorphism group has a natural action (through pull-back) on the configuration variables, it also has a corresponding canonical action on the algebra of functions on these variables. So, given a wavefunction with support on a graph $\Gamma$, we have an image of this function under the action of $\phi$ which now has support on the image of $\Gamma$, let's call it $\phi \triangleright \Gamma$. This is quite intuitive: diffeomorphism moves the support points and the new function depends on the displaced points. This action moreover is compatible with the equivalence relation, that is it is cylindrically consistent and therefore carries into the continuum state space $\mathcal{H}$.

The problem comes from the second step: can we find diffeomorphism invariant states? The answer is: there are none except from the trivial state, which is the constant wavefunction. The problem comes from non compactness of the group. Indeed, let's consider for instance a simple loop excitation. We could try and write a projector onto the diffeomorphism invariant states (as this is very well possible for gauge invariance) by average over the group. Applied to the loop excitation $\psi_\square$, this would read something like this:
\begin{equation}
\psi_{inv}[A] = \int_\textrm{diff} \psi_\square[\phi^\star A] \mathrm{d}\phi
\end{equation}
There are a lot of problems in such a writing but here are the two most stringent: the first one is the non-compactness of the group which prevents such a simple writing. But as we will see in a moment, we can evade this problem by considering the space of distributions. This would not solve the second problem however of the measure. Indeed, in the integration we used a measure on the diffeomorphism group. For the (compact) gauge group, we had the Haar measure which was natural and respected in some sense the composition law of the group. It is difficult to see what would serve as this for the diffeomorphism group. Let us deal with each problem in order.

Regarding the non-compactness of the diffeomorphism group, the natural thing to do is, as we said, to consider distributions. The idea of distributions is to consider the topological dual of a proper subset $\mathcal{V}$\graffito{More generally, we can consider \textit{habitats} that is distributional extension of the $L^2$ space.} of the interesting space $\mathcal{H}$. This will allow the construction of a triplet of space, a Gelfand triple:
\begin{equation}
\mathcal{V} \subset \mathcal{H} \subset \mathcal{V}^\star
\end{equation}
The choice of $\mathcal{V}$ is important as it must be small enough for all the interesting operators to be defined on it. But it must also be large enough so that the topological dual is larger than the original Hilbert space. A usual choice is to use Schwarz' space of smooth functions with rapid fall-off. This space is interesting as all observables of the holonomy-flux algebra are defined on it and it is stable under successive action of these observables. As the gauge group is compact, the fall-off condition is automatically satisfied. Equipped with a nice topology, the space $\mathcal{V}$ becomes a Frechet space. The choice of topology is also important for another reason: it must be fine enough to be able to define the operators on $\mathcal{V}^\star$. This means in particular that we want the interesting observables to be continuous with respect to the topology. All these technical points are somewhat important to define the theory. But, it should be noted that once a choice is done, the construction does not change. For what we are interested in, that is the imposition of the diffeomorphism invariance,  we just need to fulfill the prerequirements of size and fineness of topology.

Distributions are continuous linear forms on the space $\mathcal{V}$. In our case, they can be considered as bras (as in bra-kets). And indeed, if we go back to our problem of averaging with diffeomorphisms, and if we forget the problem of the measure (for now), though it is ill-defined for the vectors, it is somewhat better shaped in the case of distributions. This would look like:
\begin{equation}
\langle \psi_{inv} | = \int_\textrm{diff} \langle \phi \triangleright \psi_\square | \mathrm{d}\phi
\end{equation}
We mentioned that states are not necessarily orthogonal if they have different support graphs. They are orthogonal though if they have non-covering \textit{minimal} support graph, that is graphs with all the edges with trivial dependency removed. Therefore, we can hope that most of the terms above will vanish and that this scalar product can be defined at all. Still, this integral is something very ill-defined. What is the problem? Let us imagine that we test against a function whose support graph is diffeomorphic to the support graph of the function we are averaging $\psi_\square$. The support of this function is a loop. Then we should ask the question of the volume of the subset of the diffeomorphism group that leaves the loop alone. The problem is that there is a multitude of diffeomorphisms sending one graph onto another one which is diffeomorphic. And in our integral all these terms will appear and factor somehow a volume of this set of diffeomorphisms. Apart from the definition of the measure precisely, this is the problem we are facing: to define a volume for all the valid diffeomorphisms sending a graph to a covering of our loop.

A natural idea is therefore to cut this liberty in the choice of diffeomorphism out. If we are averaging a wavefunction with support graph $\Gamma$, let's consider the group $\textrm{Diff}/\Gamma$ which is the group of diffeomorphisms moving only the graph $\Gamma$ around. That is, two diffeomorphisms are considered equivalent if they act in the same way on $\Gamma$. In a way, we cut out the action outside of $\Gamma$ effectively setting our spurious volume to $1$. Then we can define:
\begin{equation}
\langle \psi_{inv} | = \sum_{\phi \in \textrm{Diff} / \Gamma} \langle \phi \triangleright \psi_\Gamma |
\end{equation}
At first sight, it might seem even worse as we traded an integral for a very ill-defined continuous sum. But the trick is that we are considering distributions which are defined by their evaluation against test functions. And now, because of the orthogonality we mentioned between functions on different support graph, the sum will always have a finite number of terms which are non-zero. In practice, if we evaluate $\psi_{inv}$ against a test function $\varphi_{\tilde{\Gamma}}$ (with support on $\tilde{\Gamma}$), there are two possible outcomes: either $\tilde{\Gamma}$ and $\Gamma$ are diffeomorphic and by construction, there is only one diffeomorphism in the sum that will work or $\tilde{\Gamma}$ and $\Gamma$ are not diffeomorphic. In that case, the situation is a bit more subtle as a coarser or finer version of a one of the graph might be diffeomorphic to a coarser or finer version of the other graph and we would be led to first case again. Or the graphs are \textit{really} not diffeomorphic, not even in their subgraphs and no terms survive giving a zero scalar product.

We can therefore define the space of diffeomorphism invariant states as the space of all such averaged states. Being a space of distributions, the space of diffeomorphism invariant states does not \textit{a priori} carry a scalar product. In our case though, it does carry a natural one. For any state $\langle \psi |$, by definition of the space, there is a writing as:
\begin{equation}
\langle \psi | = \sum_{\phi \in \textrm{Diff} / \Gamma} \langle \phi \triangleright \psi_\Gamma |
\end{equation}
for some support graphs $\Gamma$ (it does not matter which representative we choose). $\langle \psi_\Gamma |$ could be understood as a representative of the diffeomorphism invariant state. How should we take the scalar product between two states $|\psi_1\rangle$ and $|\psi_2\rangle$? If we understand the sum as a projector, then it would be written as:
\begin{equation}
\langle \psi_1 | \psi_2 \rangle = \langle \psi_{1,\Gamma_1} | \mathcal{P}^\dagger \mathcal{P} \psi_{2,\Gamma_2} \rangle
\end{equation}
where $\mathcal{P}$ represents the projector and $\Gamma_1$ and $\Gamma_2$ are the support graphs of the representatives of each of the diffeomorphism invariant states. This would be divergent. But we force the matter and accept that $\mathcal{P}$ is a projector in disguise then, we would hope:
\begin{equation}
\mathcal{P}^\dagger \mathcal{P} = \mathcal{P} \mathcal{P} = \mathcal{P} = \mathcal{P}^\dagger
\end{equation}
And this would lead to:
\begin{equation}
\langle \psi_1 | \psi_2 \rangle = \langle \psi_{1,\Gamma_1} | \mathcal{P}^\dagger \psi_{2,\Gamma_2} \rangle = \langle \psi_1 | \psi_{2,\Gamma_2} \rangle
\end{equation}
which is, incidentally, well-defined. So rather than take the previous discussion as rigorous, let's rather consider the result as defining the scalar product. This defines a bonafide Hilbert space of diffeomorphism invariant states. In a sense, we took the (divergent) projector once too many, leading to divergences and bad definitions of the scalar product.

This new scalar product solves some of the previously mentioned difficulties. First of all, the Hilbert space is now separable\graffito{The separability of the Hilbert space actually depends on the precise choice of the diffeomorphism group. Indeed, the space is separable for a generalization thereof, the star-diffeomorphism but not for the usual diffeomorphism group.}. That means that the problem we mentioned for Yang-Mills theory disappears in the case of quantum gravity theories as the final kinematical Hilbert space is a honest Hilbert space. The basis of the space also becomes countable, provided we take a somewhat larger group than diffeomorphisms. We must for that use star-diffeomorphism which are essentially diffeomorphism almost everywhere: we are allowed some non-inversibility at isolated points. This use is necessary in order to deform the graphs around the nodes, because (regular) diffeomorphisms would not allow the transformation of different directions at nodes of the graphs. Allowing for singularities at nodes of the graph makes it possible relate nodes with valency greater than two. This concludes the presentation of the kinematical space which is now well-defined.

The construction seems relatively natural. Before we continue, we should ask if the path is somewhat unique or to what extent it is. It turns out, we are pretty constrained. In a more precise manner, what we want is a representation of the holonomy-flux algebra (with or without weak continuity, these operators would make sense, so let's start there) which has a natural notion of gauge invariance and which also has a natural action of the diffeomorphism invariance. We also want the representation to be cyclic: that means every states can be obtained from a starting state, let's call it the vacuum, through successive action of the holonomy-flux algebra. Otherwise this means our representation has a bit too much states that we can't distinguish with our observables. And we want the vacuum to be diffeomorphism invariant under the previously mentioned action of the diffeomorphism group. According to the LOST theorem, with a few mathematical hypothesis (namely semi-analicity of a bunch of functions), such a representation is unique and is the one described above \cite{Lewandowski:2005jk}. If we then apply the constraints, the previous method naturally arises and it seems that, at least at the kinematical level, the theory is unique or tightly constrained.

There are ways to evade the LOST theorem though. Most of them were investigated in the context of coarse-graining, so we will come back to them. But let's rapidly announce them. First, we can consider a vacuum which is not invariant under diffeomorphism. This was considered in order to describe excitations over a given (spatial) metric \cite{Koslowski:2011vn}\graffito{Fixing a metric breaks diffeomorphism invariance in the same way that using finite temperature in \ac{QFT} breaks Poincaré invariance. In both cases, we need another representation}. In that case, we can consider a covariant rather than invariant vacuum and arrive at a different representation. A second possibility is to remove an hypothesis which is close to weak continuity: the existence of the flux observables \cite{Dittrich:2014wpa}. If we consider than only exponentials of them exists, a new algebra can be defined and without weak-continuity, a new representation can be found with a diffeomorphism invariant vacuum. This was considered in the context of developing a physical vacuum for Dittrich's programme. The BF vacuum was constructed along with its representation. In that case the vacuum state is an eigenvector of the holonomies rather than the flux and is purely flat. These alternative representations have been studied in the context of coarse-graining. For us though, we will stick to the usual Ashtekar-Lewandowski representation, at least for now, and develop the theory along this line.

%*****************************************

\section{Geometrical observables}

Now that the kinematical space is defined, but before tackling the dynamics, we should populate the space of operators. In particular, it is paramount to define geometric operators with a precise action, and the possible spectrum, on the kinematical phase space. We will concentrate on two particular operators here: the area operator associated to the area of a given surface and the volume operator  \cite{Ashtekar1997} defined in a similar way for volumes.

We face here the difficulty of dealing with diffeomorphism invariance. Indeed, because of diffeomorphism invariance, the position of points in the coordinate system is merely a choice\graffito{The problem was first mentioned when considering discrete \ac{GR} as in \cite{Regge1961}.}. On diffeomorphism invariant states, we loose this possibility of designating points and in a similar way, we loose the possibility of defining surfaces and volumes in this way \cite{Gaul1999}. This is related to Einstein's hole argument \cite{Einstein1916} as points are not physical because of the diffeomorphism invariance. The problem for us is that in spacetime, we do not define events by their coordinates but rather by physical events. For instance, an event might be a collision of two particles. This notion will be diffeomorphism invariant: whatever the diffeomorphism we apply, there will be a point corresponding to the collision and at this point, whatever the diffeomorphism, the value of Higgs field for example will be the same. This suggests a relational point of view: physical quantities are related to and are expressed with respect to other physical quantities. And when we think about it, this is indeed what we do all the time. Even when measuring time, we actually use a physical device (a clock) with respect to which we express the movement of other things. But, this seems to need the use of matter. And for now, we concentrate on pure gravity. So, how are we to solve this dilemma?

The simple strategy is simply to write operators on $\mathcal{H}_\textrm{gauge}$ rather than $\mathcal{H}_\textrm{diff}$. We will develop and define area and volume operators on the space of Gauge invariant states but not on diffeomorphism invariant states. In this way, it is pretty clear that we cannot possibility run into trouble with respect to the diffeomorphism. In this Hilbert space, the notion of space-point can very simply be mapped by coordinates and nothing prevents to consider a given surface in the manifold $\Sigma$ or a volume. In particular, we do not need to consider equivalence classes under diffeomorphism and that means that the surfaces and volumes can be well-defined. But this raises a new problem: what could possibly be the relation between these operators and the diffeomorphism invariant operators we would be interested in? Well, the precise relation would be given by the construction of Dirac observables. The basic idea is that if we can define an operator giving the coordinates of a collision for example and if we have an operator for the value of the Higgs field, say, given coordinates, we can define the value of the Higgs boson at the collision. The Dirac construction therefore works in two steps: first define diffeomorphism \textit{covariant} operators and then identify a well-suited one for a physical choice of coordinates\graffito{A nice choice of physical coordinates is exactly what we choose with the GPS system as explained in \cite{Rovelli:2001my}.}. Then, we can combine them into diffeomorphism invariant operators as wanted. Our area and volume operators will therefore correspond to the first step of the construction. But though there are not themselves diffeomorphism invariant, they will participate in the construction and interpretation of such operators.

We can turn to the definition of the area operator. The problem is of course to express it using the canonical variables we introduced: the densitized triad and the connection. As the densitized triad diagonalizes the metric (in the classical sense) and the metric is what we need, this does not seem insurmountable. Let us write down the definition of the area with the usual metric formalism:
\begin{equation}
\mathcal{A}(\mathcal{S}) = \int_\mathcal{S} \sqrt{\det\left(q_{ab} \frac{\partial x^a}{\partial \sigma^i}\frac{\partial x^b}{\partial \sigma^j}\right)_{ij}} \mathrm{d}^2\sigma
\end{equation}
where $\mathcal{S}$ is the surface of interest embedded in $\Sigma$, $\sigma^i$ are coordinates on this surface and $x^a$ is a coordinate system on $\Sigma$. We also used the loose notation of $x^a(\sigma)$ to denote the coordinates in $\Sigma$ of the point parametrized by $\sigma$ in $\mathcal{S}$. $(q_{ab})$ is the spatial metric as usual and therefore the expression corresponds to taking the determinant of the induced metric on the surface. Now, we want to express this using the densitized triad. The densitized triad is a vector of weight one, or secretly a bivector. This means that it is naturally integrated over surfaces and for our concern that its spatial indices should be contracted with normals (which are pseudo-vectors) rather than vectors. In particular, this means we would like to make terms like $E^a_i n_a$ appear where $n$ is the normal vector to the surface.

The expression is invariant under spatial diffeomorphism, let's use that in order to reexpress the area. Let us extend the coordinate system given by $\sigma^1$ and $\sigma^2$ with a new component whose direction is locally given by the normal vector. Locally, this is possible. Then, in this coordinate basis, the metric $(q_{ab})$ reads:
\begin{equation}
  (q_{ab}) = \begin{pmatrix}
    q_{00} & 0 \\
    0 & (h_{cd})
  \end{pmatrix}
\end{equation}
where $(h_{cd})$ is the induced metric on the surface. We have $h q_{00} = q$ because the metric is block diagonal in this coordinate system. $h$ is the determinant of the induced metric and therefore is the quantity we are interested in and can be written:
\begin{equation}
  h = q_{00}^{-1} q
\end{equation}
Now, we can express quite simply the inverse metric in terms of the densitized triad as follows:
\begin{equation}
q q^{ab} = E^a_i E^b_j \delta^{ij}
\end{equation}
In particular, we can consider the quantity:
\begin{equation}
E^a_i E^b_j \delta^{ij} n_a n_b = q q^{ab} n_a n_b = h
\end{equation}
because $(n_a) = (1~0~0~0)$ in our coordinate system. This gives the following expression for the area:
\begin{equation}
\mathcal{A}(\mathcal{S}) = \int_\mathcal{S} \sqrt{E^a_i E^b_j \delta^{ij} n_a n_b} \mathrm{d}^2\sigma
\end{equation}
We now only have to quantize it.

The problem is that we only have integrated version of the densitized triad in our toolkit of operators so far. But we can deal with that with a simple Riemann sum: let's split the surface $\mathcal{S}$ into $N$ smaller surfaces $\mathcal{S}_n$ with (coordinate) area shrinking to $0$ as $N$ grows. We have:
\begin{equation}
\mathcal{A}(\mathcal{S}) = \sum_{n=1}^N \int_{\mathcal{S}_n} \sqrt{E^a_i E^b_j \delta^{ij} n_a n_b} \mathrm{d}^2\sigma
\end{equation}
So far, the equality is exact but as the surface tends to zero, we can approximate as follows:
\begin{equation}
\int_{\mathcal{S}_n} \sqrt{E^a_i E^b_j \delta^{ij} n_a n_b} \mathrm{d}^2\sigma \simeq \sqrt{\int_{\mathcal{S}_n} \int_{\mathcal{S}_n} n_a(\sigma) n_b(\sigma') T^{ab}(\sigma,\sigma') \mathrm{d}^2\sigma\mathrm{d}^2\sigma'}
\end{equation}
where:
\begin{equation}
  T^{ab}(r,s) = E^a_i(r) R^{(1)}(U(A,\gamma_{rs}))^{ij} E^b_j(s)
\end{equation}
where $R$ represents parallel transport in the adjoint representation along the path $\gamma_{rs}$ which is a straight path (in coordinates) between $r$ and $s$. For small surfaces, only the first term of the holonomy survives, that is the identity and the quantity converges to the previous expression. Therefore, we have:
\begin{equation}
  \mathcal{A}(\mathcal{S}) = \lim_{N \rightarrow \infty} \sum_{n=1}^N \sqrt{\int_{\mathcal{S}_n} \int_{\mathcal{S}_n} n_a(\sigma) n_b(\sigma') T^{ab}(\sigma,\sigma') \mathrm{d}^2\sigma\mathrm{d}^2\sigma'}
\end{equation}
which, for $N$ finite only involves well-defined operator. The remaining question is how to take the limit.

This problem is solved by the ingenious choice of representation in \ac{LQG}. Indeed, let's consider a given spin network (with a given support graph) and try to evaluate the operators in the sequence defined above. When sending the surfaces to zero area (in coordinates area), because the excitations are polymer like, at some point each small surface will cross by only one link of the spin network\graffito{Technically, it could be crossed by a node also. We will come to this case later.}. Because the densitized triad acts by derivation, only the dependence of the spin network wavefunction at the crossing point will intervene. For any further refinement of the surface, the action of the operator cannot possibly change. Let us say this in another way: once the splitting of the surface is fine enough to capture each link of the spin network individually, the actions of the operators in the sequence are the same. This means that evaluated against a given spin network, the previous sequence is actually constant starting from some $N$. The limit is therefore well-defined and is the definition of the area operator in the quantum case.

This is an important point that we should underline: in general continuum operators are hard to define. But, in the case of the Ashtekar-Lewandowski representation, because the excitations are in some sense discrete, it is possible to define continuum operation by the refinement of discrete operators. Because the spin networks are discrete, once the operators are sufficiently refined they will totally capture the continuum behavior. Of course, the necessary refinement will depend on the precise spin network, but for any spin network a sufficient refinement exists. This implements the continuum limit (at the kinematical level) for the Ashtekar-Lewandowski representation.

Let us now discuss the spectrum of the area operator just defined. It is a surprising and very important fact that this spectrum is actually well understood! Let us consider a spin network with spins $j_i$ on its edge $e_i$. This spin network diagonalizes the area operator and its eigenvalues are \cite{Alexandrov2001}:
\begin{equation}
A = 8\pi \beta G \sum_{e_i \textrm{ crossing } \mathcal{S}} \sqrt{j_i(j_i +1)}
\end{equation}
That is, for each crossing of the surface by a link, we get a quantum of area given by the spin carried by the link. Actually, this is called the non-degenerate spectrum as two others cases may arise:
\begin{itemize}
\item a link might touch the surface without crossing it,
\item or it might end on the surface (if a node is on the surface).
\end{itemize}
In those case, the spectrum is a bit altered as can be seen in \cite{Gour:2004gk}. But the spectrum is still fully characterized and does show quanta of geometry in the sense of quantized surfaces.

Of course, this quantization of the spectrum is linked to a compactification, as is usual in quantum theory. In our case, this is linked to the compact gauge group ($\mathrm{SU}(2)$) used to describe the theory. This sparkles a lot of discussion (see for instance \cite{Achour:2015xga}) about the role of the Immirzi parameter and whether it should be included or if we should turn to self-dual variables. In our case, we took the bet that the Ashtekar-Barbero variables is just a clever choice of variables when using the Holst-Immirzi parameter and so we take this quantification to be physical. But further studies would of course be welcome.

As usual, there is more than one quantum operator corresponding to a given classical quantity. This problem is most known when it is manifested as ordering ambiguities and our area operator does not escape those ambiguities. Here, we used the traditional ordering used for quantum mechanical angular momentum and spin. But, of course, this is not the only possible choice and, for now, it has not been selected by experiments. Two other ordering can be found in literature and so, we will discuss them as examples of other possibilities. They come down to changing the quanta of surface $A_j$ given, in the usual ordering, by:
\begin{equation}
A_j = \sqrt{j(j +1)}
\end{equation}
We have two popular alternatives:
\begin{itemize}
\item The first one is favored by the Duflo map \cite{Corichi:2004ey} which has nice mathematical properties though their physical meaning is not totally understood. It reads:
\begin{equation}
  A_j = j + \frac{1}{2}
\end{equation}
\item The second one avoids the gap from when there is no link and reads:
\begin{equation}
  A_j = j
\end{equation}
\end{itemize}
These two orderings have some popularity since they both imply a regular spectrum with equal spacing between area quanta. At this point, it might seem difficult to choose an ordering without physical input. But there are at least good theoretical reasons to avoid some of them: and it boils down to cylindrical consistency.

Indeed, cylindrical consistency tells us that the action of the operators should not depend on the support graph of an evaluated spin network as long as it is sufficiently refined to capture the whole spin network. In particular, adding a link to a graph with spin $0$ should not change the eigenvalue of the operators \cite{Corichi:2004ey}. For the area operator, this means that a spin $0$ link should not contribute to the eigen value. This rules out the Duflo map induced ordering of the area operator. Indeed, as we can see, it would have a non-zero contribution coming from links with $0$ spin corresponding to a gap, which is forbidden by cylindrical consistency.

This argument forbids any spectrum which do not have $A_j = 0$ for $j=0$. But any other ambiguities remain. It may be that physical input is required or that a particularly clever argument is still missing. One intriguing possibility comes from the $\mathrm{U}(N)$ framework \cite{Livine:2013tsa,Livine:2013gna}. In this framework, an algebra of operators around a node has been developed and seems to correspond to a discrete equivalent of 2d area preserving diffeomorphism on the sphere. But, at the quantum level, the preserved quantity is $\sum_i j_i$ where the $i$ numbers the incoming and outgoing links from the node. This seems to favor the second ordering of $A_j = j$. But, this is only intriguing at best for now, and further studies are really needed.

\vspace{1em}

In order to close this chapter, let's comment on another geometrical operator: the volume. Sadly, the situation is much less clear. Though similar strategies can be employed (discretization and refinement to define a continuum limit), the volume operator is not unique at this point. There are currently two different constructions that are equivalent only for tetravalent nodes. In general they differ, with their differences linked to usual ambiguities in quantization \cite{Ashtekar1997} and no way to select one of them is known for now. Some argument exists since one of the operator is sensitive to the differential structure of the graph and the other is only sensitive to topological information. Therefore this might be linked to a choice of invariance group. Still, the discussion is open for now. Sadly, their spectrum is also open, though progress has been made \cite{Thiemann1996,Bianchi2011}. The situation is quite unnerving as the volume operator seems to play a great role in the theory. Still, something quite satisfactory is there: it is that it is even possible to define such an operator. With area and volumes, the natural geometrical operators are all defined (except for length which is more involved \cite{Thiemann1996b,Livine:2006xk,Bianchi2009}) and have a discrete spectrum \cite{Rovelli1995a}\graffito{The survival of the discrete spectrum at the gauge invariant level though is still discussed \cite{Dittrich:2007th} and might need a resolution of the dynamics.}. This is really interesting as it corresponds to the intuition that quantum gravity should induce a discrete nature to spacetime.

\medskip

Therefore, in this chapter, we introduced the idea of loop quantization. The idea is to use discrete excitations and to link them through refinement in order to define a continuum Hilbert space. This space can be endowed with an action of both the gauge group $\mathrm{SU}(2)$ but also with an action of the (spatial) diffeomorphism group. More interestingly, through group averaging techniques, we can define their kernel and equip those with nice scalar products, finally defining a kinematical Hilbert space for quantum gravity. This space shall now be used to define the dynamics of the theory as we will turn to in the next chapter.

%and the theory seems to be on a good start.

%*****************************************
%*****************************************
%*****************************************
%*****************************************
%*****************************************

\ctparttext{}
\part{Dynamics and phenomenology}
%*****************************************
\chapter{Hamiltonian dynamics} \label{ch:HamilDynamics}
%*****************************************

\inspiquote{People assume that time is a strict progression of cause to effect, but *actually* from a non-linear, non-subjective viewpoint - it's more like a big ball of wibbly wobbly... time-y wimey... stuff.}{The Doctor}

So far, we have considered the kinematics of quantum gravity. Though the results are impressive, considering in particular their relative uniqueness, we have still major challenges ahead. The first one is the imposition of the dynamics through the Hamiltonian constraint.

Indeed, in our quantization process, we isolated this constraint because of its peculiar form (it is non-polynomial) but also because its geometrical interpretation is less clear. Of course, it must represents diffeomorphism along the time direction, but this quite difficult to implement on the spatial manifold $\Sigma$. We don't have a canonical action as for spatial diffeomorphisms. For the other constraints, the solving method used an action of the transformation group and employed group averaging techniques in order to define the solution of the constraints. Remarkably, the spatial diffeomorphism constraints do not even exist as operators. For the Hamiltonian constraint, we don't have, at least not yet, an interpretation as a group action. Of course, this is not possible \textit{per se} as the Hamiltonian constraint does not generate a subgroup of the spatiotemporal diffeomorphisms, as the Dirac algebra testifies. Still, we could have hoped for an action of the full spatiotemporal diffeomorphisms on the Hilbert space. Such an action is not known however and the more tedious route of finding a quantization of the Hamiltonian constraint is the one we should now take.

Still, imposing the dynamics as an operator constraint might seem difficult for the least. There are several reasons to doubt such an endeavour. First, there is the problem just mentioned that spatial diffeomorphisms do not exist as operators on the kinematical Hilbert space. One would not expect the time diffeomorphisms to work any better. This is true, but is related to the distributional solutions of the constraint. Indeed, the diffeomorphisms do not exist as operators on the \textit{kinematical} Hilbert space. They might exist though on a distributional space. They do exist in particular on $\mathcal{H}_{\textrm{diff}}$ where they are represented trivially. The hope here is therefore that the Hamiltonian constraint could be implemented as a well-defined constraint on the \textit{diffeomorphism invariant} Hilbert space.

Then, there are problems concerning the quantization: indeed, we could define operators as area and volume because they only used the triad operators for which the spin networks are eigenvectors. For the spatial diffeomorphism constraints already, the quantization was not straightforward (and indeed it is impossible) because of the presence of the curvature in the constraint which does not scale well on the spin network basis. Indeed, infinitesimal operators do not exist in our representation. In particular, the curvature operator cannot be represented. We might want to approximate it via some small but finite Wilson loop. But the continuum limit will be tricky in any case because the operators are not weakly continuous as the lack of infinitesimal operators shows. In the spatial diffeomorphism case, we could avoid the problem but to implement the dynamics as a constraint, we must get back to it.

So our goal is to find a way to implement the curvature (or any infinitesimal operator we might need) that can be reasonably defined on the kinematical Hilbert space with a suitable continuum limit. The strategy will be to define discrete, or regularized, versions of our operator and then take the limit in an appropriate sense. As can be seen from a dimensional argument \cite{Bonzom:2011jv}, for such a limit to be well-defined, we need to use a density weight one operator. Only in such a case can we hope the limit operator not to depend on the regularization scale. This is quite intuitive as density weight ones are naturally integrated on volumes and as such should scale accordingly and have a nice continuum limit. This is quite the equivalent of finding discrete definition that does not depend on the precise triangulation once it is refined enough. In that case, the continuum limit is trivial. Because of the non-commuting operators, we can't define such a triangulation invariant quantity but the density weight one is next best thing. Sadly, this means that we must quantize the square root in the denominator of the Hamiltonian. This kind of quantity is of course very hard to quantize consistently and this is the main problem at this stage. It should be noted that other density weights can be argued for \cite{Tomlin:2012qz,Henderson:2012ie,Henderson:2012cu}, but apart from density weight two, the problem subsists.

In this chapter, we will first introduce Thiemann's original constraint and rapidly explain the main ideas. We will then overview the various criticism against it and how we might hope to resolve them. We will then turn to the three-dimensional scenario which is more clearly understood in order to get a sense of what the challenges may be. Finally, we will comment on the precise role of the coarse-graining in the dynamics.

%\textbf{TODO:} last constraint to be imposed: hamiltonian constraint, encodes the dynamics, discussion on dynamics as constraint, should be imposed as an operator this time (should it?), difficulty: quantization, problem: not just a function of $E$, want a continuum limit though $\rightarrow$ weight density one, this implies difficulty of square root

%*****************************************

\section{Thiemann's constraint}

In order to quantize the square root, Thiemann proposed a \textit{trick} using the fact that the Poisson bracket acts as a derivation \cite{Thiemann:1996aw,Thiemann1996a}. If it is indeed difficult to quantize the square root in the denominator, it is way easier to quantize it in the numerator as it is linked to the volume. Taking the Poisson bracket with the connection, we get the inverse volume as needed. In more precise term we use:
\begin{equation}
\{A^k_a,V\} = \frac{\epsilon_{abc} \epsilon^{ijk} E^a_i E^b_j}{\sqrt{q}}
\end{equation}
where $V = \int \sqrt{q} \mathrm{d}^3 x$ is the volume. The advantage of this is that as long as we can write down operators for the connection and the volume, their Poisson bracket can be quantized by replacing them by appropriate commutators. And we know how to quantize the volume: this was done in the previous chapter. For the connection, it is a bit more tricky as we need to consider finite refinements, that is holonomies, but we also know how to do that. Apart from problems coming from the continuum limit, this should be doable. So, this seems a good idea on how to quantize the inverse square root.

Sadly, nothing guarantees, at least \textit{a priori}, the success of such a method. Indeed, the quantities we wish to quantize are highly non polynomial, there are not even analytic. Therefore, the quantization is plagued with ambiguities. Moreover, it is not unusual for brackets to get quantum corrections of higher order in $\hbar$. This means, that nothing guaranties that the bracket is reproduced non-anomalously in the quantum theory. On the plus though, it might also work, but this must be checked afterwards, and we certainly have no reason to believe it. Indeed, in simpler settings, this trick does not work. For instance, in the spinor formalism (see appendix \ref{app:spinors} for details about the spinor formalism), the holonomy operator is way messier to write down, but it is possible. One of the difficulty comes from a norm that also appears in the denominator. Using a variant of Thiemann's trick\cite{Livine:2013zha} does not work and produce an anomalous algebra. Requiring a non-anomalous algebra on the other hand fixes the ambiguities and make it possible to define the holonomy operator in terms of spinors. This trivial example does not show anything but the fact that the result must be checked afterwards one way or the other.

Still, in a series of seminal papers \cite{Thiemann:1996aw,Thiemann1998a,Thiemann1998b,Thiemann1999,Thiemann:1997rt,Thiemann1998,Thiemann2000}, Thiemann actually defined the whole theory and the Hamiltonian constraint which still satisfied some very interesting properties:
\begin{itemize}
\item First and foremost, the Hamiltonian satisfied a natural version of covariance.
\item There was a natural notion of continuum limit for diffeomorphism invariant states which allowed a complete definition of the operator.
\item The operator was non-anomalous in a minimal sense: on diffeomorphism invariant states, the commutator of two Hamiltonian constraints vanishes. This means that \textit{on-shell}, that is when the constraints are satisfied, the algebra is indeed verified.
\end{itemize}
Such properties are indeed remarkable and are good indicators of the consistency of the theory. They are not however sufficient. In particular, two criticisms came out against Thiemann's constraint. First, the \textit{on-shell} closure of the Dirac algebra seems very weak as a test for the absence of anomaly. In other known situations, for instance in string theory \cite{Nicolai2007}, the true test is off-shell when the constraints are not verified. Now, this is quite hard to test in the quantum regime as the diffeomorphism constraints are not well-defined in those cases. Still, it is possible to explore the behavior of these operators on an extension of the diffeomorphism invariant space, called a habitat, and some troubling behaviors were found \cite{Lewandowski:1997ba}. A second problem comes from what was ultimately called the \textit{ultra-locality} of Thiemann's proposal \cite{Smolin:1996fz}. Indeed, the operator does change the graph and creates new nodes. But it does not act on the nodes it creates. This is an issue as it seems to mean that there is no way any information could propagate. This might also signify an anomaly in the theory.

%\vspace{1em}
%\hrule
%\vspace{1em}

%\textbf{TODO:} Describe action of the operator

%\vspace{1em}
%\hrule
%\vspace{1em}

%\textbf{TO REMOVE:} As it stands therefore, the problem of the dynamics is still open. But several roads are possible. In what follows, we will discuss two cases where the matter are much more resolved. First, in a simplified theory called $\mathrm{U}(1)^3$, originally introduced by Smolin \graffito{\textbf{REF}}, it was possible to define consistently a Hamiltonian constraint. This sheds light on the difficulties so far. And second, in the simpler setting of 2+1d quantum gravity, one can take advantage of the fact that the theory is topological in order to work on finite triangulation. In these cases, much advances can be made and the Hamiltonian constraint can be defined properly. Let us review these two cases.

%\textbf{TODO:}  Thiemann's trick, detail of the regularization, describe successes, still: problem of Thiemann's constraint: ultra-locality, anomaly (possibly), concerning the trick: simpler example with spinors, trick does not work, discussion of open roads

%*****************************************

\section{An anomaly free algebra}

The first problem as we saw, is the difficulty in reproducing the constraints algebra at the quantum level. What is at stake here? When we use canonical quantization \textit{à la} Dirac, it is quite natural to do a splitting of spacetime into space and time using a foliation. This splitting does not necessarily breaks diffeomorphism or Lorentz invariance (indeed we do it all the time in usual quantum field theory) but it certainly hides it, making these symmetries implicit. If we see the phase space as the space of solutions of the constraints, then such a splitting corresponds to a choice of coordinates on the phase space. And it happens that a clever choice is related to a foliation of spacetime. Still, we must retain some notion of \textit{covariance}. This means that the choice of foliation cannot possibly have physical consequences and we must have some kind of invariance under a change of foliation or more generally under a change of spatio-temporal coordinates. At the classical level, this remainder is encoded in the Dirac algebra:
\begin{equation}
  \begin{array}{rcl}
    \{D[N^b],D[M^c]\} &=& D[\mathcal{L}_{N^b}M^c] \\
    \{D[N^b],H[N]\} &=& H[\mathcal{L}_{N^b}N] \\
    \{H[N],H[M]\} &=& -D[q^{ab}(N\partial_b M - M \partial_b N)]
  \end{array}
\end{equation}
The $D$s encode spatial diffeomorphisms while the $H$ (which stands for Hamiltonian constraint) encodes the time diffeomorphisms. The fact that this algebra is satisfied guarantees that the theory has a spacetime interpretation.

What about the quantum level? We would like this algebra to be realized anomaly-free. Indeed, the absence of anomaly is paramount for the quantum theory to be well-defined. This is the case for instance in gauge theory where the absence of anomaly even restricts the possible theories. In string theory, the anomaly-free algebra imposes the dimension of spacetime (as mentioned in \cite{Nicolai2007}). The presence of anomaly on the other hand means that the symmetry is not verified at the quantum level. This is a problem as symmetry usually allows the quantum theory even to exist. In particular, symmetries allow cancellations which protect some parameters of the theory under renormalization. In general, a property of the classical theory (for instance null masses) is usually lost in the quantum theory if there is no symmetry to protect it (like the Chiral symmetry). Therefore, symmetries are a much desirable from this point of view. For quantum general relativity (and indeed any theory which implements diffeomophism invariance), that translates to the Dirac algebra being reproduced at the quantum level. That means, there should exist, in some sense, some operators such that their commutation relations reproduce the algebra aforementioned. And this is where it becomes difficult. Indeed, as we saw, there are no (spatial) diffeomorphism operators on the non-diffeomorphism invariant Hilbert space (since the representation of the diffeomorphisms is not weakly continuous). We cannot therefore hope to reproduce the algebra on this Hilbert spaces (except trivially on the diffeomorphism invariant space). But a hope might be that it is reproduced on some particular subset or generalization, with a space of distributions for example.

On this stand, the Thiemann constraint does kind of well at first sight: the algebra is indeed satisfied in a precise sense on a particular space of distributions. More precisely, it is realized on the space of diffeomorphism invariant states. Of course, it is realized in a particularly trivial sense: the diffeomorphism constraints act trivially, and so as soon as the Hamiltonian constraints commute among themselves, the algebra is satisfied. In particular, it is not really possible to check if the precise relation $\{H[N],H[M]\} = -D[q^{ab}(N\partial_b M - M \partial_b N)]$ is satisfied since the right hand side always vanishes on these states. In order to check this algebra, it was proposed \cite{Lewandowski:1997ba,Gambini:1997bc} to extend the space of diff-invariant states into a larger space which is still distributional but would include non-invariant states. The hope of course is that the Hamiltonian constraints could still be defined on such a space and that the diffeomorphism constraint would have some canonical writing on it. This proposal of an extended \textit{habitat} led to the construction of a slightly larger space defined as follows.

For diffeomorphism invariant state, the construction starts from a (gauge-invariant) spin network wavefunction and uses group averaging to get a diff-invariant state as follows:
\begin{equation}
(\psi_{\Gamma,\textit{diff}}| = \sum_{\phi \in \textrm{Diff}(\Sigma / \Gamma)} \langle \phi \triangleright \psi_\Gamma |
\end{equation}
We generalize the construction by allowing a weighted sum as follows:
\begin{equation}
(\psi_{\Gamma,f,\textit{hab}}| = \sum_{\phi \in \textrm{Diff}(\Sigma / \Gamma)} f(v_1,...v_n) \langle \phi \triangleright \psi_\Gamma |
\end{equation}
where the $\{v_i\}$ are the vertices of the graph $\Gamma$ and $f$ is a function on $\Sigma^n$ not necessarily constant. If $f$ is constant we recover a diff-invariant state since it amounts (up to a factor) to the previous construction. But, if $f$ is not constant, it becomes sensitive to diffeomorphism and thus, with a kind enough set of functions, we enlarge the space of diff-invariant states into a habitat where non-trivial things can happen.

Now if we try and evaluate the commutator of two Thiemann's Hamiltonian constraints on this habitat, it turns out to be zero. This is a problem \textit{a priori}. Of course, we don't \textit{need} the theory to work \textit{off-shell} as physical predictions are only \textit{on-shell}. And, as we don't have a precise definition for the diffeomorphisms anyway, they might also cancel so this is not expected. This means there is, \textit{properly speaking}, no inconsistency. Still, this comes out as odd as we do not expect the diffeomorphism constraints to vanish on such a space. So, while not an inconsistency \textit{per se}, this seems to point to some anomaly which has not been uncovered yet. According to Lewandowski \cite{Lewandowski:1997ba}, the surprising fact is that any density weight one, because of a simple counting argument must have a vanishing commutator in this space. This is precisely why the right hand side vanishes. On the other hand, an ultra-local Hamiltonian constraint will always have a vanishing commutator on this states space. This means that Thiemann's proposal seems to be wrong headed in at least two-different ways: the regularization procedure and the density weight. The former, while interesting, does seem to have problem with locality. We need an Hamiltonian which acts on the vertices it creates. And the latter, though leading to a natural continuum limit, seems to actually hide a lot of the structure.

There are of course further developments. To name a few, there is the master constraint programme \cite{Thiemann:2003zv} and more algebraic approaches \cite{Giesel:2006uj, Giesel:2006uk, Giesel:2006um, Giesel:2007wn}. But at this point, we would like to highlight another perspective.
%We should note here that Thiemann started a new programme trying to solve these issues of anomalies: the master constraint programme \cite{Thiemann:2003zv}. The idea is quite simple: instead of imposing an infinity of constraints $H(x)$, one for each point, we impose a simpler single constraint, called the master constraint defined as:
%\begin{equation}
%M = \int d^3 x H(x)^2
%\end{equation}
%The definition of Dirac observables and some other technicalities become a bit more subtle in this approach. At the moment, it is unclear if this solves the problem.
Rather than trying to quantize the Hamiltonian and then check the Dirac algebra, we could search what are the good properties the constraints must verify in order to check the algebra. Then, this would select the correct quantization. In particular, concerning the density weights, the density weight one was chosen to avoid the dependence on the regularization scale $\delta$. But, as it is exposed in \cite{Bonzom:2011jv}, a $\delta^{-1}$ factor on the right hand side, though leading to a divergence, might actually be handy in order to regain the derivative that we need on the right hand side. In that case, we should quantize density weight $\frac{4}{3}$. The programme is not over, and the question is still very well opened but these two elements, having a correct density weight and taking the algebra as the holy grail to obtain, coupled with the technology of the habitats might lead to the correct theory.

Though the full theory has not been resolved this way, interesting developments have happened in the context of a simpler theories \cite{Bonzom:2011jv} and especially in the $\mathrm{U}(1)^3$ theory. This theory was originally introduced by Lee Smolin \cite{Smolin:1992wj} as a $G \rightarrow 0$ limit of general relativity, in a suitable sense. The advantage of this theory is twofold:
\begin{itemize}
\item it is still a complicated theory in the sense that there is still a sense of general covariance in the theory. It is therefore an interesting toy-model for the study of theories that encompass diffeomorphism invariance;
\item it is very simple in that the gauge group is abelian and therefore a lot of the technicalities due to non-commutativity of the connection are removed.
\end{itemize}
The theory was originally introduced as a hope to expand general relativity around $G = 0$ but in a diffeomorphism invariant way. Though this project has not been successful (yet), it is certainly an interesting theory in itself, or at least as a toy-model for quantum gravity. And indeed, the Dirac algebra was reproduced in a precise sense for this theory. We refer the reader to the given papers \cite{Smolin:1992wj,Bonzom:2011jv,Henderson:2012ie,Henderson:2012cu} for the precise technicalities. But, what lessons can we learn from this quantization of $\mathrm{U}(1)^3$ theory? It seems there are four main take-home points:
\begin{itemize}
\item First, it is possible to quantize the constraints (even the spatial diffeomorphisms) as operators on a suitably chosen distributional space. The use of distributions is of course needed as the representations is not weakly continuous, but it is rather impressive that such a possibility even exists.
\item Second, the geometrical interpretation of the constraint seems to be paramount in the construction. This is reminiscent of the construction of the spatial diffeomorphism constraints which heavily relies on the geometrical action.
\item Third, the choice of density weight indeed matters and the density weight one does not work. This is because we should not expect a continuum limit on a space for which the constraint is not even defined. Indeed, we do not expect the infinitesimal action of temporal diffeomorphism to be well-defined on the kinematical space. As such, we should not expect a continuum limit on this space, though an action of finite elements of the group might be well-defined.
\item And fourth, the constraint algebra indeed helps us a lot in the right quantization. The $m-$ambiguity \cite{Borissov:1997ji} is resolved for instance in this context. The surprising feature is that the Hamiltonian in some sense depends on the coloring of the spin network not only its graph structure. By this, we mean that the regularization we must use on a given spin network depends also on the graphs coloring not only on the graph structure.
\end{itemize}
There are some drawback in this developments though. First, the algebra has only been checked on a precise habitat, only as matrix elements. But still, the result is encouraging. Then, the non-commutativity certainly helps a lot but this also means that a huge hurdle still has to be overcome when considering the full-theory. In particular, because there is no non-commutativity, the Hamiltonian constraint does not create new vertices, it only moves them around. This is kind of the equivalent in electromagnetism that the photon number is conserved because there is no self-interaction. Here, a spin network, up to nicely chosen diffeomorphism averaging and coloring, are directly solutions of the constraint. A finite graph can be a solution. In some sense, quanta of graphs are not created or destroyed, and discrete states are solutions of the theory. Upon restoring the non-commutativity, this won't be the case any more and we expect a lot of technical troubles to appear.

%\textbf{TODO:} Take as fundamental: anomaly free algebra, problem with density weight, discuss habitat, discuss advances (huge biblio)

%*****************************************

\section{The three dimensional case}

Since we are considering models with fixed graphs, we might as well consider other models which cannot possibly change the graph, that is topological model. The archetype of such a model is 2+1d quantum gravity. Indeed, in three (spatiotemporal) dimensions, gravity simplifies a lot and becomes a topological theory. From a differential geometry point of view, this is due to the fact that the Riemann curvature tensor can be formulated entirely in terms of the Ricci tensor as follows:
\begin{equation}
  R_{abcd} = f(R_{ac})g_{bd} - f(R_{ad})g_{bc} + f(R_{bd})g_{ac} - f(R_{bc})g_{ad}
\end{equation}
with:
\begin{equation}
  f(R_{ab}) = R_{ab} - \frac{1}{4}R g_{ab}
\end{equation}
This implies in particular that the vacuum equations of general relativity impose flat spacetime. From the Palatini formalism, it is even clearer that flatness is imposed, since the triad\graffito{It is common to designate the equivalent of the tetrad by a dimensionnally adapted name: \textit{triad} for three-dimensions, even sometimes \textit{diad} for two. In general, we might say \textit{vielbein}.} will act as Lagrange multipliers for the curvature. From a physics point of view, this can be seen by a counting of degrees of freedom: in 3d, the triad has $6$ independent components, and the same goes for the connection. All these are eaten up by the constraints. There are indeed $3$ constraints (corresponding to three spacetime directions) eating up $3$ components and then, the three $3$ gauge freedom generated by the constraints eat up the rest (since gauge constraints eat each $2$ degrees of freedom and not just $1$). We can say this in another fashion: as soon as the metric satisfy the constraints, there is always a gauge transformation than sends it back to the flat metric.

This does not mean though that the theory is trivial. Indeed, \textit{topology} can happen. What is meant here is that curvature is a local notion and does not totally restrict global degrees of freedom. As a simple example, in 2d, the plane and the torus are both flat but are different manifolds. Such things can also happen in 3d. We therefore get a interesting theory which does not have local degrees of freedom. In particular, this means that nothing gravitates and there are no gravitational waves, but the quantization is not totally straightforward and must also implement general covariance though in a simpler setting. The advantage of studying such a topological\graffito{By \textit{topological}, it is understood a theory without local degrees of freedom.} theory is that we do not need to capture local degrees of freedom. This led to the first successful quantum theory of gravity with combinatorial quantization \cite{Witten1,Alekseev:1994pa} which recast the theory into a Chern-Simons one. But, if we want to stay close to four dimensional gravity, let's keep our usual variables. Now, if we are to work on a triangulation, for example, of space as soon as it is refined enough to capture the topology, the discrete quantization can reproduce exactly the continuum theory. This lift the huge burden of handling varying graphs and other technical subtleties.

In our case, let's rapidly review the kinematics of the theory we want to study. The gauge group is either $\mathrm{SU}(2)$ or $\mathrm{SU}(1,1)$ depending on the signature of spacetime. As we want to stay close to the $4$ dimensional theory, we will consider 3d euclidean gravity with $\mathrm{SU}(2)$ as a gauge group. Now, we can directly develop the kinematical phase space as we've done in the 4d case, except we don't need any time-gauge. The Hilbert space will be spin network wavefunctions embedded in a 2d space, on which we impose gauge invariance. We don't impose spatial diffeomorphism just yet. We could develop the cylindrical consistency in order to implement the continuum limit of the theory, but as we said, this is exactly what we can to avoid in the case of a topological theory. So, we'll just fix a graph refined enough to capture the topology of the embedding (spatial) 2d slice. The remaining problem is indeed to implement the dynamics of the theory and to give a precise sense to the Dirac algebra. Two strategies have been implemented so far:
\begin{itemize}
\item The first one consists in implementing the diffeomorphism constraints as usual and then considering the discrete Hamiltonian constraint as a regularization. Because we are really interested in the physical results, we can as well define the projector onto physical states, which is the projector onto flat space for 3d quantum gravity.
\item The second strategy is to quantize a discretization of the constraints. In this fashion, all the constraints (spatial diffeomorphism and Hamiltonian) are discretized and a discrete Dirac algebra is found.
\end{itemize}
%These two are only possible because the regularization is actually exact thanks to the topological nature of the theory.

In the first approach, we want to define a projector onto physical states. This is not actually a projector as physical states are distributional but this is more of a technicality at this stage. In practice, we want to define something like this:
\begin{equation}
\mathcal{P} = \prod_{\square} \delta(h_\square)
\end{equation}
where the $\square$ represents plaquettes surround by loops of the spin network. The $h_\square$ represents the holonomy around such loop and the $\delta$ is the Dirac delta which enforces the holonomy to be trivial. This indeed forces the curvature to be zero and leads to the right theory. We can indeed check that the theory reproduces the Ponzanno-Regge amplitude for example \cite{Noui2005}. The problem of such an approach is that it does not share much light on the subject of the quantization of the Hamiltonian constraint: we have beautifully avoided the problem. The algebra of constraints cannot be checked as there is no algebra to check. Though interesting in the study of 3d quantum gravity, this method is not sufficient for studying the quantization of the Hamiltonian constraint. It has the advantage though of given us the exact result for 3d quantum gravity which will allow us to compare with other techniques.

In the second approach, we want to quantize the Hamiltonian constraint but because the continuum limit is difficult to obtain, we use the topological nature of the theory and try quantizing a discretization of it. To do this, we need a geometrical interpretation of the constraints that will naturally extend to the discrete setting. Let us begin with an observation: the constraints in 3d quantum gravity reads:
\begin{equation}
  \left\{
  \begin{array}{rcl}
    C_a &=& E^b_i F^i_{ab} \\
    H &=& \epsilon^{ij}_{~~k} E^a_i E^b_j F^k_{ab}
  \end{array}
  \right.
\end{equation}
while the equation of motions for the connection (obtained by the Lagrange-Euler equations) are:
\begin{equation}
  F_{ab}^i = 0
\end{equation}
and therefore, the constraints are actually the projection onto three independent directions of the flatness condition. Indeed, at each point, the triad gives a correspondence between two directions in the spatial slice and two directions in the tangential space. Therefore, there cross-product, which appears in the Hamiltonian constraint, can be interpreted as the normal in tangent space to the spatial slice as embedded in spacetime. The diffeomorphism constraint then correspond to lifting the flatness condition into tangent space for the two tangential directions to the spatial slice. This means that we need a way to project the flatness condition onto various directions given by the triad.

The solution \cite{Bonzom:2011hm} that was chosen was to interpret the Hamiltonian constraint as a continuum version of\graffito{There is also an extension to spinor language \cite{Bonzom:2011nv} which handles the full double-cover $\mathrm{SU}(2)$ rather than $\mathrm{SO}(3)$.}:
\begin{equation}
  H_\textrm{discrete} = \vec{X}_1 \dot (h\triangleright \vec{X}_2) - \vec{X}_1 \dot \vec{X}_2
\end{equation}
that is to interpret the $\epsilon$ as the action of the $\mathrm{SU}(2)$ group rather than the cross-product. In that case, the curvature is naturally understood as the first deviation from the identity by the holonomy. So given a wedge, that is a vertex and two links forming an angle, in the spin network, we can compare the dot product between the two vectors (which are the discrete equivalent of the triad) and the dot product once one is parallel transported. This indeed corresponds to a natural interpretation of the Hamiltonian constraint. But it comes with an additional treat: for a loop, there is generally more than $3$ wedges and therefore, we get three different projections for the same loop. As the vectors are not the same, this also gives the spatial diffeomorphisms or more precisely, their discrete equivalent which is moving around the point of the spin network graph. It can also be checked that the algebra closes in simple cases and that the right recurrence relations are generated for the tetrahedron. This means it is possible to reproduce the amplitudes from the Ponzanno-Regge model \cite{Bonzom:2011hm}. The question of a cosmological constant is not yet solved at the canonical level though a covariant spinfoam model exists : the Tuarev-Viro model \cite{Mizoguchi:1991hk}.

It should be noted that the quantized constraint correspond to density weight two. This was done so because the density two is algebraically really simple, avoiding all the complication of having a square root in the denominator. The constraint being polynomial, all the quantization process becomes somewhat easier. But it could remarkable that it is even possible. That actually comes from the discrete setting: because we went discrete, we don't have to take a specific density weight to have the nice continuum limit. Any density weight will have some meaning and we can as well take the simplest. This has the side problem of giving us no information at all on how to go to the continuum limit.

\section{Interplay with coarse-graining}

Considering the large unknowns with regard to the dynamics, it might seem odd that the coarse-graining is already attempted in order to find the large scale of the theory. The large scale of what theory? Coarse-graining seems premature at this point to say the least. But we know already from works in \ac{QFT} and in condensed matter that coarse-graining and renormalization has a rather larger domain of application than just studying the large scale phenomena. For instance, in condensed matter, coarse-graining can be used to understand phase transition for example. It can have similar roles in the case of quantum gravity.

As can be noted, one of the major problem is how to take into account varying graphs. And this is were coarse-graining might come in handy. Indeed, we might hope to coarse-grain the state to the same graph with additional structure. Another way of putting this is that we might hope that the dynamics on a varying graph might be mapped onto the dynamics of a single graph with additional structure, which might be easier to handle. This might not be a desperate move. Indeed, this is what we observe, though in two-dimensions, with all the work around the KPZ conjecture \cite{Knizhnik:1988ak}.

Coarse-graining can also be used to understood the continuum limit of a theory. Though, this might not be useful as such for studying quantum general relativity, it can fuel a programme starting from some discrete equivalent of general relativity and, after quantization, looking for a continuum limit. If such a limit is found, it is likely to correspond to quantum general relativity. This would correspond to the study of the continuum limit of lattice \ac{QCD} for instance. Such a technique would circumvent the problems of the density weights and even of the precise algebra of diffeomorphism. Indeed, techniques from 3D quantum gravity could even be imported and worked out.

But this point is more general than just for discrete theories approximating a continuum theory. Indeed, for usual quantum field theories, which are expressed on a continuous spacetime, the idea of renormalization is paramount to the definition of \textit{good} theories. By \textit{good} theories we mean renormalizable theories, which happen to be both the theories that have a nice behavior under renormalization and the theories that have nice symmetries, namely gauge symmetry. In other words, in quantum field theory, the dynamics can be \textit{found} using renormalization and the given theories reproduce classical symmetries in a very precise sense. We can nourish a similar hope for quantum gravity. Let us define a coarse-graining or renormalization flow for a quantum theory of geometry and the renormalizable theories will be the interesting theories and, we can hope for this, will be the theories reproducing the Dirac algebra.

There is a link between these two approaches: coarse-graining discrete theories and renormalizing continuum ones. This can be seen in lattice \ac{QCD} for instance: the renormalizable parameters of usual \ac{QCD} correspond to the \textit{relevant} parameters of the discrete theory, that is parameters that are not eaten up by taking of the continuum limit. In the case of quantum gravity, we expect the theory to be finite. But the two distinctions should remain. It might be possible to define any dynamics, even in particular, discrete ones. But only some parameters might be relevant under coarse-graining. The distinction between renormalizable and non-renormalizable could also remain in the realization of the Dirac algebra. And indeed, it was argued \cite{Perez:2005fn} that the finiteness of the theory will lure back in the quantum ambiguities (allowing more or less any discrete theory to be written) but that the Dirac algebra (or more generally the realization of a well-chosen symmetry) might select back the correct theories. This would indeed be similar to the way gauge theories are selected by renormalizability criteria. This programme seems very similar to the asymptotic safety scenario. And indeed, eventhough the context is a bit different, we would be looking for a fixed point under renormalization and a characterization of the critical surface (given all the interesting theories). The difficulty, as for the asymptotic safety programme, is also to find nice cut and nice writing of the equations so that the renormalization flow might be closed.

The two aspects of renormalization, that is the discrete approach with coarse-graining and the continuum approach are linked in the case of quantum field theory. The coarse-graining approach defines a Wilson flow while renormalization does not use (fundamental) cut-offs but both of these give access to the same physical idea and are linked in their results. We should expect something similar in the case of quantum gravity. This will be one of our focus point in this thesis. To be sure, we will mainly study the coarse-graining process, but in between we will highlight a possible link with a more continuous process which might correspond to the definition of a renormalization flow for quantum gravity.

\medskip

Therefore, in this chapter, we introduced the problem of the dynamics of loop quantum gravity. We quickly stated Thiemann's trick to quantize the Hamiltonian constraint of general relativity but we rapidly moved to the various difficulties it involves, in particular regarding the realization of the Dirac algebra which should, at the quantum level, implement the (spatiotemporal) diffeomorphism invariance. We quickly surveyed the work in this field concentrating on toy models to finally turn the three-dimensional case to study it as a guideline. Finally, we explained how coarse-graining and renormalization should enter this big picture in order to help us finding the correct dynamics for the theory.

As we explained, one of the major problem is to define nice cut-off for the renormalization flow to be closed. One way to search such cut-offs is to look at known coarse-grained dynamics as their form should be close to the expression needed. This is why we will turn to \ac{LQC} in the next chapter since it implements the most large scale theory we can think of and therefore its dynamics should be close to the kind of dynamics we expect for coarse-grained loop quantum gravity.

%But there is a second more fundamental reason for considering coarse-graining: it is precisely because the dynamics is not well-set yet. Indeed, coarse-graining might help us in finding a \textit{continuum} limit of some discrete dynamics. As such, it might help define the correct dynamics and maybe, as a side benefit, help us control the expansion. At this stage, it is of course a hope, but a tantalizing one. And this is what we are goind to study in the rest of this thesis.

%\textbf{TODO:} particular discussion: 3d case, density wieght two (because we can go discrete), definition, close to euclidean general relativity, may solve the problem of anomaly, do not separate spatial and time, so quite different

%*****************************************
%*****************************************
%*****************************************
%*****************************************
%*****************************************

%*****************************************
\chapter{Toward large scales: loop quantum cosmology} \label{ch:LQC}
%*****************************************

\inspiquote{``I thought... well, I started to think you were just a madman with a box.''}{Amy Pond}

In an ideal world, we would start with some dynamics of \ac{LQG}, study the continuum limit in some way (coarse-graining for instance) and find that the first order is classical \ac{GR} with maybe some quantum corrections. In an even better world, these deviations would be testable quite soon and would be confirmed experimentally. Whether such a world might exist or is possible is up to debate, but for sure, it is not our current world. We have no dynamics to start with and no way, currently, to study systematically the continuum limit of \ac{LQG} and certainly no definite testable predictions of it. While waiting for such an advance, we can only hope for tests and studies in simplified settings.

As we want to study coarse-graining, the natural setting is the one of cosmology. Indeed, this is large scale enough to be somehow related to coarse-graining and the setting is quite simple. A natural hypothesis for the study of the universe at large is that it is quite homogeneous and isotropic. Concentrating on homogeneous and isotropic spacetime, we might get a chance to study the dynamics of \ac{LQG}. Moreover, this setting is important with respect to deviations from usual \ac{GR}: as the density of the universe grows as we go back in time in the history of the universe, it approaches Planck density and so, we expect corrections from the quantum regime. It might even be testable for example by carefully looking in the cosmic microwave background \cite{Ashtekar:2015dja}. This is also a setting where we want deviations as the initial singularity of the Big Bang is more of a red flag indicating something wrong in the theory than a genuine prediction. It has long been hoped that a quantum theory of gravity would avoid singularities and therefore testing it in the cosmological setting seems natural.

There is another reason to consider \ac{LQC}: it is its dynamics. Indeed, two dynamics were proposed in the early usually denoted the $\mu_0$ and $\overline{\mu}$ schemes. Only one of them reproduces the large scale dynamics correctly. The problem comes from the regularization of the curvature operator which can be approximated by an holonomy around a loop. But as in coarse-graining, it might be very natural to consider loops with increasing size as the universe grows. This leads to the $\mu_0$ scheme and incorrect semi-classical behavior. And indeed, we rather expect the corrections to be around a loop of fixed size with respect to Planck units. This is the $\overline{\mu}$ scheme. As these two distinctions will come naturally in the coarse-graining process, it is interesting to study them \textit{in the redux}. Moreover some recent work, even without any coarse-graining, but just because of cylindrical consistency, shows that defining infinitesimal operators (for diffeomorphism for instance) requires some kind of $\overline{\mu}$ scheme \cite{Laddha:2011mk}.

In this chapter, we will briefly review \ac{LQC} (for a complete review see \cite{Ashtekar2003}) which is precisely the theory of a homogeneous universe quantized with \ac{LQG} techniques. Though \ac{LQC} is not strictly restricted to homogeneous and isotropic universes, in this short chapter, we will only consider to those. We will first recall the basic assumptions of homogeneous isotropic cosmology. In the second section, we'll develop the quantum version of the dynamical equations and discuss the various ambiguities in the kinematical setting that are the parallel of the choice of representation in \ac{LQG}. And finally, we will rapidly discuss how this might be related to actual states and dynamics in \ac{LQG}.

\section{The cosmological setting}

In this first section, we shall start by briefly restating what was described in the first chapter about cosmology but differently, of course, in order to carry us to \ac{LQC}. So, let's consider a homogeneous and isotropic universe. By homogeneous and isotropic, we mean that it is \textit{spatially} homogeneous and isotropic. The time direction is granted much more freedom and will rather be determined by the matter content and the dynamics. In particular, we do not assume any homogeneity in time. Otherwise, we would be rather restricted, up to topology, to flat, De-Sitter and anti-De-Sitter spaces. We restrict ourselves therefore to spacetimes with a high degree of symmetry but not a maximal one. The \textit{space} slices though have a maximal degree of symmetry corresponding to the existence of a maximal number of Killing vector fields. Each one of this field encodes an infinitesimal transformation which does not transform the metric (up to gauge). For a $3$ dimensional space, the maximum number of such fields is $\begin{pmatrix} 3+1 \\ 2 \end{pmatrix} = 6$ corresponding to three rotations and three translations or their equivalents in curved spaces.

In essence, we hope to be able to quantize the symmetry reduced version of \ac{GR}. In principle, this is not the correct way to symmetry reduce a system. Indeed, we should start by quantizing the full theory and then in the quantum regime we should reduce by symmetry. Here, we are trying to do the opposite: start by reducing with symmetry and then quantize. This means, we are hoping for the following diagram to commute:
\begin{center}
  \begin{tikzpicture}[scale=0.8]
    \coordinate (A) at (0,0);
    \coordinate (A1) at (1.5,0);
    \coordinate (A2) at (0,-1);
    \coordinate (B) at (8,0);
    \coordinate (B1) at (6.5,0);
    \coordinate (B2) at (8,-1);
    \coordinate (C) at (0,-4);
    \coordinate (C1) at (1.5,-4);
    \coordinate (C2) at (0,-3);
    \coordinate (D) at (8,-4);
    \coordinate (D1) at (6.5,-4);
    \coordinate (D2) at (8,-3);

    \draw[->,>=stealth] (A1) -- node[midway,above]{symmetry reduction} (B1);
    \draw[->,>=stealth] (A2) -- node[midway,left]{quantize} (C2);
    \draw[->,>=stealth] (B2) -- node[midway,right]{quantize} (D2);
    \draw[->,>=stealth] (C1) -- node[midway,above]{symmetry reduction} (D1);

    \draw (A) node{GR};
    \draw (B) node{cosmology};
    \draw (C) node{LQG};
    \draw (D) ++(45:0.5) node{LQC};
    \draw (D) ++(45:-0.5) node{?};
    \draw (D) ++(135:0.7) -- ++ (-45:1.4);
  \end{tikzpicture}
\end{center}
Of course, it does not always work. But it also works in some examples. For instance, the hydrogen atom spectrum is excruciatingly hard to derive from quantum electrodynamics. But if we start from the non-relativistic limit and froze nearly all the degrees of freedom of the electromagnetic field, the quantization is possible and the spectrum can be derived quite easily. As it was argued in \cite{Ashtekar2003}, we are hoping for something similar here and this might not be so misguided. So let's start studying cosmology by first reducing and then quantizing.

From the metric point of view, we consider a metric of the form given by the following line element:
\begin{equation}
  \mathrm{d}s^2 = -N(t)^2 \mathrm{d}t^2 + a(t)\left(\frac{\mathrm{d}r^2}{1-kr^2} + r^2 \mathrm{d}\Omega^2\right)
\end{equation}
The metric is expressed in some polar coordinates for space. $r$ is a (fiducial) radius and $\mathrm{d}\Omega$ represents an infinitesimal angle. The function of $t$, $a(t)$ represents the size of the universe. It is called the scale parameter or scale factor. In a spherical universe, it can even be interpreted as the radius of the universe. But otherwise, it is more though of as a fiducial reference: it has a physical sense only up to a factor and keeps track of the evolution of the size of the universe. Finally, the parameter $k$ gives the sign of the curvature. It can be $0$ for flat space, $+1$ for spherical universes and $-1$ for hyperbolic universes\graffito{Note that the curvature $k$ is the curvature of \textit{space} and not \textit{spacetime}. In particular, the genius of \ac{GR} is to understand that even for $k=0$ spacetime can be curved and this precisely encodes the expansion.}. In the previous expression for the metric, the parameter $t$ is a time parameter but is totally arbitrary. Indeed, this is why we kept the lapse $N(t)$ which will change accordingly when changing the time parametrization. We do this to keep the temporal diffeomorphisms a symmetry of the theory, even though we've fixed the gauge for the spatial parts. Keeping this invariance allows us to stay close to some difficulties raised in \ac{LQG} and see how they are resolved in this context. 

Classically, the equations of motion are\graffito{There are actually two equations of motion but because of the symmetry, they are equivalent.}:
\begin{equation}
\left(\frac{\dot{a}}{a}\right)^2 + \frac{k}{a^2} = \frac{8\pi G}{3 a^3} H_\textrm{matter}(a)
\end{equation}
where $H_\textrm{matter}$ is the matter Hamiltonian. They can be derived either by considering the symmetry reduction on Einstein equations or by first simplifying the action and then deriving the equation with respect to the variation of $N$ in the Lagrangian. Our concern is to quantize this equation. For this, we should turn to the Hamiltonian formalism. The variables are $a$ and its conjugate momentum $p_a$. A standard analysis will reveal:
\begin{equation}
p_a = - \frac{3V_0}{4\pi G} \frac{a \dot{a}}{N}
\end{equation}
where $V_0$ is the volume of the universe at some time $t_0$ (or the volume of a fiducial cell if it is not well-defined)\graffito{The \textit{fiducial cell} is a device for studying flat or more generally non-compact universes. Because their volume is not well-defined, we need to define a small  block whose evolution can be tracked. This is precisely the fiducial cell.}. $N$ does not appear as a dynamical variable but is rather a Lagrange multiplier. Therefore its conjugate momentum checks the constraint $p_N = 0$ which generates time diffeomorphisms. All this can be factored in the following Hamiltonian constraint:
\begin{equation}
H = -\frac{2\pi G}{3} \frac{p_a^2}{V_0 a} - \frac{3}{8\pi G} V_0 ak
\end{equation}
This is the Hamiltonian we want to quantize. We've already quickly looked over the standard quantization but we will now consider another avenue which is inspired by \ac{LQG} techniques.

%\textbf{TODO:} ideal case: coarse-graining perfectly done, continuum limit recovered and it is general relativity in the semi-classical regime, in practice - at least in the short run - should only hope for tests, first test is cosmology (large scale + testable + simplifying assumptions)

\section{Loop quantization}

In order to connect to \ac{LQG}, we should first consider connection variables as they are the fundamental variables. The number of variables should not change though. Of course, we might add some variables because of the local Lorentz invariance but for the gauge invariant variables, their number cannot change.\graffito{Again \cite{Ashtekar2003} is an excellent reference for the toying around part.} After toying around, we get that the triad is entirely given by its norm $\tilde{p}$ and the connection by its component $\tilde{c}$. They can be expressed in term of the metric variables as:\graffito{Note that the triad gets an orientation information which is missed by the metric.}
\begin{equation}
  \begin{array}{rcl}
    |\tilde{p}| &=& \frac{a^2}{4} \\
    \tilde{c} &=& \frac{1}{2}\left(k+\gamma \dot{a}\right)
  \end{array}
\end{equation}
where $\gamma$ is the Immirzi parameter. Their conjugation relations can be written as:
\begin{equation}
\{\tilde{c},\tilde{p}\} = \frac{8\pi \gamma G}{3V_0}
\end{equation}
It is customary to reabsorb the volume factor by defining:
\begin{equation}
  \begin{array}{rcl}
    p &=& V_0^{2/3} \tilde{p} \\
    c &=& V_0^{1/3} \tilde{c}
  \end{array}
\end{equation}
which finally gives:
\begin{equation}
\{c,p\} = \frac{8\pi \gamma G}{3}
\end{equation}
The natural step would now be to define wavefunction over $c$ and define $\hat{p}$ as the conjugate operator (up to a factor) which acts by derivation. But as we saw earlier, this would lead to the unique representation of quantum mechanics from von Neumann's theorem \cite{v.Neumann1931} which we know to fail or at least not to resolve the singularity problem.

We must therefore consider another representation and find the correct hypothesis of the theorem we should let go of. Actually, the uniqueness of the representation is not the uniqueness of the Heisenberg \textit{algebra} but of the Heisenberg \textit{group}. It means the natural operators are in fact exponentiated version of $c$ and $p$. Let us therefore define:
\begin{equation}
  \begin{array}{rcl}
    V_\lambda &=& \exp \left(\mathrm{i}\lambda c\right) \\
    W_\mu &=& \exp \left(\mathrm{i}\mu p\right)
  \end{array}
\end{equation}
with the corresponding commutation relations:
\begin{equation}
\{V_\lambda,W_\mu\} = \frac{8 \mathrm{i}\pi \gamma G}{3} (\lambda + \mu) V_\lambda W_\mu
\end{equation}
The advantage of these variables is that they will naturally be represented (quantum mechanically) by well-defined operators on the whole of the Hilbert space (as multiplying by an imaginary exponential does not change the convergence properties of an $\mathrm{L}^2$ function).

Now, the crucial property of the Von Neumann representation is the (weak) continuity of the operators $V_\lambda$ and $W_\mu$ in $\lambda$ and $\mu$ respectively. If we were to forget at least one of these continuity hypothesis, other representations would be allowed. Indeed, in \ac{LQG}, only holonomies operator exist, not connection. Therefore, a very natural continuity hypothesis to forget is the continuity of $V_\lambda$ since we do not expect $c$ to exist as an operator if it comes out of \ac{LQG}. This means that only the exponentials of $c$ are well-defined at the quantum level. $p$ is still represented by a derivative though. We will now develop this new representation but with an added convention to make the notations easier:
\begin{equation}
\frac{8\pi \gamma G}{3} = 1
\end{equation}
This way $c$ and $p$ are exactly conjugate without any factor to worry about.

We have two operators.\graffito{Though we use an exponential in the notations, we shouldn't think that it is indeed an exponential. There is no operator $\hat{c}$ of which this operator is the exponential. This is just, once again, a bad notation.} The first one is the exponentiated operator $\widehat{\mathrm{e}^{\mathrm{i}\lambda c}}$ and the second one is the conjugate momentum $\hat{p}$. Formally, we have:
\begin{equation}
  \hat{p} = -\mathrm{i} \frac{\mathrm{d}}{\mathrm{d}c}
\end{equation}
Now the Hilbert space has a natural basis:
\begin{equation}
  \{|\lambda \rangle = |\mathrm{e}^{\mathrm{i}\lambda c}\rangle\}_\lambda
\end{equation}
The operators act on this basis as:
\begin{equation}
  \begin{array}{rcl}
    \widehat{\mathrm{e}^{\mathrm{i}\lambda c}} |\mu \rangle &=& |\mu + \lambda\rangle \\
    \hat{p} |\mu \rangle &=& \mu |\mu  \rangle
  \end{array}
\end{equation}
which means that the basis is constituted by the eigenvectors of $\hat{p}$ (they are planewaves) and the exponentiated operators act as the intuition would command that is as translation operators. The most important now is that this basis is normalizable. More precisely:
\begin{equation}
  \langle \lambda | \mu \rangle = \delta_{\lambda,\mu}
\end{equation}
This is a paramount point. Indeed, note here that this is a Kronecker's delta, not a Dirac delta. This normalizability means a few things. First it prevents the weak continuity of the translation operators. It can be seem quite easily as follows:
\begin{equation}
  \langle \mu | \widehat{\mathrm{e}^{\mathrm{i}\lambda c}} |\mu \rangle = 0, ~\forall \lambda\neq 0
\end{equation}
In particular:
\begin{equation}
  \langle \mu | \widehat{\mathrm{e}^{\mathrm{i}\lambda c}} |\mu \rangle \xrightarrow[\lambda \rightarrow 0]{}  0
\end{equation}
where we would have expected $1$ under weak continuity. It also means that the representation is \textit{unitarily} inequivalent to the usual Schrodinger representation of quantum mechanics. Therefore, we might expect a different behavior especially in the small $c$ regime.

In the next section, we will indeed see the differences. The first point will be to find a natural quantization of the Hamiltonian (constraint) in this language and see how this corrects the singular behavior of Friedman-Lemaitre cosmologies.

\section{Deviations from standard quantum cosmology}

Let us simplify the problem as much as we can. As we saw in the first chapter, it is customary to use a scalar field as a physical clock. The Hamiltonian constraint can then be expressed as a flow along increasing (or decreasing scalar field). Up to a factor (which can be absorbed in a redefinition of the field), the gravitational part of this Hamiltonian is:
\begin{equation}
  H = p c
\end{equation}
It can also be considered as the Hamiltonian constraint of pure \ac{LQC} that is \ac{LQC} with no matter.

Up to ordering, this Hamiltonian is easy to quantize in the Schrodinger picture. But, in our case, we don't have any $\hat{c}$ operator and therefore this needs correction. The problem is more stringent that it looks: as there is no $\hat{c}$ operator, there cannot be any way of reproducing the (would-be) action of $\hat{c}$ by combination of well-defined operators. We are forced to change the equation of motions.

The natural thing to do is to consider approximations of $c$ using exponentials. For instance, it might be possible to write:
\begin{equation}
  c \simeq \frac{\mathrm{e}^{\mathrm{i}\lambda c} - \mathrm{e}^{-\mathrm{i}\lambda c}}{2\mathrm{i}\lambda}
\end{equation}
where the approximation is better when $c \rightarrow 0$. This amounts to replacing $c$ by a sine. Because $c$ encodes the curvature, this is quite natural. Indeed, the curvature is evaluated in \ac{LQG} using \textit{finite}\graffito{``finite'' here means ``not infinitesimal''. Maybe ``cofinite'' would be a better suited word as it is rather the inverse which is finite.} loops. This means that the connection is exponentiated. Intuitively, we expect the typical length of the loop to be around Planck size and we should look for a dimensionful parameter which would help us select this size. Sadly $\lambda$ is dimensionless and therefore it cannot be used for such a scheme. We could still try changing $H$ to a corrected version $H_0$:
\begin{equation}
  H_0 = p \frac{\sin \lambda c}{\lambda}
\end{equation}
This is called the $\mu_0$ scheme and does not lead to the correct classical limit. It can be seen quite easily: no scale appears in this Hamiltonian. Quantizing this, the quantum effect will be of order of $\sqrt{p}$ that is the size of the universe and will grow with it. So there is definitely something wrong with this scheme.

We can therefore constraint the scale the other way around: knowing that the quantum effects should be of Planck scale $\ell_p$ we can choose $\lambda$ to be:
\begin{equation}
\lambda = \frac{\ell_p}{\sqrt{|p|}}
\end{equation}
This guarantees that the quantum effects are controlled. The new parameter $\ell_p$ is dimensionful and control the size of the probing loop. At the classical level, this gives the new regularized Hamiltonian:
\begin{equation}
  \overline{H} = \mathop{sgn} (p)\frac{|p|^{\frac{3}{2}}}{\ell_p} \sin \frac{\ell_p c}{\sqrt{|p|}}
\end{equation}
which of course has the correct limit when $c \rightarrow 0$ and is the $\overline{\mu}$ scheme Hamiltonian.

But this creates a new problem: what is the quantum operator corresponding to $\sin \frac{\ell_p c}{\sqrt{|p|}}$? Of course, we cannot put simply the newly found $\lambda$ factor into the translation operator as the coefficient now depends on the basis vector. There is a solution to this conundrum. The quantity $\frac{c}{\sqrt{|p|}}$ (which is the small $c$ limit of the sine) is conjugated to the volume $\mathop{sgn}(p)|p|^{\frac{3}{2}}$ therefore, the operator should acts as a simple shift on the eigenvectors $|V\rangle$ of the volume. But these eigenvectors are just a relabelling of the eigenvectors of $\hat{p}$. Using a normalized version of the volume:
\begin{equation}
  \nu = \frac{V}{2\pi \gamma G \ell_p}
\end{equation}
we simply have:
\begin{equation}
  \widehat{\exp\left(\mathrm{i}\frac{\ell_p c}{\sqrt{|p|}}\right)} |\nu\rangle = |\nu + 2 \rangle
\end{equation}
Allowing a full definition of the Hamiltonian at the quantum level, up to ordering ambiguities as usual.

We were mostly interested in the different scheme and how they are resolved. But let's refer the reader to more detailed and more general accounts on LQC. There are various reviews on the subject detailing recent developments \cite{Ashtekar:2015dja,Ashtekar2011,Bojowald2008}. The mathematical details are most explicit in \cite{Ashtekar2003}. One of the major advantage of LQC is its simpler structure which allows the study of genuinely difficult question in the full theory. In particular, it should be noted that using self-dual variables is possible in this setting \cite{Wilson-Ewing:2015lia}.

\section{Relation to Loop Quantum Gravity}

For us, there is still a question though: what is the connection with full \ac{LQG}? How similar the theories really are? To what extent does homogeneous \ac{LQG} resemble \ac{LQC}? This can be answered in two ways. First, we can focus on the genuine resemblance and establish some link, though tentative for some parts. Second, we can consider how \ac{LQC} might be a simpler framework telling us how to solve difficult questions in the full theory.

For the first point, it should be noted that the mathematical structure is actually quite close to that of \ac{LQG}. Indeed, apart from the continuum limit issue (and all the technology around it as the cylindrical consistency conditions), the process of construction of the states of the theory is rather similar. The algebra of observable is taken to be the exponentials and the derivative in LQC which is the direct parallel of considering the algebra of holonomies and the fluxes (which act as derivation). The states created by the holonomies or the exponentials are all normalizable and the connection operator ($c$ in \ac{LQC}) does not exist. Even the Hilbert spaces has some similarities. Indeed, the space defined for \ac{LQC} is the space of all quasi-periodic wavefunctions over $\mathbb{R}$. According to Gelfand theory, there exists a set $\overline{\mathbb{R}}$ such as its algebra of square integrable functions $\mathrm{L}^2(\overline{\mathbb{R}})$ (with the Haar measure) is indeed the space of quasi-periodic functions on $\mathbb{R}$. As this algebra is unital, $\overline{\mathbb{R}}$ is compact. It is therefore a compactification of $\mathbb{R}$, called the Bohr compactification of the real line, different from the usual point-at-infinity compactification. From a Fourier perspective, the dual to $\overline{\mathbb{R}}$ is the set $\mathbb{R}$ with discrete topology, which explains the normalizability of the plane waves. This compactification of $\mathbb{R}$ is the parallel of switching from the set of connections to the set $\overline{A}$ of \textit{distributional} connections on the manifold in \ac{LQG}.

For the second point, this is where LQC is really interesting for us, in the context of coarse-graining. Indeed, the possibility of writing the dynamic should guide us here. First, in the coarse-graining process, it is very natural to define loop operators that are \textit{macroscopic} meaning of the size of the coarse-grained graph. Using a simple holonomy cannot work there as is now obvious from \ac{LQC}. We will need an equivalent of the $\overline{\mu}$ scheme. How this scheme is devised (using the algebra of the volume) might also be a nice clue on how to write the dynamics of coarse-grained models. Even without coarse-graining, the simple fact that we are dealing with discrete excitations for a continuum theory seems to imply $\overline{\mu}$-like formulations, for instance for the diffeomorphism constraints \cite{Laddha:2011mk}.

There is a second interest of \ac{LQC}: as it is way simpler than the full theory, eventhough it is derived in a non-rigorous manner (symmetry reducing before quantizing), we can be quite confident in the results it provides. In particular, it means that \ac{LQC} can be a goal to obtain in some models. In other words, a good test of techniques applied to \ac{LQG} is to apply them in the cosmological setting and compare them to \ac{LQC}. This was done for example with the condensate approach to cosmology \cite{Gielen:2013naa} quite popular in \ac{GFT}.

\vspace{1em}

In this chapter, we have introduced the idea of \ac{LQC}. This theory studies the quantum corrections to cosmology by considering the classically reduced system and quantize it in a \textit{loopy} fashion. The resulting theory is interesting for at least two regards (with respect to coarse-graining): it helps finding what the actual large-dynamics is, and how to implement it on a coarse graph. Because the operators we need in the Hamiltonian are not all well-defined, we are forced to regularize the curvature operator. This leads to two quantization scheme. The first scheme, called $\mu_0$, is more naive one and does not lead to the right semi-classical limit. The second scheme, $\overline{\mu}$ is the right one as far as we know today and have interesting properties which are linked to problems in coarse-graining. If this represents the top-down approach to coarse-graining, we are still clueless about a bottom-up point of view where we would start with coarse-graining and find a natural large scale description. These two methods may actually meet and be used in concert. In the next chapter, we will therefore turn to the problem of coarse-graining from the point of view of the full-theory.

%*****************************************
%*****************************************
%*****************************************
%*****************************************
%*****************************************

%*****************************************
\chapter{Coarse-graining and gauge invariance} \label{ch:GaugeFixCG}
%*****************************************

\inspiquote{The universe is big, its vast and complicated, and ridiculous. And sometimes, very rarely, impossible things just happen and we call them miracles.}{The Doctor}

Until now, we've reviewed the kinematics of \ac{LQG}. Concerning the dynamics, we have been less ambitious, as this is still an open problem, but we have overviewed several developments in this direction. And there, coarse-graining appears in at least two aspects. From a pragmatic point of view first, even if we were able to devise the correct dynamics, we would still need coarse-graining in order to predict the large scale behavior of the theory. This is why we review in the last chapter the expected results devised from \ac{LQC}. But we might also need coarse-graining in another regard: in helping us find the correct continuum dynamics for the theory. In this regard, cosmology helps us again by signaling toward natural dynamics integrating a cut-off. But this won't help us at all if we have no systematic way of relating scales.

In this chapter, we will therefore concentrate on coarse-graining more precisely. In the context of \ac{LQG}, or of any theory of spacetime, coarse-graining is actually harder than anticipated. Indeed, the natural notion of coarse-graining from usual quantum field theory cannot simply be reproduced as they are usually based on the definition of an energy scale (or at least distance scale). It is difficult to use the same trick in \ac{GR}. Indeed, distance scales are set dynamically which makes the endeavour, if not ill-defined, much more complicated. We might introduce a background metric to measure the energy scales but that breaks diffeomorphism invariance. Therefore, we need a new approach.

A possible way is to make use of the discrete nature of \ac{LQG} and try and adapt techniques from condense matter. Indeed, coarse-graining in this context does not correspond to \textit{scales} but to the \textit{fineness} of excitations considered. The rationale behind this is that small momenta can be captured more or less exactly via a small number of points, while large momenta require finer structures. This observation can be brought back into \ac{LQG} by considering discrete excitations. The more refined correspond to smaller scales while more coarse excitations correspond to large scales. This way, there is no mention of a scale and the scale of excitations is entirely coded into the state by its fineness. The idea can pushed further by considering \textit{exact} dynamics for the discrete excitations. Indeed, in the case of \ac{QFT}, it is possible to reproduce exactly low momenta excitations on a discrete setting. The idea is to do something similar: define cut-offs compatible with the dynamics so that the coarse description remain exact as times goes. This means that a (discrete) truncation of the theory can encode the full continuum dynamics. This is the idea behind \textit{perfect discretizations} \cite{Bahr:2011uj,Bahr:2009qc}.

Note that the discreteness of the excitations might not be linked with a natural discreteness of space. Indeed, this is already the case in quantum field theory. If we want to describe low momenta excitations of a free field on a discrete lattice, it is possible but the map from this theory to the full continuum theory does not send the excitations to point excitations of the field. A more natural description therefore is to link different scales by mapping states of a coarse description into a finer one. Because the cut-offs are compatible with the dynamics, the new degrees of freedom of the finer description must be in a minimum excitation state, a kind of \textit{physical vacuum}. All this is very similar to the kinematical constructions used so far and indeed the property of such a representation and of a corresponding dynamics was dubbed \textit{dynamical} cylindrical consistency \cite{Dittrich:2012jq,Dittrich:2013xwa}. The physical vacuum then corresponds to a fundamental state for the theory and can be constructed explicitly in some cases as for BF theory \cite{Dittrich:2014wpa,Bahr:2015bra}.

Our goal is this chapter is more technical. It is to devise a way to systematically coarse-grain a geometrical theory. It is linked to the whole program but is more on the technical side. Indeed, for a more complex theory than BF theory or a free field on flat spacetime, physical vacua are hard to construct. But the process of coarse-graining can help define such states asymptotically. But for this we need a procedure to link different scales in the current and for this, we present a possible method.

The construction is called coarse-graining by gauge-fixing, was originally introduced in \cite{Livine:2013gna} and is a way to define \textit{effective} vertices in a large scale graph from a small scale one. The method in itself do not loose any degrees of freedom and therefore does not give any truncation procedure. This will be studied in more details in our original work in the following chapters. The method consists in gauge fixing the $\mathrm{SU}(2)$ data inside a bounded region using synchronization trees. The resulting information can naturally be collapsed into \textit{petals} around an effective vertex and unfolding information. This is therefore a natural way to collapse regions of a spin network into a single point.

In this chapter, because the geometrical intuition is paramount in this construction, we will first review the geometrical interpretation of spin networks as twisted and spinning geometries. We will then describe the various problems that arise if we try to coarse-graining naively such a theory, particularly the problem of the closure defect. We will see how this comes naturally from a problem of parallel transport which will allow us to finally consider the starting point of our work: coarse-graining by gauge-fixing.

%\vspace{1em}
%\textbf{PROBLEM:} Where do review BF vacuum, and others????

%*****************************************

\section{The geometrical interpretation of spin networks}

So, as advertised, let's consider the geometrical interpretation of spin networks. As we have seen, spin networks are quantum states of the theory encoded in (knotted) graphs colored with spins on the links and intertwiners on the nodes. These states can be understood as the quantum states diagonalizing some discrete operators \cite{Barbieri1998}. The relevant discrete operators are:
\begin{itemize}
\item the (open) holonomy along the link of the graphs and,
\item the integrated triad over surfaces dual to the links.
\end{itemize}
These operators do have a counter part in a classical but still discrete setting. This does not correspond to a straightforward classical limit of the continuum theory. We can find a similar situation in electromagnetism if we consider the classical limit but with the number of photon fixed. Normally, the relevant states for the macroscopic world are coherent states. But we can also consider a fixed number of photon with the limit of large quantum numbers giving a classical limit to the fixed number of photon theory.

Here, we do the same thing: we fix a graph structure and consider the large quantum number limit. In that case, the classical limit that appears is a discrete theory with support on the given graph. Algebraically, we have a vector on each end of the links and a group element on each link. We have furthermore two constraints:
\begin{itemize}
\item the closure constraint: at each vertex, the sum of outgoing vectors must sum up to zero. This encodes the local gauge invariance;
\item the matching constraint: on a given link, the two vectors must be images of one another (up to a sign) through the group element. This implies in particular that the norm of the two vectors must match. Geometrically, this means that the area of the surface separating the two nodes is the same from whatever point of view we choose.
\end{itemize}
Now that this structure is classical, we can try and interpret it classically. Indeed, our structure is combinatorial, being based on a graph and algebraic data supported by this graph. But as we are entering the realm of geometric theories in the discrete setting, we expect an interpretation of the structure as a discrete manifold or as a parametrization of some set of discrete manifolds. And indeed, it is.

The natural way to interpret the construction is through Minkowski's theorem \cite{Alexandrov1958,Minkowski1897}. It states the following: given a convex polygon in flat space, we can define all its normal as the vectors that are outgoing orthogonal to their face and with a norm equal to the area of the face. Then, the sum of all the normals sums up to zero. This is quite natural, at least from the integrated triad point of view, as in flat space we can gauge fix the triad and then its integration is exactly the normals. And of course, flat space is a solution of the Einstein equations, in particular, it satisfies the gauge condition. We can even see here that the convex requirement for the polygon is not necessary. It is however needed for the converse statement: given a set of vectors such that their sum is zero, there is a unique \textit{convex} polyhedron in flat space such that these vectors correspond to the normals of the polyhedron. This is Minkowski's theorem.

This mandates the natural interpretation of spin networks as twisted geometries: each node is interpreted as a convex polyhedron given by its normals and they are connected according to the connection data given by the links. The group element carried by the link is the natural discrete connection encoding the curvature at the gluing face. The geometry is called twisted because the matching faces do not necessarily have the same shape, though they must have the same area. This can be interpreted as torsion or discontinuities of the metric at the gluing point \cite{Freidel:2010aq}.

There is a second interpretation, developed by Freidel \textit{et al.}, which will be of particular significance later on: the spinning geometry interpretation \cite{Freidel2014}. The existence of such another interpretation should not be surprising: when we are interpreting the discrete setting, we are looking for a subset of continuum configurations that matches the discrete phase space. It is not surprising that several possibilities exist. If twisted geometry is such a possibility, spinning geometry is a second one. The starting point of the spinning geometry interpretation is actually the same one: the closure constraint. But it is then interpreted as a Bianchi identity for some connection. Therefore, it is possible to reconstruct a connection such that its holonomies give the triad vectors. The advantage of such a construction is that, the generalization of normal vectors can now be computed for non flat surfaces as long as they are embedded in flat space. In fact, it allows for spinning edges of the polyhedron, hence the name. In this construction, the metric can be continuous as long as we allow for torsion.

In any case, the interpretations look similar and maybe there is a global theme here for all imaginable interpretations: a polyhedron or more generally a solid in flat space is associated to each node, this association being granted by the closure constraint. They are then connected with the connection of the links, the connection being impersonated by the matching constraint. The global picture is one of discrete geometry with flat chunks of spaces connected to each other allowing curvature at hinges through the gluing.

%\textbf{TODO:} idea of coarse-graining, geometrical interpretation of spin networks, new networks with chunks of geometry, description of twisted geometry, word on spinning geometry

%*****************************************

\section{Large scale curvature}

Now armed with a geometrical interpretation, what should coarse-graining should look like in this language? Let us consider some (somewhat vague) coarse-graining scheme imported from condensed matter. If we consider a given graph, the coarse-grained theory would be expressed on a simpler graph following this kind of process: we group together some vertices and decide to describe them collectively. There is now a new graph of the connections of these regions, which have some internal structure. The best case scenario of course would be that the macroscopic degrees of freedom (whatever their precise meaning in a diffeomorphism invariant theory) are precisely described by a theory on this graph. Our goal in this section is to underline two points regarding this: first, it is impossible for such a theory to exist except for very special cases and second, it is possible to interpret geometrically the source of this impossibility.

So what is the problem? If we consider such a coarse-grained graph, as we sketch in the previous paragraph, we will find that the closure condition can no longer hold for the coarse graph. This is easy to see on a simple example. Let us illustrate it on the simplest possible graph where this happens:
\vspace{1em}
\begin{center}
  \begin{tikzpicture}
    \coordinate (A) at (0,0);
    \coordinate (B) at (-60:2);
    \coordinate (C) at (-120:2);

    \draw (A) -- node[midway,sloped]{$>$} (B) -- node[midway,sloped]{$<$} (C) -- node[midway,sloped]{$>$} (A);
    \draw (A) -- node[midway,sloped]{$>$} ++(0,1);
    \draw (B) -- node[midway,sloped]{$>$} ++(-45:1);
    \draw (C) -- node[midway,sloped]{$<$} ++(-135:1);

    \draw[->,>=stealth] (A) -- node[midway,below]{$v_1~~$} ++(-80:1);
    \draw[->,>=stealth] (B) -- node[midway,below]{$v_2$} ++(-160:1);
    \draw[->,>=stealth] (C) -- node[midway,above left]{$v_3$} ++(80:1);

    \draw[->,>=stealth] (A) -- node[midway,above right]{$~v_A$} ++(70:1);
    \draw[->,>=stealth] (B) -- node[midway,above right]{$v_B$} ++(-25:1);
    \draw[->,>=stealth] (C) -- node[midway,below right]{$v_C$} ++(-115:1);
    
    \draw (A) node{$\bullet$} node[right]{$~A$};
    \draw (B) node{$\bullet$} node[above right]{$B$};
    \draw (C) node{$\bullet$} node[above left]{$C$};

    \draw[->,>=stealth,very thick] (2.5,-1.333) -- (4.5,-1.333);

    \coordinate (O) at (6,-1.333);

    \fill (O) circle (0.2);

    \draw (O) -- node[midway,sloped]{$>$} ++(0,2);
    \draw (O) -- node[midway,sloped]{$>$} ++(-45:2);
    \draw (O) -- node[midway,sloped]{$<$} ++(-135:2);

    \draw[->,>=stealth] (O) -- node[midway,above right]{$~\tilde{v}_A$} ++(70:1.2);
    \draw[->,>=stealth] (O) -- node[midway,above right]{$\tilde{v}_B$} ++(-25:1.2);
    \draw[->,>=stealth] (O) -- node[midway,below right]{$\tilde{v}_C$} ++(-115:1.2);
  \end{tikzpicture}
\end{center}
\vspace{1em}
We introduced here a loop and a minimal amount of outside edges, that is three. The fluxes at the source of the edges are noted $v$ with an index labelling the edge. The outgoing edges are labelled by their source vertex. The inner edges are labelled by numbers. We want to coarse-graining this previous graph of a triangle into a single point with edges going out of it. The new fluxes should be related to the original fluxes up to parallel transport. The question is: can the external vectors associated to the outside edges sum up to zero? We must recognize first that we cannot sum these vectors directly because we want a gauge covariant quantity, with a well-defined transformation law. A way out of this conundrum is to use parallel transport.

Let us use point $A$ as a reference point and parallel transport quantities following the path $C \rightarrow B \rightarrow A$. The natural sum for the external vectors is:
\begin{equation}
  \overrightarrow{S} = \overrightarrow{v}_A + g_{BA} \triangleright \overrightarrow{v}_B + g_{BA} g_{CB} \triangleright \overrightarrow{v}_C
  \label{eq:defect_sum}
\end{equation}
where $\triangleright$ is the natural action of $\mathrm{SU}(2)$ on vectors and $g_{IJ}$ is the transport from point $I$ to point $J$. This new quantity $\overrightarrow{S}$ is gauge covariant and simply transforms with the natural action of the gauge group at the point $A$. It is called the \textit{closure defect} as it encodes how much of the closure is missing.

Now, in order to compute this quantity, we will use the closure condition at each point which state:
\begin{equation}
  \left\{
  \begin{array}{rcl}
    \overrightarrow{v}_A + \overrightarrow{v}_1 - g_{CA} \triangleright \overrightarrow{v}_3 &=& \overrightarrow{0} \\
    \overrightarrow{v}_B + \overrightarrow{v}_2 - g_{AB} \triangleright \overrightarrow{v}_1 &=& \overrightarrow{0} \\
    \overrightarrow{v}_C + \overrightarrow{v}_3 - g_{BC} \triangleright \overrightarrow{v}_2 &=& \overrightarrow{0}
  \end{array}
  \right.
\end{equation}
Substituting in (\ref{eq:defect_sum}), we find:
\begin{equation}
  \overrightarrow{S} = g_{BA} g_{CB} \triangleright \overrightarrow{v}_3 - g_{CA} \triangleright \overrightarrow{v}_3 = (H - \mathbb{1}) g_{CA} \triangleright \overrightarrow{v}_3
\end{equation}
where $H$ is the holonomy around the whole loop ad reads:
\begin{equation}
  H = g_{BA} g_{CB} g_{AC}
\end{equation}

As we can see, the closure defect exists precisely because the holonomy around the loop is not trivial. In fact, it is possible to have closure and still a non-trivial holonomy if the action of the loop is considered on a correctly aligned vector. But in the general case, the curvature of the loop will cause the closure condition to be relaxed on the large scale.

This actually calls for a natural geometric interpretation. Indeed, as we said in the previous section, curvature is encoded in hinges, that is on loops of the graph. And we saw here that this curvature is contained in the gluing of the polyhedra. Therefore, we should have expected the failure of the closure condition: curvature is encoded in the finer graph and, with a naive coarse graph, is entirely missed in the coarse description. It would not even be allowed if we were to enforce the closure condition. Or to put it quite simply: curvature can build up at large scales (see figure \ref{fig:curvature_build}). Interestingly, this also supports the geometrical interpretations of the nodes as flat pieces, as the closure condition applies to them.

\begin{figure}[h!]

  \centering

  \begin{tikzpicture}

    \def \scale {1.2}
    \def \d {0.2}

    \coordinate (A1) at (1*\scale,\d,0);
    \coordinate (A2) at (0.5*\scale,\d,0.866*\scale);
    \coordinate (A3) at (-0.5*\scale,\d,0.866*\scale);
    \coordinate (A4) at (-1*\scale,\d,0);
    \coordinate (A5) at (-0.5*\scale,\d,-0.866*\scale);
    \coordinate (A6) at (0.5*\scale,\d,-0.866*\scale);
    \coordinate (B) at (0,0.8+\d,0);

    \draw (A6) -- (A1) -- (A2) -- (A3) -- (A4);
    \draw[dashed] (A4) -- (A5) -- (A6);

    \draw[blue,thick] (A1) -- (B);
    \draw (A2) -- (B);
    \draw (A3) -- (B);
    \draw (A4) -- (B);
    \draw[dashed] (A5) -- (B);
    \draw (A6) -- (B);

    \draw (A1) node {$\bullet$};
    \draw (A2) node {$\bullet$};
    \draw (A3) node {$\bullet$};
    \draw (A4) node {$\bullet$};
    \draw (A5) node {$\bullet$};
    \draw (A6) node {$\bullet$};
    \draw (B) node {$\bullet$};

    \draw[->,>=stealth,thick] (2,0.5) -- (4,0.5);

    \coordinate (P) at (6.5,0.5);

    \def \step {50}
    \def \r {1.4}

    \coordinate (C1) at ($(P) + (0*\step:\r)$);
    \coordinate (C2) at ($(P) + (1*\step:\r)$);
    \coordinate (C3) at ($(P) + (2*\step:\r)$);
    \coordinate (C4) at ($(P) + (3*\step:\r)$);
    \coordinate (C5) at ($(P) + (4*\step:\r)$);
    \coordinate (C6) at ($(P) + (5*\step:\r)$);
    \coordinate (C7) at ($(P) + (6*\step:\r)$);

    \fill[lightgray] (C7) -- (P) -- (C1) -- cycle;
    \draw[red,thick,dashed] (C1) -- (C7);

    \draw[blue,thick] (P) -- (C1);
    \draw (P) -- (C2);
    \draw (P) -- (C3);
    \draw (P) -- (C4);
    \draw (P) -- (C5);
    \draw (P) -- (C6);
    \draw[blue,thick] (P) -- (C7);

    \draw (C1) -- (C2) -- (C3) -- (C4) -- (C5) -- (C6) -- (C7);

    \draw[red,thick] ($(P) + (6*\step:0.3*\r)$) arc (6*\step:360:0.3*\r);

    \def \s {0.01}
    \def \si {0.9}

    \draw[->,>=stealth] ($(P) + (0.5*\step:\r*\si)$) -- ++(0.5*\step:\step*\s);
    \draw[->,>=stealth] ($(P) + (1.5*\step:\r*\si)$) -- ++(1.5*\step:\step*\s);
    \draw[->,>=stealth] ($(P) + (2.5*\step:\r*\si)$) -- ++(2.5*\step:\step*\s);
    \draw[->,>=stealth] ($(P) + (3.5*\step:\r*\si)$) -- ++(3.5*\step:\step*\s);
    \draw[->,>=stealth] ($(P) + (4.5*\step:\r*\si)$) -- ++(4.5*\step:\step*\s);
    \draw[->,>=stealth] ($(P) + (5.5*\step:\r*\si)$) -- ++(5.5*\step:\step*\s);

    \draw[->,>=stealth,red,thick,dashed] ($(P) + ({180+3*\step}:1.2)$) -- ++({180+3*\step}:{360*\s-6*\step*\s});

    \draw (P) node {$\bullet$};

    \draw (C1) node {$\bullet$};
    \draw (C2) node {$\bullet$};
    \draw (C3) node {$\bullet$};
    \draw (C4) node {$\bullet$};
    \draw (C5) node {$\bullet$};
    \draw (C6) node {$\bullet$};
    \draw (C7) node {$\bullet$};

  \end{tikzpicture}

  \caption{On this figure, we represented the dual graph of a 2d trivalent graph. The curvature at the vertex manifests itself as a defect in the closure condition. This can be seen by flattening the triangulation, which amounts to gauge-fix the variables. The curvature manifests itself as a gap (in gray on the figure) at some edge (in blue on the figure) in the flattened manifold. The closure defect can be seen as the missing normal coming from the closure of the flattened polygon (in red on the figure).}
  \label{fig:curvature_build}

\end{figure}
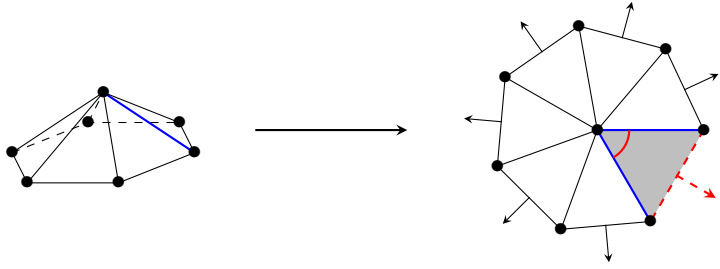

More physically, this means that gravity gravitates. Gravity interacts with itself (it is a non-linear theory after all) and can be a source of gravity. This can be seen for example in the mass of a star which is usually heavier than its constituents. We expect therefore gravity to have some weigh when curvature is non-zero. As this manifests itself through the closure failure, it is an interesting road to consider that the relevant macroscopic quantity would be the closure defect.\graffito{Of course, it is usual for a system to be heavier than the sum of its parts because of the mass-energy equivalence. But note that in \ac{GR}, nearly everything counts as matter. Photons count as matter. Even gluons do. So the added energy is gravitational, necessarily.} At this stage, however, let's just note that this problem is ubiquitous to non-linear theories and therefore to non-abelian gauge theories. We do have the same problem in QCD and the same kind of defect appear for instance in lattice QCD.

%\textbf{TODO:} coarse-graining procedure in principle, intuitive problem: might be curved, equivalent in algebraic terms: closure defect, example in simple graph, explicit link to curvature, interpretation as charge of gravity (gravity weighs)

%*****************************************

\section{Coarse-graining by gauge fixing}

In this section therefore, we will concentrate on the kind of structure we need for the coarse-graining of \ac{LQG}. We need two important properties. We need the stability of the structure under coarse-graining, that is, at each step, the structure must accommodate the necessary information of the internal structure of the coarse-grained vertices. And we need some notion of completude: the structure must keep enough information to be able to write the dynamics for the coarse graph. We will describe here a very straight forward way to accommodate both: we will consider how a structure can naturally arise that keep all the internal structure.

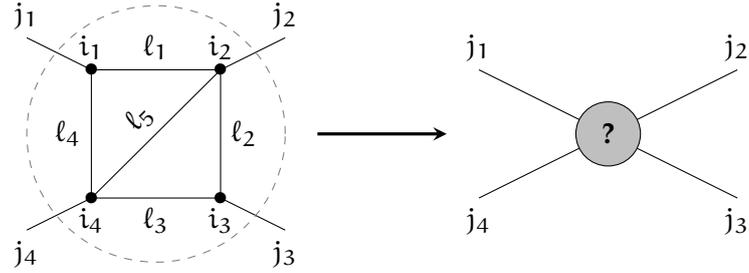
\begin{figure}
  \centering

  \begin{tikzpicture}[scale=0.85]
    \coordinate(A) at (0,0);
    \coordinate(B) at (2,0);
    \coordinate(C) at (2,-2);
    \coordinate(D) at (0,-2);

    \draw (A) -- (B) node[midway,above] {$\ell_1$};
    \draw (B) -- (C) node[midway,right] {$\ell_2$};
    \draw (C) -- (D) node[midway,below] {$\ell_3$};
    \draw (D) -- (A) node[midway,left] {$\ell_4$};
    \draw (B) -- (D) node[midway,sloped,above] {$\ell_5$};

    \draw (A) -- ++(-1,0.5) node[above] {$j_1$};
    \draw (B) -- ++(1,0.5) node[above] {$j_2$};
    \draw (C) -- ++(1,-0.5) node[below] {$j_3$};
    \draw (D) -- ++(-1,-0.5) node[below] {$j_4$};

    \draw (A) node {$\bullet$} node[above]{$i_1$};
    \draw (B) node {$\bullet$} node[above]{$i_2$};
    \draw (C) node {$\bullet$} node[below]{$i_3$};
    \draw (D) node {$\bullet$} node[below]{$i_4$};

    \draw[gray,dashed] (1,-1) circle(2);

    \draw[->,>=stealth,very thick] (3.5,-1) -- (5.5,-1);

    \coordinate(O) at (8,-1);
    \draw (O) -- ++(-2,1) node[above] {$j_1$};
    \draw (O) -- ++(2,1) node[above] {$j_2$};
    \draw (O) -- ++(2,-1) node[below] {$j_3$};
    \draw (O) -- ++(-2,-1) node[below] {$j_4$};

    \draw[fill=lightgray] (O) circle(0.5);
    \draw[scale=2] (O) node{\textbf{?}};

  \end{tikzpicture}

  \caption{We coarse-grain bounded connected regions into \textit{effective} vertices. Presumably, because the curvature carried by the loops in the collapsed bounded region leads to curvature, a new structure is needed to describe these vertices.}
  \label{fig:coarsegraining_ex}
\end{figure}

The main point we have to consider is what are the internal degrees of freedom of our coarse-grained vertices? Let us consider a graph with several nodes and links and a bounded region of the graph that we want to coarse-grain into a single point. We will assume that this portion is also connected as illustrated on the figure \ref{fig:coarsegraining_ex}. We will now gauge-fix in a systematic manner the group elements on the graph. Indeed, the degrees of freedom are far less numerous than it appears at first sight because of gauge invariance. Gauge-fixing unravels this and renders more clear the true degrees of freedom. The coarse-graining procedure goes as follows:
\begin{itemize}
\item First, choose a maximal tree in the region with a given root vertex also in the region.

  Because the region is connected, the maximal tree will circle over every vertices in it. This way, we have defined a unique path linking every vertex in the region to the root vertex.

\item Then, gauge-fix iteratively the group elements along the links of the tree to the identity.

  Indeed, as the wavefunction is gauge invariant (or as the states do not depend on the gauge classically), we can apply a gauge transformation at any vertex and still get the same state. In particular, starting from the root vertex and considering all the outgoing links, we can gauge-fix them to the identity by acting at the opposite ends. We can then pursue this iteratively acting on every vertex in the tree, and therefore because the tree is maximal, on every vertex in the region. The tree acts as a synchronization network: we fixed the reference frames at all the vertices and connected them to the reference frame living at the root of the tree. This procedure works because the path from the root to any vertex is unique and therefore the synchronization is well-defined. This well-definition can also be seen by the absence of loops in the synchronization tree. Once, the graph is gauge-fixed, there is still a residual action of $\mathrm{SU}(2)$ at the root vertex representing the choice of frame there. This will be upgraded to local $\mathrm{SU}(2)$ invariance of the coarse vertex in the coarse-grained graph and so is desirable.

\item Finally, collapse the region to the root vertex. All edges in the tree are collapsed while the other links in the region become loops starting and ending on the root vertex and now label the independent loops of the regions.

  These self-loops (illustrated on figure \ref{fig:selfloops}) carry the curvature living in the bounded region. It is clear that the flux-vectors living on the boundary of the region do not necessarily satisfy the closure condition as the loops also contribute to the sum. The loops therefore induces a closure defect which is the sum of the flux-vectors from their both ends. It should be noted that the collapse seems natural from a quantum mechanical perspective where the wavefunction depends only on the configuration space and the identity elements can be safely removed. But from a classical perspective, this might be a bit odd, as we might wonder on how we should reconstruct the flux-vectors that we lost in the process. This answer is of course the closure constraints: as soon as we know the expansion tree, the closure constraint at each vertex allows us to reconstruct iteratively all the flux-vectors.
\end{itemize}
This gauge-fixing procedure allows to clearly identify and distinguish between the degrees of freedom of the internal geometry of the considered bounded region of space to coarse-grain. The loops describe the internal structure and in particular, the closure failure that is induced by the curvature. The boundary data are quite natural and correspond to the original boundary data. The closure condition is somewhat lifted to a generalization involving the self-loops.

\begin{figure}[h!]
  \centering

  \begin{tikzpicture}[scale=0.8]
    \coordinate(A) at (0,0);
    \coordinate(B) at (2,0);
    \coordinate(C) at (2,-2);
    \coordinate(D) at (0,-2);

    \draw[blue,thick] (A) -- (B) node[midway,above] {$\ell_1$};
    \draw[blue,thick] (B) -- (C) node[midway,right] {$\ell_2$};
    \draw[red,very thick] (C) -- (D) node[midway,below] {$\ell_3$};
    \draw[red,very thick] (D) -- (A) node[midway,left] {$\ell_4$};
    \draw[red,very thick] (B) -- (D) node[midway,sloped,above] {$\ell_5$};

    \draw (A) -- ++(-1,0.5) node[above] {$j_1$};
    \draw (B) -- ++(1,0.5) node[above] {$j_2$};
    \draw (C) -- ++(1,-0.5) node[below] {$j_3$};
    \draw (D) -- ++(-1,-0.5) node[below] {$j_4$};

    \draw (A) node {$\bullet$} node[above]{$i_1$};
    \draw (B) node {$\bullet$} node[above]{$i_2$};
    \draw (C) node {$\bullet$} node[below]{$i_3$};
    \draw[red] (D) node[scale=2] {$\bullet$} node[below]{$i_4$};

    \draw[gray,dashed] (1,-1) circle(2);

    \draw[->,>=stealth,very thick] (3.5,-1) -- (5.5,-1);

    \coordinate(O) at (8,-1);

    \draw (O) -- ++(-2,1) node[above] {$j_1$};
    \draw (O) -- ++(2,1) node[above] {$j_2$};
    \draw (O) -- ++(2,-1) node[below] {$j_3$};
    \draw (O) -- ++(-2,-1) node[below] {$j_4$};
    \draw[blue,thick,scale=3] (O) to[loop] (O);
    \draw[blue,thick] (O) ++(0,1) node {$k_1$};
    \draw[blue,thick,scale=3,rotate=180] (O) to[loop] (O);
    \draw[blue,thick] (O) ++(0,-1) node {$k_2$};

    \draw[red] (O) node[scale=2] {$\bullet$} ++(0.35,0) node{$i$};
  \end{tikzpicture}

  \caption{Coarse-graining via gauge-fixing: we can gauge-fix the subgraph using a maximal subtree (in red). The remaining edges (in blue) correspond to loops on the coarse-grained vertex. There is a residual gauge-freedom at the coarse-grained vertex that corresponds to the action of the gauge group at the root of the tree (red vertex on the figure).}
  \label{fig:selfloops}
\end{figure}
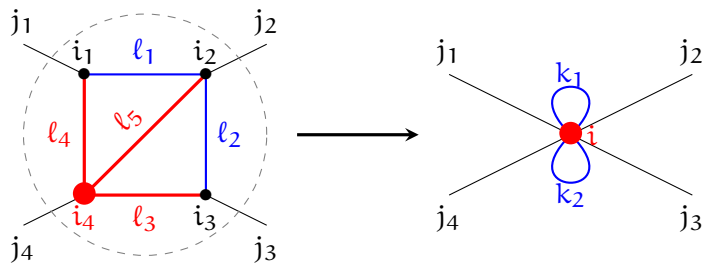

We should note here that with respect to the problem of the closure condition, these loopy generalizations carry too much information. Indeed, we only need to carry the information of the closure defect in order to restore the closure constraint (as we will do with \textit{tagged} spin network in chapter \ref{ch:Bosons}). These loops, augmented with the unfolding information actually carry \textit{all} the degrees of freedom of the internal geometry of a coarse-grained region. Some cuts will be needed for the coarse-graining process. But we should not be hasty: we don't know yet what are the relevant degrees of freedom for the dynamics and the different scales might be heavily interdependent. Still, it should be noted that our initial problem is not motivational enough for the introduction of full-fledged loopy spin networks.

An interesting point with regard to this problems of defect closure and loopy spin networks is that, even with only one loop, the loops carry more information than the closure defect. They are the degrees of freedom corresponding to the rotations carried by the loop along the axis of the flux-vectors. These leave the flux-vectors invariant and therefore the loop does not contribute to the closure defect. This can be seen in our formula for the closure defect in the case of a triangle:
\begin{equation}
  \overrightarrow{S} = (H - \mathbb{1}) g_{CA} \triangleright \overrightarrow{v}_3
\end{equation}
If $H$ is along the axis of $g_{CA} \triangleright \overrightarrow{v}_3$, then $\overrightarrow{S} = \overrightarrow{0}$. We might then wonder at this point if these degrees of freedom are relevant in the coarse-graining process.

It should be noted that these degrees are ubiquitous in that kind of theory and remain whenever the closure condition is satisfied. In particular, they can also appear in linear theories like quantum electrodynamics where the closure condition is always checked. They correspond in these cases to localized magnetic field excitations. As they appear in a linear theory (which can therefore be coarse-grained in a very naive way), they might actually decouple from other modes and this might lead us to consider that only the defect closure is relevant as a macroscopic degree of freedom. % The relevance of such arguments is precisely what we will tackle in the next section.
%\textbf{TODO:} what degrees of freedom?, gauge-fix, space of loopy, loops correspond to curvature excitation, can have closure with curvature: explicit case of ``torsion'', from gauge invariance point of view corresponds to independant gauge invariance, separate degrees of freedom: loops or defect + torsion
%*****************************************
%\section{A consistent cut}

\medskip

It is therefore time to see what cuts should be done. It is an important step to find a consistent cut of the degrees of freedom which separates the macroscopic and the microscopic at least in some approximation, in order to make predictions. As we mentioned in the previous sections, some degrees of freedom appear even in linear theories and it is therefore reasonable to expect them to decouple. In the case of quantum gravity, we can also find arguments for natural cut-offs. For instance, from a more geometrical standpoint, we expect homogeneous curvature to be particularly relevant. We should be able therefore to isolate the degrees of freedom linked to homogeneous curvature of a vertex. From the loopy spin networks perspective, this means that we should be able either to reduce the number of loops and capture some mean component.

This gives us two directions for the programme of coarse-graining: either we could start from some geometrical intuition and try and find natural ways to describe coarse-grained geometry. Or we could go the other way: starting from a natural construction of coarse-graining, try to find algebraically the relevant degrees of freedom for the dynamics. If our geometric intuition is anywhere close to right, the two directions should meet at some point. It could also be that the two directions are needed: the first part of the programme tells us what kind of variables we should look at, and the second part helps us implement them concretely.

\vspace{1em}

In what follows, we will develop these two parts. In the first part, we will consider homogeneously curved geometries and write down generalization of the relations on flat geometry. In particular, because this is the problem when coarse-graining, we will write down generalizations of the closure constraint for curved geometries. In the second part, we will consider a more algebraic approach, doing explicitly coarse-graining and trying progressively to meet the first part of the programme.

%\textbf{TODO:} coarse-graining should reduce the number of degrees of freedom, from the geometrical standpoint, should isolate degrees of freedom suited for description of homogeneously curved manifold, from algebraic standpoint, should be able to reduce the number of loops, at some point these should meet, rapid description of the programme

%*****************************************
%*****************************************
%*****************************************
%*****************************************
%*****************************************

\part{The kinematics of homogeneously curved spaces}
%*****************************************
\chapter{The closure constraint as a Bianchi identity: spinning geometries} \label{ch:ClosureBianchi}
%*****************************************

\inspiquote{Geronimo!}{The Doctor}

Let us start in this chapter our programme of coarse-graining. As we explained earlier, there are two aspects to consider. One is the definition of large and small scales in the context of diffeomorphism invariant theories, \textit{i.e.} a way to define a proper coarse-graining step. A second aspect is the definition of natural large scale observables which must be adapted to the precise theory. We will concentrate on this second point in this chapter and the following.

Our intent can therefore be summed up in the following manner: we want to find natural observables for large homogeneous blocks of curved spacetime. Indeed, first we expect large scale observables to be able to describe curved backgrounds as curvature can built up on large scales in \ac{GR}. Second, in the context of coarse-graining, we hope that the inhomogeneities might be irrelevant or at least entail higher order effects and that the homogeneous component becomes the relevant one. Intuitively, this corresponds to a first order approximation of the curvature. The problem we will have to face comes from the closure constraint: in the usual \ac{LQG} framework, the closure constraint is enforced on every vertex of the supporting graph of a spin network. It can be interpreted as a flatness constraint, effectively encoding the fact that the quantum space results from the gluing of quantum (convex) \textit{flat} polyhedra. But in a larger framework of coarse-graining where curved blocks are needed, this closure constraint will be a problem.

In this regard, we want to define a generalization of the closure constraint, that will correspond to curved geometries. We will consider this in two steps. In this first chapter, we will only try to point out relevant properties of the closure constraint for a natural generalization towards curved spaces. In the next chapter, we will apply this in the context of hyperbolically curved spaces and find natural descriptions of hyperbolic tetrahedra and more generally of hyperbolic polyhedra. Therefore, we are actually considering flat spaces but in a new way before tackling the case of curved spaces.

We will interpret the (usual) closure constraint in a new fashion that will make this generalization easier. We will build on the work done by Freidel \textit{et al.} \cite{Freidel2014} on spinning geometry and show that the closure constraint can be understood as a Bianchi identity. We will in particular see how the normals can arise as holonomies and what are the important properties of the underlying connection for such an interpretation.

This chapter is grounded in our work from \cite{Charles:2015lva} and most of its content is available in our recent paper \cite{Charles:2016xzi}. It is organized as follows: we will first reconsider the spinning geometry interpretation and show in the discrete setting how the closure constraint can naturally be interpreted as a Bianchi identity. In the second section, we will review the continuum formulation in the flat case, underlying its main properties that will become important in order to select interesting closures. And finally, we will try and generalize this formulation to other groups even in the flat case. This will help us understand and interpret the results of the curved case of the next chapter.

\section{A discrete point of view}

Let us restate the problem and the usual interpretation of it. For each link of the graph defining our state of quantum space, there is a natural operator associated to the integrated triad. Around a vertex, these operators sum up to zero as the Gauß constraint implies. Classically, this means we have $N$ vectors $\vec{v}_i$ around a vertex supported by the $N$ links going in and out of it. And we have:
\begin{equation}
  \sum_{i=1}^N \vec{v}_i = \vec{0}
\end{equation}
which is the closure constraint. Geometrically, we can interpret this as defining a convex polyhedron according to Minkowski's theorem. The vectors are then interpreted as normals: they are orthogonal to the faces and their norms encode the areas of the faces. Therefore, giving such $N$ vectors is equivalent to parametrizing the space of convex polyhedra with $N$ faces.

There is another interpretation coming from spinning geometries. The closure can be understood as a discrete Bianchi identity. We can see this quite easily by comparing it with electromagnetism. Indeed, it is not uncommon to call the densitized triad field the electric field, since it appears in the same place as the electric field of electromagnetism in the action as the conjugate of the connection. It also satisfies the same kind of constraints hence the name of the Gauß constraints. Now, in electromagnetism, when there is no source (as is the case for us here), the equations are completely symmetrical in the electric and magnetic fields. Usually, we choose to write down the fields as:
\begin{equation}
  \left\{ \begin{array}{rcl}
    \overrightarrow{E} &=& -\overrightarrow{\nabla}\phi - \frac{\partial \overrightarrow{A}}{\partial t} \\
    \overrightarrow{B} &=& \overrightarrow{\nabla} \times \overrightarrow{A}
  \end{array} \right.
\end{equation}
where $\overrightarrow{E}$ is the electric field, $\overrightarrow{B}$ is the magnetic field and $\phi$ and $\overrightarrow{A}$ are the scalar and vector potentials. But it is also possible (as long as there is no source) to write:
\begin{equation}
  \left\{ \begin{array}{rcl}
    \overrightarrow{E} &=& \overrightarrow{\nabla}\times \overrightarrow{A}_B \\
    \overrightarrow{B} &=& -\overrightarrow{\nabla}\phi_B - \frac{\partial \overrightarrow{A}_B}{\partial t} \\
  \end{array} \right.
\end{equation}
where $\phi_B$ and $\overrightarrow{A}_B$ are new scalar and vector potentials whose existence shows the symmetry of the equations. In particular, the fact that $\overrightarrow{E}$ can be written as a rotational comes from the Gauß constraint:
\begin{equation}
\overrightarrow{\nabla}\cdot \overrightarrow{E} = 0
\end{equation}
This second way of writing is of course natural if we were to include magnetic monopoles, but not electric ones. So, the Gauß constraint can be interpreted as a structure equation saying that the electric field is constructed from a connection. We can do the same thing in the gravitational case.

From the discrete point of view, that is if we consider a triangulation dual to the graph, such a construction correspond to associating information to the edges of the triangulation. Then, the normal of a given face would be reconstructed by summing over the information at each edge around that face. In fact, this relation is of course well known in discrete geometry. Consider a surface, which is homeomorphic to the sphere, and a graph on it, which is therefore a planar graph. If we define a vector quantity for each face of the graph such that the total sums up to zero, we can decompose the quantities on the edges, up to an addition at the graph vertices. There is no cocycle contribution, since the graph is planar. This fact can be checked quite easily thanks to the gauge invariance at the vertices of the graph. We can fix the gauge by selecting a maximal tree as a synchronization tree. Let us consider the simple case of a tetrahedron to illustrate this.

\begin{figure}[h!]

  \centering

  \begin{tikzpicture}[scale=1.5]
    \coordinate (A) at (0,0,0);
    \coordinate (B) at (1,2,-1);
    \coordinate (C) at (-1,2,-1);
    \coordinate (D) at (0,2,0.5);

    \draw (A) node[below]{$A$};
    \draw (B) node[right]{$B$};
    \draw (C) node[left]{$C$};
    \draw (D) node[below right]{$D$};

    \draw[thick] (B) -- (C) -- (D) -- cycle;
    \draw[red,thick] (A) -- (B);
    \draw[red,thick] (A) -- (C);
    \draw[red,thick] (A) -- (D);

    \draw[ball color=red] (A) circle (0.05);
  \end{tikzpicture}

  \caption{We consider a tetrahedron. $A$ is considered as the root vertex of a maximal tree on the tetrahedron (in red). The parallel transport along the edges of this tree is gauge fixed to the identity. This allows to directly reconstruct the connection for the remaining edges of the form $IJ$ as the holonomy of the face $AIJ$.}
  \label{fig:gaugefixing_tetra}

\end{figure}
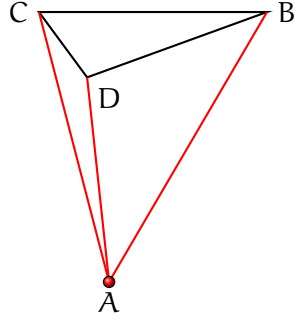

A tetrahedron has six edges. Let us select a maximal tree on the tetrahedron by choosing a vertex, which we will call the \textit{root} and by including every edge coming out this root as illustrated in figure \ref{fig:gaugefixing_tetra}. Therefore if we label the four vertices of the tetrahedron $A$, $B$, $C$ and $D$ and if we select $A$ as our root, then the edges of the graph will be $AB$, $AC$ and $AD$. Now, we can gauge-fix in the following manner: let's act at the vertices $B$, $C$ and $D$ so that the connection on the edges of the graph are sent to zero. Therefore, the integrated connections will be:
\begin{equation}
  \left\{ \begin{array}{rcl}
    g_{AB} &=& \vec{0} \\
    g_{AC} &=& \vec{0} \\
    g_{AD} &=& \vec{0}
  \end{array} \right.
\end{equation}
This is always possible if the decomposition exists. Now, the full holonomy for each face touching $A$ is carried by the edge opposite to the root. We can define the integrated connection on the various edges to be:
\begin{equation}
  \left\{ \begin{array}{rcl}
    g_{BC} &=& \vec{N}_{ABC} \\
    g_{CD} &=& \vec{N}_{ACD} \\
    g_{DB} &=& \vec{N}_{ADB}
  \end{array} \right.
\end{equation}
We only have to check that the holonomy around the last face $BCD$ is indeed its normal. This means that we have to check that:
\begin{equation}
g_{CB} + g_{BD} + g_{DC} = \vec{N}_{CBD} \Leftrightarrow -\vec{N}_{ABC} - \vec{N}_{ACD} - \vec{N}_{ADB} = \vec{N}_{CBD}
\end{equation}
which is precisely the closure condition. Therefore, the construction works.

We should note here that we can give a more precise sense to the root and to the normals associated to the edges. Consider for instance the edge $BC$. Now consider a particle moving at constant speed along $BC$ and compute its angular momentum with respect to the point $A$. It will precisely be the normal of $ABC$. Indeed, angular momentum precisely encodes the idea of swept area\graffito{Note that as angular momentum encodes swept area, this is precisely how the second law of Kepler comes out from conservation of angular momentum.}. This interpretation still works for the edges $AB$, $AC$ and $AD$ as their angular momentum will be zero. We see here that we can also solve the problem by consider a root vertex and angular momenta. This will give the area of the swept surface for any edge and for a surface, this will give the normal, precisely because of the closure condition stated for the polyhedron formed by the surface plus the root vertex.

%\textbf{TODO:} first point, try to describe homogenesouly curved geometries, we want only one parameter, the curvature, and everything else should be alike, in particular closure constraint, look for a generalization of the closure constraint, main idea: closure = discrete bianchi identity, cocycle point of view, construct from selection of a point and closure, interpretation as angular momentum (link with area law)

\section{Freidel's connection}

Let us now turn to the continuum version of this construction. The advantage of such a continuum construction is that we can better understand the relation between the elements carried by straight and curved edges. In Freidel's work, this was important for defining spinning geometries. But for us, because we turn to curved geometries, this is a mandatory step. We can also use this opportunity to see what are the relevant properties of the connection.

Let us define the following connection on flat space:
\begin{equation}
  A_\textrm{Fr}(\vec{x}) = \frac{1}{2}\epsilon^i_{jk} T_i x^j e^k
\end{equation}
Here $e$ is the triad\graffito{Note that the spin connection does not appear in $A_\textrm{Fr}$. This is due to its triviality in flat space. If we were to use non-cartesian coordinates, it should be included.}, the $T$s are the generators of the $\mathbb{R}^3$ group and the $\vec{x}$ gives the point on flat space. We can see that this connection reproduces the previous idea in an infinitesimal setting. Indeed, the point $\vec{x} = \vec{0}$ is an arbitrarily selected origin of space. So, we are exactly doing what we told before: given an infinitesimal edge, we take the cross-product therefore computing the infinitesimal angular momentum which is also the normal to the infinitesimal triangle formed by the edge and the origin. Exactly as in the previous (finite) case, the choice of origin is irrelevant to the computation of holonomy. And because the connection is commutative, we have no action of the group on the holonomies and the holonomies are truly gauge-invariant.

So, how do we compute the holonomy around a closed loop? We can simplify our problem here, since the group is commutative. Therefore, according to a fundamental theorem, the holonomy around a closed loop is equal to the integration of the curvature on (any of) the enclosed surface(s). In particular, if the loop is planar, there is a flat surface enclosed by it for which we can define the notion of a normal. Checking that the holonomy does give the normal just amounts to checking that the curvature has an interpretation as an infinitesimal normal. Therefore, we just have to compute the curvature of the connection. We find:
\begin{equation}
  F[A_\textrm{Fr}] = \epsilon^i_{jk} T_i e^j \wedge e^k
\end{equation}
which is exactly what we wanted: given two directions, we lift them in tangent space thanks to the triad and take the cross-product to have a normal. Because we used the triad, the result is directly proportional to the area. And, because we have the commutativity of the group, the results directly extends to the finite case. We can also integrate, out of a whim, the connection along a straight line and find that it indeed does give the angular momentum as discussed in the discrete setting. Therefore, we have here the continuum version of the previous discussion in the discrete case, as long as we consider flat polyhedra.

Let us note that the closure condition becomes trivial in this setting. Indeed, because $A_\textrm{Fr}$ is a connection, its curvature satisfies a Bianchi identity. In our case, it is a pretty simple one as the group is abelian:
\begin{equation}
  \mathrm{d}\vec{F}[A_\textrm{Fr}] = \vec{0}
\end{equation}
Integrated over a volume, we will find the closure condition. Indeed, the very fact that the connection exists given the closure condition is just the discrete equivalent of :
\begin{equation}
  \mathrm{d}F = 0 \Leftrightarrow F = \mathrm{d}A
\end{equation}
as soon as the embedding space is simply connected (which is of course the case for the interior of a polyhedron which is homeomorphic to the ball).

Introducing such a connection (or its equivalent in discrete terms) makes some relations trivial (like the closure condition) but raises new questions. This is similar to going from Maxwell’s equation, governing the electric and magnetic field, to a more modern and covariant way of writing electrodynamics using the gauge connection of $\mathrm{U}(1)$. In particular, the gauge transformation does not affect the electromagnetic fields, since the gauge group is abelian, but they do affect the connection (and the matter fields) and even have an interpretation as a transformation on them. Now that we have a connection to construct normals, it is natural to wonder what the gauge transformations correspond to. We should notice that, as for electromagnetism, the gauge group being abelian, and therefore the holonomies do not transform under gauge transformation. This is reflected in the fact that there is also no parallel transport needed in any of our constructs so far. This implies that the gauge transformation will not correspond to some rotation of space (otherwise they would transform the normals) or, except for a few possible - but implausible - exceptions (translations), to any geometrical transformation of the tetrahedron.

Let’s dive a bit more into the precise transformations. Let us consider a generic gauge transform $\vec{\phi}$ which will encode a shift at each point of space. More precisely, the fields transform as follows:
\begin{equation}
  \left\{
  \begin{array}{rcl}
    A_\textrm{Fr} &\rightarrow& A_\textrm{Fr} + \mathrm{d}\phi^i T_i \\
    \overrightarrow{v_{AB}} &\rightarrow& \vec{\phi}_B + \overrightarrow{v_{AB}} - \vec{\phi}_A
  \end{array}
  \right.
\end{equation}
where we put $\overrightarrow{v_{AB}}$ the integrated (open) holonomy between two points $A$ and $B$. So, we expect each (open) holonomy to be transformed by two contributions coming from each end of the segment and depending only on the shift at these points. To what could possibly this corresponds to? Though, it cannot cover every possible transformation, we see that a change of the origin point in our connection can indeed be expressed like this. For instance, let's define $\overrightarrow{v_{AB,O}}$ and $\overrightarrow{v_{AB,O'}}$ the integrated connection along $AB$ but using different reference points $O$ and $O'$. We have:
\begin{equation}
  \overrightarrow{v_{AB,O}} = \frac{1}{2} \overrightarrow{OA} \times \overrightarrow{AB} = \frac{1}{2}\overrightarrow{OO'} \times \overrightarrow{OB} + \overrightarrow{v_{AB,O'}} - \frac{1}{2} \overrightarrow{OO'}\times\overrightarrow{OA}
\end{equation}
And so, a change of origin from $O$ to $O'$ corresponds to a gauge transform with $\vec{\phi}_I = \frac{1}{2} \overrightarrow{OO'}\times\overrightarrow{OI}$. So, though they do not cover all of possible gauge transform, a change of the origin just amount to a gauge transform. It is easy to see that the only terms that appear this way are orthogonal to $\overrightarrow{OI}$ as the cross-product implies. But, this calls for a decomposition over this case plus a vector proportional to $\overrightarrow{OI}$ that would not otherwise be obtainable. We could therefore write a more general gauge transform:
\begin{equation}
  \vec{\phi}_I = \frac{1}{2}\overrightarrow{OO'}\times\overrightarrow{OI} + \alpha \overrightarrow{OI}
\end{equation}
with general $O'$ and $\alpha$. Of course, these can depend on the selected point in general, reflecting the local nature of a gauge transform. But this gives us a neat geometrical interpretation, that we already allude to in the discrete case, for the closure of the gauge transformed connection: one term disappears because of the independence of the origin point (which is basically translation invariance) and the second term disappear because of the closure of polygons in flat space.

It is now important to consider what are the important properties of Freidel's connection, as we will care about its possible generalization. So as a first question, we might wonder if another connection could precisely give us the normals. The answer is of course \textit{no} as all the gauge invariant information is contained in the curvature of the connection. But it is possible to consider generalization of the normals (as we will be forced to do in a curved context anyway). In that case, the major point of our connection for its interpretation as normals is:
\begin{equation}
  F[A_\textrm{Fr}] = \epsilon^i_{jk} T_i e^j \wedge e^k
\end{equation}
This totally defines Freidel's connection and gives it its interpretation as a normal. As we said, the dependence on the triad gives the dependence on the area and the cross-product gives the right notion of direction. Taken with a commutative group of dimension $3$, it naturally scales and gives the full construction.

%\textbf{TODO:} definition of a connection, spinning geometry, discussion of gauge invariance, discussion of reconstruction and stuff (equivalent of Minkowski), important features of the connection

\section{Non-abelian constructions}

Our end goal is to consider curved geometries, and more specifically hyperbolic geometries as we will see. But, as we will soon realize, the natural group will not be an abelian group and we must therefore relax the condition of commutativity. Instead of tackling these new subtleties in the general case, let's try and consider them in the simplified setting of flat space. So, in this section, we will investigate a non-abelian connection, still on the flat space. The goal is to have a kind of toy-model to sharpen our intuition and understanding of the geometry linked to what we might call non-abelian normals but in a still controlled environment. So let’s try to define such a non-abelian connection, in flat $\mathbb{R}^3$ , along with its curvature, which should give us a definition of non-abelian normals. Such a proposal seems counter-intuitive at first. Indeed, the non-commutative nature of a connection is linked to problems with parallel transport. In flat space, we do not expect such complications to arise. Moreover, we do not want our connection to be arbitrary so that it reflects no interesting geometrical property.

Let us consider what interesting properties we should keep. First, we need three generators to correspond to the three directions of space. We also want an action of the rotation group to be well-defined on the (deformed) normals. We also want some kind of homogeneity for the connection as the normal definition should not depend on the point, except maybe for some gauge transformation artifacts. All this boils down to finding a group, with a good action of $\mathrm{SU}(2)$ for which a connection $A_{nc}$ can be written such that:
\begin{equation}
  F[A_\textrm{nc}] = \epsilon^i_{jk} J_i e^j \wedge e^k
\end{equation}
where the $J$s are the $3$ generators of the group. Of course, a natural group appears and it is the rotation itself or its double-cover $\mathrm{SU}(2)$. We should note moreover that the cross-product was extensively used in the definition of Freidel's connection. But the structure of the cross-product is somewhat linked to $\mathrm{SU}(2)$ as it involves the structure constants of the algebra $\mathfrak{su}(2)$. This makes this group doubly natural.

There is of course a natural connection on flat space that is $\mathfrak{su}(2)$ valued: it is the natural spin connection, which is torsionless and compatible with the triad. However, it is trivial on flat space (by definition of flat space) and therefore does not seem so well fitted for our case, to say the least. We must therefore look for deformation of this connection that stays homogeneous. The natural way to search for this is to allow torsion and see if we can find a connection with a non-trivial and interesting curvature.

A possibility is to work with a connection inspired from that of Ashtekar-Barbero, namely:
\begin{equation}
  A_\textrm{AB} = \Gamma + \beta K
\end{equation}
where $K$ is the extrinsic curvature. The extrinsic curvature depends on the embedding of our manifold. The natural embedding of flat space however is in $\mathbb{R}^4$ where its extrinsic curvature is trivial. In any case, for such constructions, we should care more about intrinsic properties as we are describing intrinsic geometry. Still, the idea to add a connection valued in the tangent space is one we can emulate. Consider, the following connection:
\begin{equation}
  A_\textrm{nc} = \Gamma + a J_i e^i = a J_i e^i
\end{equation}
where the $J$s are now the generators of $\mathrm{SU}(2)$ and $a$ is a real coefficient. The term $\Gamma$ is the usual spin connection and is needed for gauge invariance, but we can gauge-fix this on flat space and sent it to zero.

Geometrically, this parallel transport according to this connection gives a twist along the direction of propagation. As we said earlier, the connection will have torsion, which we can compute:
\begin{equation}
\mathrm{d}_{A_\textrm{nc}} e^i = \mathrm{d}e^i + a \epsilon^i_{~jk} e^j \wedge e^k = a \epsilon^i_{~jk} e^j \wedge e^k
\end{equation}
We already see that a notion of cross-product was encoded in the connection. This is interesting because it gives the link with Freidel's abelian connection. Let us go to the infinitesimal level to uncover this. We have:
\begin{equation}
F[A_{nc}] = a^2 \epsilon^i_{~jk} J_i e^j \wedge e^k
\end{equation}
We should remember here that the curvature of a connection gives the first order of its holonomy around infinitesimal surfaces. Getting out the $a^2$ factor and forgetting that $J_i \neq T_i$ , we see here that the curvature of $A_{nc}$ is the same as the curvature of $A_{Fr}$ . In more precise terms, it means that at the infinitesimal level, the holonomy precisely encodes the normal as a rotation around the normal axis. The angle is proportional to the area (we do not have to worry about compactness at the infinitesimal level). We see here that $A_{Fr}$ encodes the first order of $A_{nc}$.

We should not be surprised by this: the cross-product encodes the infinitesimal action of the rotation\graffito{The link between rotation and cross-product is of course evident when one knows about angular momentum in quantum mechanics as the $\mathrm{SU}(2)$ algebra is reproduced by the angular momentum which are cross-product.}. If we go to the first non-trivial order, we are to get a cross-product. But this implies something maybe more interesting in the context of coarse-graining: the structure of the $\mathrm{SU}(2)$ connection is more natural than the structure of $\mathbb{R}^3$ in the following sense: no additional structure coming from outside of the group is put. This is kind of similar to the development of the $\mathrm{U}(1)^3$ model of quantum gravity: the $\epsilon$ comes from nowhere in the structure of the model except from the theory that it tries to emulate. But, for the full theory, it just comes out of a commutator of two group elements. The same goes here: by going to some abelian limit, we lost structure that we put kind of artificially back into the game. This is important for coarse-graining because the abelian connection only exists locally (around the vertex), but it is not inconceivable that some $\mathrm{SU}(2)$ might survive the long distance and gives some notion of closure for large chunks of space.

Now, this definition does not precisely gives a normal, it gives, at best, a deformed normal. Let us go further and study the deviation. This can be done quite easily because the holonomy of $A_{nc}$ can be computed exactly at least for (flat) faces of polyhedra. For definiteness, let’s consider a triangle $ABC$. The closed holonomy around the triangle will be the composition of three open holonomies corresponding to each edge. For an (oriented) edge $\overrightarrow{AB}$, we define the holonomy $g_{AB}$ which can be computed exactly as:
\begin{equation}
  g_{AB} = \exp \left( \frac{\mathrm{i} a}{2} \overrightarrow{AB} \cdot \overrightarrow{\sigma} \right)
\end{equation}
where the $\sigma$s are the Pauli matrices. Note that the holonomy depends on the choice of gauge. But here, even the closed holonomy will. It is not a problem though as we have some good gauge invariant quantity as the angle of the rotation and, for now, as we consider flat spaces, we have natural gauge-fixing conditions. Let us now write the holonomy around the full triangle $ABC$. It is:
\begin{equation}
h_{ABC} = g_{CA} g_{BC} g_{AB} = \exp\left(\frac{\mathrm{i} a}{2} \overrightarrow{CA} \cdot \overrightarrow{\sigma}\right) \exp\left(\frac{\mathrm{i} a}{2} \overrightarrow{BC} \cdot \overrightarrow{\sigma}\right) \exp\left(\frac{\mathrm{i} a}{2} \overrightarrow{AB} \cdot \overrightarrow{\sigma}\right)
\end{equation}
In order to have a precise definition of a normal as a vector, we can now take the rotation axis. We defined, the vector $\vec{n}^a$ (which depends on the scale $a$) by:
\begin{equation}
h_{ABC} = g_{CA} g_{BC} g_{AB} = \exp\left(\frac{\mathrm{i} a}{2} \vec{n}^a \cdot \overrightarrow{\sigma}\right)
\end{equation}
This will act as a definition for our deformed normal and is well-suited on a least one respect: it has a good infinitesimal limit when we consider small triangles.

Note here that we introduced a squared factor $a^2$, rather than just a linear dependency in $a$. Indeed, we now that the factor already appears squared in the curvature and therefore is the first order. The linear order of $a$ must therefore cancel in the expansion. Of course, the expansion does not have to involve only squared terms and it doesn't. This implies some other interesting property: the vector $\vec{n}^a$ depends on $a$ and has a limit when $a$ goes to $0$. This limit, because we factored out the first $a^2$ is non-zero and is actually the undeformed normal. This sharpens our intuition of $a$: it is a scale factor that tells us at what scale the deformation of the normal kicks in. It will appear for a triangle of an area of order $\frac{1}{a^2}$.

We will not compute exactly $\vec{n}^a$ . The interested reader ca turn to our paper \cite{Charles:2016xzi}. We can summarize the full result as follows in:
\begin{equation}
  h_{ABC} = \mathbb{1} + \frac{\mathrm{i}a^2}{4} \overrightarrow{n} \cdot \overrightarrow{\sigma} + \frac{\mathrm{i} a^3}{12}\left( \overrightarrow{CA}^2\overrightarrow{CA} + \overrightarrow{BC}^2 \overrightarrow{BC} + \overrightarrow{AB}^2 \overrightarrow{AB} \right)\cdot \overrightarrow{\sigma} + \mathcal{O}(a^4)
\end{equation}
And so, at first order, we find:
\begin{equation}
\vec{n}^a = \overrightarrow{n} + \frac{a}{3} \left( \overrightarrow{CA}^2\overrightarrow{CA} + \overrightarrow{BC}^2 \overrightarrow{BC} + \overrightarrow{AB}^2 \overrightarrow{AB} \right) + \mathcal{O}(a^2)
\end{equation}
Though the geometrical interpretation of this added term is not totally clear, it shows that the shape of the triangle influences our deformed normal. Of course, we expected corrections due to the topology of $\mathrm{SU}(2)$. For instance, a triangle with lengths of integer multiples of $2\pi$ a necessarily has a trivial holonomy. This is due to the periodic nature of $\mathrm{SU}(2)$ which, therefore, cannot distinguish all the triangles. 

Let’s finish with the closure of this connection, which is guaranteed, as we now all understand, by the Bianchi identity. As in the previous case, given any polyhedron, we can associate to each of its face a holonomy $h_f$. The closure condition will look like:
\begin{equation}
  h_n ... h_2 h_1 = \mathbb{1}
\end{equation}
But now, there are new difficulties which are linked to parallel transport. Indeed the connection is no longer abelian and therefore we must consider the transport and transformation of the $h_f$. For this, we will need a reference point and a path for each holonomy as the parallel transport depends on the precise path because of curvature. Let us recap the procedure for a tetrahedron as an example: three of the faces share a common point, which can be chosen as the origin. For the corresponding holonomies, it is always possible to start the closed loop on this origin and all the holonomies are therefore expressed at the same point. In these cases therefore, the point is always the same and no path (or the trivial path) is involved. But we have a remaining face and as this one does not touch the origin, we will need an edge to transport the corresponding holonomy. This path, containing only one edge, has to start at the origin and must land on one of the vertices of the opposite face. Any edge of the tetrahedron satisfies this criteria and so any edge will do the job. We must just be careful in the order of composition so that the first three holonomies start and finish on this edge as illustrated on figure \ref{fig:Transport}. This whole procedure can be thought of as a gauge fixing of the tetrahedron. We can generalize this to any polyhedron by choosing a path going through each faces and composing accordingly.

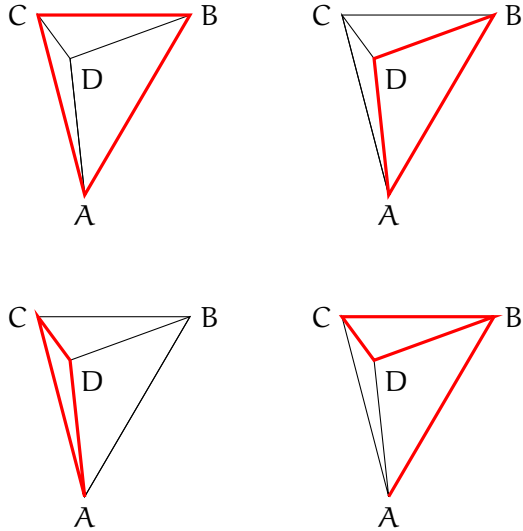
\begin{figure}[h!]

  \centering

  \begin{tikzpicture}

    \coordinate (O1) at (0,0,0);
    \coordinate (A1) at (1,2,-1);
    \coordinate (B1) at (-1,2,-1);
    \coordinate (C1) at (0,2,0.5);

    \draw (O1) node[below]{$A$};
    \draw (A1) node[right]{$B$};
    \draw (B1) node[left]{$C$};
    \draw (C1) node[below right]{$D$};

    \draw (O1) -- (B1) -- (C1) -- cycle;
    \draw (O1) -- (C1) -- (A1) -- cycle;
    \draw[red,very thick] (O1) -- (A1) -- (B1) -- cycle;

    \coordinate (O2) at (4,0,0);
    \coordinate (A2) at (5,2,-1);
    \coordinate (B2) at (3,2,-1);
    \coordinate (C2) at (4,2,0.5);

    \draw (O2) node[below]{$A$};
    \draw (A2) node[right]{$B$};
    \draw (B2) node[left]{$C$};
    \draw (C2) node[below right]{$D$};

    \draw (O2) -- (B2) -- (C2) -- cycle;
    \draw (O2) -- (A2) -- (B2) -- cycle;
    \draw[red,very thick] (O2) -- (C2) -- (A2) -- cycle;

    \coordinate (O3) at (0,-4,0);
    \coordinate (A3) at (1,-2,-1);
    \coordinate (B3) at (-1,-2,-1);
    \coordinate (C3) at (0,-2,0.5);

    \draw (O3) node[below]{$A$};
    \draw (A3) node[right]{$B$};
    \draw (B3) node[left]{$C$};
    \draw (C3) node[below right]{$D$};

    \draw (O3) -- (C3) -- (A3) -- cycle;
    \draw (O3) -- (A3) -- (B3) -- cycle;
    \draw[red,very thick] (O3) -- (B3) -- (C3) -- cycle;

    \coordinate (O4) at (4,-4,0);
    \coordinate (A4) at (5,-2,-1);
    \coordinate (B4) at (3,-2,-1);
    \coordinate (C4) at (4,-2,0.5);

    \draw (O4) node[below]{$A$};
    \draw (A4) node[right]{$B$};
    \draw (B4) node[left]{$C$};
    \draw (C4) node[below right]{$D$};

    \draw (O4) -- (B4) -- (C4) -- cycle;
    \draw[red,very thick] (O4) -- (A4);
    \draw[red,very thick] (A4) -- (B4) -- (C4) -- cycle;

  \end{tikzpicture}

  \caption{When defining the holonomies for the closure constraint, we need to use the same reference point (root) for each holonomy. In the case of a tetrahedron for instance, this means that the last holonomy must have some parallel transport along an edge as shown on the forth figure.}
  \label{fig:Transport}

\end{figure}

As we just saw, the gauge fixing needs a choice of origin. What is the link with the previous choice of origin in $A_{Fr}$? Mathematically, the transport is simply done through conjugation. This is exactly the operation we would do if we gauge transform the holonomy. And we saw in the previous case of $A_{Fr}$ that gauge transform was linked, at least partially, to a change of origin. This is implemented here in a much more concrete sense as an origin must be selected to even define the holonomy. So eventhough we do not have an origin selected in the definition of the connection, a trace remains in the parallel transport needed for the closure relation.

\section{Duality and $\mathrm{ISU}(2)$ closure}

There is still one point we haven’t discuss or illustrated in the flat case. Indeed, in the hyperbolic case, we will look for $\mathrm{SL}(2, \mathbb{C})$ connections rather than $\mathrm{SB}(2, \mathbb{C})$. This technicality, which is linked to the problem of having a connection that transforms well under rotation, does not have an equivalent in what we have seen just yet. But we could develop one using a connection of the isometry group rather that just the rotation group.

The full isometry group of flat space is $\mathrm{ISO}(3)$ or, equivalently for our concerns, its double cover $\mathrm{ISU}(2)$ and is the semi-direct product of the rotation group and the translation group. It is $6$ dimensional. This last point might seem problematic since we have twice too many dimensions in our group to reproduced the construction of a normal at the infinitesimal level. However, the group naturally splits into two three-dimensional parts and therefore, we can extend the idea.

So let's construct our new $\mathrm{ISU}(2)$ connection:
\begin{equation}
  A_{ISU} = a J_i \mathrm{e}^i + b T_I \mathrm{e}^I
\end{equation}
where $a$ and $b$ are two real parameters. The indices $i$ and $I$ both run from $1$ to $3$ but we've written them in a different style to underline their differences. We named the $a$ parameter in the same way as in $A_{nc}$ to highlight its very same role. So, in a sense, our new connection is really a generalization of the previous one where $b$ was set to $0$. It also turns out that $A_{ISU}$ is a generalization of $A_{Fr}$ but in a more subtle sense, that we will elaborate on. Let us also note that no extra structure is imposed and everything will come from commutation relationships in $\mathrm{ISU}(2)$, especially the interpretation as normals.

Let us compute the curvature to see what kind of meaning we can associate to this connection and more particularly to its holonomies:
\begin{equation}
  F[A_{ISU}] = (ab) \epsilon^i_{~jk} T_i \mathrm{e}^j \wedge \mathrm{e}^k + (a^2) \epsilon^I_{~JK} J_I \mathrm{e}^J \wedge \mathrm{e}^K
\end{equation}
Here, we have naively split into the $\mathbb{R}^3$ and $\mathrm{SU}(2)$ part using the same exponents and indices notations as before. It now becomes clear that, at the infinitesimal level, the curvature gives the same result as the $A_{Fr}$ connection (for the translational part) and as the $A_{nc}$ connection (for the rotational part). In a sense, we've bundled the two together. But, more importantly, once again this holonomy will have a nice geometrical meaning of normals, as is clear at at the infinitesimal level. Of course, we expect deformations at the finite level.

The way this connection is a generalization of the previous two is also now clear: it gives them both in the appropriate limit. First, as we saw, if $b=0$, then we just fall back to the $\mathrm{SU}(2)$ connection. Freidel's connection however appears when we sent $a$ to $0$ but keeping $ab$ constant. This means that the two-parameter family of connections we just defined is actually an interpolation between the two previous one-parameter families of connections.

Let us now highlight a property of the writing of the closure that will be of particular significance for the hyperbolic case. Let us use the fact that $\mathrm{ISU}(2)$ is the semi-direct product of $\mathbb{R}^3$ and $\mathrm{SU}(2)$. In particular this means that any $\mathrm{ISU}(2)$ group element has a unique decomposition: $g_i = (N_i, h_i)$ where $N_i \in \mathbb{R}^3$ and $h_i \in \mathrm{SU}(2)$. We used the same notation as in the previous section to highlight the belonging in the rightful groups but we should \textit{not} think that they indeed corresponds to holonomy of the $\mathbb{R}^3$ or the $\mathrm{SU}(2)$ connections. Now, the Bianchi identity which encodes the closure can be written as:
\begin{equation}
  g_n ... g_2 g_1 = \mathbb{1}
\end{equation}
with the same parallel transport problem as before, for the $\mathrm{SU}(2)$ case, and with the same solution. There is however the new problem of splitting the holonomy into two parts and how this relates to the closure. Let us simply write the decomposition onto the closure itself:
\begin{equation}
  \left\{
  \begin{array}{rcl}
    N_n + h_n \triangleright N_{n-1} + ... + (h_n ... h_3) \triangleright N_2 + (h_n...h_3h_2) \triangleright N_1 &=& \vec{0} \\
    h_1 h_2 ... h_n &=& \mathbb{1}
  \end{array}
  \right.
\end{equation}
where $\triangleright$ is the natural action of $\mathrm{SU}(2)$ onto $\mathbb{R}^3$ (the adjoint representation). We can see that a \textit{braiding} appears in the composition of normals. This behavior might appear odd at first but it will be of paramount importance in the $\mathrm{SL}(2,\mathbb{C})$ case we will develop for the hyperboloid. Note also that such braiding does not appear for the $\mathrm{SU}(2)$ closure. This is due to our peculiar situation - working on a flat manifold - but will not survive the generalization. Finally, let's note that we got two closure conditions here. Their precise link with the closure conditions defined before is not straightforward (except for the $\mathrm{SU}(2)$ one which is undeformed) but an interesting idea is that there is some duality. Indeed, if it is possible to reconstruct the polyhedron from only the $\mathrm{SU}(2)$ closure for instance, then the second closure can be reconstructed which tells us that there might be a correspondence between $\mathbb{R}^3$ closure and $\mathrm{SU}(2)$ closure. This kind of duality is interesting, particularly in the context of coarse-graining where the fundamental variables are in $\mathbb{R}^3$ and the coarse-grained description fits better in $\mathrm{SU}(2)$. This idea, though interesting will not be the focus of our inquiry.

\medskip

In this chapter, we have introduced a new interpretation of the closure condition as a Bianchi identity. We have developed this by studying it in the context of spinning geometry and we have extracted the relevant information for a generalization. We also explored simple generalization in the flat case as some prep work for the curved case. We saw how different structure groups, linked to the symmetries of the space considered can be natural for such construction, paving the way for the hyperbolic case. We will now turn to this latter case and see how all this work can enable us to define closure conditions for curved geometries.

%\textbf{TODO:} $\mathrm{ISU}(2)$, definition, curvature, discussion on dimensions and complex $\Lambda$, closure and brading, other cases as limits

%\textbf{TODO:} generalization from connection, $\mathrm{SU}(2)$ case, discussion, in particular possibility of having normals elsewhere than in $\mathbb{R}^3$, then $\mathrm{ISU}(2)$ and other cases as limits

%*****************************************
%*****************************************
%*****************************************
%*****************************************
%*****************************************

%*****************************************
\chapter{A closure for the hyperbolic tetrahedron} \label{ch:SB2C}
%*****************************************

\inspiquote{That’s how I see the universe. Every waking second I can see what is, what was, what could be, what must not.}{The Doctor}

In the previous chapter, we concentrated on the interpretation of the closure constraint and illustrated various generalizations in flat space. The goal though was to consider curved spaces for the coarse-graining programme. This is what will be started in this chapter. Only a specific case will be considered: the case of hyperbolic homogeneous spaces. Other cases, like spherically curved spaces, will be left for further study. One of the reason is that their expected (quasi) Poisson structure is much more complicated if we are to believe that they match quantum groups structure.  We do not foresee however any obstacle in the geometrical construct itself. So, the work done here might be plainly generalizable to the spherical case.

Therefore, the goal is to generalize the notion of twisted geometries or spinning geometries to curved spaces, more specifically to hyperbolically curved spaces. In the image of twisted geometries, we will consider a space constructed from fundamental blocks which are curved polyhedra glued together \textit{via} matching conditions of some sort which should be made precise later on. These (curved) polyhedra will be described by an appropriate generalization of normals as introduced in the previous chapter and a corresponding closure condition. Therefore, as advertised, we will continue the programme started in the previous chapter and intensively interpret Bianchi identities as closure conditions. In this chapter, we will make explicit constructions of such interesting connections with some normal interpretation.

A good generalization should have the following three properties: there should be a notion of closure, it should have a nice geometrical interpretation and we should have some reconstruction procedure available \textit{à la} Minkowski. The first point will be guaranteed by the interpretation as a Bianchi identity of the closure constraint. The second point will be harder to check and apart from special cases will remain open. But it is satisfied in a minimal sense as the curvature (but not finite holonomies) gives the normal. And the third is blatantly left open in this thesis, though we have hinted at possible reconstruction procedures in our published work \cite{Charles:2016xzi}.

This chapter is taken from the work done in \cite{Charles:2016xzi} and is organized as follows: in the first section, we will pose definitions and notations for our work on the hyperboloid. Then, we will discuss what group we must expect for the normals and the various problems we will encounter with naive approaches. This will allow us to finally define our $\mathrm{SL}(2,\mathbb{C})$ connection generalizing the $\mathrm{ISU}(2)$ construction in the flat case. We will discuss various limit cases. And finally, we will step back to get some perspective and discuss the implications for coarse-graining.

\section{Hyperbolic geometries}

Let us define the framework. We want to define polyhedra on an hyperbolic manifold. Therefore, we must first define the 3d hyperboloid. The simplest definition for our purpose is as follows: the 3d hyperboloid is the set of points at square distance $\kappa^2$ from the origin of Minkowski space, which maybe augmented by some sign condition in order to avoid having a two-sheet hyperboloid. $\kappa$ is the radius of curvature of the hyperboloid. In equation form, this means that the hyperboloid $\mathcal{H}$ is the set of points with coordinates $(t,x,y,z)$ in $\mathbb{R}^{3,1}$ satisfying:
\begin{equation}
  t^2 - (x^2 + y^2 + z^2) = \kappa^2, \quad t \ge 0
\end{equation}
where the positivity condition on $t$ selects the upper sheet of the hyperboloid.

The isometry group of the hyperboloid is the transformations of Minkowski space that preserves the previous quadratic form. This is exactly the Lorentz transformations of $3+1$d spacetime. Because, we will need some idea on how to decompose the transformations of this group in some notions of translational and rotational part, let's make precise the action of the $\mathrm{SL}(2,\mathbb{C})$ group onto this hyperboloid.

We note, first, that the points of Minkowski space are, as a vector space, in one-to-one correspondence with the hermitian $2\times2$ matrices. Indeed, any matrix $M$ from $\mathrm{H}_2(\mathbb{C})$ can be written as:
\begin{equation}
  M = \begin{pmatrix}
    t+z & x - \mathrm{i} y \\
    x + \mathrm{i} y & t-z
  \end{pmatrix}
\end{equation}
This writing seems a bit artificial at first (though it proves the existence of the bijection), but it is quite natural when we consider the determinant of the matrix:
\begin{equation}
  \det M = t^2 - (x^2 + y^2 + z^2)
\end{equation}
which perfectly reproduces the quadratic form on Minkowski's space. Therefore, the hyperboloid $\mathcal{H}$ can now be seen as a set of matrices:
\begin{equation}
  \mathcal{H} \simeq \{ M \in \mathrm{H}_2(\mathbb{C}) ~/~ \det M = \kappa^2 ~ \& ~ \mathop{Tr} M \ge 0 \}
\end{equation}
The positivity condition on the trace corresponds to the selection of the upper sheet as $\mathop{Tr} M = 2t$. Using this writing of the coordinates, the action of the Lorentz group is now quite simple. We have:
\begin{equation}
\forall \Lambda \in \mathrm{SL}(2,\mathbb{C}),~\forall M \in \mathrm{H}_2(\mathbb{C}),~ \Lambda \triangleright M = \Lambda M \Lambda^\dagger
\end{equation}
where $\triangleright$ is the action of the group and $\dagger$ denote the transconjugate. This action preserves the determinant, as well as the sign of the trace and therefore, defines an action on the hyperboloid.

Let us note here, that for any point $M$ on the hyperboloid, there is a Lorentz group element sending the origin point ($t = \kappa$) of the hyperboloid onto this point, or more generally given two points on the hyperboloid, there is always a group element sending one onto the other. In usual special relativity, this corresponds to the fact that there is always a Lorentz transform sending one frame of reference onto another, a natural transformation being the boost. For us, this is important because, we can now try and find a translational part of the group which must be three dimensional but such that this subset still has a transitive action on the hyperboloid. Two choices are quite natural and correspond to different slicing of $\mathrm{SL}(2,\mathbb{C})$ with different properties. First, we can, as was just mentioned, consider boosts. The set of boosts enables us to cover the full hyperboloid starting from one point and therefore gives a decomposition of any elements into two parts (boost plus a rotational part corresponding to the stabilizer of a point):
\begin{equation}
  \forall \Lambda \in \mathrm{SL}(2,\mathbb{C}), ~\exists! (B,H)\in \mathrm{SH}_2(\mathbb{C})\times \mathrm{SU}(2),~\Lambda = BH
\end{equation}
This is called the left Cartan decomposition. We can of course change the order and get a right decomposition instead. But the set of boosts is not a group which might come in hard to handle. In particular, for our closure condition, we won't use them.

The second choice is the Borel subgroup of $\mathrm{SL}(2,\mathbb{C})$ usually written $\mathrm{SB}(2,\mathbb{C})$. It corresponds to lower triangular matrices as follows:
\begin{equation}
  \forall \ell \in \mathrm{SB}(2,\mathbb{C}),~\exists (\lambda,\omega) \in \mathbb{R}_+\times \mathbb{C},~\ell = \begin{pmatrix}
    \lambda & 0 \\
    \omega & \lambda^{-1}
  \end{pmatrix}
\end{equation}
This is indeed a group, and with the previous action we can show that it acts transitively on the hyperboloid. The drawback of using such a group is that its elements transform in a very odd way under rotation making some writings hard to do. Still, we will stick with them. Indeed, as we will see in the next section, if we take the structure of quantum groups to be the right one, we can articulate what kind of algebraic laws, especially transformation laws might be suited and therefore requested for our normals. It turns out that this is precisely the one we will uncover for the Borel subgroup.  In particular, the Borel subgroup $\mathrm{SB}(2,\mathbb{C})$ must somehow appear. This select another decomposition which is the (left) Iwasawa decomposition. It is written as follows:
\begin{equation}
  \forall \Lambda \in \mathrm{SL}(2,\mathbb{C}), ~\exists! (L,H)\in \mathrm{SB}(2,\mathbb{C})\times \mathrm{SU}(2),~\Lambda = LH
\end{equation}
Once again, the order can be inverted to get the right decomposition.

To conclude this section on the 3d hyperboloid, let's finally define our question for the chapter. We want to consider polyhedra on this hyperboloid. These polyhedra will be defined, presumably, by points on the hyperboloid and geodesic arcs and faces between them. Our goal will now be to define a connection on the hyperboloid such that its holonomies around the faces of polyhedra will have some geometric notion of normals and more importantly, such as these normals satisfy a natural and nice closure condition.

%*****************************************

\section{Quantum deformed Loop Quantum Gravity}

The first kind of normals we might consider looking at is normals in $\mathrm{SU}(2)$ directly given by the natural connection on the hyperboloid. Indeed, this idea was explored in our work \cite{Charles:2015lva} as well as in \cite{Haggard:2014xoa,Haggard:2015yda,Haggard:2015ima} and is natural in several regards:
\begin{itemize}
\item First, this idea fits well into the scheme we developed so far: Bianchi identities give natural closure and on a curved background, such a holonomy is non-trivial.
\item Regarding algebraic behavior, the holonomy is also attractive: because the connection is an $\mathrm{SU}(2)$ connection, it naturally gets an $\mathrm{SU}(2)$ action on its holonomies corresponding to gauge-transform. This means that we can naturally define parallel transport on links of the coarse-grained graph.
\item And for the geometric interpretation, this is where the idea shines. Indeed, in a homogeneously curved background (such as an hyperboloid), the curvature naturally encodes a notion of normal and of area. Indeed, it can be shown that the axis of rotation of an holonomy is normal to the (geodesic) surface it surrounds. And as the deficit angle is related to the surface, the angle of rotation (which also happens to be gauge-invariant) encodes the area of the surface.
\end{itemize}
It has drawbacks though. One property of using such a connection is that it looses the sign of the curvature. It must be reconstructed from boundary data. This can also be seen as an attractive feature, especially in the context of coarse-graining where we might want to encode blocks of different curvature, but here, we will concentrate on another endeavour. We will look for descriptions that differ depending on the sign of the curvature, either hyperbolic or spherical. Let us note also that we expect our flat case result to be generalized to the hyperbolic case. In the flat case, the connection appeared as a deformation of the spin connection, it could be that this works also in the hyperbolic case.

Another route should also warn us that a more general structure should be expected and it comes from $2+1$d quantum gravity. Indeed, when we want to include a cosmological constant in $2+1$d gravity, the Ponzanno-Regge model gets deformed into the Tuarev-Viro model. This model is based on the representation theory of quantum groups. Including this deformation in the canonical framework has been investigated for some time now \cite{Dupuis:2013haa,Bonzom:2014wva,Bonzom:2014bua,Dupuis:2014fya}. At the classical limit, the main idea is that the element of $\mathrm{T}^* \mathrm{SU}(2)$ on each link of the graph gets replaced by an element of $\mathrm{SL}(2,\mathbb{C})$ (in case of hyperbolic curvature) with the appropriate symplectic structure coming from a Drinfeld double construction. As the group $\mathrm{T}^* \mathrm{SU}(2)$ gets separated into a translational ($\mathbb{R}^3$) and rotational part ($\mathrm{SU}(2)$), the new group also gets decomposed in a similar way. The Iwasawa decomposition alluded to in the previous section is used and for each link, we get an $\mathrm{SU}(2)$ element corresponding to the parallel transport and two $\mathrm{SB}(2,\mathbb{C})$ elements, one for each end of the link, with some matching condition. The spin network can then be naturally generalized to a \textit{ribbon} network as illustrated in figure \ref{fig:ribbon}. The closure condition are imposed on every loop of $\mathrm{SB}(2,\mathbb{C})$ elements and flatness\graffito{The word ``flatness'' might be misleading as the curvature is really negative. Still as far as the holonomies are concerned, there are sent to the identity.} conditions can be imposed on $\mathrm{SU}(2)$ loops. The $\mathrm{SB}(2,\mathbb{C})$ therefore have a similar role than the normals in the $4$d theory. It has been conjectured for some time now that they would also appear in a hyperbolic setting as a deformation of $\mathbb{R}^3$ for the $4$d theory \cite{Dupuis:2013haa}. The advantage of importing such technology from the $3$d case is that the algebraic data is very well understood. Therefore we can use their properties to select interesting normals. Our programme therefore corresponds to giving a geometric interpretation to the quantum deformed framework.

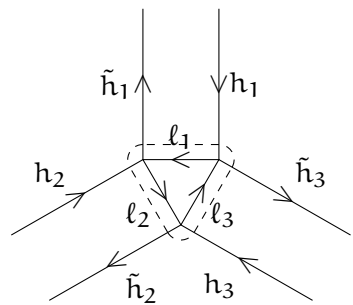
\begin{figure}[h!]
  \centering

  \begin{tikzpicture}
    \def \d{0.2}
    
    \coordinate (A) at (0,0);
    \coordinate (A1) at (-30:\d);
    \coordinate (A2) at (-150:\d);
    \coordinate (B) at (60:1);
    \coordinate (B1) at ($(B) + (-30:\d)$);
    \coordinate (B2) at ($(B) + (90:\d)$);
    \coordinate (C) at (120:1);
    \coordinate (C1) at ($(C) + (-150:\d)$);
    \coordinate (C2) at ($(C) + (90:\d)$);

    \draw (A) -- node[midway,sloped]{$>$} node[midway,below right]{$\ell_3$} (B) -- node[midway,sloped]{$<$} node[midway,above]{$\ell_1$} (C) -- node[midway,sloped]{$>$} node[midway,below left]{$\ell_2$} (A);

    \draw (B) -- node[midway,sloped]{$<$} node[midway,right]{$h_1$} ++(90:2);
    \draw (C) -- node[midway,sloped]{$>$} node[midway,left]{$\tilde{h}_1$} ++(90:2);

    \draw (B) -- node[midway,sloped]{$>$} node[midway,above right]{$\tilde{h}_3$} ++(-30:2);
    \draw (A) -- node[midway,sloped]{$<$} node[midway,below left]{$h_3$} ++(-30:2);

    \draw (A) -- node[midway,sloped]{$<$} node[midway,below right]{$\tilde{h}_2$} ++(-150:2);
    \draw (C) -- node[midway,sloped]{$>$} node[midway,above left]{$h_2$} ++(-150:2);

    \draw[dashed] (A1) -- (B1);
    \draw[dashed] (A2) -- (C1);
    \draw[dashed] (B2) -- (C2);

    \draw[dashed] (A2) arc (-150:-30:\d);
    \draw[dashed] (B1) arc (-30:90:\d);
    \draw[dashed] (C2) arc (90:210:\d);
   
  \end{tikzpicture}
  
  \caption{In quantum deformed loop quantum gravity, the spin networks are replaced by ribbon spin networks. Each vertex is replaced by a loop of edges carrying $\mathrm{SB}(2,\mathbb{C})$ elements. And each link is replaced by two links carrying (in general) different $\mathrm{SU}(2)$ elements for the two directions of propagation. Therefore a vertex in usual spin networks correspond to the whole encircled region on the graph.}
  \label{fig:ribbon}
\end{figure}

We should remember that these arguments are heuristic, but still, we have some independent pieces of evidence pointing to such a construction. One comes from the deformation which is natural and the second one comes from a much more understood theory. We will continue this chapter therefore by supposing that indeed, the deformed normals should be encoded in $\mathrm{SB}(2,\mathbb{C})$ elements. In the next section, we will see how this naturally leads us to considering $\mathrm{SL}(2,\mathbb{C})$ connections.

%\textbf{TODO:} now, what kind of normals should we be looking for ? possibility of $\mathrm{SU}(2)$, natural with respect to a few points (normal interpretation, area), but another perspective: quantum deformed theory, quick survey of the theory (in 3d), works in 3d, just a guideline in 4d, but following this, we would like $\mathrm{SB}(2,\mathbb{C})$ normals

%*****************************************

\section{$\mathrm{SB}(2,\mathbb{C})$ transformation laws}

In the 3d theory, from which we will now borrow the algebraic content, the $\mathrm{SB}(2,\mathbb{C})$ elements come from the $\mathrm{SB}(2,\mathbb{C})$ part of an $\mathrm{SL}(2,\mathbb{C})$ element carried by the link of the graph. This $\mathrm{SL}(2,\mathbb{C})$, carried by the edge of the graph of a (deformed) spin network, must not be confused with the one we will introduce for our connection. Indeed, we will now introduce a new $\mathrm{SL}(2,\mathbb{C})$ element as a holonomy of a connection. Both the $\mathrm{SU}(2)$ and the $\mathrm{SB}(2,\mathbb{C})$ carry information about the normal. Whereas the $\mathrm{SU}(2)$ part coming from the deformed quantum geometry correspond to parallel transport. This means that the deformed normals in $\mathrm{SB}(2,\mathbb{C})$ get transported using $\mathrm{SU}(2)$ elements. In this part, we will concentrate on this particular: the parallel transport of $\mathrm{SB}(2,\mathbb{C})$ elements with $\mathrm{SU}(2)$ elements.

The classical limit of the quantum group deformation can be seen as endowing a even-dimensional group with a (symplectic) Poisson bracket. This can be done using the Drinfeld double construction with an $R$-matrix as done in \cite{Bonzom:2014wva}. In our case, the phase space is $\mathrm{SL}(2, \mathbb{C})$, where we define the Poisson bracket for the elements $D$ of the phase space as follows:
\begin{equation}
  \{D_1,D_2\} = -r D_1 D_2 - D_1 D_2 r^\dagger
\end{equation}
As usual in the quantum group literature, the indices mark places in the tensor product. Therefore $D_1 = D \otimes  \mathbb{1}$ and $D_2 = \mathbb{1} \otimes D$ and the previous bracket should really be thought as the bracket of all possible brackets of all the pair of matrix elements of $D$. $r$ in the previous equation is the $r$-matrix encoding the deformation and is given by:
\begin{equation}
  r = \frac{\kappa}{4} \sum_i \tau_i \otimes \sigma_i = \frac{\mathrm{i}\kappa}{4}\begin{pmatrix}
    1 &  0 &  0 & 0 \\
    0 & -1 &  0 & 0 \\
    0 &  4 & -1 & 0 \\
    0 &  0 &  0 & 1
  \end{pmatrix}  
\end{equation}
in terms of the Pauli matrices $\sigma_i$, which are up to a factors the generators of $\mathrm{SU}(2)$, and $\tau_i = \mathrm{i}(\sigma_i - \frac{1}{2}[\sigma_3,\sigma_i]) = (\mathrm{i}\sigma_i + \epsilon^k_{3i} \sigma_k)$ which are the generators of $\mathrm{SB}(2,\mathbb{C})$. Now, this Drinfeld double structure induces a bracket on the $\mathrm{SB}(2,\mathbb{C})$ part of this phase space. Any $\mathrm{SB}(2,\mathbb{C})$ element $\ell$ can be written as:
\begin{equation}
  \ell = \begin{pmatrix}
    \lambda & 0 \\
    \omega & \lambda^{-1}
  \end{pmatrix},
  \quad \lambda > 0,
  \quad \omega \in \mathbb{C}
\end{equation}
And then, with these notations, the induced Poisson bracket reads:
\begin{equation}
  \{\lambda,\omega\} = \frac{\mathrm{i}\kappa}{2} \lambda \omega,\quad
  \{\lambda,\overline{\omega}\} = -\frac{\mathrm{i}\kappa}{2} \lambda \overline{\omega},\quad
  \{\omega,\overline{\omega}\} = \mathrm{i}\kappa (\lambda^2 - \lambda^{-2})
\end{equation}
If we use the brackets, the closure condition for the deformed normals, namely:
\begin{equation}
  \ell_n ... \ell_2 \ell_1 = \mathbb{1}
\end{equation}
generates $\mathrm{SU}(2)$ rotations. The action is in fact non-linear but can still be generated at the infinitesimal level in the following way:
\begin{equation}
  \exp \left(\prod_k \lambda_k^{-2} \{\mathop{Tr} V \mathcal{G} \mathcal{G}^\dagger, \cdot\}\right) \ell = \ell^{(\textrm{rot})}
\end{equation}
where the $\lambda_k$ are the $\lambda$ component of the $k^\textrm{th}$ $\mathrm{SB}(2,\mathbb{C})$ matrix from the (deformed) Gauss constraint $\mathcal{G}$. The finite action can be computed explicitly to find:
\begin{equation}
  \ell \rightarrow k \ell \tilde{k}^{-1}
\end{equation}
where $\tilde{k}$ is the unique $\mathrm{SU}(2)$ element such as the previous quantity is indeed in $\mathrm{SB}(2,\mathbb{C})$. If we were to take the four normals associated to a tetrahedron ($\ell_1$ to $\ell_4$) we would have the following transformations for a rotation $k$:
\begin{equation}
  \left|
  \begin{array}{rcl}
    \ell_4 &\rightarrow& k \ell_4 (k^{(1)})^{-1} \\
    \ell_3 &\rightarrow& k^{(1)} \ell_3 (k^{(2)})^{-1} \\
    \ell_2 &\rightarrow& k^{(2)} \ell_2 (k^{(3)})^{-1} \\
    \ell_1 &\rightarrow& k^{(3)} \ell_1 k^{-1}
  \end{array}
  \right.
\end{equation}
where the $k^{(i)}$ must all be chosen as needed for the elements to fall in the right group. Note that the last element is noted $k$ and not $k^{(4)}$. In principle, for a generic set of vectors, this has no reason to be. But as our tetrahedron satisfies the closure, we have the guaranty that $\ell_4 ... \ell_1 = \mathbb{1}$. As the identity transforms by conjugation, if the transformation is consistent, then the action on the right for $\ell_1$ will be $k$. This can also be checked explicitly \cite{Bonzom:2014wva}. This transformation rule is the difficulty we have to face in our construction. It is what makes the passage to $\mathrm{SL}(2,\mathbb{C})$ connections necessary.

Indeed, if we were to develop an $\mathrm{SB}(2,\mathbb{C})$ connection, it is very difficult to see how the $\mathrm{SU}(2)$ group would act. Indeed, if we were to consider the holonomy around the same loop twice for instance, the transformation could not apply. Let us label the holonomy around a loop once $\ell$ and the holonomy around it twice $L$. We would have:
\begin{equation}
  L = \ell^2
\end{equation}
But this quantity does not transform well under rotation. Indeed, if $\ell$ is seen as the generalization of a vector, its end point and its start point do not transform in the same way. So, the product would need some compensation in term of an $\mathrm{SU}(2)$ matrix for instance between the too composition. Something like $L = \ell h \ell$ with $h$ having nice transformation property. Precisely, we want $h$ to behave like:
\begin{equation}
  \left| ~ h \rightarrow k^{(1)} h k^{-1} \right.
\end{equation}
This idea leads us to construct elements into two parts, one in $\mathrm{SB}(2,\mathbb{C})$ and one in $\mathrm{SU}(2)$ such as their product nicely transforms by conjugation. As the product of $\mathrm{SB}(2,\mathbb{C})$ and $\mathrm{SU}(2)$ is $\mathrm{SL}(2,\mathbb{C})$, this leads us to an $\mathrm{SL}(2,\mathbb{C})$ connection.

But let's convince ourselves first, that such a construction of an $\mathrm{SL}(2,\mathbb{C})$ normal is natural and possible in the discrete setting. For definiteness, let's consider a tetrahedron and four (deformed) normals given by four $\mathrm{SB}(2,\mathbb{C})$ elements labeled from $\ell_1$ to $\ell_4$ and satisfying:
\begin{equation}
\ell_4 \ell_3 \ell_2 \ell_1 = \mathbb{1}
\end{equation}
Can we construct an $\mathrm{SL}(2,\mathbb{C})$ closure from there? Indeed, as any element $\ell$ of $\mathrm{SB}(2,\mathbb{C})$ can be written uniquely using the Cartan decomposition as:
\begin{equation}
  \ell = bh
\end{equation}
where $b$ is a boost, that is a matrix from $\mathrm{H}_2(\mathbb{C})$ and $h$ is a rotation matrix from $\mathrm{SU}(2)$. The advantage of such a transformation is that these new elements transform as follows:
\begin{equation}
  \left|
  \begin{array}{rcl}
    b &\rightarrow& kbk^{-1} \\
    h &\rightarrow& k h \tilde{k}^{-1}
  \end{array}
  \right.
\end{equation}
We see here that the boost has the transformation we wanted and the $h$ transforms exactly as the padding we hoped. We can now start transformation our closure:
\begin{equation}
  \begin{array}{rcl}
    \ell_4 \ell_3 \ell_2 \ell_1 &=& \mathbb{1} \\
    \Leftrightarrow b_4 h_4 \ell_3 \ell_2 \ell_1 &=& \mathbb{1} \\
    \Leftrightarrow b_4 h_4 \ell_3 (h_4^{(1)})^{-1} h_4^{(1)} \ell_2 \ell_1 &=& \mathbb{1}
  \end{array}
\end{equation}
In the last line, we just introduce a non-trivial writing of the identity, so that we get a transformed $\mathrm{SB}(2,\mathbb{C})$ element namely $h_4 \ell_3 (h_4^{(1)})^{-1}$. This element can now be decomposed using the Cartan decomposition:
\begin{equation}
  \begin{array}{rcl}
    \ell_4 \ell_3 \ell_2 \ell_1 &=& \mathbb{1} \\
    \Leftrightarrow b_4 b_3 h_3 h_4^{(1)} \ell_2 \ell_1 &=& \mathbb{1}
  \end{array}
\end{equation}
Continuing the process we find:
\begin{equation}
  b_4 b_3 b_2 b_1 h_1 h_2^{(1)} h_3^{(2)} h_4^{(3)} = \mathbb{1}
\end{equation}
This does not seem to have taken far. But actually, we now have an $\mathrm{SL}(2,\mathbb{C})$ closure with elements that transform by conjugation. They read:
\begin{equation}
  b_4,\quad b_3,\quad b_2\quad\textrm{and}\quad (b_1 h_1 h_2^{(1)} h_3^{(2)} h_4^{(3)})
\end{equation}
Of course, the converse construction can be undertaken: if we have an $\mathrm{SL}(2,\mathbb{C})$ closure we can use the Iwasawa decomposition to find now two closures, one in $\mathrm{SB}(2,\mathbb{C})$ and one in $\mathrm{SU}(2)$. Indeed, the $\mathrm{SL}(2,\mathbb{C})$ connection associated to the $\mathrm{SB}(2,\mathbb{C})$ closure is not unique. But only one is needed to find the $\mathrm{SB}(2,\mathbb{C})$ closure. This condition that the connection transforms by conjugation will be our main guide. And as we will see, some natural connections come out.

%*****************************************

\section{An $\mathrm{SL}(2,\mathbb{C})$ connection}

Let’s dwell into the technical side. We want to define explicitly an $\mathrm{SL}(2, \mathbb{C})$ connection. Because this will ease our life later on, let’s note here that $\mathrm{SL}(2, \mathbb{C})$ is the complexification of $\mathrm{SU}(2)$, that is: $\mathrm{SL}(2, \mathbb{C}) \simeq \mathrm{SU}_\mathbb{C}(2)$. Accordingly, the generators of $\mathrm{SL}(2, \mathbb{C})$ (which is 6-dimensional) can be written as $J_1$ , $J_2$ , $J_3$ , $\mathrm{i}J_1$, $\mathrm{i}J_2$ and $\mathrm{i}J_3$ where $J_1$, $J_2$ and $J_3$ are the generators of $\mathrm{SU}(2)$. Therefore, an $\mathrm{SL}(2, \mathbb{C})$ connection can be thought of as a complex $\mathrm{SU}(2)$ connection. Let’s now define the following connection on the hyperboloid:
\begin{equation}
  A_\textrm{SL} = \left(\Gamma^i + \frac{\beta}{\kappa} \mathrm{e}^i\right)J_i
  \label{eq:connection}
\end{equation}
where the $J_i$ are the generators of the $\mathrm{SU}(2)$ group, which can be represented (up to a $\frac{1}{2}$ factor) by the Pauli matrices, $\beta \in \mathbb{C}$ is a parameter and $\Gamma^i$ is the unique spin connection on the hyperboloid compatible with the metric and the triad and without torsion. The $\kappa$ parameter is put here just to keep $\beta$ dimensionless and to make a clearer comparison with the Immirzi-Barbero parameter of the Ashtekar-Barbero connection. This connection is written as a complex connection of the $\mathrm{SU}(2)$ group and can therefore be reinterpreted as an  $\mathrm{SL}(2,\mathbb{C})$ connection. To further point out that it is indeed an $\mathrm{SL}(2,\mathbb{C})$ connection, we can also write:
\begin{equation}
  A_\textrm{SL} = \left(\Gamma^i + \frac{\Re (\beta)}{\kappa} e^i\right) J_i + \frac{\Im (\beta)}{\kappa} e^I B_I
\end{equation}
where the $B_I$ are the boosts generators and we used $B_I = \mathrm{i} J_I$.

The connection we just defined is, in a certain sense, the Ashtekar-Barbero connection which we will write $A_\textrm{A-B}$. This is not true in a general sense as the connection we just defined is totally intrinsic, depending only on intrinsic geometry quantities like the spin connection or the triad, and the Ashtekar-Barbero connection is defined for an embedded surface in a 4d spacetime. But with respect to the embedding of the hyperboloid in Minkowski space, the two connections match. Indeed, because the hyperboloid is homogeneous, we have:
\begin{equation}
  \frac{e^i}{\kappa} = K^i
\end{equation}
where $K^i$ is the extrinsic curvature. There are therefore subtle differences between $A_\textrm{SL}$ and $A_\textrm{A-B}$. The first one is that $A_\textrm{SL}$ is intrinsic, that is depends only on intrinsic geometry. $A_\textrm{A-B}$ on the other hand explicitly depends on extrinsic data and therefore on the embedding. So, while they coincide on what we might call \textit{on-shell}, the two connections are actually quite different and would differ if the embedding were to change. For instance, if we were to embed the hyperboloid in some curved space, like Anti-De-Sitter space for instance, the two connections would not match. A second difference will be of relevance latter on for the Poisson structure. There are natural non-trivial brackets between $A_\textrm{A-B}$ and the triad because of the presence of the extrinsic curvature. This is not the case for $A_\textrm{SL}$ which depends only on the triad and therefore commutes with it.

With this in mind, it seems way more natural to consider $A_\textrm{SL}$ for geometrical interpretation on the hyperboloid as it does not rely on a particular embedding. For instance, the curvature of $A_\textrm{SL}$ should have a natural geometric meaning:
\begin{equation}
  F[A_\textrm{SL}] = \frac{1+\beta^2}{\kappa^2} \epsilon^i_{~jk} J_i \mathrm{e}^j \wedge \mathrm{e}^k
\end{equation}
We see here that, up to a (complex) factor, there is once again an interpretation as a normal. That is, at least, at the infinitesimal level. After integration, the non-commutativity of the group might induce non trivial deformation from a canonical normal, but still, this is a very natural generalization of Freidel's connection. This also appears to be a generalization of the usual $\mathrm{SU}(2)$ connection (which indeed encoded a normal) when we note that we can write the curvature as:
\begin{equation}
  F^{ij}[A_\textrm{SL}] = \Lambda \mathrm{e}^i \wedge \mathrm{e}^j
\end{equation}
where $\Lambda = \frac{1+\beta^2}{\kappa^2}$ acts as a cosmological constant but complex. This in particular shows that this connection is homogeneous. This will be of great relevance for its property under rotation.

The precise behaviour of $A_\textrm{SL}$ will of course depend on the value of $\beta$ (or equivalently, on the value of $\Lambda$). Though we will not go through a complete survey of the possible values and behaviours, we should note some specific instances with interesting properties:
\begin{itemize}
\item First, as for the Ashtekar-Barbero variables \textit{per se}, the values $\beta = \pm \mathrm{i}$ are very specific and induce very specific properties. In that case, the connection is the self-dual $\mathrm{SL}(2,\mathbb{C})$ connection, which can be thought of as the natural connection induced by the flat Minkowski connection. This connection is entirely flat, as can be seen from the value of $\Lambda = 0$. In particular, no information at all is preserved in the holonomies. We do have a closure but only because it is a trivial closure.

  We should note here that this is not particularly surprising. There are also peculiar values of $a$ and $b$ in $A_\textrm{ISU}$ which makes the connection trivially flat. But the values are simply $a=b=0$. Still, they exists and find their equivalent precisely in $\beta = \pm \mathrm{i}$.
\item If $\beta \in \mathbb{R}$, the connection is pure $\mathrm{SU}(2)$. The special case $\beta = 0$ corresponds to the usual metric-compatible torsion-free connection. But we have a whole class of new $\mathrm{SU}(2)$ connections here which still have closure. They are a generalization of $A_\textrm{nc}$, which was a deformation of the flat connection on the plane.
\item One particularly interesting choice is $\Lambda \in \mathrm{i}\mathbb{R}$, that is $1 + \beta^2$ purely imaginary. In this case, $\beta$ is on the unit hyperboloid in the complex plane. At the infinitesimal level, this implies that the holonomy is a pure boost. Granted that the finite case might be a bit more convoluted, this is in some sense an \textit{orthogonal} version of a pure $\mathrm{SU}(2)$ connection. What we mean here is that, as in the flat case we had two natural sets of generators that where in bijection, in the curved case also, we have a natural duality between the rotation generators and the boosts generators. And this construction seems to be the natural dual construction corresponding to the pure $\mathrm{SU}(2)$ case.

  This choice seems to be one of the possible generalization of $A_\textrm{Fr}$. Indeed, choosing $\Lambda$ purely imaginary corresponds to choosing $a \rightarrow 0$ in $A_\textrm{ISU}$ in the flat case, which sends the connection to $A_\textrm{Fr}$. And in both cases, the class of connection at least appears to have a free (real) parameter (contrary to the usual $\mathrm{SU}(2)$ connection), the parameter being $b$ in $A_\textrm{ISU}$ and $\mathrm{i}\Lambda$ for $A_\textrm{SL}$.

\item We should note a final possibility which also carries some interest: when $\beta \in \pm\mathrm{i} + \mathbb{R}$. This case corresponds to a very natural geometrical construction.

  Indeed, let's consider two points on the hyperboloid. What holonomy along the geodesic could we possibly attribute? A natural choice is the unique boost sending the first point to the second. But there is a whole lot class of possible holonomies that match this geometrical intuition. This is because the $\mathrm{SL}(2,\mathbb{C})$ transformation sending the first point to the second is in fact not unique since any rotation around the final point can be added. This is indeed the case by definition of the rotation subgroup which is the stabilizer of a point.

  If we set the rotation to be around the axis of the boost, we find the previous connection with $\beta = \pm\mathrm{i} + \lambda$. The $\lambda$ then is a helix parameter telling us how much we wind up around the axis for a specific length. The geometrical resemblance with spinning geometry is kind of cunning and might point to something deeper. Note here that there is a natural equivalent in the flat case, once again. This is when we set $b=\pm 1$ and $a$ is let free in $A_\textrm{ISU}$.
  \end{itemize}

Now that we have developed a more precise intuition about the meaning of the connection $A_\mathrm{SL}$, especially at the infinitesimal level, let's now turn to the finite case. Our main interest is the study of the hyperbolic tetrahedron. Its faces are hyperbolic triangle and as such we are interested in the computation of the holonomy of $A_\textrm{SL}$ around each of this triangle. The holonomy can now be computed exactly around a finite triangle as detailed in \cite{Charles:2016xzi}. The result has a surprisingly simple form, as the holonomy $h$ around a triangle $ABC$ simply reads:
\begin{equation}
  h = -R_{Y,\pi-\hat{a}} B_{\mathrm{i} \beta \ell_{AC}} R_{Y,\pi - \hat{c}} B_{\mathrm{i} \beta \ell_{BC}} R_{Y,\pi -\hat{b}} B_{\mathrm{i} \beta \ell_{AB}}
\end{equation}
where:
\begin{equation}
  B_\ell = \begin{pmatrix}
    \mathrm{e}^{-\frac{\ell}{2\kappa}} & 0 \\
    0 & \mathrm{e}^{\frac{\ell}{2\kappa}}
  \end{pmatrix}\ ,\quad
  R_{Y,\alpha} = \begin{pmatrix}
    \cos \frac{\alpha}{2} & \sin \frac{\alpha}{2} \\
    -\sin \frac{\alpha}{2} & \cos \frac{\alpha}{2}
    \end{pmatrix}
\end{equation}
are the boosts and rotation associated to a given length or angle and $\ell_{AB}$ is the length of the geodesic from $A$ to $B$ and $\hat{a}$ is the angle at the point $A$.

This expression is quite nice as each term has a clear geometrical interpretation. For each edge, there is an exponential of the length. The complex and imaginary part corresponds to the boost and rotation part around the axis and are controlled by the $\beta$ parameter. For each wedge, there is a rotation of the corresponding angle. So the holonomy is quite simply built by turning around the triangle and composing every relevant term associated to the geometrical element being passed on. The expression presented above corresponds to a triangle in a specific plane with some gauge choice. It is if course possible to generalize for an arbitrary triangle. Indeed, we just have to conjugate the expression by the appropriate rotation so that the $\sigma_z$ coordinate in tangent space corresponds to the direction of the normal of the triangle. But we can do better. Let us write $h$ in the following manner:
\begin{equation}
  h = R B_{CA}^{\mathrm{i}\beta} B_{BC}^{\mathrm{i}\beta} B_{AB}^{\mathrm{i}\beta}
\end{equation}
where $R$ is the rotation around the triangle, that is the holonomy of the natural spin connection on the hyperboloid, $B_{AB}$ is the boost sending $A$ onto $B$ and the exponentiation must be understood as a quick hand notation for:
\begin{equation}
  \left(\exp\left(\frac{\eta}{2} \hat{u} \cdot \overrightarrow{\sigma}\right)\right)^{\alpha} = \exp\left(\frac{\alpha \eta}{2} \hat{u} \cdot \overrightarrow{\sigma}\right)
\end{equation}
which is strictly defined only for boosts, though $\alpha$ can be complex. We, of course, recover the usual holonomy this way when $\beta = 0$. More interestingly, we recover that $h=1$ if $\beta = \pm\mathrm{i}$ since the expression of the holonomy is precisely the boosts closure condition found for a hyperbolic triangle in \cite{Charles:2015lva}. More generally, we see in a quite precise sense that the connection thus defined is indeed a deformation of the usual spin connection.

Now, let's turn back to the 3-dimensional problem. We want to study the closure for the tetrahedron. The closure condition appears exactly as in the $\mathrm{ISU}(2)$ (flat) case. To be a little more precise, and fix the notation, if we have a hyperbolic tetrahedron $ABCD$, let's define a root for the holonomies. We choose the point $A$ for this. Around each face, we can define the $\mathrm{SL}(2,\mathbb{C})$ holonomy. Let us write it $\Lambda_i$ where $i$ is the name of the opposite vertex. So, for the face $ABC$, the holonomy is called $\Lambda_D$. Then, we have:
\begin{equation}
  \left\{\begin{array}{rcl}
  \Lambda_B &=& g_{AD}^{-1} g_{CD} g_{AC}  \\
  \Lambda_C &=& g_{AB}^{-1} g_{BD}^{-1} g_{AD} \\
  \Lambda_D &=& g_{AC}^{-1} g_{BC} g_{AB}
  \end{array}\right.
\end{equation}
each holonomy being rooted in $A$. The $g${\scriptsize s} correspond to the holonomy for open path. Only three holonomies are given in the previous equation. Indeed, the last holonomy is bit more complicated as we have to parallel transport along an edge. Choosing to parallel transport along $AC$, we have:
\begin{equation}
  \Lambda_A = g_{AC}^{-1} \left(g_{CD}^{-1} g_{BD} g_{BC}^{-1}\right) g_{AC}
\end{equation}
The discrete Bianchi identity then reads:
\begin{equation}
  \Lambda_D \Lambda_C \Lambda_B \Lambda_A = \mathbb{1}
\end{equation}
This is the $\mathrm{SL}(2,\mathbb{C})$ closure constraints derived from the connection introduced in the previous paragraph. It is the direct parallel of the $\mathrm{ISU}(2)$ closure constraint developed on flat space in the first section.

Now, the $\mathrm{SL}(2,\mathbb{C})$ connection is not the full story. Indeed, we were interested in an $\mathrm{SB}(2,\mathbb{C})$ (the Borel subgroup) closure. So what is the link between the two? The link is to be found in the Iwasawa decomposition. Indeed, every element $\Lambda \in \mathrm{SL}(2,\mathbb{C})$ can be written uniquely as a product of elements in $\mathrm{SB}(2,\mathbb{C})$ and $\mathrm{SU}(2)$. More formally, this can be written:
\begin{equation}
  \forall \Lambda \in \mathrm{SL}(2,\mathbb{C}),\ \exists! (L,H) \in \mathrm{SB}(2,\mathbb{C})\times\mathrm{SU}(2), \textrm{ such that }\Lambda = L H
\end{equation}
This decomposition highlights the fact that the $\mathrm{SL}(2,\mathbb{C})$ group can be understood as a semi-direct product: $\mathrm{SL}(2,\mathbb{C}) \simeq \mathrm{SB}(2,\mathbb{C}) \rtimes \mathrm{SU}(2)$. This is the exact equivalent, in the hyperbolic case, of the decomposition of the $\mathrm{ISU}(2)$ connection introduced in the flat case, since we have $\mathrm{ISU}(2) \simeq \mathbb{R}^3 \rtimes \mathrm{SU}(2)$. In the flat case, $\mathbb{R}^3$ is the translation subgroup of the isometries and $\mathrm{SU}(2)$ is the rotation group (or more precisely its double cover). We have a natural corresponding interpretation here: $\mathrm{SB}(2,\mathbb{C})$ can be understood as a (deformed) translational group on the hyperboloid. Indeed, it is three-dimensional, it is a group and its action on the 3d hyperboloid is both transitive and faithful. The $\mathrm{SU}(2)$ part can naturally be interpreted as a rotation group as it is indeed the stabilizer of a point.

Now, using the Iwasawa decomposition, we can split the $\mathrm{SL}(2,\mathbb{C})$ closure into two closures, one in $\mathrm{SB}(2,\mathbb{C})$ and a second one in $\mathrm{SU}(2)$. For this, let's write down the following decomposition for our group elements:
\begin{equation}
  \left\{\begin{array}{rcl}
  \Lambda_D &=& L_D H_D \\
  \left(H_D\right) \Lambda_C \left(H_D\right)^{-1} &=& L_C H_C \\
  \left(H_C H_D\right) \Lambda_B \left(H_C H_D\right)^{-1} &=& L_B H_B \\
  \left(H_B H_C H_D\right) \Lambda_A \left(H_B H_C H_D\right)^{-1} &=& L_A H_A
  \end{array}\right.
\end{equation}
And thus, we have:
\begin{equation}
  \Lambda_D \Lambda_C \Lambda_B \Lambda_A = L_D L_C L_B L_A H_A H_B H_C H_D = \mathbb{1}
\end{equation}
The decomposition being unique, we have precisely:
\begin{equation}
  \left\{\begin{array}{rcl}
  L_D L_C L_B L_A &=& \mathbb{1} \\
  H_A H_B H_C H_D &=& \mathbb{1}
  \end{array}\right.
\end{equation}
So, the Lorentz closure induces two closures, one for the Borel subgroup and one for the rotational subgroup. The two closures can of course be assembled back to the original closure and do not carry extra information.

Note here, that in order to do this we used parallel transport through the rotational part of the Lorentz elements. This introduces a kind of twisting or braiding of the relations which is not that surprising. Indeed, this twist was already present in the $\mathrm{ISU}(2)$ case in flat space and is in fact introduced by the semi-direct product. The twist extends here in the $\mathrm{SU}(2)$ part because of the non-linear action of $\mathrm{SU}(2)$ over $\mathrm{SB}(2,\mathbb{C})$. This means in particular that the behaviour of the elements under rotation are rather non-trivial. But this behaviour is actually wanted as we saw earlier with the quantum group structure.

%\textbf{TODO:} transformation laws $\rightarrow$ difficulties with holonomies, idea = $\mathrm{SL}$ plus Iwasawa, look for homogeneous connection of $\mathrm{SL}$, turns out natural connection for this, quick survey, explicit computes, question of the recontruction, numerical studies, limit cases ($\mathrm{SU}$)

%*****************************************

\section{Link with coarse-graining}

Let us close this chapter by going back to coarse-graining considerations. First let's note that our connection is quite natural to consider in the context of coarse-graining. That is the relevant degrees of freedom for a coarse-grained region might very-well be better described by holonomies of some connection on its boundary rather than the flux-vector data. We could be tempted to use directly the Ashtekar-Barbero connection but as we underlined the geometric interpretation is rather different. Still, the link between the triad and the extrinsic curvature is due to the flatness of the \textit{embedding} space. It is not unreasonable to think that such a link might still have some reasonable sense in the context of the equivalence principle which guarantees that, locally, spacetime can be considered as flat. In that case, it might be possible that the direct usage of the Ashtekar-Barbero connection might be relevant for describing curved surfaces.

This work will be of help to identify the relevant degrees of freedom when coarse-graining. We should already remark that such a mixing of the flux-vectors and the holonomy are now ubiquitous. We refer in particular to the recent work of Freidel \textit{et al.} \cite{Donnelly:2016auv,Freidel:2015gpa} on the phase space of surfaces in quantum gravity. Indeed, their work shows that it is natural to consider surface degrees of freedom, which encode degrees of freedom related to the $2d$ metric of the surface and which, in our context, can be understood as surface holonomies. It is also interesting to note that similar structures (described by pair of holonomies) appear in the description of defects in the context of 3D quantum gravity \cite{Delcamp:2016yix}.

All this suggest a new structure to describe coarse-grained quantum geometries: \textit{dual spin networks}. They are more or less usual spin networks where the graph structure is the combination of a graph \textit{and} its dual (as illustrated on figure \ref{fig:DualSpinNetwork}). Therefore, a dual spin network naturally carries information about holonomies that are \textit{tangential} to a surface as well as \textit{transversal}. Note, that this is \textit{per se} a generalization of spin networks, as non-existent link on spin networks correspond to link with a trivial spin, but introducing such structure helps us devising natural coarse-grained operators and ways to think about coarse-grained space. 

\begin{figure}[h!]
  \centering

  \begin{tikzpicture}[scale=1]
    \coordinate (O) at (0,0,0);

    \coordinate (A) at (0,1.061,0);
    \coordinate (B) at (0,-0.354,1);
    \coordinate (C) at (-0.866,-0.354,-0.5);
    \coordinate (D) at (0.866,-0.354,-0.5);

    \draw[red] (A) -- ++(A);
    \draw[red] (B) -- ++(B);
    \draw[red] (C) -- ++(C);
    \draw[red] (D) -- ++(D);
    
    \draw[red] (A) -- (B);
    \draw[red] (A) -- (C);
    \draw[red] (A) -- (D);
    \draw[red] (B) -- (C);
    \draw[dashed,red] (C) -- (D);
    \draw[red] (D) -- (B);

    \draw[dotted,blue] (O) -- ++(0,-0.531,0);
    \draw[blue] (0,-0.531,0) -- ++(0,-0.531,0);

    \draw[blue,dotted] (O) -- ++(0,0.177,-0.5);
    \draw[blue,dashed] (0,0.177,-0.5) -- ++(0,0.177,-0.5);
    \draw[blue] (0,0.177,-0.5) ++(0,0.177,-0.5) -- ++(0,0.177,-0.5);

    \draw[blue,dotted] (O) -- ++(0.433,0.177,0.25);
    \draw[blue] (0.433,0.177,0.25) -- ++(0.433,0.177,0.25);

    \draw[blue,dotted] (O) -- ++(-0.433,0.177,0.25);
    \draw[blue] (-0.433,0.177,0.25) -- ++(-0.433,0.177,0.25);

    \draw[blue] (O) node{$\bullet$};
    \draw[red] (A) node{$\bullet$};
    \draw[red] (B) node{$\bullet$};
    \draw[red] (C) node{$\bullet$};
    \draw[red] (D) node{$\bullet$};
  \end{tikzpicture}

  \caption{The coarse-graining procedure suggest a new structure of \textit{dual spin network}. The idea is that instead of only having edges \textit{transversal} to the surface of interests (in blue on the figure), we will also need the edges carrying \textit{tangential} holonomies (in red on the figure). Reported to the whole graph, we should carry the usual excitations \textit{plus} the dual graph.}
  \caption{}
  \label{fig:DualSpinNetwork}
\end{figure}

Let us note also the importance of the Immirzi parameter in this construction and in particular, the role of a complex Immirzi parameter. The relevance of complex Immirzi has already been stated in general in \ac{LQG} \cite{Achour:2015xga, Achour:2013gga} but our work suggest we could even consider non self-dual Ashtekar-Barbero connections and get interesting results. And more importantly, to have a non-trivial $\mathrm{SB}(2,\mathbb{C})$, having a complex Immirzi is mandatory. Finally, the Immirzi parameter plays here a peculiar role which might be linked to coarse-graining. For instance, we might wonder how two descriptions, using two different Immirzi parameters, might be linked. We want to suggest that this may be linked to different descriptions at different energy scales. In the chapter, we will pause on this, and see how it might be made a bit more concrete a proposal.

\medskip

In this chapter, we continued the work done on spinning geometry, generalizing it to the context of hyperbolically curved manifolds. We found interesting ways to define closure constraints for hyperbolically curved geometries. These constraints have the very nice property of giving a geometrical interpretation to quantum deformed hyperbolic geometries as the transformation laws all fit the scheme. We also reflected on the role of this construction with respect to coarse-graining suggesting in particular that a new structure of \textit{double spin network} might be hinted at here. Finally, we discussed the role of the Immirzi parameter and the role of its complexity.

%\textbf{TODO:} discussion: quite natural in coarse-graining since the connection is Ashtekar-Barbero but with complex parameter, might also be linked to deformed connection (definitionn open in 3+1d, but may be done by Thiemann's trick or Duflo map \textit{a priori}), this might lead us to the correct algebraic data to keep for coarse-graining

%*****************************************
%*****************************************
%*****************************************
%*****************************************
%*****************************************

%*****************************************
\chapter{The renormalization of the Immirzi parameter} \label{ch:Immirzi}
%*****************************************

\inspiquote{Do what I do. Hold tight and pretend it’s a plan!}{The Doctor}

In the previous chapter, we highlighted the possible role of the Immirzi parameter in the coarse-graining of the theory. This possible role is only suggestive at this stage and we want to explore it more and elaborate on that. We showed in the previous chapter that the Immirzi parameter (or a parameter playing a similar role) may have profound implications. For example, it appears that the compactness of $\mathrm{SU}(2)$ forbids the full reconstruction of tetrahedra in the flat case if they are to be constructed with deformed holonomies. In the coarse-graining context (which might take advantage of such deformed holonomies), this is a problem. But even more so, the same kind of phenomena appears even with the usual Ashtekar-Babebero connection and we wish to comment on this fact and consider how this might impact the coarse-graining scheme.

In essence, what we aspire to show is that the choice of variables for \ac{LQG}, because the gauge group is compact, might be important. This is of course true for \textit{euclidean} quantum gravity which involves compact gauge group even for self-dual variables. But because of the usual choice of the Ashtekar-Barbero connection, this also carries to the \textit{lorentzian} case. And we will show, through a very simple calculation, that a given holonomy cannot resolve all possible curvatures. This is a not a problem \textit{per se}, since it only means that for high curvature several holonomies are needed. It seems however to have implications with regard to two points. First, this upper bound on the curvature might be what makes the theory finite. Because the Immirzi parameter seems sometimes \textit{ad hoc}, this might triggers discussions on the physicality of all this. Of course, it is also totally possible to believe that it is precisely the Immirzi parameter which saves the day. Second, this will be a problem when coarse-graining as we want to consider larger and larger holonomies. Indeed, if this is done naively, this might lead to a faulty theory.

We will suggest in this chapter that renormalization might therefore entail a \textit{renormalization} of the Immirzi parameter. We will discuss more precisely how this idea comes around. But for now, it may be plausible from a quantum field theory perspective (after all the Immirzi parameter is a parameter of the theory) as considered in \cite{Benedetti:2011yb}. The problem however in \ac{LQG} is that the Immirzi parameter appears in the definition of the connection and therefore affects the gauge structure. In particular, the spectrum of the geometrical operators seems to depend on the Immirzi parameter making the definition of the renormalization of the Immirzi parameter tricky to say the least. Therefore, the argument made in this chapter is only suggestive, taking inspiration in techniques from \ac{LQC}. But if the programme is supposed to work, this means we would have to find yet another representation\graffito{Indeed, after the recent work of Dittrich \textit{et al.} \cite{Dittrich:2014wpa,Dittrich:2014wda,Bahr:2015bra}, and as was already emphasized \cite{Dittrich:2007th}, it becomes even more clear that the spectrum of the operators is highly dependant on the representation of the theory.} of quantum geometry.

%In this chapter, we will stand a bit astray of the full thesis plan. Indeed, in the last two chapters, we consider the definition of normals for hyperbolic tetrahedra. This led us to consider connections on the hyperboloid or even on the plane that looked like the Ashtekar-Barbero connection. In particular, we saw that some periodicity, due to the compactness of $\mathrm{SU}(2)$, appears in the normals and can lead to problem for reconstruction. 

% Note however that though this problems appear in 3+1D with the Immirzi parameter, it is not confined to this case. In particular, the very same problem can appear in self-dual \textit{euclidean} gravity which has a compact gauge group also. 

%   Naively, we might then expect some kind of \textit{renormalization} of the Immirzi parameter. It may be plausible from a  This is also difficult to make precise as the Immirzi parameter appears in the definition of the gauge structure. This is why we will consider how this problem is solved in \ac{LQC} (as was advertised in chapter \ref{ch:LQC}). In particular, we will see how the problem can be solved in toy models of \ac{LQG} which are designed to reproduce \ac{LQC}.

This chapter is inspired from work in \cite{Charles:2015rda}. It is organized as follows: in a first section, we will restate the results of \cite{Charles:2015rda} and comment on them, in particular on their \textit{would-be} consequences on the renormalization flow. In a second section, we will develop a very simple model of quantum cosmology. The goal is to see how the problem of coarse-graining is solved in some \ac{LQC} fashion. We will explore this precisely in the third section. In a fourth and final section, we will expand our main idea and how it might be encoded in a full coarse-graining programme.

%*****************************************

\section{The Immirzi parameter as a cut-off}

In \ac{LQG}, the Immirzi parameter $\beta$ plays a crucial role. It can be seen as a new coupling constant entering the Palatini action for general relativity in front of an almost-topological term \cite{Holst1996}. But, at a deeper level, it implements a canonical transformation from the original complex self-dual Ashtekar connection, which we will call $\mathcal{A}$, and the real Ashtekar-Barbero $\mathfrak{su}(2)$-connection $A$ \cite{Rovelli:1997na}. This allowed both to work with a compact gauge group $\mathrm{SU}(2)$ (instead of the non-compact Lorentz group $\mathrm{SL}(2,\mathbb{C})$) and to avoid the issue of the reality conditions. Indeed, as the Ashtekar connection $\mathcal{A}$ is complex and is thus not equal to its complex conjugate, it cannot be simply quantized as a multiplicative operator if the scalar  product is simply defined by the Gaussian measure. One needs to modify in a non-trivial way either the scalar product or the action of the connection operator (see e.g. \cite{Soo:2001qf}).

At a more effective level, the Immirzi parameter enters the \ac{LQG} dynamics in a non-trivial way. It also appears to control the couplings to fermionic fields and possible quantum gravity induced CP violation \cite{Perez:2005pm, Freidel:2005sn, Mercuri:2006um,Mercuri:2006wb}. The main drawback of the Immirzi parameter is that the fact that the Ashtekar-Barbero connection is not a space-time connection anymore and resulting in an apparent loss of covariance \cite{Alexandrov2002a} (see also the more recent \cite{Geiller:2012dd,Achour:2013gga}). Nevertheless, this does not cause any problem in practice as the kinematical operators are well defined. Therefore, one can perfectly define the kinematical Hilbert space of the theory and transition amplitudes between spin network states either by a canonical Hamiltonian \cite{Thiemann:1996aw,Bonzom:2011jv} or by a spinfoam path integral amplitude \cite{Engle:2007wy}. It can however be tempting to go back to the original complex formulation, given by the specific imaginary choice of Immirzi parameter $\beta =\pm i$, and attempt to define an analytic continuation of the real formulation of \ac{LQG} \cite{Achour:2015xga}.

The two important points that we would like to underline in this section are:
\begin{itemize}
\item The connection $A$ is not a space-time connection, except in the special case of the (anti-) self-dual Ashtekar connection ${\cal A}$ for the purely imaginary choice $\beta=\pm i$. It depends on the space-time embedding of the canonical hypersurface $\Sigma$. Considering a Wilson loop $\gamma$, its value will change if we embed it in different canonical space-like hypersurfaces $\Sigma$.

\item The  Ashtekar-Barbero connection, for real $\beta$, is in some sense a projection of the non-compact Lorentz connection into the compact $\mathrm{SU}(2)$ group. We lose some information, due to the periodicity in the extrinsic curvature. At the classical level, different extrinsic curvatures will still lead to the same value of the Wilson loop. This appears to impose a cut-off on the possible excitations of the geometry, more precisely on the extrinsic curvature, i.e. on the speed/momentum of the 3d intrinsic geometry.
\end{itemize}

We will illustrate these two points with the example of a closed loop embedded in space-like hyperboloids with variable curvature within the flat 4d space-time. We will discuss the dependence of the Wilson loop on the curvature of the hyperboloid, to show both how the Ashtekar-Barbero connection depends on the space-time embedding and how we can recover the extrinsic curvature from the value of the holonomy.

\begin{figure}[h!]
  \centering
  \includegraphics[scale=0.4]{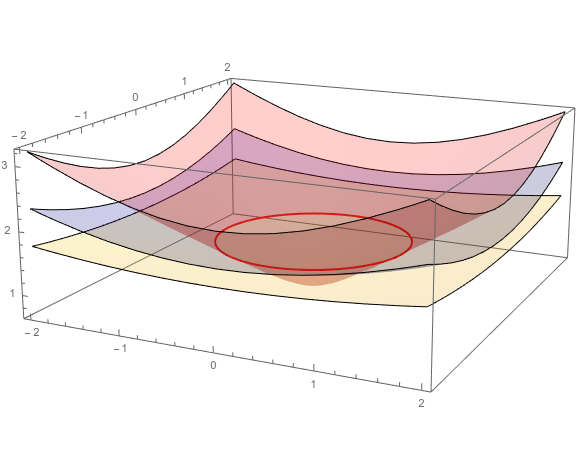}
  \caption{\label{fig:plot3d} This shows several hyperboloids of different curvature all containing the same loop in flat spacetime. The curvature of the embedding hyperboloid affects however the curvature of the Ashtekar-Barbero connection.}
\end{figure}

The hyperboloid has been defined in the previous chapter. We can consider a slight generalization in order to have several different hyperboloids of different curvature. Let us start with  the flat 3+1d Minkowski space-time with signature (-+++) and consider the upper sheet of the space-like hyperboloid,
\begin{equation}
-(t-t_{0})^{2}+(x^{2}+y^{2}+z^{2})=-\kappa^{2},\quad
t\ge t_{0}\,,
\end{equation}
with an arbitrary curvature radius $\kappa>0$ and a possible time shift $t_{0}\in\mathbb{R}$. We would like to look at the Ashtekar-Barbero holonomy around a loop of radius $R$, say
\begin{equation}
\gamma\equiv\{t=T,\,x^{2}+y^{2}+z^{2}=R^{2}\}\,,
\end{equation}
where $T$ and $R$ are arbitrarily fixed.
As illustrated on fig.{\ref{fig:plot3d}}, we embed this loop in the whole family of hyperboloid of arbitrary curvature radius $\kappa$ by adjusting their time shift in terms of $\kappa$,
\begin{equation}
t_{0}=T-\sqrt{R^{2}+\kappa^{2}}\,.
\end{equation}
This setting is very similar to \cite{Samuel:2000ue}, but we extend that calculation explicitly to arbitrary curvature $\kappa$.

\begin{figure}[h]
  \centering
  \includegraphics[scale=0.43]{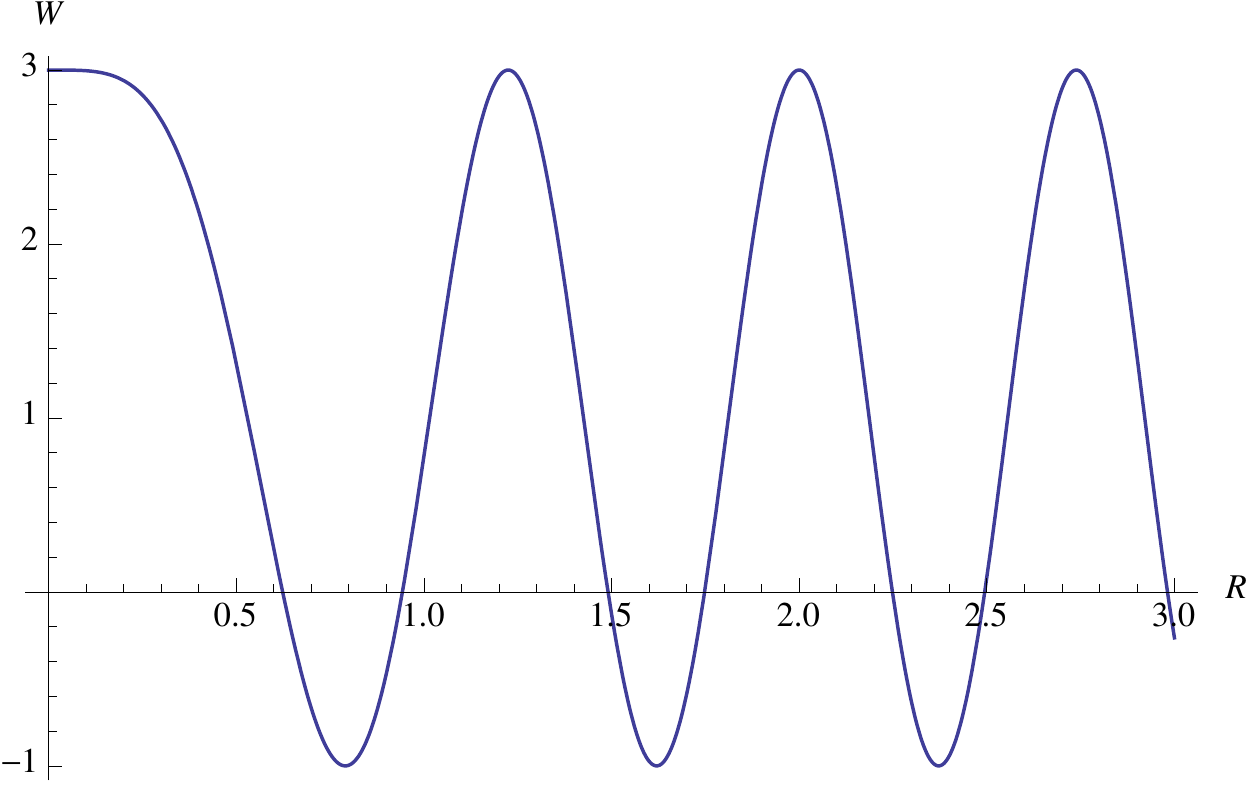}\quad
  \includegraphics[scale=0.43]{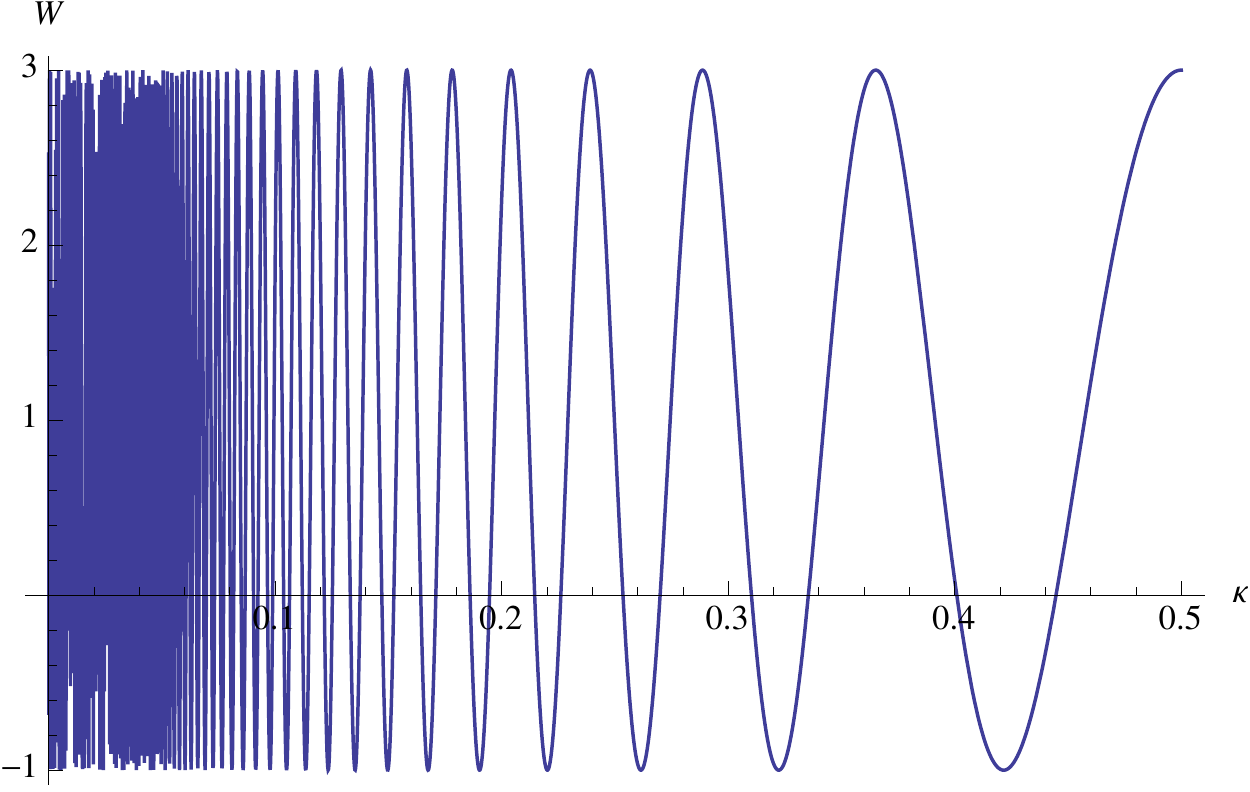}
  \caption{\label{fig:W}The Wilson loop $W$ plotted in terms of the loop size $R$ (in units of $\kappa$) in the upper graph, and in terms of the curvature radius $\kappa$ (in units of $R$) in the lower graph, both for a Immirzi parameter set to $\beta=1$.}
\end{figure}

All the computations can now be done. We refer the reader interested in the details to the corresponding paper \cite{Charles:2015rda}. Let us concentrate here on the result. For the holonomy around the loop $\gamma$, we obtain for the spin-1 Wilson loop (for the 3-dimensional representation, where the holonomy is represented as a $\mathrm{SO}(3)$ group element):
\begin{equation}
W_{\kappa}(R)=1+2\cos\left(
2\pi\sqrt{1+(1+\beta^{2})\frac{R^{2}}{\kappa^{2}}}
\right)
\end{equation}
We see a clear dependence of the size of the loop in units of the curvature radius of the hyperboloid, as illustrated on the plots in fig.\ref{fig:W}.

This term further depends on the Immirzi parameter $\beta$. For $\beta^{2}=-1$, this extra term vanishes and we recover $W_{\kappa}(R)=3$, which signals a flat connection. This is indeed the case for the complex (anti-) self dual Ashtekar connection, which is a space-time connection and sees that the initial space-time here is flat. However in general, the Ashtekar-Barbero connection is not flat, even though the space-time is flat, and contains information about the curvature $\kappa$ of the hyperboloid.

This means that extrinsic curvature information should be reconstructed from the connection. Indeed, one can invert the relation above and obtain the dimensionless ratio $\kappa/R$ from $W$. However, this ratio cannot be fully determined. Indeed, there is some periodicity due to the compactness of $\mathrm{SU}(2)$. In general, we only have:
\begin{equation}
\frac{\kappa^{2}}{R^{2}} = \frac{1+\beta^{2}}{(\varphi+k)^{2}-1},\,\,k\in\mathbb{Z}
\end{equation}
with the angle $\varphi$ given in terms of the Wilson loop by:
\begin{equation}
2\pi\varphi = \cos^{-1}\left(\frac{W-1}{2}\right)\,\,\in[0,\pi]\,.
\end{equation}
Therefore, the curvature is not uniquely fixed but determined up to a period $k\in\mathbb{Z}$.

The periodicity implies an ambiguity in the determination of the curvature from the Wilson loop. One could  decide to take the lowest value of the curvature, i.e the highest value of the curvature radius, typically given by the natural choice $k=1$. But this would mean obviously neglecting the possibility of higher curvature fluctuations. In this sense, we see that fixing a real Immirzi parameter leads to a cut-off in curvature in the context of \ac{LQG}. This highlights our major difficulty: by approximating the spin network by coarse-grained version of it, we will have access only to small curvature information, due to the compactness of $\mathrm{SU}(2)$. Note here, that even if the problem could be avoided with self-dual variables for instance, it still appears in euclidean quantum gravity or with Yang-Mills theory. A quantum theory of gravity coupled to matter fields will still have this problem if they are Yang-Mills interactions. This problem is the one of the major focus of the present chapter. And in what follows, we will consider how it is solved in \ac{LQC} (which is arguably coarse-grained) and how this might suggest ways to solve it.

%There is a simple way out as suggested by our paper \cite{Charles:2015rda}: it is to renormalize the Immirzi parameter (or an Immirzi-like parameter), though a precise implementation of it is still missing. This however, will not cover the full scale of the problem as presented here if no such parameter is available. It is then particularly enlightning to consider coarse-grained model in \ac{LQG}. \ac{LQC} in particular should have exactly the same problem. But as can be seen in the theory, the Immirzi parameter is still present and is not set to a particular value, which would be a fixed point of the renormalization process. However, the problem seems to be lurking in with respect to the difference between the $\mu_0$ and the $\overline{\mu}$ scheme. And this is, what we will investigate in the remaining of this chapter.

%*****************************************

\section{The single loop model of cosmology}

The problem that we underlined in the previous section can be restated as follows: capturing curvature with large loops (with respect to the curvature radius) is not possible because of the cut-off imposed by the Immirzi parameter. This is actually why \ac{LQC} needs the $\overline{\mu}$ scheme which implements loops of a given size and not a loop growing with the universe. How can we solve this problem in a coarse-grained setting? We will consider a model in this section, using technology from \ac{LQG}, whose goal is to encode cosmological dynamics which will allow a comparison with \ac{LQC}. The model is inspired from the $\mathrm{U}(N)$ \cite{Borja:2010rc,Borja:2010gn,Livine2013} model but stripped out of all the non-necessary ingredients for our purpose. It is in some sense justified only \textit{a posteriori} after the $\mathrm{U}(N)$ can be devised as a simple model reproducing the same kind of dynamics. We refer the reader interested in the $\mathrm{U}(N)$ model to the corresponding appendix \ref{app:spinors}.

The $\mathrm{U}(N)$ model is a simple model aiming at describing homogeneous and isotropic universes using a fixed graph. The graph has only two vertices linked together by a fixed number of edges. To describe a homogeneous and isotropic universe, this has still too many degrees of freedom. It can be reduced by using the $\mathrm{U}(N)$ action on the vertices. Indeed, each vertex can be endowed with a collection of observables which satisfy a $\mathrm{U}(N)$ algebra. Enforcing that the equality between these observables on the two vertices can be understood as a homogeneity and isotropy condition. Because the constraints are first class, we can study the remaining degrees of freedom by simplectic reduction. The observables have the following action: they carry quanta of surfaces from one edge to the other, conserving the total area separating the two vertices. The simplectic reduction leads then to two degrees of freedom, which are conjugated: the total area (which is conserved under transformation) and the angle of the transformation on each link (which is the same on each link between the two vertices when the constraints are imposed).

These two degrees of freedom correspond more or less to the degrees of freedom we are interested in cosmology: the scale factor and its conjugate. They have even a nice correspondence with the variables from \ac{LQC}. Indeed as was presented in the chapter \ref{ch:LQC}, the symmetry reduction of \ac{LQC} using the variables of \ac{LQG} leads to two variables $p$ and $c$. $p$ is the squared scale factor and $c$ is its conjugate which is linked to (extrinsic \textit{and} intrinsic) curvature. These are precisely the variables which are uncovered in the $\mathrm{U}(N)$ symmetric model.

We will not consider the full $\mathrm{U}(N)$ model since it clutters the discussion. We will therefore concentrate on a much simpler model which is sufficient to underline the problem. We simply called it the \textit{single loop model}, as it seems quite fit. Let us consider a single vertex and a single link starting and ending on it, making it a loop. On this graph, we have two variables classically which are conjugated : the area of the transverse surface to the loop and the angle of the rotation along the loop. They are the exact equivalent (up to factors) of the two variables of the $\mathrm{U}(N)$ model and do correspond to the variables of \ac{LQC}. This can be described quantum mechanically by wave functions over $\mathrm{SU}(2)$ which are invariant under conjugation, that is functions of the form:
\begin{equation}
\psi : \mathrm{SU}(2) \rightarrow \mathbb{C}
\end{equation}
such as:
\begin{equation}
\forall h,g \in \mathrm{SU}(2),\ \psi(ghg^{-1}) = \psi(h)
\end{equation}
A natural basis for this space of states is given by the characters of $\mathrm{SU}(2)$ which we will write $\chi_j$ where $j$ is the half-integer labelling representations. Therefore, the state can be written:
\begin{equation}
\psi(g) = \sum_{j \in \frac{\mathbb{N}}{2}} \psi_j \chi_j(g)
\end{equation}
There are natural operators on this space that will nicely do the job as observables for curvature and area. There are the holonomy in the fundamental representation and the casimir. We write them as:
\begin{equation}
  \left\{\begin{array}{rcl}
  \hat{\chi} \psi(g) &=& \chi_{\frac{1}{2}}(g) \psi(g) \\
  \hat{C} \psi (g) &=& \sum_{j \in \frac{\mathbb{N}}{2}} j(j+1) \psi_j \chi_j(g)
  \end{array}\right.
\end{equation}
Note that these are not precisely the operators for $p$ and $c$. Apart from ordering ambiguities, $\hat{p}$ is simply linked to the casimir by $\hat{p}^2 = \hat{C}$. But for $c$, there is no well-defined operator. This is precisely the point, as we only want exponentials of the operator to be defined. It is therefore appropriate to think of $\hat{\chi}$ as $\cos c$ (the even part of the exponential).

Now, let's forget about the numerical factors. As we saw in chapter \ref{ch:LQC}, the classical Hamiltonian for cosmology is (expressed in $p$ and $c$):
\begin{equation}
  H = (pc)^2
\end{equation}
Because $c$ does not exist as an operator, we will need to change the Hamiltonian slightly. The most basic thing to do is to consider:
\begin{equation}
  \tilde{H} = 2 p^2 \left(\cos c - 1\right)
\end{equation}
This corresponds more or less to the $\mu_0$ scheme, which we sadly know to be false. But it has the advantage of being writable in the model we just exposed. Indeed, this is precisely why we are doing all this: the connection does not exist as an operator and we are led to use holonomy operators. The problem is, we only have access to functions of $c$ of the form $\cos nc$ with $n$ an integer. They correspond to the different representations of the holonomy or, alternatively, to powers of the fundamental one. This means that, at least at first sight, we cannot represent correctly the  $\overline{\mu}$ scheme. Indeed, as we saw, we need more than just the integer exponentials of the connection but all exponentials for all real values of a prefactor. And this cannot be implemented on our compact configuration space (namely $\mathrm{SU}(2)$).

In a more \ac{LQC} like language, everything is behaving as if we were considering $\mathrm{U}(1)$ as a configuration space rather than $\mathbb{R}_\textrm{Bohr}$. Most operators are not well-defined in this context and more importantly, the conjugate to the volume is not well-defined which is however needed for the $\overline{\mu}$ scheme. Indeed, usually the geometrical interpretation goes as follows: when the scale factor goes larger, the loops, relative to the size of the universe, actually scales down. Therefore, the connection, which is expressed in the comoving coordinates should scale down accordingly. For the fundamental theory, this means the dynamics should be graph changing and rather than a fixed graph with a compact configuration space on each link. This is precisely the problem we face with coarse-graining and renormalization, here in the redux. This is also the equivalent of the problem mentioned in the previous section and is manifest as a problem of periodicity. Indeed, all our observables are periodic (due to the compactness of $\mathrm{SU}(2)$) but the observables we would like for quantum cosmology are non-periodic. This is the natural problem we have to face.

%*****************************************

\section{The $\overline{\mu}$ scheme in a coarse-grained fashion}

Can we bypass this problem and find a way to write down equivalents of the $\overline{\mu}$ scheme operators? Let us start by considering a simpler problem: let's consider a scalar field in a one dimensional space. It can be modeled by a function from $\mathbb{R}$ and valued in $\mathbb{C}$ (if we consider a complex scalar field). A coarse-graining of it could be represented by a function from $\mathbb{Z}$ into $\mathbb{C}$. The relation between the discretization and the continuum theory is left open here. But in principle, something similar to the perfect discretization process \cite{Bahr:2009qc} should give an ideal description of the low energy states of the \textit{continuum} theory. Still, there is a similar issue to the one we encountered above: because of the discretization, the Fourier space becomes compact and it is now impossible to describe the non-periodic dispersion relations on the Brillouin zone. The problem of compactification appears therefore once more but here in the Fourier space.

Technically, let's note here that the problem is a bit different as this is not a problem of compactness of the \textit{configuration} space. Mathematically however, it is quite similar and it will be quite illustrative. How can we resolve the problem? There are several possibilities.
%Still, it can be recast in this language in the following manner: let's rather consider one discrete theory (with fields from $\mathbb{Z}$ into $\mathbb{C}$) and a second one but with a length between points doubled (so in a sense with fields from $2\mathbb{Z}$ into $\mathbb{C}$ though of course this changes nothing mathematically speaking). We can consider the coarse-graining problem of going from one scale to the other. Incidently, the fixed points should represent the continuum theory (or at least the low energy states of it).
Consider the dispersion relation of a (relativistic) free scalar field:
\begin{equation}
E^2 = p^2 + m^2
\end{equation}
with the usual conventions of $c=\hbar=1$ and where $E$ is the energy and $p$ the momentum. It can also be represented as on figure \ref{fig:usual_dispersion}. Now, back to our discretization: because it is only spatial, energy is not compactified, but momentum will be. This is a problem however as the dispersion relation is clearly not periodic. Indeed, on a lattice, only periodic function of $p$ will exist, that is only those which depend on $p$ up to $\frac{2\pi}{a}$. This is the equivalent in this simpler setting of our compactification problem. A solution though is to simply concentrate on low energy excitations as illustrated on figure \ref{fig:dispersion_low}. Indeed, if we consider momentum that are between $-\frac{\pi}{a}$ and $\frac{\pi}{a}$ they can be described both by the original dispersion relation or by a periodic function. But this requires to concentrate on low energy excitations. It could in principle be possible to extend to higher energies by defining energy \textit{bands} but the practicality does not seem straight forward (see figure \ref{fig:dispersion_levels}).

\begin{figure}[h!]
  \centering

  \begin{subfigure}[t]{.45\linewidth}
    \centering
    
    \begin{tikzpicture}[scale=0.5]

      \begin{axis}[
          axis lines*=middle,
          xmin=-3,
          xmax=3,
          ymin=0,
          ymax=5,
          xlabel=$p$,
          ylabel=$E$,
          samples=50
        ]
        \addplot[blue, ultra thick] (x,{sqrt(x*x+1)});
      \end{axis}
      
    \end{tikzpicture}
    
    \caption{In the usual continuum theory, the dispersion relation is the well-known one: quadratic near $p \simeq 0$ and linear at infinity.}
    \label{fig:usual_dispersion}
    
  \end{subfigure} \hfill
  \begin{subfigure}[t]{.45\linewidth}
    \centering
    
    \begin{tikzpicture}[scale=0.5]

      \begin{axis}[
          axis lines*=middle,
          xmin=-3,
          xmax=3,
          ymin=0,
          ymax=5,
          xlabel=$p$,
          ylabel=$E$,
          samples=50
        ]
        \addplot[red, ultra thick,domain=-2:2] (x,{sqrt(x*x+1)});
        \addplot[black, dashed] (-2,x);
        \addplot[black, dashed] (2,x);
      \end{axis}
      
    \end{tikzpicture}
    
    \caption{A truncation of the theory can be defined at low energies. Incidentally, this theory can be defined on a discrete space. Here $\frac{\pi}{a} = 2$ (arbitrary units) where $a$ is the lattice spacing.}
    \label{fig:dispersion_low}
    
  \end{subfigure} \\
  \begin{subfigure}[t]{.9\linewidth}
    \centering
    
    \begin{tikzpicture}[scale=0.5]

      \begin{axis}[
          axis lines*=middle,
          xmin=-3,
          xmax=3,
          ymin=0,
          ymax=5,
          xlabel=$p$,
          ylabel=$E$,
          samples=50
        ]
        \addplot[red, ultra thick,domain=-2:2] (x,{sqrt(x*x+1)});
        \addplot[red, ultra thick, dashed, domain=-2:2] (x,{sqrt((x-4)*(x-4)+1)});
        \addplot[red, ultra thick, dashed, domain=-2:2] (x,{sqrt((x+4)*(x+4)+1)});
        \addplot[black, dashed] (-2,x);
        \addplot[black, dashed] (2,x);
      \end{axis}
      
    \end{tikzpicture}
    
    \caption{It is possible in principle to describe any excitation on a discrete space as long as the larger momentum are considered as energy levels, creating \textit{bands} in a condensed matter fashion. The lattice spacing is $a = \frac{\pi}{2}$ (arbitrary units).}
    \label{fig:dispersion_levels}
  \end{subfigure}

  \caption{Representation of the dispersion relation $E=f(p)$, the energy as a function of momentum, for a massive scalar field in different frameworks with a mass $m=1$ (arbitrary units).}
  
\end{figure}

Can we do the same thing for \ac{LQG}? Classically, it is possible, at least in the context of \ac{LQC}. We can start with the $\overline{\mu}$ Hamiltonian and expand it in terms of sine and cosine of $c$ as shown here:
\begin{equation}
  \forall c \in ]-\pi;\pi[,\ \cos \frac{\ell_p c}{\sqrt{p}} = \frac{\sqrt{p}}{\pi \ell_p} \sin \frac{\ell_p \pi}{\sqrt{p}} + \sum_{n>0} \frac{(-1)^n \ell_p \sqrt{p}}{p n^2 - \ell_p^2} \cos nc
\end{equation}
and we would use:
\begin{equation}
\overline{H} = \frac{2p^3}{\ell_p^2}\left(\cos \frac{\ell_p c}{\sqrt{p}} - 1 \right)
\end{equation}
The low energy limit is here replaced by a low curvature limit: as long as $c$ is small (that is $|c| < \pi$), the expansion will be exact. We are quite lucky here as this is one of the property of the $\overline{\mu}$ scheme: it tends to keep $c$ small. For larger values though, the expansion is not directly usable, apart maybe from analytical continuations. It might seems a way out but we should keep in mind at least two things here. First, the expansion is highly singular and the proper definition of corresponding operators might be tricky. Second, such an expansion is available for \textit{any} function restricted to a compact interval which includes the standard Hamiltonian $H = p^2 c^2$. This means that using such expansion might harm the singularity resolution of the theory.

Though, we have not explored this the other way around, there are also troubles linked to the spectrum of the area operator which is quantized and linked to the Immirzi parameter. All this reveal the real difficulty, that we have to work with varying graphs. Indeed, the problem of how varying graphs enable the right dynamics has been recently in the context of \ac{GFT} \cite{Gielen:2013naa}. This is expressed in the formalism of condensed state where the universe consists of many tetrahedron all in the same state\graffito{Note that these condensed states are not gauge-invariant and this is therefore an approximation.}. The fact that the number of tetrahedron is variable allows the right dynamics of \ac{LQC} to be found \cite{Oriti:2016ueo}. We will develop this line of thought in the next part of this thesis.

But as we suggested in \cite{Charles:2015rda}, because the Immirzi parameter acts as the cut-off scale, we might try and see if it is possible to renormalize the Immirzi parameter itself. This is what we will consider and develop in the following section.

\section{Loop quantum gravity for all Immirzis}

The rationale behind the renormalization of the Immirzi parameter comes from the following question: in a well-defined quantum theory (without infinities) to what does the renormalization process correspond? An interesting point of view is that it can be understood as a unitary transform relating different scales. In usual \ac{QFT}, this is difficult to see because of all the infinities. In condensed matter, which usually comes with a cut-off, such a unitary transform is not possible because it would change the cut-off scale. But, what we should expect in a well-defined theory of \textit{continuous} quantum gravity? We might indeed have a unitary transform corresponding to scales transformation.

As we will see, the Immirzi canonical transform is a perfectly good candidate. Indeed, classically it is generated by:
\begin{equation}
\mathcal{C} = \int K^I_a E^a_I d^3 x
\end{equation}
This acts on the densitized triad (and therefore on all spatial information) as a scale transform. This is precisely the behavior we would like for the equivalent of the renormalization procedure. \graffito{The very similar form with the (volume preserving) conformal constraint of shape dynamics (see \cite{Mercati:2014ama} for a review) also validates the interpretation as a scale transform.}Therefore the Immirzi parameter has a role even more special than previously thought: it turns out precisely to be precisely the parameter selecting the scale of the theory.

Of course, this is not as straightforward in practice. First of all, this Immirzi canonical transformation is not represented unitarily in the (current) quantum theory, as can be seen from the spectrum of the area operator. Let us note right away that this argument falls when considering self-dual variables. Indeed, they are arguments that indicates that the spectrum of the area operator might be continuous in this framework \cite{Alexandrov2001}. This would possibly allow scale transforms to be implemented. But of course nothing guarantees that. And it leaves the case of \textit{euclidean} quantum gravity open anyway. Another possibility might be the existence of other representations of quantum geometry. Indeed, one of the major problem comes from the non-existence of the connection operator. A representation insisting on exponentiated fluxes (rather than connection), maybe similar to the BF representation \cite{Dittrich:2014wpa}, might be able to represent the transformation.

Now let's admit that the problem is solved, that we can find a new representation implementing unitarily the Immirzi canonical transform. Or we could also imagine that the Immirzi parameter \textit{per se} is a topological parameter but that we find a natural way of implementing a scale transform at the quantum level. Then, what would the renormalization programme look like? It would be remarkably similar to the \ac{AS} programme of \ac{LQG}, with the infinities of \ac{QFT} removed. Indeed, the Immirzi canonical transform would generate the change in scale and could be applied (for instance) to the Hamiltonian. A general Hamiltonian could be written and its transformation under renormalization resolved. Of course, we would now have the same problems as in lattice Yang-Mills, since all couplings could in principle be written. We could hope that the ambiguities might be resolved \textit{via} the fulfilling of some spatiotemporal diffeomorphism algebra \cite{Perez:2005fn}. In the context of renormalization, this is where the asymptotic safety scenario is joined, since that would correspond to fixed points and there associated critical surface. So, in practice, we would look for Hamiltonians invariant under the renormalization flow, that is a fixed point of the renormalization flow. From there, in the general space of possible Hamiltonians, we would concentrate on the critical surface under the flow connected to the fixed point. The dimensionality of the critical surface would give the number of parameters necessary to describe quantum gravity.

But what is the link between this approach and the more general coarse-graining programme we started to study? The link would be similar to the one between standard renormalization in \ac{QFT} and the Wilson flow. One is done without cut-offs between scales, the other one explicitly use one. But the flow of both should be related. We should also note that if the Immirzi transform is generated by something like:
\begin{equation}
\mathcal{C} = \int K^I_a E^a_I d^3 x = \int (A^I_a - \Gamma^I_a(E)) E^a_I d^3 x 
\end{equation}
then we do not expect spin networks, or combination of them with a given support, to be eigenvectors of the transformation. Indeed, the presence of the connection in the generator shows that the transformation would be graph changing. So, in the case of continuum quantum gravity, it might very well be the case that the standard renormalization flow and the Wilson flow are actually the exact same thing. In practice, the unitary transform would be graph changing and therefore select a natural way to coarse-graining the theory.

We can of course hope that such a natural way would be found \textit{by hand}, that is without the guidance of some, properly implemented, scale transform. Such a coarse-graining might be obtainable by geometrical consideration, as we have done in the previous chapters. But, the argument suggest that we should now concentrate more on varying graphs techniques. We will therefore need a new structure, of course capable of describing large distances as coarse-graining would imply, but more specifically which would encode finer graphs and varying graphs. The importance of varying graphs comes therefore through different channel: it has been discussed in chapters \ref{ch:LQC} and \ref{ch:GaugeFixCG} but also appears from scale transformation consideration. It has not however been fully explored yet in this thesis and this what we will now do in the next part and the remaining chapters.

%Even more, we might be able to import some of the information from this Immirzi parameter study into the coarse-graining approach. Indeed, at least we have highlighted the importance of varying graphs. This point to the possible necessity of a new structure capable of describing large blocks of space but also carrying the representation of an algebra of \textit{small} loops capturing the curvature.
%Such new structures have already been hinted at in the two previous chapters, where we highlighted the possibility of describing large homogeneously curved blocks. 

\medskip

In this chapter, we show that the Ashtekar-Barbero connection is not a spacetime connection and we show that its holonomy cannot totally resolve the \textit{extrinsic} curvature of its embedding manifold. We explored, in the context of \ac{LQC}, how this is linked to the problem of varying graphs. But the specific role of the Immirzi parameter also suggested a new route towards renormalization, which we discussed. We will now turn to our last part of the dissertation, concentrating of the problem -- underlined in the current chapter -- of handling varying graphs in the full theory.

\part{Loopy spaces: toward the coarse-graining of Loop quantum gravity}
%*****************************************
\chapter{The space of loopy spin networks} \label{ch:Loopy}
%*****************************************

\inspiquote{You want weapons? We're in a library. Books are the best weapon in the world. This room's the greatest arsenal we could have.}{The Doctor}

\newcommand{\C}{{\mathbb C}}
\newcommand{\N}{{\mathbb N}}
\newcommand{\R}{{\mathbb R}}
\newcommand{\Z}{{\mathbb Z}}

\newcommand{\cA}{{\mathcal A}}
\newcommand{\cB}{{\mathcal B}}
\newcommand{\cE}{{\mathcal E}}
\newcommand{\cF}{{\mathcal F}}
\newcommand{\cG}{{\mathcal G}}
\newcommand{\cJ}{{\mathcal J}}
\newcommand{\cI}{{\mathcal I}}
\newcommand{\cK}{{\mathcal K}}
\newcommand{\cL}{{\mathcal L}}
\newcommand{\cH}{{\mathcal H}}
\newcommand{\cM}{{\mathcal M}}
\newcommand{\cN}{{\mathcal N}}
\newcommand{\cR}{{\mathcal R}}
\newcommand{\cO}{{\mathcal O}}
\newcommand{\cP}{{\mathcal P}}
\newcommand{\cT}{{\mathcal T}}
\newcommand{\cV}{{\mathcal V}}
\newcommand{\cD}{{\mathcal D}}
\newcommand{\cC}{{\mathcal C}}
\newcommand{\cS}{{\mathcal S}}
\newcommand{\cU}{{\mathcal U}}

\newcommand{\SU}{\mathrm{SU}}
\newcommand{\ISU}{\mathrm{ISU}}
\newcommand{\Spin}{\mathrm{Spin}}
\newcommand{\SL}{\mathrm{SL}}
\newcommand{\SO}{\mathrm{SO}}
\newcommand{\U}{\mathrm{U}}
\newcommand{\ISO}{\mathrm{ISO}}
\newcommand{\SB}{\mathrm{SB}}
\newcommand{\SH}{\mathrm{SH}}
\renewcommand{\H}{\mathrm{H}}

\newcommand{\be}{\begin{equation}}
\newcommand{\ee}{\end{equation}}
\newcommand{\beq}{\begin{eqnarray}}
\newcommand{\eeq}{\end{eqnarray}}
\newcommand{\bes}{\begin{eqnarray}}
\newcommand{\ees}{\end{eqnarray}}

\newcommand{\mat} [2] {\left ( \begin{array}{#1}#2\end{array} \right ) }

\newcommand{\su}{{\mathfrak{su}}}
\newcommand{\spin}{{\mathfrak{spin}}}
\renewcommand{\u}{{\mathfrak{u}}}
\renewcommand{\sl}{{\mathfrak{sl}}}
\newcommand{\so}{{\mathfrak{so}}}
\newcommand{\g}{{\mathfrak{g}}}

\newcommand{\la}{\langle}
\newcommand{\ra}{\rangle}

\newcommand{\sgn}{\mathrm{sgn}}
\newcommand{\tr}{{\mathrm{Tr}}}
\newcommand{\im}{{\mathrm{Im}}}
\newcommand{\f}{\frac}
\newcommand{\tl}{\widetilde}
\def\nn{\nonumber}
\def\pp{\partial}
\newcommand{\w}{\wedge}

\def\ka{\kappa}
\def\vphi{\varphi}
\def\eps{\epsilon}

\newcommand{\id}{{\mathbb{I}}}
\def\ka{\kappa}
\def\ot{\otimes}
\def\one{{\bf 1}}
\def\act{\triangleright}

\def\vN{\vec{N}}
\def\vC{\vec{C}}
\def\vx{\vec{x}}
\def\vJ{\vec{J}}
\def\bz{\bar{z}}
\def\vsigma{\vec{\sigma}}
\def\hu{\hat{u}}
\def\hv{\hat{v}}
\def\hN{\hat{N}}
\def\tM{\widetilde{M}}
\def\tell{\widetilde{\ell}}
\def\arr{\rightarrow}
\def\Om{\Omega}
\def\tH{\widetilde{H}}
\def\tDelta{\widetilde{\Delta}}
\def\tF{\widetilde{F}}
\def\tA{\widetilde{A}}
\def\cHl{{\cal H}^{\mathrm{loopy}}}
\def\cHs{{\cal H}^{\mathrm{sym}}}

%Theorems
\newtheorem{theorem}{Theorem}[section]
\newtheorem{lemma}[theorem]{Lemma}
\newtheorem{proof}[theorem]{Proof}
\newtheorem{prop}[theorem]{Proposition}
\newtheorem{conj}[theorem]{Conjecture}
\newtheorem{corollary}[theorem]{Corollary}
\newtheorem{definition}[theorem]{Definition}

\def\restriction#1#2{\mathchoice
  {\setbox1\hbox{${\displaystyle #1}_{\scriptstyle #2}$}
    \restrictionaux{#1}{#2}}
  {\setbox1\hbox{${\textstyle #1}_{\scriptstyle #2}$}
    \restrictionaux{#1}{#2}}
  {\setbox1\hbox{${\scriptstyle #1}_{\scriptscriptstyle #2}$}
    \restrictionaux{#1}{#2}}
  {\setbox1\hbox{${\scriptscriptstyle #1}_{\scriptscriptstyle #2}$}
    \restrictionaux{#1}{#2}}}
\def\restrictionaux#1#2{{#1\,\smash{\vrule height .8\ht1 depth .85\dp1}}_{\,#2}}

\newcommand*\circled[1]{\tikz[baseline=(char.base)]{
    \node[shape=circle,draw,inner sep=1pt] (char) {#1};}}

\def\tGamma{\widetilde{\Gamma}}
\def\tpsi{\widetilde{\psi}}
\def\tf{\widetilde{f}}
\def\tk{\tilde{k}}
\newcommand{\binomial} [2] {\left ( \begin{array}{c}#1 \\ #2\end{array} \right ) }
\def\hchi{\widehat{\chi}}
\def\vv{\vec{v}}
\def\dd{\mathrm{d}}
\def\tg{\tilde{g}}

In the last chapter, we got out of the main flow of this thesis to talk about the role of the Immirzi parameter. This got us back, surprisingly, to coarse-graining. If this previous chapter defines a more or less equivalent of the standard continuum renormalization procedure, we still have not defined the equivalent of the Wilson flow in \ac{LQG}. Even assuming that indeed the right variables for large scale descriptions are the surfaces holonomies and holonomies across the surface, we should now try and explore a way of extracting these degrees of freedom in a coarse-graining procedure. The idea of coarse-graining is to integrate out the microscopic degrees of freedom, by an iterative procedure, up to some given energy or length scale to get the effective dynamics of the macroscopic degrees of freedom. And to do this, we must implement, \textit{concretely}, what was presented in chapter \ref{ch:GaugeFixCG}, namely coarse-graining by gauge fixing. Therefore, we will now turn to the coarse-graining of spin networks for \ac{LQG}. 

Because of the discrete nature of spin networks, it is a good idea to search for inspiration in condensed matter models. In these models, one typically works on a regular lattice with degrees of freedom living on its edges and/or nodes and one can decimate consistently the variables, integrating out one node out of two for example, and thus derive an effective Hamiltonian on the coarser lattice. The length scale is set by the lattice spacing. In quantum field theory, the renormalisation group scheme integrates out quantum fluctuations of the field of high momentum and energy to derive an effective dynamics on the low momentum degrees of freedom. The idea is therefore always to separate scales and to define a way of integrating out some irrelevant scale. But, in \ac{GR}, the main difficulty is that the space-time geometry itself has become dynamical thus leading to some serious obstacles: in a background independent context, we face the problem of defining consistently a length or energy scale and of properly localizing perturbations and degrees of freedom both in position and momentum. Of course, these conceptual difficulties do not wash away when going to the quantum theory and are actually all the more difficult to handle.

We already suggested a way to generalize the renormalization process in the previous chapter using a varying Immirzi parameter. But even in quantum field theory, renormalization can be thought of in various manners and appears in different forms. As advertised, we want to consider a more \textit{Wilsonian} version of the renormalization flow. Such a flow is less speculative than our previous idea and is quite natural in the theory. Indeed, in \ac{LQG}, the natural graph structure of the theory makes it simpler to tackle the coarse-graining of the theory, and least conceptually. Putting aside the huge complication of fluctuating graphs and graph superpositions, in a coarse-graining process, the graph underlying the spin network will no longer represent \textit{fundamental} degrees of freedom, but the coarse degrees of freedom we are interested in. Therefore, a natural coarse-graining procedure on a fixed graph is to subdivide it into a partition of bounded (usually connected) regions and to collapse those subgraphs to single points. The internal geometrical information carried by the spin network state on those subgraphs would be coarse-grained to some effective data living at the new node of the coarser graph, as illustrated on fig.\ref{fig:coarsegraining}.  Integrating over these local degrees of freedom would lead to new effective dynamics on the coarser graph. Such a  procedure would then be iterated  to obtain a tower of effective theories \textit{\`a la} Wilson for \ac{LQG} towards a large scale limit.
\begin{figure}[h!]

  \centering

  \begin{tikzpicture}[scale=0.3]
    \coordinate(A) at (0,0);
    \coordinate(B) at (2,0);
    \coordinate(C) at (2,-2);
    \coordinate(D) at (0,-2);

    \coordinate(E) at (6,0);
    \coordinate(F) at (6,-2);
    \coordinate(G) at (8,0);
    \coordinate(H) at (0,-6);
    \coordinate(I) at (2,-6);
    \coordinate(J) at (2,-8);
    \coordinate(K) at (6,-6);

    \draw (A) -- (B);
    \draw (B) -- (C);
    \draw (C) -- (D);
    \draw (D) -- (A);
    \draw (B) -- (D);

    \draw (A) -- ++(-1,1);
    \draw (B) -- (E);
    \draw (C) -- (F);
    \draw (C) -- (I);
    \draw (D) -- (H);

    \draw (G) -- ++(1,1);
    \draw (E) -- (F) -- (G) -- (E);

    \draw (H) -- ++(-1,-1);
    \draw (J) -- ++(-1,-1);
    \draw (H) -- (I) -- (J) -- (H);

    \draw (I) -- (K);
    \draw (J) -- (K);
    \draw (K) -- (F);

    \draw (K) -- ++(1,-1);

    \draw (A) node {$\bullet$};
    \draw (B) node {$\bullet$};
    \draw (C) node {$\bullet$};
    \draw (D) node {$\bullet$};
    \draw (E) node {$\bullet$};
    \draw (F) node {$\bullet$};
    \draw (G) node {$\bullet$};
    \draw (H) node {$\bullet$};
    \draw (I) node {$\bullet$};
    \draw (J) node {$\bullet$};
    \draw (K) node {$\bullet$};

    \draw[gray,dashed] (1,-1) circle(2);
    \draw[gray,dashed] (1.3,-6.7) circle(1.7);
    \draw[gray,dashed] (6.7,-0.7) circle(1.7);
    \draw[gray,dashed] (K) circle(1);

    \draw[->,>=stealth,very thick] (9,-4) -- (13,-4);

    \coordinate(O1) at (15,-1);
    \coordinate(O2) at (15.3,-6.7);
    \coordinate(O3) at (20.7,-0.7);
    \coordinate(O4) at (20,-6);
    \draw (O1) -- ++(-1,1);
    \draw (O2) -- ++(-1.7,-0.5);
    \draw (O2) -- ++(-0.5,-1.7);
    \draw (O3) -- ++(1,1);
    \draw (O4) -- ++(1,-1);

    \draw (O1) to[bend left] (O2);
    \draw (O1) to[bend right] (O2);
    \draw (O1) to[bend left] (O3);
    \draw (O1) to[bend right] (O3);
    \draw (O2) to[bend left] (O4);
    \draw (O2) to[bend right] (O4);
    \draw (O4) -- (O3);

    \draw[fill=lightgray] (O1) circle(0.3);
    \draw[fill=lightgray] (O2) circle(0.2);
    \draw[fill=lightgray] (O3) circle(0.2);
    \draw[fill=lightgray] (O4) circle(0.1);
  \end{tikzpicture}

  \caption{\label{fig:coarsegraining}
    We coarse-grain a graph by partitioning it into disjoint connected subgraphs. We will reduce each of these bounded region of space by a single vertex of the coarser graph. Since each of these regions of space had some internal geometrical structure and were likely carrying curvature, the natural question is whether spin network vertices carry each data to account for these internal structure and curvature. We will see that standard spin network vertices can be interpreted as flat and that we need to introduce some new notion of ``curved vertices'' carrying extra algebraic information and define new extensions of spin network states more suitable to the process of coarse-graining loop quantum gravity.}

\end{figure}
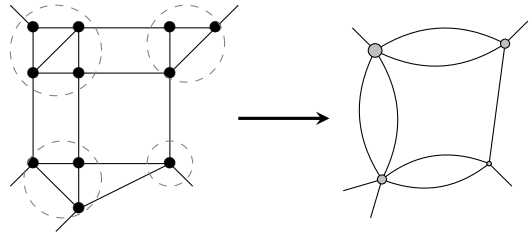

We propose a truncation of the theory based on this idea. The new effective graph can be understood as a background lattice over which various excitations at the vertices are possible and correspond to the different possible spin networks for the covered region of each effective vertex. This might seem to break diffeomorphism invariance, but the choice of a background graph does not have to be arbitrary. For instance, a given observer might choose some geometry to probe and a support graph accordingly. Therefore, the lattice is not considered as the fundamental graph underlying the physical spin network state. Instead, since the observer is assumed to have a finite resolution, its nodes represent bounded regions of space whose internal geometry can fluctuate. Then, if we consider a spin network states based on a graph with a very fine structure, we will coarse-grain it onto our chosen lattice. Such a scheme has the additional benefit of allowing the taking into account of graph fluctuations and superpositions while actually working on a fixed lattice. Indeed, considering a superposition of graphs, it will live, by cylindrical consistency, on a finer graph containing both graphs. Then we will coarse-grain the quantum geometry state on the finer graph until it lives on our reference lattice.

A key step of this procedure is the coarse-graining of subgraphs to nodes. We use the ``coarse-graining through gauge-fixing'' procedure introduced in \cite{Livine:2006xk, Livine:2013gna} and also exploited in \cite{Dittrich:2014wpa,Bahr:2015bra} to reformulate the algebra of geometrical observables in \ac{LQG}. This is based on the  gauge-fixing  for spin networks defined earlier in \cite{Freidel2003}, which allows to collapse an arbitrary subgraph to a \textit{flower} , that is a single vertex with self-loops -or petals- attached to it. These loops account  for the building-up of the curvature and thus of the gravitational energy density within these microscopic bounded regions which we will coarse-grain to single points on the measurement lattice chosen by the observer. This chapter and the next one are both inspired from our published work \cite{Charles:2016xwc}.

\section{Coarse-graining and flower graphs}

Let us give a closer look to this gauge-fixing procedure and the resulting coarse-graining of spin networks.
At the classical level, a spin network state is given by the graph dressed with discrete holonomy-flux data: each oriented edge carries a $\mathrm{SU}(2)$ group element $g_{e}\,\in\mathrm{SU}(2)$ while each edge's extremity around  a vertex is colored with a vector $X^{v}_{e}\in\mathbb{R}^{3}$. So one edge carries two vectors, one living at its source vertex and the other living at its target vertex, respectively $X^{s,t}_{e}\equiv X^{s,t(e)}_{e}$. The group element gives the parallel transport of the vectors along the edges, that is $X^{t}_{e}=-\,g_{e}\triangleright X^{s}_{e}$ with the action of $g_{e}$ as a $\mathrm{SO}(3)$-rotation on the flat 3d space.
%
%Provided with its holonomy $g_{e}$ and its flux $X^{s}_{e}$, each edge carries a copy of $T^{*}\mathrm{SU}(2)$ and is endowed with the corresponding Poisson bracket.
%
This obviously forces the two vectors to have equal norm, $|X^{t}_{e}|=|X^{s}_{e}|$, which is called the (area-)matching constraint.
One requires another set of constraints: we impose the closure constraint at each vertex $v$, so that the sum of the fluxes around the vertex vanishes, $\sum_{e\ni v} X^{v}_{e}=0$. This holonomy-flux data can be interpreted as some discrete geometry in the framework of twisted geometries
%generalizing Regge geometries
\cite{Freidel:2010aq,Freidel2014}.
This is achieved through Minkowski's theorem stating that the closure constraint determines a unique convex polyhedron in flat 3d space dual to each vertex $v$, such that the fluxes $X^{v}_{e}$ are the normal vectors to the polyhedron faces.

Curvature appears as non-trivial holonomies around loops $\mathcal{L}$ of the graph, when $\overrightarrow{\prod_{e\in\mathcal{L}}}g_{e}\,\ne\mathbb{1}$. As pointed out in \cite{Livine:2013gna}, coarse-graining a subgraph carrying non-trivial curvature leads to an effective vertex breaking the closure constraint. This underlines the fact that a generalization of spin network states is required in order to properly carry out a coarse-graining procedure: we need an extended structure allowing for \textit{curved} vertices.

Indeed, let us consider a bounded region in space and the normals to its boundary.
%
%Let us consider the normals at the boundary of the region, which can be considered as being constitued of flat pieces.
%
Due to gauge-invariance, if the region contains a single vertex, the sum of the normals will sum up to zero. But if there are loops inside the region, the parallel transport around these loops might introduce non-trivial rotations. And indeed, as soon as the parallel transport around the loops is non-trivial, the sum of the normals is no longer zero, leading to a closure defect \cite{Livine:2013gna}. This is natural and translates the fact that curvature is carried by the loops of the spin network. And this must be taken  into account when coarse-graining.

A rigorous way to make this explicit is to gauge-fix the spin network state, following the procedure devised in \cite{Freidel2003}.
Let us consider a bounded region of a larger spin network, defined as a finite connected subgraph $\gamma$ of the larger graph $\Gamma$.%,  as in fig.\ref{fig:GaugeFix}.
The procedure goes as follow:
\begin{enumerate}

\item Choose arbitrarily a root vertex $v_{0}$ of the subgraph and select a maximal tree $T$ of the region:

  The subgraph being connected, the maximal tree  goes through every vertex of the region and defines a unique path of edges from the root vertex $v_{0}$ to any vertex of the subgraph.

\item Gauge-fix iteratively all the group elements along the edges of the tree $g_{e\in T}=\mathbb{1}$:

  Using the gauge-invariance of the wave-functions with gauge transformations acting at every vertex by $\mathrm{SU}(2)$ group elements $h_{v}$ as $g_{e}\rightarrow h_{s(e)}^{-1}g_{e}h_{t(e)}$, we can start from the root of the tree $v_{0}$ and progress through the tree until we reach the boundary of our subgraph. We define the appropriate gauge transformations $h_{v}$ at every vertex in order to fix all the group elements along the edges of the tree to the identity $\mathbb{1}$. The absence of loops in the tree, by definition, guarantees the consistency of this gauge-fixing. We can somewhat interpret this maximal tree as a synchronization network: we set all the parallel transports along the tree edges to the identity, thus synchronizing the reference frames at all the vertices and identifying them to a single reference frame living at the root of the subgraph. This realizes the coarse-graining of the subgraph $\gamma$ to its chosen root vertex $v_{0}$.
  The action of $\mathrm{SU}(2)$ gauge transformations inside the region is not entirely gauge-fixed and we are still left with the $\mathrm{SU}(2)$ gauge transformations at the root vertex.
  %
  %Note that there is still a spurious gauge freedom: the action of $SU(2)$ at the root of the tree leads to an action of $SU(2)$ on the curved vertex which acts by conjugation on the outgoing loops.

\item Having collapsed the subgraph $\gamma$ to its root vertex $v_{0}$, the edges of the subgraph $\gamma$ which are not in the tree, $e\in\gamma\setminus T$ label all the (independent) loops of the subgraph and lead to self-loops attached to the $v_{0}$:

  These self-loops or little loops carry the holonomies around the loops of the original subgraph $\gamma$, that is the curvature living in the bounded region. The flux-vectors living on the boundary edges, linking the region to the outside bulk, generically do not satisfy the closure constraint anymore since the effective vertex does satisfy a closure constraint which takes into account the flux-vectors of those boundary edges but also of the internal loops. The closure defect, induced by the little loops, thus reflects the non-trivial internal structure  of the coarse-grained subgraph and curvature developed in the corresponding region of the spin network state. The interested reader can find details and proof in the previous work \cite{Livine:2013gna}.

  %The remaining degrees of freedom correspond to the edges within the bounded region that are not in the tree. They can be represented as loops coming out of the collapsed vertex.

\end{enumerate}

This gauge-fixing procedure allows to clearly identify and distinguish between the degrees of freedom of the internal geometry of the considered bounded region of space to coarse-grain. The tree encodes the internal combinatorial structure of the region and describes the network of points and links within: they provide the bulk structure on which we can create curvature.  The little loops and the $\mathrm{SU}(2)$ group elements coloring them are the  excitations of the parallel transport and curvature. Together, tree and little loops attached to a vertex describe all its internal structure and are the extra data needed to define \textit{curved vertices} for the effective coarse-grained theory. These curvature excitations create a closure defect for the flux-vectors living on the boundary edges linking the coarse-grained vertex -the root vertex- to the rest of the spin network (obtained by the actually satisfied closure constraint between boundary edges and little loops)

\section{A hierarchy of structures}

When coarse-graining in practice, we do not want to retain all the information about the internal geometry, but only want to retain the degrees of freedom most relevant to the dynamics and interaction with the exterior geometry. In the next section, we will therefore introduce a hierarchy of extensions of spin network states with \textit{curved vertices}, from the finest notion of spin networks decorated with both trees and little loops to the coarser notion of spin networks with a simple tag at each vertex recording the induced closure defect.

In \ac{LQG}, we start with spin network states, which are graphs decorated with spins on the edges and intertwiners at the vertices:
\begin{equation}
  {\mathcal{H}}_\Gamma = \bigoplus_{\{j_e,i_v\}} \mathbb{C}|j_e,i_v\rangle\,.
\end{equation}
Curvature is carried loops of the graph.
We have argued that coarse-graining these networks should naturally lead to extended spin networks that can carry localized curvature excitations at the vertices. Following the coarse-graining through gauge-fixing procedure, we propose a hierarchy of three possible extensions of the spin network states, which depend on how much extra information and structure are added to each vertex:
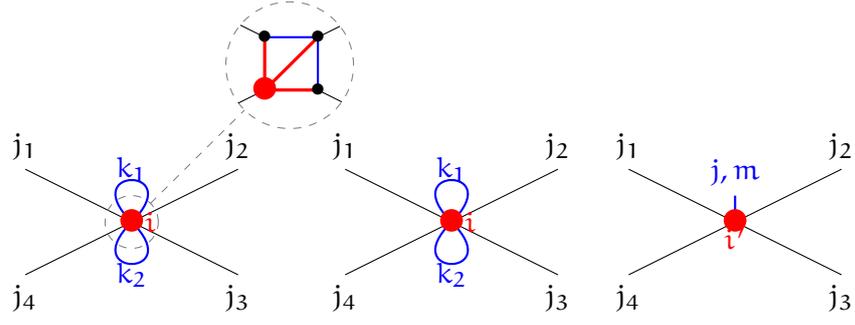
\begin{figure}
  \begin{subfigure}[t]{.33\linewidth}
    \centering
    \begin{tikzpicture}[scale=0.7]
      \coordinate(O1) at (0,0);
      \coordinate(O2) at (2.5,3.5);
      \coordinate(A) at ($(O2)+(0,0)$);
      \coordinate(B) at ($(O2)+(1,0)$);
      \coordinate(C) at ($(O2)+(1,-1)$);
      \coordinate(D) at ($(O2)+(0,-1)$);

      \draw (O1) -- ++(-2,1) node[above] {$j_1$};
      \draw (O1) -- ++(2,1) node[above] {$j_2$};
      \draw (O1) -- ++(2,-1) node[below] {$j_3$};
      \draw (O1) -- ++(-2,-1) node[below] {$j_4$};
      \draw[blue,thick] (O1) to[loop,scale=3] (O1) ++(0,1) node {$k_1$};
      \draw[blue,thick] (O1) to[loop,scale=3,rotate=180] (O1) ++(0,-1) node {$k_2$};
      
      \draw[red] (O1) node[scale=2] {$\bullet$} ++(0.35,0) node{$i$};
      
      \draw[gray,dashed] (O1) circle (0.5) ++(45:0.5) -- (45:3) ++(45:1.2) coordinate (O3) circle (1.2);

      \clip (O3) circle (1.2);
      
      \draw[blue,thick] (A) -- (B);
      \draw[blue,thick] (B) -- (C);
      \draw[red,very thick] (C) -- (D);
      \draw[red,very thick] (D) -- (A);
      \draw[red,very thick] (B) -- (D);
      
      \draw (A) -- ++(-2,1);
      \draw (B) -- ++(2,1);
      \draw (C) -- ++(2,-1);
      \draw (D) -- ++(-2,-1);
      
      \draw (A) node {$\bullet$};
      \draw (B) node {$\bullet$};
      \draw (C) node {$\bullet$};
      \draw[red] (D) node[scale=2] {$\bullet$};
    \end{tikzpicture}
    
    \caption{All the information can be preserved by carrying the $SU(2)$ labels and an unfolding tree describing the inner details of the coarse-grained vertex.}\label{fig:loopy_a}
  \end{subfigure}%
  \hspace{2mm}
  \begin{subfigure}[t]{.28\linewidth}
    \centering
    \begin{tikzpicture}[scale=0.7]
      \coordinate(O1) at (0,0);
      
      \draw (O1) -- ++(-2,1) node[above] {$j_1$};
      \draw (O1) -- ++(2,1) node[above] {$j_2$};
      \draw (O1) -- ++(2,-1) node[below] {$j_3$};
      \draw (O1) -- ++(-2,-1) node[below] {$j_4$};
      \draw[blue,thick] (O1) to[loop,scale=3] (O1) ++(0,1) node {$k_1$};
      \draw[blue,thick] (O1) to[loop,scale=3,rotate=180] (O1) ++(0,-1) node {$k_2$};
      
      \draw[red] (O1) node[scale=2] {$\bullet$} ++(0.35,0) node{$i$};
    \end{tikzpicture}
    \caption{The particular subgraph can be forgotten and only the $SU(2)$ information is preserved.}\label{fig:loopy_b}
  \end{subfigure}
  \hspace{2mm}
  \begin{subfigure}[t]{.28\linewidth}
    \centering
    \begin{tikzpicture}[scale=0.7]
      \coordinate(O1) at (0,0);
      
      \draw (O1) -- ++(-2,1) node[above] {$j_1$};
      \draw (O1) -- ++(2,1) node[above] {$j_2$};
      \draw (O1) -- ++(2,-1) node[below] {$j_3$};
      \draw (O1) -- ++(-2,-1) node[below] {$j_4$};
      \draw[blue,thick] (O1) to ++(0,0.5) node[above]{$j,m$};
      
      \draw[red] (O1) node[scale=2] {$\bullet$} ++(0,-0.3) node{$i'$};

    \end{tikzpicture}
    \caption{Everything except the closure defect is forgotten. Only a ``tag'' remains.}\label{fig:loopy_c}
  \end{subfigure}
  \caption{The hierarchy of possible coarse-graining frameworks}\label{fig:loopy}
\end{figure}
\begin{enumerate}
  
  %%%
\item {\textbf{Folded spin networks~:}}
  %%%
  
  In the first scenario, we follow the gauge-fixing procedure but we do a minimal coarse-graining, retaining as much information as possible on the original state.
  Each vertex is allowed with an arbitrary number of little loops attached to it and is endowed with a tree connecting the ends of the external edges and of the internal loops, as represented in fig.\ref{fig:loopy_a}. This tree can be seen as a circuit telling us how to unfold the vertex, reversing the gauge-fixing procedure and recovering the original (finer) graph.
  This Hilbert space $\mathcal{H}_{\Gamma}^\mathrm{folded}$ can be written formally as:
  \begin{equation}
    %\mathcal{H}_\mathrm{Trees}
    \mathcal{H}_{\Gamma}^\mathrm{folded}
    = \bigoplus_{\{j_e,j^{(v)}_{\ell},i_v,\mathcal{T}_v\}} \mathbb{C}\,|j_e,j^{(v)}_{\ell},i_v,\mathcal{T}_v\rangle\,.
  \end{equation}
  $\mathcal{T}_v$ is the unfolding tree for each vertex, $j^{(v)}_{\ell}$ are the spins carried by the additional loops labeled by the index $\ell$ and the intertwiners $i_{v}$ now lives in the tensor product of the spins $j_{e}$ of the edges linking to the other neighboring vertices and (twice) the spins $j^{(v)}_{\ell}$ living on the internal loops (because each loop has its two ends at the vertex).
  
  With such an internal space at each vertex, we actually lose no information at all on the internal degrees of freedom. Starting with a spin network state living on a finer graph $\tGamma$, we simply gauge-fix it to a spin network on our coarser graph $\Gamma$. And we can follow the reverse path. Using the tree at each vertex, we can fully reconstruct the original finer graph $\tGamma$ thus simply perform generic gauge transformations to recover the fully gauge-invariant spin network state.
  
  Thus the chosen graph $\Gamma$ can be considered as a skeleton graph, to which we can add extra information to represent spin network states living on any (finer) graph. In a sense, we have not done any coarse-graining yet. The truncation of the theory will happen when defining the dynamics on the folded spin network Hilbert space, distinguishing actual edges and spins of our  skeleton lattice -the background- from spins and edges on the unfolding trees and little loops, when the fundamental dynamics would have considered them on equal footing.
  
  %%%
\item {\textbf{Loopy spin networks~:}}
  %%%
  
  In a second scenario, we coarse-grain the internal structure of the effective vertices by discarding the unfolding trees. We keep the curvature excitations living on the little loops, but we discard the combinatorial information of the internal subgraph: we forget that the  vertex effectively represents an actual extended region of space and we localize all the internal curvature degrees of freedom on that coarse-grained vertex. This leads to loopy spin networks, with an arbitrary number of loops at each vertex but no unfolding tree data:
  \begin{equation}
    \mathcal{H}^{\mathrm{loopy}}_{\Gamma} = \bigoplus_{\{j_e,j^{(v)}_{\ell},i_v\}} \mathbb{C}|j_e,j^{(v)}_{\ell},i_v\rangle\,,
  \end{equation}
  where the $j^{(v)}_{\ell}$ are the spins living on the little loops attached to the vertex $v$ and the intertwiners $i_{v}$ live again in the tensor product of the spins carried by the graph edges attached to the vertex $v$ and the spins carried by its little loops.
  
  Now our chosen graph $\Gamma$ for loopy spin network states is to be considered as a background graph. The little loops are explicit local excitations of the gravitational fields located at each vertex of the graph. A given loopy spin network comes from the coarse-graining of several possible finer spin network states living on finer graph, but we lack the unfolding tree information to recover the original more fundamental state.

     The truncation of full theory is clear. Spin network states on the ``loopy graphs'' living on top on $\Gamma$, that is the base graph $\Gamma$ plus an arbitrary number of self-loops at every vertices, are already in the Hilbert space of the \ac{LQG}, although we do not usually focus on such graphs. Restricting ourselves to this subset of states is a clear truncation of the full Hilbert space. The difference with the standard interpretation is that we think here of the base graph $\Gamma$ as embedded in the space manifold, while the little loops are abstract objects decorating the base graph vertices.

     Since we have local degrees of freedom, carried by the little loops, we need to discuss their statistics, which leads to a few variations of this theme:

     \begin{enumerate}
     \item {\textit{Distinguishable loops~:}}
       First, it is natural to consider that the loops are distinguishable as they come from a substructure. The loops do come from different edges of a finer graph and create curvature excitations  at different places within the coarse-grained bounded region. As a result, we should distinguish them and allow to number and order them.
       %
       %This will lead to the straightforward Hilbert space just described and to a canonical application of projective techniques.

     \item {\textit{Indistinguishable bosonic loops~:}}
       A second possibility is to push further along the logic  of coarse-graining  and to consider that the loops indistinguishable since we do not have access anymore to the specific substructure. This should lead to bosonic statistics, as expected for gravitational field excitations. Formally, this can be written as the identification:
       \begin{equation}
         |j_e,i_v,j^{(v)}_{\ell}\rangle = |j_e,i_v,j^{(v)}_{\sigma_{v}(\ell)}\rangle
       \end{equation}
       for any permutation $\sigma_{v}\in\,S_{\#\ell}$ in the symmetric group of order $\#\ell$ when the vertex $v$ has $\#\ell$ loops. This point of view is compatible with considering the action of space diffeomorphisms on the little loops around the vertex as gauge transformations.

     \item  {\textit{Anyonic statistics~:}}
       We can easily imagine other statistics, for instance by allowing for a phase in the equality above (i.e a non-trivial representation of the permutation group). In fact, instead of thinking of the vertex as a mere point, we can represent the boundary of the bounded region as a sphere and consider the little loops as living on a sphere around it. Then the diffeomorphism invariance on the sphere will lead to an action of the braiding group leading to interesting anyonics statistics, similarly to the punctures of a Chern-Simons theory as  already explored  in the case of black holes in \ac{LQG} \cite{Pithis:2014uva}.

     \end{enumerate}

     %%%
   \item {\textbf{Tagged spin networks~:}}
     %%%

     In this third and last scenario, we fully coarse-grain the internal geometry of the bounded region now reduced to a graph vertex. We discard the unfolding tree, used in the gauge-fixing and unfixing procedure, and we integrate out the little loops attached to the vertex. All we retain is the closure defect induced by the non-trivial holonomies and spins carried by those little loops. The fact that coarse-graining spin networks, or their classical counterpart of twisted geometries, leads to closure defect, accounting for the presence of a non-trivial curvature within the coarse-grained region was already pointed out in \cite{Livine:2013gna}.
     Here, the simplest method to see how this comes about is to use the intermediate spin decomposition of the intertwiner at the vertices, as illustrated on fig.\ref{fig:intermediatespin}, introducing a fiducial link separating the external edges from the internal loops:
     \begin{equation}
       \begin{array}{c}
         \mathrm{Inv}_{\mathrm{SU}(2)}\,\Big{[}
           \bigotimes_{e}\mathcal{V}^{j_{e}}
           \otimes
           \bigotimes_{\ell}\big{(}\mathcal{V}^{j_{\ell}}\otimes\mathcal{V}^{j_{\ell}}\big{)}
           \Big{]} \\
         \,=\, \\
         \bigoplus_{J}
         \,
         \mathrm{Inv}_{\mathrm{SU}(2)}\,\Big{[}
           \mathcal{V}^{J}
           \otimes
           \bigotimes_{e}\mathcal{V}^{j_{e}}
           \Big{]}
         \otimes
         \mathrm{Inv}_{\mathrm{SU}(2)}\,\Big{[}
           \mathcal{V}^{J}
           \otimes
           \bigotimes_{\ell}\big{(}\mathcal{V}^{j_{\ell}}\otimes\mathcal{V}^{j_{\ell}}\big{)}
           \Big{]}
         \end{array}
     \end{equation}
     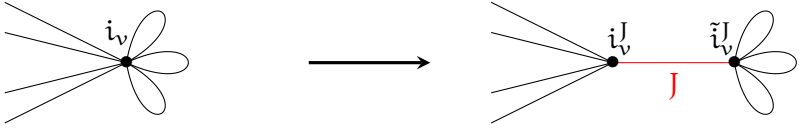
\begin{figure}[h!]

       \centering

       \begin{tikzpicture}[scale=0.8]
         \coordinate(O1) at (0,0);

         \draw (O1) -- ++(-2,1);
         \draw (O1) -- ++(-2,0.5);
         \draw (O1) -- ++(-2,-0.5);
         \draw (O1) -- ++(-2,-1);
         \draw (O1) to[in=-25,out=25,loop,scale=3] (O1);
         \draw (O1) to[in=30,out=80,loop,scale=3] (O1);
         \draw (O1) to[in=-80,out=-30,loop,scale=3] (O1);
         %\draw (O1) to[loop,scale=3,rotate=-90] (O1);
         %\draw (O1) to[loop,scale=6,rotate=-90] (O1);

         \draw (O1) node {$\bullet$} ++(-0.15,0.5) node{$i_v$};

         \draw[->,>=stealth,very thick] (3,0) -- (5,0);

         \coordinate(O2) at (8,0);
         \coordinate(O3) at (10,0);

         \draw (O2) -- ++(-2,1);
         \draw (O2) -- ++(-2,0.5);
         \draw (O2) -- ++(-2,-0.5);
         \draw (O2) -- ++(-2,-1);

         \draw[in=-25,out=25,scale=3] (O3) to[loop] (O3);
         \draw[in=30,out=80,scale=3] (O3) to[loop] (O3);
         \draw[in=-80,out=-30,scale=3] (O3) to[loop] (O3);

         \draw[red] (O2) -- (O3) node[midway,below]{$J$};
         \draw (O2) node {$\bullet$} ++(0.12,0.4) node{$i_v^{J}$};
         \draw (O3) node {$\bullet$} ++(-0.2,0.4) node{$\tilde{i}_v^{J}$};

       \end{tikzpicture}

       \caption{We represent a loopy vertex $v$, here with three little loops attached to it. The intertwiner $i_{v}$ can be decomposed onto the intermediate spin basis, where we introduce a fiducial edge between the external legs and the internal loops. This orthogonal basis is labeled by the intermediate spin $J$, and two intertwiners $i^{J}_{v}$ and $\tilde{i}^{J}_{v}$ intertwining between that intermediate spin  and respectively the external legs or the internal loops.}
       \label{fig:intermediatespin}

     \end{figure}
     This spin $J_{v}$ living at the vertex $v$ encodes the closure defect and is the only extra information with which we decorate the graph.
     We call it the \textit{tag} and amounts to adding an open leg to every vertex of the graph. This open edge is colored with the spin $J_{v}$ and a vector in that $\mathrm{SU}(2)$ representation. Using the standard spin basis labeled by magnetic moment number $M$, the Hilbert space of \textit{tagged spin networks} on the base graph $\Gamma$ is then formally defined as:
     \begin{equation}
     \mathcal{H}_\Gamma^{\mathrm{tag}} = \bigoplus_{\{j_e,J_v,M_v,i_v\}} \mathbb{C}|j_e,J_v,M_v,i_v\rangle\,,
     \end{equation}
     where the intertwiner $i_{v}$ at the vertex $v$ now lives in the tensor product of the spins $j_{e\ni v}$ on the external edges $e$ attached to the vertex and of the vertex tag $J_{v}$.

     The state $|J_v,M_v\rangle$ is the quantized version of the closure defect vector. Indeed, at the classical level, as shown in \cite{Livine:2013gna}, the sum of the flux-vectors living on the external edges $e\ni v$ does not vanish anymore and should be balanced by the sum of the flux-vectors living on the internal loops. This defect vector means that there is  no convex polyhedron dual to the vertex, as usual in twisted geometries. One way to go is to try to open the polyhedron somehow, which wouldn't have a clear geometrical interpretation. Instead we propose to interpret it as the dual convex polyhedron should not be embedded in flat space but in a (homogeneous) curved space, the curvature radius depending on the actual value of the closure defect. Progress in this direction has been achieved in the study of hyperbolic and spherical tetrahedra \cite{Bonzom:2014wva,Charles:2015lva,Haggard:2015ima} but we do not yet have an explicit  embedding and formula relating the curvature to the norm of the defect. It would ultimately be enlightening to relate this tag $J_{v}$ to the spectrum of some quasi-local energy operator in \ac{LQG} (e.g. \cite{Yang:2008th}), which would allow to view it as a measure of the gravitational energy density within the bounded coarse-grained region.

   \end{enumerate}

   \medskip

   These three extended spin network structures are the heart of our present proposal for studying effective truncations for the coarse-graining of \ac{LQG}. The goal would be to reformulate the dynamics of \ac{LQG} on these new structures and study their renormalisation flow under coarse-graining. An important point is that these folded, loopy and tagged spin networks sidestep the problem of fluctuating graph dynamics and allow to project the whole dynamics on a fixed background graph, or skeleton, interpreted as the lattice postulated by the observer. Note that the background lattice can then be adapted to the studied models. We could choose a regular lattice or a much simpler graph, such as a flower with a single vertex and an arbitrary number of little loops. Such simple graphs could reveal useful in the study of highly symmetric problems as is the case in cosmology or in the study of Einstein-Rosen waves \cite{Korotkin:1997ps,Ashtekar:1996cm}. In general though, we have local excitations of the geometry, representing the internal fluctuations of the gravitational field in the coarse-grained regions, living at the graph vertices and represented by the new information attached to them, respectively unfolding trees, little loops or tags. The use of a background lattice, which might be regular, would simplify greatly the setting of a systematic coarse-graining of \ac{LQG}.

   The folded spin networks are mathematically a simple gauge-fixing of spin networks onto the skeleton graph. In the following sections, we will focus on providing a clean mathematical definition of loopy and tagged spin networks and exploring the definition of a Fock space of loopy spin networks with bosonic statistics for the little loops living at every graph vertex.
   
%\textbf{TODO:} second point of view: algebraic, idea of flower graphs, for coarse-graining might be useful for fixed point search, should be able to describe a changing graph dynamics on some kind of fixed graph with additionnal info, the additionnal information is the loops

\section{Projective limit proper}

Here we  would like to define properly loopy spin networks and investigate their properties.
Choosing a fixed graph $\Gamma$ with $E$ edges, and given numbers of little loops $N_{v}$ at each vertex $v$, we consider the following space of wave-functions on $\mathrm{SU}(2)^{\times\,(E+\sum_{v}N_{v})}$ invariant under $\mathrm{SU}(2)$  gauge transformations acting at every vertices:
\begin{equation}
\psi\Big{(}\{g_{e}\,,\,h^{v}_{\ell}\}_{e,v\in\Gamma}\Big{)}
\,=\,
\psi\Big{(}\{a_{s(e)}g_{e}a_{t(e)}^{-1}\,,\,a_{v}h^{v}_{\ell}a_{v}^{-1}\}\Big{)}
\,,\qquad
\forall a_{v}\in\mathrm{SU}(2)^{\times V}\,.
\end{equation}
The $\mathrm{SU}(2)$ gauge transformations act as usual on the edges $e$ of the graph, while they act by conjugation as expected on the little loops.
A basis is provided by the spin decomposition on functions in $L^{2}(\mathrm{SU}(2))$ as with standard spin networks. The loopy spin network basis states are labeled with a spin $j_{e}$ on each edge $e$, a spin $k^{v}_{\ell}$ on each little loop $\ell$ attached to a vertex $v$, and an intertwiners $i_{v}$ at each vertex leaving in the tensor product of the attached edges and of the loop spins:
\begin{equation}
i_{v}\in\,\mathrm{Inv}_{\mathrm{SU}(2)}\,
\Big{[}
  \bigotimes_{e\ni v}\mathcal{V}^{j_{e}}
  \,\otimes\,
  \bigotimes_{\ell \ni v}(\mathcal{V}^{k^{v}_{\ell}}\otimes\bar{\mathcal{V}}^{k^{v}_{\ell}})
  \Big{]}\,
\end{equation}
so that the Hilbert space of loopy spin networks on the graph $\Gamma$ with given number $N_{v}$ of little loops at every vertex is, as announced in the previous section presenting the hierarchy of extended spin network structures:
\begin{equation}
\mathcal{H}_{\Gamma,\{N_{v}\}}^{\mathrm{loopy}}
\,=\,
L^{2}\big{(}
\mathrm{SU}(2)^{\times\,(E+\sum_{v}N_{v})}
\,/\,
\mathrm{SU}(2)^{\times V}
\big{)}
\,=\,
\bigoplus_{\{j_{e},k^{v}_{\ell},i_{v}\}}\,
\mathbb{C}\,|j_{e},k^{v}_{\ell},i_{v}\rangle\,.
\end{equation}
What needs to be properly defined and analyzed is the Hilbert space of states with arbitrary number of little loops, allowing $N_{v}$ to run all over $\mathbb{N}$ and summing over all these possibilities. To this purpose, the full graph structure $\Gamma$ does not intervene and we can ignore it and focus on the space of little loops around a single vertex. Thus, for the sake of simplifying the discussion, we will focus on a single vertex with no external, but with an arbitrary umber of little loops attached to it. This is the \textit{flower} graph.

In this section, we will assume the little loops to be distinguishable. We define the spin network states with a given number of loops -the flower graph with fixed number of petals- and we then discuss the whole Hilbert space of states with arbitrary number of excitations by a projective limit. We define and analyze the holonomy operators acting on that space and we finally  implement the BF theory dynamics on that space as a first application of our framework and a consistency check.
We will tackle the case of indistinguishable little loops in the next chapter, imposing bosonic statistics and defining the holonomy operator on symmetrized  states.

Let us start with the flower graph with a fixed number $N$ of petals, that is a single vertex with $N$ little loops attached to it as drawn on fig.\ref{fig:singlevertex}. We are going to define the wave-functions on that graph, the corresponding decomposition on the spin and intertwiner basis and the action of the holonomy operators.
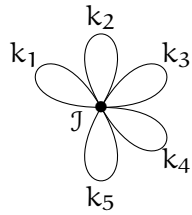
\begin{figure}[h!]
  \centering
  \begin{tikzpicture}
    \coordinate(O1) at (0,0);

    \draw (O1) to[in=-30,out=+30,loop,scale=3,rotate=90] (O1) ++(0,1.2) node {$k_2$};
    \draw (O1) to[in=-30,out=+30,loop,scale=3,rotate=30] (O1) ++(1,0.7) node {$k_3$};
    \draw (O1) to[in=-30,out=+30,loop,scale=3,rotate=-30] (O1) ++(1,-0.7) node {$k_4$};
    \draw (O1) to[in=-30,out=+30,loop,scale=3,rotate=150] (O1) ++(0,-1.2) node {$k_5$};
    \draw (O1) to[in=-30,out=+30,loop,scale=3,rotate=-90] (O1) ++(-1,0.7) node {$k_{1}$};
    %\draw (O1) to[in=-30,out=+30,loop,scale=3,rotate=90] (O1) ++(0,1) node {$k_1$};

    \draw (O1) node[scale=1] {$\bullet$} ++(-0.32,-0.18) node{$\mathcal{I}$};

  \end{tikzpicture}
  \caption{We consider the class of special graph, flowers, with a single vertex and an arbitrary number $N$ of little loops attached to it. Here we have drawn a flower with $N=5$ petals. The spin network states on such graphs are labeled by a spin on each loop, $k_{\ell=1..N}$, and an intertwiner $\mathcal{I}$ living in the tensor product $\bigotimes_{\ell=1}^{N} (\mathcal{V}^{k_{\ell}}\otimes\bar{\mathcal{V}}^{k_{\ell}})$.}
  \label{fig:singlevertex}
\end{figure}

Wave-functions are gauge-invariant functions of $N$ group elements, that is functions on $\mathrm{SU}(2)^{\times N}$ invariant under the global action by conjugation:
\begin{equation}
\Psi(h_1,...,h_N) = \Psi(gh_1g^{-1},...,gh_Ng^{-1})\,.
\end{equation}
The scalar product is defined by integration with respect to the Haar measure on $\mathrm{SU}(2)$ and the resulting Hilbert space is:
\begin{equation}
\mathcal{H}_{N}=
L^2\,\Big{(}\mathrm{SU}(2)^{\times N}/\mathrm{Ad}\,\mathrm{SU}(2)\Big{)}\,.
\end{equation}
%
%The space of wave-functions is thus:
%\begin{equation}
%\mathcal{H}_N \simeq \{\Psi : SU(2)^N \rightarrow \mathbb{C},~\Psi\textrm{ continuous},~ \forall g \in SU(2) ~ \mathcal{G}(g)\Psi = \Psi\}
%\end{equation}
%where $N$ is the number of edges (or loops) and $\mathcal{G}$ is the gauge transformation which, in this case, can be written:
%\begin{equation}
%\left(\mathcal{G}(g)\Psi\right)(h_1,...,h_N) = \Psi(gh_1g^{-1},...,gh_Ng^{-1})
%\end{equation}
%Thus, $\mathcal{H}_N$ is the space of \textit{loopy} intertwiners, that is the space of intertwiners between loops.
%
A basis of this space is provided as usual by the spin network states, labeled by a spin on each loop, $k_{\ell=1..N}\,\in\frac{\mathbb{N}}{2}$, and an intertwiner $\mathcal{I}$ living in the tensor product $\bigotimes_{\ell=1}^{N} (\mathcal{V}^{k_{\ell}}\otimes\bar{\mathcal{V}}^{k_{\ell}})$ and invariant under the action of $\mathrm{SU}(2)$:
\begin{equation}
\Psi^{\{k_{\ell},\mathcal{I}\}}\big{(}\{h_{\ell}\}_{\ell=1..N}\big{)}
\,=\,
\langle h_{\ell}\,|\,k_{\ell},\mathcal{I}\rangle
\,=\,
\tr\,\Big{[}
  \mathcal{I}\otimes\bigotimes_{\ell=1}^{N}D^{k_{\ell}}(h_{\ell})
  \Big{]}\,,
\end{equation}
where the trace is taken over the tensor product $\bigotimes_{\ell=1}^{N} (\mathcal{V}^{k_{\ell}}\otimes\bar{\mathcal{V}}^{k_{\ell}})$. To underline that each spin representation is doubled and that $\mathcal{I}$ is an intertwiner between the loops around the vertex, we can dub it a \textit{loopy intertwiner}

%What operators can be written on this space? A very simple and natural operator, and one on which we will focus for most of this paper, is the holonomy trace. It is naturally gauge-invariant. In this very simple setting, it can be defined quite easily:
%
The holonomy operator is the basic gauge-invariant operator of \ac{LQG}. It can shift and increase the spins along the edges on which it acts and so is used in practice as a creation operator. We define the holonomy operators $\hat{\chi}_\ell$ along the loops around the vertex as acting by multiplication on the wave-functions in the group representation:
\begin{equation}
(\hat{\chi}_\ell \triangleright\Psi)\,(h_1,...,h_N)
\,=\,
\chi_\frac{1}{2}(h_\ell) \Psi(h_1,...,h_N)
\end{equation}
where $\chi_\frac{1}{2}$ is the trace operator in the fundamental two-dimensional representation of $\mathrm{SU}(2)$. We can of course also consider holonomy operators that wrap around several loops around the flower:
\begin{equation}
(\hat{\chi}_{i,j,k,l,...}  \triangleright\Psi)\,(h_1,...,h_N)
\,=\,
\chi_\frac{1}{2}(h_i h_j h_k h_l ...) \Psi(h_1,...,h_N)\,,
\end{equation}
where the $i,j,k,l,..$ indices label loops. These operators are obviously still gauge-invariant, and we can further take the inverse or arbitrary powers of each group element.
%
% and $i$ is the number of the loop being considered. This simple loop operator favours the canonical loop around the vertices. But in our strategy, these loops correspond to a particular gauge-fix of an internal subgraph. By changing the gauge-fixing (or coarse-graining) procedure, we would end up with different loops. As a consequence, we should define the holonomy for any path around the loops. This is actually a straightforward extension and, adopting the notation $h_{-i} = h_i^{-1}$ and $h_0 = 1$, it goes as:
%\begin{equation}
%\hat{\chi}_{i,j,k,l,...} : \left\{\begin{array}{rcl}
%\mathcal{H}_N^0 &\rightarrow& \mathcal{H}_N^0 \\
%\Psi &\mapsto& ((h_1,...,h_N) \mapsto \chi_\frac{1}{2}(h_i h_j h_k h_l ...) \Psi(h_1,...,h_N))
%\end{array}\right.
%\end{equation}
%This operator represents the trace of the holonomy following the path passing by the loop $|i|$, $|j|$, $|k|$, $|l|$,... in that order and with the sign specifying if we follow the orientation of the loop or we go through it in reverse. This definition coincides, of course, with the previous one when the path contains only one loop.
%
There are two remarks we should do about these multi-loop operators. First, they can be decomposed as a composition of single loop operators combining both holonomy operators and grasping operators (action of the $\mathfrak{su}(2)$ generators as a quantization of the flux-vectors) by iterating the following 2-loop identity:
\begin{equation}
\chi_\frac{1}{2}(h_i h_j)
\,=\,
\frac{1}{2}\,\left[
  \chi_\frac{1}{2}(h_i)\chi_\frac{1}{2}(h_j)
  +\sum_{a=1}^{3}\chi_\frac{1}{2}(h_i\sigma_{a})\chi_\frac{1}{2}(h_j\sigma_{a})
  \right]\,,
\end{equation}
where the $\sigma_{a}$'s are the three Pauli matrices, normalized such that their square is equal to the identity matrix.
Second, if the loopy spin network state comes from the gauge fixing of a more complicated graph down to a single vertex, we had chosen a particular maximal tree on that graph to define the gauge-fixing procedure. The loops around the coarse-grained vertex correspond to the edges that didn't belong to the folding tree. Changing the tree actually maps the single loop holonomies onto multi-loop holonomies \cite{Livine:2013gna}. So, from the coarse-graining perspective, there is no special reason to prefer single loops over multi-loop operators.

We would like to allow for an arbitrary number of loops $N$, with possibly an infinite number of loops, and superpositions of number of loops. We will apply the usual projective limit techniques used in \ac{LQG}.
We assume here that the little loops are all distinguishable, so we avoid all symmetrization issue. The case of indistinguishable loops will be dealt with in the next chapter \ref{ch:Bosons}. We discuss the countable infinity of loops around the vertex, so we can number them using the integers $\mathbb{N}$. The point, as with standard spin networks, is that a state with a spin-0 on an edge does not actually depend on the group element carried by that edge and is thus equivalent to a state on the flower without that edge. Reversing this logic, a state built on a finite number of loops is equivalent to a state with an arbitrary larger number of loops carrying a spin-0 on all the extra edges , which will allow to define it in the projective limit as a state on the flower with an infinite number of loops.

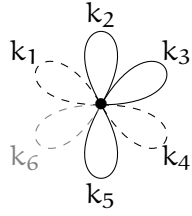
\begin{figure}[h!]
  \centering
  \begin{tikzpicture}

    \coordinate(O1) at (0,0);

    \draw (O1) to[in=-30,out=+30,loop,scale=3,rotate=90] (O1) ++(0,1.2) node {$k_2$};
    \draw (O1) to[in=-30,out=+30,loop,scale=3,rotate=-90] (O1) ++(0,-1.2) node {$k_5$};
    \draw (O1) to[in=-30,out=+30,loop,scale=3,rotate=30] (O1) ++(1,0.7) node {$k_3$};
    \draw[dashed] (O1) to[in=-30,out=+30,loop,scale=3,rotate=-30] (O1) ++(1,-0.7) node {$k_4$};
    \draw[dashed] (O1) to[in=-30,out=+30,loop,scale=3,rotate=150] (O1) ++(-1,0.7) node {$k_{1}$};
    \draw[black!50,dashed] (O1) to[in=-30,out=+30,loop,scale=3,rotate=-150] (O1) ++(-1,-0.7) node {$k_6$};
    %\draw[black!25,dashed] (O1) to[in=-30,out=+30,loop,scale=3,rotate=-150] (O1) ++(-1,-0.7) node {$k_6$};

    \draw (O1) node[scale=1] {$\bullet$} ;
    %\draw (O1) node[scale=1] {$\bullet$} ++(-0.32,-0.18) node{$\mathcal{I}$};

  \end{tikzpicture}

  \caption{We consider a loopy spin network state with a varying number of loops as a superposition of states with support over different loops.}
  \label{fig:variousloops}

\end{figure}

Let us consider  the set $\mathcal{P}_{<\infty}(\mathbb{N})$  of all finite subsets of $\mathbb{N}$. A flower with a finite number of loops corresponds to a finite subset $E\in \,\mathcal{P}_{<\infty}(\mathbb{N})$ of indices labeling its loops. Since we keep the loops distinguishable, we do not identify all the subsets with same cardinality and keep on distinguishing them. We define the Hilbert space of gauge-invariant wave-functions on the flower corresponding to $E$:
\begin{equation}
  \begin{array}{c}
    \mathcal{H}_{E}
    \,=\,
    L^2\,\Big{(}\mathrm{SU}(2)^{E}/\mathrm{Ad}\,\mathrm{SU}(2)\Big{)}\,, \\
    \Psi(\{h_{\ell}\}_{\ell\in E})=\Psi(\{gh_{\ell}g^{-1}\}_{\ell\in E})
    \quad
    \forall g\in\mathrm{SU}(2)\,.
  \end{array}
\end{equation}
We would like to consider arbitrary superpositions of states with support on arbitrary subsets $E$ of loops, but we do not wish to brutally consider the direct sum over all $E$'s. We still require cylindrical consistency. Indeed, a function on $\mathrm{SU}(2)^{E}$ which actually  does not depend at all on the loop $\ell_{0}\in E$ can legitimately be considered as a function on $\mathrm{SU}(2)^{E\setminus\ell_{0}}$. We introduce the equivalence relation making this explicit. For two subsets $E\subset F$, and two functions $\Psi$ and $\widetilde{\Psi}$ respectively on  $\mathrm{SU}(2)^{E}$ and  $\mathrm{SU}(2)^{F}$, the two wave-functions are defined as equivalent if:
\begin{equation}
  \begin{array}{l}
    E\subset F\,,
    \quad
    \Psi:\mathrm{SU}(2)^{E}\rightarrow\mathbb{C}\,,
    \quad
    \widetilde{\Psi}:\mathrm{SU}(2)^{F}\rightarrow\mathbb{C}\,, \\
    \quad\Psi\sim\widetilde{\Psi}
    \quad\Leftrightarrow\quad
    \widetilde{\Psi}(\{h_{\ell}\}_{\ell\in F})
    \,=\,
        {\Psi}(\{h_{\ell}\}_{\ell\in E})\,,
  \end{array}
\end{equation}
that is the function $\widetilde{\Psi}$ on the larger set $F$ does not depend on the group elements $h_{\ell}$ for $\ell\in F\setminus E$ and coincides with the function $\Psi$ on the smaller set $E$. More generally, when the two subsets $E$ and $F$ do not contain one or the other, we transit trough their intersection $E\cap F$.

The space of wave-functions in the projective limit  is defined as the union over all subsets $E$ of functions on  $\mathrm{SU}(2)^{E}$, quotiented by this equivalence. We similarly define the projective limit of the integration measure over $\mathrm{SU}(2)$. We use this measure to define the Hilbert space $\mathcal{H}^{\mathrm{loopy}}$ of states on the flower with an arbitrary number of loops. All the rigorous mathematical definitions and proofs are given in the appendix  \ref{app:proj}.

The practical way to see this Hilbert space is to use the spin network basis and understand that a loop carrying a spin-0 means that the wave-function actually does not depend on the group element living on that loop. For every state, we can thus reduce its underlying graph to the minimal possible one removing all the loops with trivial dependency. Following this logic, for every subset $E$, we define the space of proper states living on $E$, that is without any spin-0 on its loops. This amounts to removing all possible 0-modes:
\begin{equation}
\mathcal{H}_{E}^0
\,=\,
\Bigg{\{}
\Psi\in\mathcal{H}_{E}
\,:\,
\forall \ell_{0}\in E\,,\,\,\int_{\mathrm{SU}(2)}\mathrm{d}h_{\ell_{0}}\,\Psi =0
\Bigg{\}}\,.
\end{equation}
We can decompose the Hilbert space of states on the subset $E\subset\mathbb{N}$ of loops onto proper states:
%%%%%
\begin{prop}
  \label{proper}
  The Hilbert space $\mathcal{H}_{E}$ on loopy intertwiners on the set of loops $E$ decomposes as a direct sum of the Hilbert spaces of proper states with support on every subset of $E$:
  \begin{equation}
    \mathcal{H}_{E} \simeq \bigoplus_{F \subset E} \mathcal{H}_{F}^0
    %\quad\subset\cH^{\mathrm{loopy}}
    \,.
  \end{equation}
  This isomorphism is realized through the projections $f_{F}=P_{E,F}f\in\mathcal{H}^{0}_{F}$, acting on wave-functions $f\in\mathcal{H}_{E}$, defined for an arbitrary subset $F\subset E$:
  \begin{equation}
    \begin{array}{l}
      f_{F}\big{(}
      \{h_{\ell}\}_{\ell\in F}
      \big{)}
      \,=\,\\
      \quad\sum_{\tF \subset  F}
      (-1)^{\#\tF}
      \int \prod_{\ell\in E\setminus F}\mathrm{d}g_{\ell}
      \prod_{\ell\in\tF}\mathrm{d}k_{\ell}\,
      f\big{(}
      \{h_{\ell}\}_{\ell\in F\setminus \tF},
      \{k_{\ell}\}_{\ell\in\tF},
      \{g_{\ell}\}_{\ell\in E\setminus F}
      \big{)}\,.
    \end{array}
  \end{equation}
  These projections realize a combinatorial transform of the state $f\in\mathcal{H}_{E}$:
  \begin{equation}
    f=\sum_{F\subset E} f_{F}\,,
    \qquad
    f_{F}\in\mathcal{H}_{F}^{0}\,,
    \qquad
    \forall\ell\in F\,,\quad\int \mathrm{d}h_{\ell}\,f_{F}=0\,.
  \end{equation}
\end{prop}
%%%%%
This decomposition is straightforward to prove. It will also be crucial in the case of indistinguishable loops and  symmetrized states, as we will see in the next section \ref{ch:Bosons}.
Then, as we show in the appendix  \ref{app:proj},  the Hilbert space of loopy spin networks on the flower, with an arbitrary number of distinguishable loops, defined as the projective limit of the Hilbert spaces $\mathcal{H}_{E}$ is realized as  the direct sum of those spaces of proper states:
\begin{equation}
  \mathcal{H}^{\mathrm{loopy}} \simeq \bigoplus_{F \in \mathcal{P}_{<\infty}(\mathbb{N})} \mathcal{H}_{F}^0\,,
  \qquad
  \mathcal{H}_{F}^0 = \bigoplus_{j_{\ell \in F} \neq 0, \mathcal{I}} \mathbb{C} |j_{\ell \in F}, \mathcal{I}\rangle\,.
\end{equation}

We can revisit the definition of the holonomy operators on our Hilbert space $\mathcal{H}$ of states with arbitrary number of loops. In order to identify a complete set of operators acting on the Hilbert space $\mathcal{H}^{\mathrm{loopy}}$, we should further consider multi-loops holonomy operators or grasping operators or deformation operators such as $\mathrm{U}(N)$ operators \cite{Borja:2010rc}, but in the first exploration we propose, in this paper, we decide to focus on the single-loop holonomy operator. Let us consider the loop $\ell_{0}\in\mathbb{N}$ and define the corresponding holonomy operator $\hat{\chi}_{\ell_{0}}$. Looking at its action on a state $\Psi$ with finite number of loops living in the Hilbert space $\mathcal{H}_{E}$, we have two possibilities: either the loop $\ell_{0}$ belongs to the subset $E$ or it doesn't. If the acting loop $\ell_{0}$ is already a loop of our state $\Psi$, then the holonomy operator acts on as before by multiplication:
\begin{equation}
  \begin{array}{l}
    \ell_{0}\in E\,,
    \quad
    \Psi\in\mathcal{H}_{E}\,,
    \quad
    \hat{\chi}_{\ell_{0}}\,\Psi\in\mathcal{H}_{E}\,, \\
    \qquad (\hat{\chi}_{\ell_{0}}\Psi)\,(\{h_{\ell}\}_{\ell\in E})
    \,=\,
    \chi_{\frac{1}{2}}(h_{\ell_{0}})\,\Psi\,(\{h_{\ell}\}_{\ell\in E})\,.
  \end{array}
\end{equation}
If the acting loop doesn't belong to the initial subset $E$, we use the cylindrical consistency equivalence relation and we embed both the new loop and the initial loops in a larger graph, say $E\cup\{\ell_{0}\}$,
\begin{equation}
  \begin{array}{l}
    \ell_{0}\notin E\,,
    \quad
    \Psi\in\mathcal{H}_{E}\,,
    \quad
    \hat{\chi}_{\ell_{0}}\,\Psi\in\mathcal{H}_{E\cup\{\ell_{0}\}}\,, \\
    \qquad
    (\hat{\chi}_{\ell_{0}}\Psi)\,(\{h_{\ell}\}_{\ell\in E})
    \,=\,
    \chi_{\frac{1}{2}}(h_{\ell_{0}})\,\Psi\,(\{h_{\ell}\}_{\ell\in E})\,,
  \end{array}
\end{equation}
with the holonomy operator $\hat{\chi}_{\ell_{0}}$ acting as a creation operator, creating a new loop and curvature excitation.
Since the $\mathrm{SU}(2)$ character $\chi_{\frac{1}{2}}$ is real and bounded by two, $|\chi_{\frac{1}{2}}|\le 2$, we can check that the holonomy operators $\hat{\chi}_{\ell}$ are Hermitian, bounded and thus essentially self-adjoint.

The holonomy operator $\hat{\chi}_{\ell}$  is Hermitian and has a component acting as a creation operator. It must have an annihilation counterpart. The best way to see this explicitly is to write its action on proper states, consistently removing the zero-modes. Indeed, if a loop carries a spin $\frac{1}{2}$, then it gets partly annihilated by the holonomy operator:
\begin{equation}
  \ell_{0}\in E\,,
  \quad
  \Psi\in\mathcal{H}_{E}^{0}\,,
  \quad
  \hat{\chi}_{\ell_{0}}\,\Psi\in\mathcal{H}_{E}^{0}\oplus\mathcal{H}_{E\setminus\{\ell_{0}\}}^{0}\,, \nonumber
\end{equation}
\begin{equation}
  \begin{array}{rcl}
    %\qquad
    (\hat{\chi}_{\ell_{0}}\Psi)\,(\{h_{\ell}\}_{\ell\in E})
    &\,=\,&
    \underset{\in\,\mathcal{H}_{E}^{0}}{\underbrace{\Bigg{[}\chi_{\frac{1}{2}}(h_{\ell_{0}})\Psi\,(\{h_{\ell}\}_{\ell\in E})
          -\int\mathrm{d}h_{\ell_{0}}\chi_{\frac{1}{2}}(h_{\ell_{0}})\Psi\,(\{h_{\ell}\}_{\ell\in E})\Bigg{]}}} \\
    &+&\underset{\in\,\mathcal{H}_{E\setminus\{\ell_{0}\}}^{0}}
             {\underbrace{\Bigg{[}\int\mathrm{d}h_{\ell_{0}}\chi_{\frac{1}{2}}(h_{\ell_{0}})\Psi\,(\{h_{\ell}\}_{\ell\in E})\Bigg{]}}}
             \,.
  \end{array}
\end{equation}
This way, it is clear that the holonomy operator $\hat{\chi}_{\ell_{0}}$ creates transition adding and removing one loop. This proper state decomposition of the holonomy operator will become essential when defining it on the Fock space of symmetrized loopy spin networks in the next chapter.

\medskip

In this chapter, we have introduced the concept of loopy spin network. They are generalization of spin networks with the extra possibility of having local excitations at vertices of the graph. This framework is particularly well-suited in the context of coarse-graining and is especially designed for the application of coarse-graining by gauge-fixing as advertised in chapter \ref{ch:GaugeFixCG}. In this chapter, we had a first survey of the structure of loopy spin networks, by defining them properly (using projective limits techniques) and by considering various operators on them.

But, in the context of coarse-graining, we should hope for the definition of a clever cut-off between scales. This is what we will consider in the next chapter: in the context of loopy spin networks, what cut-offs are natural for the description of local excitations? 

% The goal is to have a quasi-systematic way of reducing any spin network into a coarse-grained one but keeping the gauge invariance. We defined it precisely making use of the already well-used projective limits. In practice, we should now try and find a way to define a systematic coarse-graining step using these spaces and applying them to find the holonomies advocated in our work on hyperbolic geometry.

%But, for now, we will concentrate on these new structures. Indeed, in the context of coarse-graining we still have an aweful lot of information on these caorse-grained spin network. Studying the symmetry reduction is the first step we will consider in the next chapter.
    
%\textbf{TODO:} definition of a first space (non symmetric), operators on it, projective limit in the redux, discussion of operators geom meaning, comparison with Fock space

%*****************************************
%*****************************************
%*****************************************
%*****************************************
%*****************************************

%*****************************************
\chapter{Bosonics loops and various cut-offs} \label{ch:Bosons}
%*****************************************

\inspiquote{You should always waste time when you don't have any. Time is not the boss of you. Rule 408.}{The Doctor}

In the previous chapter, we introduced \textit{loopy spin networks} as a natural structure appearing when coarse-graining spin networks which are the natural states of \ac{LQG}. But we have not considered the issue of introducing \textit{cut-offs}. To be clear, in this context, cut-offs are not necessarily linked to energy scales or even distance scales but consist in a truncation of the degrees of freedom down to the relevant one for some description. In this chapter, we will consider two natural cut-offs of the previous construction: imposing bosonic statistics on the loops, making them indistinguishable, and removing them, keeping only the closure defect. Note that this chapter is mainly technical, focusing on the subtleties of such a procedure and consist in an exploration of the natural data and operators on such spaces.

The first idea of bosonics loops leads to a symmetrized space. The space of loops will naturally gt a structure close to that of Fock spaces (with a few subtleties) and will therefore describe a background skeleton graph with indistinguishable little loops living at its vertices. From the perspective of coarse-graining, the little loops represent curvature excitations within the bounded region coarse-grained to a single vertex. This symmetrization can be understood as follows: as we do not know the unfolding tree of a given vertex, we should not privileged one loop over the other. It should be more or less impossible to distinguish between loops corresponding to different spatial position: incoming energy at the vertex would then equally excite any of those loops, irrespective to their a priori different localization on the internal subgraph that we coarse-grained. This is the rationale behind the symmetrization. A second idea will be explored by the end of this chapter, it is that of \textit{tagged} spin networks. The idea is to remove as much data as we can (from a kinematical perspective) and only keep the closure defect in order to preserve gauge-invariance.

The main difficulty in this chapter won't be the definition of the spaces themselves. In fact, it will be the definition of various operators, in particular symmetric operators well-suited for these spaces, that will be the focus making the present chapter quite technical. We will define natural symmetrized holonomy operators as well as flux operators. Though, the question will be much simpler, we will also study natural operators for tagged spin networks. The space of bosonic loops will still have subtleties of its own. The difficulty reside in the compatibility of the symmetrization with the cylindrical consistency. Indeed, a little loop carrying a spin-0 is considering as a non-existing loop, and vice-versa. We must therefore be careful in our symmetrization process not to include these non-existent loops which will, by nature, always outnumber the non-zero spin excitations. We would have to update the definition of the symmetrization to take this new fact into account. Here, we will show how to systematically subtract the 0-modes components of the loopy spin network states, symmetrize over non-trivial little loops and define an appropriate holonomy operator acting on symmetrized states. A resulting subtlety is that we will be led to distinguish three components of the holonomy operator, that respectively conserves the number of loops, acts as a creation operator adding one little loop or as an annihilation operator removing a loop. The parallel with a Fock space will therefore be quite apparent.

This chapter is taken from our published work \cite{Charles:2016xwc}. In the first three sections, we will concentrate on bosonic loops defining in succession: the space itself, (symmetrized) holonomy operators and flux operators. In the final section, we will present tagged spin network and the associated operators.

%, and introduce two different operators.

%Again, we are developping these loopy spin networks for coarse-graining. The main goal is therefore to define curved vertices that retain some information, but not everything, about the coarse-grained region. But so far, the loopy spin networks we have defined have a huge flaw in this regard: they have exactly the same degrees of freedom as the full graph regarding the $SU(2)$ variables. This is manifest in at list one major way: loops are distinguishable. In the coarse-graining perspective, it is not natural: we are looking at bounded regions from a distance. Therefore, if we were to take coarse-graining seriously, we should not be able to distinguish precisely which loops is excited or not. We should see that a loop is excited (creating curvature) but not which one. Mathematically, there is a natural way to describe this: it is to consider undistinguishable loops.
%
%This is what we develop in this section. We impose statistics on the loops. We should consider that the canonical loops behave as bosons. We leave the exploration of other statistics to future work. Note also that the choice of canonical loops here might be selected by the details of the coarse-grained region. We also postpone the precise study of how these canonical loops are selected to future work.

\section{Symmetrization, proper}

We would like to define symmetrized loopy intertwiner states in $\cHl$. A direct way would be to work directly on states with an arbitrary number of loops. We would use an extension of the finite symmetry groups $S_{n}$ to the group of permutations of integers which only act non-trivially on a finite subset:
  $$
  S_\infty = \{f : \mathbb{N} \rightarrow \mathbb{N}, f \textrm{ bijective and }\exists n \in \mathbb{N}, \forall m>n, f(m)=m\}\,.
  $$
  We would use the canonical action of $S_\infty$  on the Hilbert spaces of loopy spin networks $\cH_{E}$ with finite number of loops:
  $$
  \sigma : \cH_{E} \rightarrow \cH_{\sigma(E)}\,,
  \qquad
  (\sigma \triangleright f)(\{h_{e_i}\}) = f(\{h_{\sigma^{-1}(e_i)}\})\,.
  $$
This action is compatible with the cylindrical consistency conditions and naturally extends to the projective limit.
However, requiring invariance of states $|\Psi\rangle \in\cHl$ under permutations, $\forall \sigma \in S_\infty,~\sigma \triangleright |\Psi\rangle = |\Psi\rangle$, only provides non-normalizable states. This forces us to work on the dual space to define symmetrized states and creates unnecessary technicalities for our present purpose. So we follow a more constructive approach and work with finite number of loops, symmetrize and then allow for  varying number of loops.

We start from the definition of the loopy states in terms of proper states, $\cHl=\bigoplus_{E\in\mathcal{P}_{<\infty}(\mathbb{N})}\cH_{E}^{0}$. This decomposition has removed all spin-0 and avoids all of the redundancies due to the cylindrical consistency. We can now symmetrize the states. For each number of loops $N$, we consider gauge-invariant wave-functions, symmetric under the exchange of the $N$ loops and such that no loop carries a vanishing spin. The full symmetrized Hilbert space $\cHs$ will then be the direct sum over $N$ of all the finite symmetrized states.

Let us realize this programme explicitly. We start with the Hilbert space $\cHs_{N}$ of wave-functions, $f\in L^{2}(\SU(2)^{\times N})$, gauge-invariant and  symmetrized  on $N$ loops:
%\be
%f\in\cHs_{N}\,,%\qquad
%\left|\begin{array}{l}
%\forall k \in \SU(2),\quad
%f(h_1,...,h_N) = f(kh_1k^{-1},...,kh_N k^{-1})\,,\\
%\forall \sigma \in S_N, \quad
%f(h_1,...,h_N) = f(h_{\sigma(1)},...,h_{\sigma(N)})\,.
%\end{array}\right.
%\ee
\beq
%f\in\cHs_{N}\,,%\qquad
\forall k \in \SU(2),&&
f(h_1,...,h_N) = f(kh_1k^{-1},...,kh_N k^{-1})\,,\nn\\
\forall \sigma \in S_N, &&
f(h_1,...,h_N) = f(h_{\sigma(1)},...,h_{\sigma(N)})\,.
\eeq
We define the subspace of proper states, removing the 0 mode:
\be
\cH_{N}^{0}=
\Bigg{\{}
f\in\cHs_{N}
\,:\,
\int \mathrm{d}h_{1}\,f(h_1,...,h_N) =0
\Bigg{\}}\,.
\ee
We only need to impose one integration condition, since the function is invariant under permutation of its arguments.
We have a simplified version of the decomposition onto proper states given in lemma \ref{proper}:
%%%%%
\begin{lemma}
  \label{propersym}
  The Hilbert space of symmetrized states on $N$ loops decomposes as a direct sum of the Hilbert spaces of proper symmetrized states on at most $N$ loops:
  \be
  \cHs_{N}=\bigoplus_{n=0}^{N}\cH_{n}^{0}\,.
  \ee
  This isomorphism is realized through a combinatorial transform of the wave-functions:
  \be
  \label{symresum}
  \forall f\in\cHs_{N}\,,\quad
  f=\sum_{n=0}^{N}\,\,
  \sum_{1\le i_{1}<..<i_{n}\le N}f_{n}
  \big{(}
  h_{i_{1}},..,h_{i_{n}}
  \big{)}
  \,,\quad
  %f=\sum_{n=0}^{N}f_{n}\,,\quad
  f_{n}\in\cH_{n}^{0}
  \ee
  \be
  \begin{array}{l}
  f_{n}(h_{1},..,h_{n}) 
  \,=\, 
  \sum_{m=0}^{n}
  (-1)^{n-m}
  \int\prod_{i=m+1}^{n}\mathrm{d}k_{i}\,\prod_{i=n+1}^{N}\mathrm{d}g_{i}\, \\
  \qquad \times\sum_{1\le i_{1}<..<i_{m}\le n}
  f\big{(}
  h_{i_{1}},..,h_{i_{m}},k_{m+1},..,k_{n-m},g_{n+1},..,g_{N}
  \big{)}\,.
  \end{array}
  \ee
  The scalar product is given by the integration with respect to the Haar measure. The integral condition (absence of 0-mode) for the proper states implies that two proper states with different support are immediately orthogonal:
  \be
  \label{scalarN}
  \begin{array}{rcl}
  \forall f,\tf\,\in\cHs_{N}\,,\quad
  \la f|\tf\ra_{N}
  & \,=\, &
  \int \prod_{i=1}^{N}\mathrm{d}h_{i}\,\overline{f(h_{1},..,h_{N})}\,\tf(h_{1},..,h_{N}) \\
  & \,=\, &
  \sum_{n=0}^{N} \binomial{N}{n}\,\la f_{n}|\tf_{n}\ra_{n}\,.
  \end{array}
  \ee
\end{lemma}
%%%%%
In the resummation formula \eqref{symresum} above, the sum over labels  $1\le i_{1}<..<i_{n}\le N$ corresponds to the sum over all subsets with $n$ elements -or $n$-uplets- among the first $N$ integers $\{1,..,N\}$. And the injection of the proper state Hilbert space $\cH_{n}^{0}$ in the larger symmetrized space $\cHs_{N}$ requires this sum over all possible choices of $n$-uplets. This leads to the binomial coefficient in the scalar product formula \eqref{scalarN}. This is a clear remnant of having distinguishable loops. Once the little loops are assumed to be bosonic and fully indistinguishable, there is no reason to distinguish a state $f_{n}(h_{a_{1}},..,h_{a_{n}})$ from $f_{n}(h_{b_{1}},..,h_{b_{n}})$ with different choice of $n$-uplets.

Therefore, to define bosonic states in the projective limit $N\arr \infty$, we will keep the decomposition as a direct sum of vector spaces $\cHs_{N}=\bigoplus_{n=0}^{N}\cH_{n}^{0}$ defining the tower of symmetrized states, but we will modify the scalar product to remove its dependence on $N$ and make it compatible with the projective limit:
\be
\label{bosonicN}
\la f|\tf\ra_{N}^{\mathrm{bosonic}}
\,=\,
\sum_{n=0}^{N} \la f_{n}|\tf_{n}\ra_{n}\,.
\ee
This is achieved by simply including the symmetrizing factor in the definition of the injection $I_{N,N+1}:\,\cHs_{N}\hookrightarrow\cHs_{N+1}$ of wave-functions of $N$ loops seen as wave-functions of $(N+1)$ loops:
\be
\begin{array}{l}
f\,\in\cHs_{N}
%f(h_{1},..,h_{N})\in\cHs_{N}
\,\mapsto\,
I_{N,N+1}f\,\in\cHs_{N+1}\,, \\
\qquad
\big{(}I_{N,N+1}f\big{)}(h_{1},..,h_{N+1})
\,=\,
\f1{N+1}\,
\sum_{i=1}^{N+1}f(h_{1},..,\widehat{h_{i}},..,h_{N+1})\,,
\end{array}
\ee
where the element $\widehat{h_{i}}$ means that we omit it from the list of arguments. This generalizes to injections $\cHs_{N}\hookrightarrow\cHs_{N+p}$ using the binomial coefficients:
\be
\begin{array}{c}
f\,\in\cHs_{N}
%f(h_{1},..,h_{N})\in\cHs_{N}
\,\mapsto\,
I_{N,N+p}f\,\in\cHs_{N+p}\,,
\big{(}I_{N,N+p}f\big{)}(h_{1},..,h_{N+p})
\,=\, \\
\binomial{N+p}{N}^{-1}\,
\sum_{1\le i_{1}<..<i_{N}\le N+p}f(h_{i_{1}},..,h_{i_{N}})\,.
\end{array}
\ee
These factors compensate the binomial factors from the scalar product formula \eqref{scalarN}.
As we will see a little bit further, this scalar product $\la f|\tf\ra_{N}^{\mathrm{bosonic}}$  on symmetric states is the one which makes the holonomy operator(s) Hermitian.
Then we can define the Fock space of loopy spin networks with bosonic little loop excitations.
\begin{definition}
  The full Fock space of symmetrized loop states is defined as the projective limit of the Hilbert spaces $\cHs_{N}$, endowed with the bosonic scalar product \eqref{bosonicN}, which amounts to the direct sum of the spaces of proper states:
  \be
  \cHs\,\equiv\,\bigoplus_{N\in\N}\cH_{N}^{0}\,,\qquad
  \forall f,\tf\,\in\cHs\,,\quad
  \la f|\tf\ra
  \,=\,
  \sum_{N=0}^{\infty} \la f_{N}|\tf_{N}\ra\,.
  \ee
\end{definition}
This describes bosonic excitations of the holonomy at each vertex of the base graph for the loopy spin network states.
This Fock space of little loops at a vertex have states for an arbitrary number of indistinguishable loops, that can be created and annihilated, each of them carrying a spin $j_{\ell}\in\N/2$ encoding the corresponding excitation of the geometry (area quanta). The spin carried by a loop is similar to the momentum carried by a particle. One must nevertheless keep in mind two differences with the usual Fock space construction used in standard quantum field theory:
\begin{itemize}

\item {\textit{0-modes are pure gauge:}}
  First, we have implemented explicitly the cylindrical consistency requirement in the definition of the Fock space of loopy spin networks. A little loop carrying a spin-0 is identified to a vanishing excitation, i.e. a non-existing loop, so we have systematically removed them using proper states. This is similar to removing particle states with 0-momentum.

\item {\textit{Non-trivial intertwiner structure:}}
  Second, for a given number of loops carrying some given spins, the loopy spin network state still contains more information: the state requires the data of an intertwiner linking all these little loops together (and to the external legs of the vertex). Each time we create a loop, the intertwiner space at the vertex is further enlarged. This extra structure implies that factorized states do not constitute a basis of the Fock space of loopy intertwiners.

\end{itemize}
After describing factorized states below, we will define the holonomy operators acting on the Fock space of symmetrized states and show how they shift the number of loops and become the basic creation and annihilation operators.

It is interesting to check how factorized state, with no correlations between the loops, get decomposed onto proper states. Let us consider a integrable function $\vphi$ on $\SU(2)$. We assume it to be invariant under conjugation, so that it can be decomposed over the $\SU(2)$-characters for all spins:
\be
\forall h,g\,\in\SU(2)\,,\quad
\vphi(h)=\vphi(ghg^{-1})\,,
\qquad
\vphi(h)=\sum_{j\in\f\N2}\vphi_{j}\chi_{j}(h)\,.
\ee
We consider the $N$-loop symmetric  state $\vphi^{\otimes N}$ and check its proper state decomposition by the combinatorial formula given above in the lemma \ref{propersym}:
\beq
\vphi^{\otimes N}_{0}&=&\left(\int\vphi\right)^{N}=\vphi_{0}^{N}\,, \\
%\qquad
\vphi^{\otimes N}_{1}(h)&=&
%\left(\int\vphi\right)^{N-1}\,\vphi(h)-\left(\int\vphi\right)^{N}
%=
\vphi_{0}^{N-1}
\big{[}
  \vphi(h)-\vphi_{0}
  \big{]}
\,,\nn\\
\vphi^{\otimes N}_{2}(h_{1},h_{2})&=&
\vphi_{0}^{N-2}
\big{[}
  \vphi(h_{1})\vphi(h_{2})
  -\vphi_{0}\vphi(h_{1})
  -\vphi_{0}\vphi(h_{2})
  +\vphi_{0}^{2}
  \big{]}\nn \\
&=&
\vphi_{0}^{N-2}
\big{[}
  \vphi(h_{1})-\vphi_{0}
  \big{]}
\big{[}
  \vphi(h_{2})-\vphi_{0}
  \big{]}
\,,\nn\\
\vphi^{\otimes N}_{n}(h_{1},..,h_{n})&=&
\vphi_{0}^{N-n}
%\big{[}
\,\sum_{m=0}^{n}
(-1)^{n-m}\sum_{1\le i_{1}<..<i_{m}\le n}
\vphi_{0}^{n-m}\,\vphi(h_{i_{1}})..\vphi(h_{i_{m}}) \nn\\
&=&
\vphi_{0}^{N-n}\prod_{i=1}^{n}
\big{[}
  \vphi(h_{i})-\vphi_{0}
  \big{]}
\,.\nn
%\big{]}
\eeq
We can check the scalar product formula \eqref{scalarN}:
$$
\begin{array}{rcl}
\Bigg{(}\int |\vphi|^2\Bigg{)}^{N}
&=&
\Bigg{(} |\vphi_{0}|^2+\int\big{|}\vphi- \vphi_{0}\big{|}^2\Bigg{)}^{N} \\
&=&
\sum_{n}^{N}\binomial{N}{n}|\vphi|_{0}^{2(N-n)}\,\Bigg{(}\int\big{|}\vphi- \vphi_{0}\big{|}^{2}\Bigg{)}^n \\
&=&
\sum_{n}^{N}\binomial{N}{n} \int \big{|}\vphi_{n}^{\otimes N}\big{|}^2
\end{array}
$$

First, we notice that the proper state projections are still factorized. We are merely consistently removing the spin-0 component from all the loops, without creating any correlation during  the process. Second, if we normalize the one-loop wave-function $\vphi_{0}=\int \vphi=1$, then the projections of the factorized state do not depend anymore on the number of loops $N$ and we can take the projective limit. We can define a factorized state $\vphi^{\otimes \infty}$ with support on an infinite number of loops by taking the limit $N\arr\infty$. We define its components, dropping the useless $\infty$ label:
\be
\vphi=1+\widetilde{\vphi}\,,\quad
\int\widetilde{\vphi}=0\,,\qquad
\vphi_{0}=1\,,
\quad
\vphi_{n}=\widetilde{\vphi}^{\otimes n}\,.
\ee
We can for instance apply this to the $\delta$-distribution and define the flat holonomy state in our Fock space of symmetrized states:
\be
\label{delta-def}
\widetilde{\delta}=\delta-1
=\sum_{j\ne 0}(2j+1)\chi_{j}
\,,\qquad
\delta_{0}=1\,,\quad
\delta_{n}=\widetilde{\delta}^{\otimes n}\,.
\ee

\section{Holonomy operators on the symmetrized space}

Now that we have defined the Fock space of loopy intertwiners, we would like to have the basic  operators creating and annihilating loop excitations. This is naturally achieved by the (one-loop) holonomy operator. We start, as in the case of distinguishable loops, with multiplying wave-functions by the spin-$\f12$ character $\chi(h_{\ell})$ applied to the group element $h_{\ell}$ living on a little loop $\ell$. Then we will to distinguish three cases: the loop $\ell$ does not belong the existing loops and the operator creates a new loop, or the loop $\ell$ is already excited, in which case it can act on a spin-$\f12$ excitation and actually annihilate the loop, or the operator will generically act on all other spin excitations by simple multiplication. This leads us to defining three components of the holonomy operator $\hat{\chi}$ acting on symmetrized states:
\begin{definition}
  \label{AAB-def}
  We define three operators $A,\tA,B$ acting on the Fock space of symmetrized loopy intertwiners $\cHs$. They act on an arbitrary state $(f_{N})_{N\in\N}$ as:
  \be
  (Af)_{N}(h_{1},..,h_{N})
  \,=\,
  \int \mathrm{d}k\,\chi_{\f12}(k)\,f_{N+1}(h_{1},..,h_{N},k)\,,
  \ee
  \be
  \begin{array}{l}
  (Bf)_{0}=2f_{0}\,,
  \quad
  \forall N>0\,,\,\,
  (Bf)_{N}(h_{1},..,h_{N})
  \,=\, \\
  \qquad \f1N\sum_{i=1}^{N}\Bigg{[}
    \chi_{\f12}(h_{i})f_{N}(h_{1},..,h_{N})
    -\int \mathrm{d}k_{i}\,\chi_{\f12}(k_{i})\,f_{N}(h_{1},..,k_{i},..,h_{N})
    \Bigg{]}
  \end{array}
  \ee
  %\be
  %(\tA f)_{0}=0\,,
  %\quad
  %\forall N>0\,,\,\,
  %(\tA f)_{N}(h_{1},..,h_{N})
  %\,=\,
  %\f1{N}\sum_{i=1}^{N}\chi_{\f12}(h_{i})f_{N-1}(\{h_{j}\}_{1\le j\le N,j\ne i})
  %\ee
  \be
  \begin{array}{l}
  (\tA f)_{0}=0\,,
  \quad
  \forall N>0\,,\,\,
  (\tA f)_{N}(h_{1},..,h_{N})
  \,=\, \\
  \qquad \f1{N}\sum_{i=1}^{N}\chi_{\f12}(h_{i})f_{N-1}(h_{1},..,\widehat{h_{i}},..,h_{N})
  \end{array}
  \ee
  The operator $B$ is the usual action of the holonomy operator by multiplication by the character up to the subtraction of the resulting spin-0 component.
  The operator $A$ is the annihilation operator, removing one loop, while the operator $\tA$ creates a new loop. We have the following relations on $\cHs$:
  \be
  \tA=A^{\dagger}\,,\qquad
  B=B^{\dagger}
  \,.
  \ee
  Finally the one-loop holonomy operator for spin $\f12$  is defined as the sum of these three components and is self-adjoint:
  \be
  \hchi_{\f12}\,\equiv\,\f12\big{(}A+\tA+B\big{)}\,.
  \ee
\end{definition}
The convention $(Bf)_{0}=2f_{0}$ follows the logic that the 0-component $f_{0}$, with no loop, represents by default a flat holonomy and thus should be multiplied by $\chi_{\f12}(\id)=2$.
Here is the proof for the Hermicity relations:
\begin{proof}
  We compare the action of $A$ and $\tA$:
  $$
  \la \phi|A\psi\ra
  \,=\,
  \sum_{N\in\N}
  \int [\mathrm{d}h_{i}]_{i=1}^{N}\mathrm{d}k\,
  \chi_{\f12}(k)\,\overline{\phi_{N}(h_{1},..,h_{N})}\,\psi_{N+1}(h_{1},..,h_{N},k)\,,
  $$
  $$
  \begin{array}{l}
  \la \tA \phi|\psi\ra
  \,=\, \\
  \quad \sum_{N>0}
  \int [\mathrm{d}h_{i}]_{i=1}^{N}
  \f1N\sum_{i=1}^{N}\chi_{\f12}(h_{i})\,
  \overline{\phi_{N-1}(h_{1},..,\widehat{h_{i}},..,h_{N})}\,
  \psi_{N}(h_{1},..,h_{N})\,.
  \end{array}
  $$
  We shift the sum over $N$ in $\la \tA \phi\,|\,\psi\ra$ and we use the invariance of $\psi_{N}$ under permutation of its arguments to conclude that these two expressions coincides, $\la \phi|A\psi\ra=\la \tA \phi|\psi\ra$. As for the operator $B$, we compute:
  \be
  \begin{array}{c}
  \la \phi|B\psi\ra
  = \\
  2\overline{\phi_{0}}\,\psi_{0}
  \,+\,
  \sum_{N>0}
  \int [\mathrm{d}h_{i}]_{i=1}^{N}
  \f1N\sum_{i}^{N}
  \chi_{\f12}(h_{i})\overline{\phi_{N}(h_{1},..,h_{N})}\,\psi_{N}(h_{1},..,h_{N})\nn\\
  -\,\int [\mathrm{d}h_{i}]_{i=1}^{N}
  \f1N\sum_{i}^{N}
  \int \mathrm{d}k_{i}\,\chi_{\f12}(k_{i})\overline{\phi_{N}(h_{1},..,h_{i},..,h_{N})}\,\psi_{N}(h_{1},..,k_{i},..,h_{N})\,.\nn
  \end{array}
  \ee
  The last term vanishes due to the absence of 0-mode, $\int \mathrm{d}h_{i}\,\phi_{N}=0$. This ensures that $\la \phi|B\psi\ra=\la B\phi|\psi\ra$ and thus $B$ is a Hermitian operator.
\end{proof}

To ensure that the operators $A$ and $B$ are well-defined and that the holonomy operator $\hchi_{\f12}$ is  self-adjoint, it is enough to check that it is bounded. And we show below that it is indeed bounded by 2 as in the usual framework.
\begin{lemma}
  %Boundedness
  The two parts of the holonomy operators are both bounded by 2, that is for all states $\phi\in\cHs$, we have the two inequalities:
  \be
  |\la \phi|\,(A+\tA)\,|\phi\ra|\,\le\,2\la \phi|\phi\ra
  \,,\qquad
  |\la \phi|B|\phi\ra|\,\le\,2\la \phi|\phi\ra\,.
  \ee
  This ensures that they are both self-adjoint. The holonomy operator $\hchi_{\f12}$ is then also bounded by 2 and self-adjoint.
\end{lemma}

\begin{proof}
  Let us start with the operator $B$. The analysis is simpler since it doesn't shift the number of loops:
  $$
  \begin{array}{l}
  \la \phi|B|\phi\ra=2|\phi_{0}|^{2}+\sum_{N>0}B_{N}\,,\\
  B_{N}=
  \f1N\sum_{i}^{N}
  \int [\mathrm{d}h_{i}]_{i=1}^{N}\chi_{\f12}(h_{i})\overline{\phi_{N}}(h_{1},..,h_{N})\,\phi_{N}(h_{1},..,h_{N})\,.
  \end{array}
  $$
  The extra term in the action of $B$ on the state $\phi$ vanishes as earlier due to the integral condition on proper states, $\int \mathrm{d}h_{i}\,\phi_{N}=0$ for all $i$'s. Since the character $\chi_{\f12}$ is bounded by 2, it is direct to conclude:
  $$
  \begin{array}{l}
  |B_{N}|\le \f2N\sum_{i}^{N}\int |\phi_{N}|^{2}= 2\int |\phi_{N}|^{2}\,, \\
  \la \phi|B|\phi\ra
  \le 2|\phi_{0}|^{2}+2\sum_{N>0}\int |\phi_{N}|^{2}=2\la \phi|\phi\ra\,.
  \end{array}
  $$
  We can proceed similarly with the operator $A+\tA$:
  $$
  \begin{array}{l}
  \la \phi|\,(A+\tA)\,|\phi\ra=
  \la \phi|A|\phi\ra+\la \phi|A^{\dagger}|\phi\ra=\la \phi|A|\phi\ra+\overline{\la \phi|A|\phi\ra}\,,\\
  \qquad |\la \phi|\,(A+\tA)\,|\phi\ra|\le 2 |\la \phi|A|\phi\ra|\,,
  \end{array}
  $$
  \be
  \begin{array}{l}
  \la \phi|A|\phi\ra=\sum_{N}A_{N}\,,\\
  A_{N}=\int \prod_{i=1}^{N}\mathrm{d}h_{i}\,\mathrm{d}k\,
  \chi_{\f12}(k)\,\overline{\phi_{N}}(h_{1},..,h_{N})\,\phi_{N+1}(h_{1},..,h_{N},k)
  \end{array}
  \ee
  As long as the components $\phi_{N}$'s are square-integrable, we can use the Cauchy-Schwarz inequality to bound these integrals:
  $$
  \begin{array}{rcl}
  \left|
  A_{N}
  \right|
  &\le&
  \sqrt{\int  \prod_{i}^{N}\mathrm{d}h_{i}\,\mathrm{d}k\,\chi_{\f12}(k)^{2}\,\big{|}\phi_{N}(h_{1},..,h_{N})\big{|}^{2}}\, \\
  &\times& \sqrt{\int  \prod_{i}^{N}\mathrm{d}h_{i}\,\mathrm{d}k\,\big{|}\phi_{N+1}(h_{1},..,h_{N},k)\big{|}^{2}}\,.
  \end{array}
  $$
  We use that the $\SU(2)$ character is normalized, $\int \chi_{\f12}^{2}=1$,  and then apply the inequality bounding a product $ab\le (a^{2}+b^{2})/2$:
  \be
  \begin{array}{rcl}
  \left|
  A_{N}
  \right|
  &\le&
  \f12{\int  \prod_{i}^{N}\mathrm{d}h_{i}\,\big{|}\phi_{N}(h_{1},..,h_{N})\big{|}^{2}} \\
  &+&\f12{\int  \prod_{i}^{N}\mathrm{d}h_{i}\,\mathrm{d}k\,\big{|}\phi_{N+1}(h_{1},..,h_{N},k)\big{|}^{2}} \\
  &=& \f12\,\Big{[}\la \phi_{N}|\phi_{N}\ra+\la \phi_{N+1}|\phi_{N+1}\ra\Big{]}\,.
  \end{array}
  \ee
  Summing over $N\in\N$, this allows us to conclude that $|\la \phi|A|\phi\ra|\le \la\phi|\phi\ra-\f12|\phi_{0}|^{2}\le \la\phi|\phi\ra$ and thus reproduces the expected bound $|\la \phi|\,(A+\tA)\,|\phi\ra|^{2}\le 2\la\phi|\phi\ra$.
\end{proof}

Although we consider the holonomy operator $\hchi_{\f12}$ to be the averaged sum of the two self-adjoint components $(A+A^{\dagger})$ and $B$, each of these is a legitimate operator in itself. We could push this logic further and state that we have defined two different holonomy operators,
on the one hand, a holonomy operator $(A+A^{\dagger})$ that acts as a ladder operator, creating and annihilating loops, and on the other hand, a holonomy operator $B$ which acts as spin shifts on existing loops (i.e. modifies the area quanta carried by each loop).

The important consistency check, which will be essential for the analysis of the BF theory dynamics, is that the flat state is an eigenvector of the one-loop holonomy operator:
\begin{prop}
  \label{flat-prop1}
  The flat state $\delta$, defined in \eqref{delta-def} by its proper state projections, $\delta_{0}=1$ and $\delta_{N}=(\delta-1)^{\otimes N}$ for $N\ge 1$, is an eigenvector of the spin-$\f12$ one-loop holonomy operator  $\hchi_{\f12}$ with the highest eigenvalue on $\cHs$:
  \be
  \hchi_{\f12}\,
  %\delta
  |\delta\ra
  \,=\, 2\,|\delta\ra
  \,.
  \ee
  This distributional flat state is also an eigenvector of the loop annihilation operator $A$ and of the loop creation operator $(B+A^{\dagger})$:
  \be
  A|\delta\ra
  \,=\,
  (B+A^{\dagger})|\delta\ra
  \,=\,
  2\,|\delta\ra
  \,.
  \ee
\end{prop}
\begin{proof}
  We compute the action of the three parts of the holonomy operators acting on the flat state defined explicitly as
  $$
  \delta_{0}=1\,,\quad
  \delta_{N}(h_{1},..,h_{N})
  =\prod_{i}^{N}\big{[}\delta(h_{i})-1\big{]}\,.
  $$
  For the no-loop component, we get:
  $$
  (A\delta)_{0}=\int \mathrm{d}k\,\chi_{\f12}(k)\,\big{[}\delta(k)-1\big{]}=2
  \,,\quad
  (A^{\dagger}\delta)_{0}=0
  \,,\quad
  (B\delta)_{0}=2\,,
  $$
  while we compute for all other components:
  \beq
  (A\delta)_{N}(h_{1},..,h_{N})&=& 2\prod_{i}^{N}\big{[}\delta(h_{i})-1\big{]}=2\delta_{N}(h_{1},..,h_{N})
  \\
  (A^{\dagger}\delta)_{N}(h_{1},..,h_{N})&=&
  \f1N\sum_{i}^{N}\chi_{\f12}(h_{i})\,\prod_{\ell\ne i}^{N}\big{[}\delta(h_{\ell})-1\big{]}
  \\
  (B\delta)_{N}(h_{1},..,h_{N})&=&
  2\prod_{i}^{N}\big{[}\delta(h_{i})-1\big{]} \\
  &-&\f1N\sum_{i}^{N}\chi_{\f12}(h_{i})\,\prod_{\ell\ne i}^{N}\big{[}\delta(h_{\ell})-1\big{]}\nn
  \eeq
  Adding these three contributions, we get as expected for all number of loops $(\hchi_{\f12}\,\delta)_{N}=2\delta_{N}$.

\end{proof}

\smallskip

Since we have three operators built in the holonomy operator, it is natural to investigate their commutation algebra. It actually involves higher spin operators. We generalize the definition of the operators $A$, $A^{\dagger}$ and $B$ to arbitrary spins: one simply replaces in their definition \ref{AAB-def} the character in the fundamental representation $\chi_{\f12}$ by the higher spin character $\chi_{j}$ for any $j\in\N^{*}/2$, thus producing new operators $A_{j}$ annihilating a loop excitation of spin $j$,  $A^{\dagger}_{j}$ creating a new loop carrying a spin $j$ and $B_{j}$ acting with a spin $j$ excitation on an existing loop.

Then acting on a an arbitrary state $f$, we get:
$$
(ABf)_{N}=\f{N}{N+1}(BAf)_{N}+\f1{N+1}(A_{1}f)_{N}
\,,
$$
$$
(BA^{\dagger}f)_{N}=\f{N-1}{N}(A^{\dagger}Bf)_{N}+\f1{N}(A_{1}^{\dagger}f)_{N}\,,
%\,,\quad
$$
$$
(AA^{\dagger}f)_{N}
=\f{N}{N+1}(A^{\dagger}Af)_{N}+\f1{N+1}f_{N}\,.
$$
Remembering that the number of loops $N$ is not constant on the Fock space of loopy intertwiners and should be treated as an operator $\hat{N}$, these translate into commutation relations, being careful about the operator ordering:
\be
\hN B = B\hN
\,,\quad
(\hN+1)A=A\hN
\,,\quad
\hN A^{\dagger}= A^{\dagger}(\hN+1)\,,
\ee
\be
\label{commAB1}
\begin{array}{c}
AB\hN=\hN BA+A_{1}
\,,\quad
\hN B A^{\dagger}=A^{\dagger}B\hN+A_{1}^{\dagger}
\,,\\
A\hN A^{\dagger}=\id+A^{\dagger}A\hN=\id+\hN A^{\dagger}A
\,.
\end{array}
\ee
These generalize to the whole tower of higher spin operators, for all spins $a,b\in\f{\N^{*}}2$:
\be
\label{commAB2}
\begin{array}{c}
A_{a}B_{b}\hN=\hN B_{b}A_{a}+A_{a\otimes b}
\,,\quad
\hN B_{b} A_{a}^{\dagger}=A_{a}^{\dagger}B_{b}\hN+A_{a\otimes b}^{\dagger}
\,,\\
A_{a}\hN A_{b}^{\dagger}=\delta_{ab}\id+A_{b}^{\dagger}A_{a}\hN=\delta_{ab}\id+\hN A_{b}^{\dagger}A_{a}
\,,
\end{array}
\ee
where we use the (natural) convention of the tensor product of spins for the annihilation operator:
\be
A_{a\otimes b}\,\equiv\,
\sum_{c=|a-b|}^{a+b}A_{c}\,.
\ee
We give the last commutation relation:
\be
\hN\,[B_{a},B_{b}]
\,=\,
A^{\dagger}_{b}A_{a}-A^{\dagger}_{a}A_{b}\,.
\ee

%Up to now, we have considered the one-loop holonomy operator for spin $\f12$. It is natural to investigate the definition of higher spin operators as well as multi-loop holonomy operators. Actually in the present indistinguishable little loops framework, these two questions are closely related.

We can combine these higher spin creation and annihilation operators to define a spin-$j$ holonomy operator $\hchi_{j}$ as the average sum of those operators as for the fundamental representation:
\be
\hchi_{j}=\f12\,\big{(}
A_{j}+A^{\dagger}_{j}+B_{j}
\big{)}\,.
\ee
This rather natural definition unfortunately doesn't ensure that the operators $\hchi_{j}$'s for different spins $j$'s commute with each other. Using the algebra computed above, the commutator of two holonomy operators $\hchi_{a}$ and $\hchi_{b}$ actually looks like a mess.
Nevertheless we can simplify the expressions by introducing suitable number of loops factors. Inserting the operator $\hN$ in the character, we find:
\be
\big{[}
  \hN\hchi_{a},\hN\hchi_{b}
  \big{]}
=\f{\hN}2\,(A^{\dagger}_{b}A_{a}-A^{\dagger}_{a}A_{b})
=\f{\hN^{2}}2\,[B_{a},B_{b}]\,.
\ee
This combination $\hN\hchi_{a}$ isn't Hermitian, but this can be easily remedied to by considering $\sqrt{\hN}\hchi_{a}\sqrt{\hN}$ instead. This commutator doesn't vanish, but we can easily find other combinations of the creation and annihilation operators that do:
\be
\big{[}
  \hN(B_{a}+A_{a}^{\dagger}-A_{a}),\hN(B_{b}+A_{b}^{\dagger}-A_{b})
  \big{]}
\,=\,
0
\,.
\ee
This suggests using the operators $A_{a}$ and $(B_{a}+A^{\dagger}_{a})$ as more fundamental as the holonomy operators. Although they are not Hermitian, the flat state is an eigenvector of both operators and we will exploit this fact in defining flatness constraints for BF theory in the  following section \ref{BFsym}.

\medskip

The other way to proceed to defining higher spin holonomy operators is to reproduce the classical algebra of the $\SU(2)$ characters. For instance, a spin-1 is obtained from the tensor product of two spin-$\f12$ representations:
$$
\chi_{1}(h)=\chi_{\f12}(h)^{2}-1\,.
$$
We propose to promote these relations to the quantum level:
\be
\hchi_{1}^{\,\mathrm{full}}
\,\equiv\,
\hchi_{\f12}^{\,2}-1
=
\f14\Big{[}
  A^{2}+AB+BA+AA^{\dagger}+A^{\dagger}A+B^{2}+A^{\dagger}B+BA^{\dagger}+A^{\dagger}{}^{2}
  \Big{]}-1
\,.
\ee
This new spin-1 holonomy operator is already a multi-loop operator: it has a component $A^{2}$ annihilating two loops and its adjoint component $A^{\dagger}{}^{2}$ creating two loops, and so on.
%We will not work out the detail of the action of this holonomy operator.
We then define all the other spin-$j$ holonomy operators by recursion as polynomials of the fundamental $\hchi_{\f12}$ operator:
\be
\chi_{\f12}^{4}=\chi_{2}+3\chi_{1}+2
\,\,\Rightarrow\,\,
\hchi_{2}^{\,\mathrm{full}}
\,\equiv\,
\hchi_{\f12}^{\,4}-3\hchi_{\f12}^{\,2}+1\,,
\quad\dots
\ee
and so on with $\hchi_{j}^{\,\mathrm{full}}$ constructed from $\hchi_{\f12}^{2j}$ and smaller powers.
This construction clearly ensures that all the holonomy operators commute with each other. This method closely intertwines the definition of higher spin operators with  multi-loop holonomies. These multi-loop operators create spins $\f12$ (and then higher spins too) excitations on several loops at once.

\section{Flux operators}

To conclude the exploration of the basic \ac{LQG} operators, we should also deal with the symmetrized flux operators (and scalar products) with the $\su(2)$ generators acting as derivations on the wave-functions, and check their commutation relations with our new holonomy operators.
The flux and grasping operators are especially important since they allow to explore the intertwiner structure at the vertices. Indeed, acting with one-loop holonomy operators will only create decoupled loops at the vertex, while a generic intertwiner will couple them. So, even though we postpone the detailed analysis of the action of flux operators on loopy spin network to future investigation, we discuss below multi-loop holonomy operators that allow for coupled loops and thus explore the space of (loopy) intertwiners at the vertex.

For instance, considering two loops with group elements $h_{1}$ and $h_{2}$, we would like to excite the overall holonomy instead of the two independent holonomies, that is act with $\chi(h_{1}h_{2})$ instead of $\chi(h_{1})\chi(h_{2})$. Proceeding similarly to the one-loop holonomy operator, we define a two-loop holonomy operator $\hchi_{\f12}^{(2)}$, which creates and annihilates pairs of coupled loops:

\begin{definition}
  \label{C-def}
  We define the following five operators $C_{-2,-1,0,+1,+2}$ on the Fock space of symmetrized loopy intertwiners $\cHs$. They act on an arbitrary state $(f_{N})_{N\in\N}$ as:

  \be
  (C_{-2}f)_{N}(h_{1},..,h_{N})
  \,=\,
  \int \mathrm{d}k \mathrm{d}k\,\chi_{\f12}(k\tk)\,f_{N+2}(h_{1},h_{N},k,\tk)
  \ee
  \be
  \begin{array}{rcl}
  (C_{-1}f)_{N}(h_{1},..,h_{N})
  &=&
  \f1N\sum_{i}^{N}
  \Bigg{[}
    \int \mathrm{d}k\,\chi_{\f12}(kh_{i})\,f_{N+1}(h_{1},..,h_{N},k) \\
    &-&\int \mathrm{d}k\, \mathrm{d}k_{i}\,\chi_{\f12}(kk_{i})\,f_{N+1}(h_{1},..,k_{i},..,h_{N},k)
    \Bigg{]}
  \end{array}
  \ee
  \beq
  (C_{0}f)_{N}(h_{1},..,h_{N})
  &=&
  \f2{N(N-1)}\sum_{i<j}^{N}\Bigg{[}
    \chi_{\f12}(h_{i}h_{j})f_{N}(h_{1},..,h_{N}) \Bigg{.} \nn\\
    &+& \int \mathrm{d}k_{i} \mathrm{d}k_{j}\,\chi_{\f12}(k_{i}k_{j})\,f_{N}(h_{1},..,k_{i},..,k_{j},..,h_{N}) \nn\\
    &-&\int \mathrm{d}k_{i}\,\chi_{\f12}(k_{i}h_{j})\,f_{N}(h_{1},..,k_{i},..,h_{N}) \nn\\
    &-&\int \mathrm{d}k_{i}\,\chi_{\f12}(h_{i}k_{j})\,f_{N}(h_{1},..,k_{j},..,h_{N})
    \Bigg{]}
  \eeq
  \be
  \begin{array}{rcl}
  (C_{+1}f)_{N}(h_{1},..,h_{N})
  &=&
  \f1{N(N-1)}\sum_{i\ne j}^{N}
  \Bigg{[}
    \chi_{\f12}(h_{i}h_{j})
    f_{N-1}(h_{1},..,\widehat{h_{i}},..,h_{N}) \\
    &-&\int \mathrm{d}k_{j}\,\chi_{\f12}(h_{i}k_{j})
    f_{N-1}(h_{1},..,\widehat{h_{i}},..,k_{j},..,h_{N})
    \Bigg{]}
  \end{array}
  \ee
  \be
  (C_{+2}f)_{N}(h_{1},..,h_{N})
  \,=\,
  \f2{N(N-1)}\sum_{i<j}^{N}
  \chi_{\f12}(h_{i}h_{j})f_{N-2}(h_{1},..,\widehat{h_{i}},..,\widehat{h_{j}},..,h_{N})
  \ee
  We complete this definition with the ``initial conditions'' for $N=0$ and $N=1$:
  \be
  (C_{-1}f)_{0}=\int \chi_{\f12}f_{1}
  \,,\quad
  (C_{0}f)_{0}=2f_{0}
  \,,\quad
  (C_{+1}f)_{0}=(C_{+2}f)_{0}=0\,,
  \ee
  \be
  \begin{array}{c}
  (C_{0}f)_{1}(h)=\chi_{\f12}(h)f_{1}(h)-\int \chi_{\f12}f_{1}
  \,,\quad
  (C_{+1}f)_{1}(h)=\chi_{\f12}(h)f_{0}
  \,, \\
  (C_{+2}f)_{0}=0\,.
  \end{array}
  \ee
  They satisfy the Hermicity relations:
  \be
  C_{-2}=C_{2}^{\dagger}\,,\,\,
  C_{-1}=C_{1}^{\dagger}\,,\,\,
  C_{0}=C_{0}^{\dagger}
  \,,
  \ee
  and the bounds for the $L^{2}$-norm:
  %\be
  %|\la f|\,(C_{-2}+C_{+2})\,|f\ra|\,\le 2 \la f|f\ra\,,\quad
  %|\la f|\,(C_{-1}+C_{+1})\,|f\ra|\,\le 2 \la f|f\ra\,,\quad
  %|\la f|\,C_{0}\,|f\ra|\,\le 2 \la f|f\ra\,,
  %\ee
  \be
  ||C_{-2}+C_{+2}||_{2}\le 2\,,\quad
  ||C_{-1}+C_{+1}||_{2}\le 2\,,\quad
  ||C_{0}||_{2}\le 2\,.
  \ee
  Finally the two-loop holonomy operator $\hchi_{\f12}^{(2)}$ for spin $\f12$  is defined as the sum of these five components:
  \be
  \hchi_{\f12}^{(2)}\,\equiv\,
  \f14
  \big{(}C_{-2}+2C_{-1}+C_{0}+2C_{+1}+C_{+2}\big{)}\,.
  \ee
  This operator is essentially self-adjoint and bounded by 2.
\end{definition}

So the spectrum of the two-loop holonomy operator is once again bounded by 2. An important consistency check is that this bound is saturated by the flat state. The proof is a straightforward computation, with special care to the initial conditions for $N=0$ and $N=1$.
\begin{prop}
  The flat state $\delta$, defined by $\delta_{0}=1$ and $\delta_{N\ge 1}=(\delta-1)^{\otimes N}$, is an eigenvector of the spin-$\f12$ two-loop holonomy operator  $\hchi_{\f12}^{(2)}$ with the highest eigenvalue on $\cHs$:
  \be
  \hchi_{\f12}^{(2)}\,
  %\delta
  |\delta\ra
  \,=\, 2\,|\delta\ra
  \,.
  \ee
\end{prop}

We could then similarly define an operator $\chi(h_{i}h_{j}^{-1})$ with a loop reversal or multi-loop operators $\chi(h_{i_{1}}..h_{i_{n}})$ taking care of properly symmetrizing the group elements. We can also generalize our construction replacing the spin-$\f12$ character by an arbitrary spin $j$ and define the spin-$j$ two-loop holonomy operator $\hchi^{(2)}_{j}$, and so on for more loops.
We will not go into these details, although we do not foresee any obstacle (beside the inflation of indices and sums).

\section{A word on tagged spin network}

In our published work \cite{Charles:2016xwc}, we considered another possibility: \textit{tagged spin network}. After folded spin networks, which retains the internal combinatorial structure inside coarse-grained regions, and loopy spin networks, which keep local curvature excitations as little loops attached to the vertices of the base graph and which we have explored in great details in the previous sections, this last structure only keeps the \textit{curvature defect}. We will spend this last section briefly describing them.

When integrating out the connection group elements inside a bounded region, as discussed in \cite{Livine:2013gna,Dittrich:2014wpa,Bahr:2015bra}, and coarse-graining the region to a single vertex, we naturally break the local gauge invariance at the resulting coarse-grained vertex. This  also happens as soon as we introduce fermionic matter fields, which act as sources for \ac{LQG}'s Gauss law and thus create non-trivial closure defects.
At the classical level, this is reflected by a non-vanishing closure defect: the sum of the flux vectors around the effective vertex does not vanish anymore and is actually balanced by the internal fluxes living on the loops living inside the bounded region and carrying non-trivial holonomies. Overall, the gauge invariance is restored if we take into account the internal degrees of freedom of the region however, once we have coarse-grained it, the breaking of the gauge invariance reflects the geometry excitations which have developed in the region's bulk and which we have traced out.
At the quantum level, the closure defect becomes a tag, attached to each vertex, as drawn on fig.\ref{tagged}. This internal degree of freedom is defined as an extra spin coupling to the actual spin living on the links and edges attached to the effective vertex and connecting the coarse-grained region to its exterior. This tag allows to relax the gauge invariance in a controlled way.

Mathematically, we are thus led to consider the whole space of non-gauge-invariant cylindrical functionals of the connection on a given fixed background graph $\Gamma$. The tagged spin networks will provide a basis of that space, with the tags record how much the local gauge-invariance is broken: when the tags vanish, we recover the usual gauge-invariant spin network basis states. This allows to account for graph changing dynamics in an effective manner. Even though the graph changes and might get more complex as the geometry evolves, we keep on coarse-graining the state projecting it onto the fixed base graph (chosen by the observer), then the internal degrees of freedom and non-trivial curvature developed inside the coarse-grained regions gets translated into excitations of the effective tag degree of freedom attached to the base graph vertices.

The space of tagged spin networks is naturally quite simple and we believe it offers a useful framework for the study of coarse-graining of the \ac{LQG} dynamics.

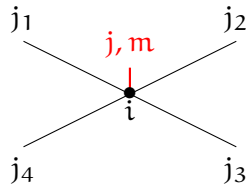
\begin{figure}[h!]

  \centering

  \begin{tikzpicture}[scale=0.7]
    \coordinate(O1) at (0,0);

    \draw (O1) -- ++(-2,1) node[above] {$j_1$};
    \draw (O1) -- ++(2,1) node[above] {$j_2$};
    \draw (O1) -- ++(2,-1) node[below] {$j_3$};
    \draw (O1) -- ++(-2,-1) node[below] {$j_4$};
    \draw[thick,red] (O1) to ++(0,0.5) node[above]{$j,m$};

    \draw (O1) node {$\bullet$} ++(0,-0.3) node{$i$};

  \end{tikzpicture}

  \caption{We consider tagged vertices: a vertex with the additional tag corresponding to a closure defect. The representations $j_{1},..,j_{4}$ living on the graph edges linked to the vertex do not form an intertwiner on their own, but they recouple to the spin $j$ defining the tag.
    \label{tagged}}

\end{figure}

We consider the space $\cH_{\Gamma}^{tag}$ of non-gauge-invariant wave-functions on the (oriented and connected) graph $\Gamma$. This is simply the space of $L^{2}$ functions on $\SU(2)^{\times E}$, where $E$ is the number of edges or links of $\Gamma$, with no further assumption. Considering such a function, we can project onto the usual space of gauge-invariant states by group averaging:
\be
\psi(\{g_{e}\}_{e\in\Gamma})
\,\in\cH_{\Gamma}^{tag}
\quad\mapsto\quad
\psi_{0}(\{g_{e}\}_{e\in\Gamma})
\,=\,
\int_{\SU(2)^{V}}\dd h_{v}
\psi(\{h_{s(e)}^{-1}g_{e}h_{t(e)}\}_{e\in\Gamma})
\,\,\in\cH_{\Gamma}\,.
\ee
We can generalize this projection to non-trivial recouplings at every vertex and get an exact decomposition of the full non-invariant state:
\be
\begin{array}{l}
\psi
\,=\,
\sum_{\{J_{v}\}\in\f\N2}
\prod_{v}(2J_{v}+1)\,P_{\{J_{v}\}}\psi\,, \\
\quad
P_{\{J_{v}\}}\psi\,(\{g_{e}\}_{e\in\Gamma})
\,=\,
\int \dd h_{v}\,
\prod_{v}\chi_{J_{v}}(h_{v})\,
\psi(\{h_{s(e)}^{-1}g_{e}h_{t(e)}\}_{e\in\Gamma})\,.
\end{array}
\ee
The spin $J_{v}$ is the tag living at the vertex $v$ and provides a measure of how much gauge-invariance is relaxed at that vertex. It is the variable conjugate to the group averaging variable $h_{v}$.
Following this logic, we can make all states in   $\cH_{\Gamma}^{tag}$ gauge-invariant by adding  $h_{v}$ as an actual argument of the wave-function. This provides a isomorphism between  $\cH_{\Gamma}^{tag}=L^{2}(\SU(2)^{\times E})$ and $\cH_{\Gamma}^{ext}=L^{2}(\SU(2)^{\times E}\times\SU(2)^{\times V}/\SU(2)^{\times V})$ where $ext$ stands for ``extended'':
\be
\begin{array}{l}
\Psi(\{g_{e},h_{v}\}_{e,v\in\Gamma})
\,=\,
\Psi(\{k_{s(e)}^{-1}g_{e}k_{t(e)},k_{v}^{-1}h_{v}\}_{e,v\in\Gamma}) \\
\quad\mapsto\quad
\psi(\{g_{e}\}_{e\in\Gamma})
\,=\,
\Psi(\{g_{e},h_{v}=\id\}_{e,v\in\Gamma})\,.
\end{array}
\ee
We define tagged spin networks as basis states for $\cH_{\Gamma}^{ext}$ thus providing through this gauge-fixing map a basis for generic non-gauge-invariant states. These generalizations of spin networks are labeled by spins $j_{e}$ on every edge $e$, the tag spin $J_{v}$ an magnetic momentum $M_{v}$ at every vertex, as well as an intertwiner $\cI_{v}$ recoupling at each vertex between the tag  and the spins on the edges attached to that vertex:
the incoming and outgoing edges attached to the vertex $v$:
$$
\cI_{v}:\cV^{J_{v}}\otimes\bigotimes_{e|s(e)=v}\cV^{j_{e}}\longrightarrow\bigotimes_{e|t(e)=v}\cV^{j_{e}}\,,
$$
\be
\label{tagbasis}
\begin{array}{c}
\Psi_{\{j_{e},J_{v},M_{v},\cI_{v}\}}
(\{g_{e},h_{v}\})
\,\equiv\, \\
\prod_{v}D^{J_{v}}_{m_{v}M_{v}}(h_{v})\,
\prod_{e}\la j_{e}m_{e}^{s}|\,g_{e}\,|j_{e}m_{e}^{t}\ra\, \\
\times \prod_{v}\la \otimes_{e|t(e)=v}j_{e}m_{e}^{t}|\,\cI_{v}\,|J_{v}m_{v}\otimes_{e|s(e)=v}j_{e}m_{e}^{s}\ra\,,
\end{array}
\ee
with an implicit over the magnetic momenta $m_{e}^{v}$ and $m_{v}$. In simple words, we work with spin network on graphs with an extra open edge at every vertex. The spins carried by those open edges are the tags.

\medskip

The whole question is the physical interpretation of these tags, which we added to the usual spin network states. It is mathematically clear how the closure defects arise from coarse-graining and that the tags reflect non-trivial holonomies around the loops of the subgraph within the coarse-graining regions. The next challenge would be to show that they can be related to some physical notions of (quasi-)local energy density or mass (see e.g. \cite{Yang:2008th} for a definition of the quasi-local energy operator in \ac{LQG}).

Let us show how starting from a loopy spin network and tracing out the little loops attached to the vertices leads naturally to a reduced density matrix defined in terms of tagged spin networks. So we consider a gauge-invariant loopy state defined on the base graph $\Gamma$ with a certain number of loops $n_{v}$ attached to each vertex $v$:
$$
\phi(\{g_{e},h^{v}_{\ell}\})
\,=\,
\phi(\{k_{s(e)}^{-1}g_{e}k_{t(e)},k_{v}^{-1}h^{v}_{\ell}k_{v}\})\,\,
\forall k_{v}\in\SU(2)^{\times V}\,,
$$
where the group elements $h^{v}_{\ell}$ live on the little loops $\ell$ attached to the vertex $v$, and we integrate out the loops:
\be
\rho(\{g_{e},\tg_{e}\}_{e\in\Gamma})
\,=\,
\int \prod_{v}\prod_{\ell=1}^{n_{v}}\dd h_{\ell}^{v}\,\,
\overline{\phi(\{g_{e},h^{v}_{\ell}\})}\,
\phi(\{\tg_{e},h^{v}_{\ell}\})\,.
\ee
Let us compute the reduced density matrix using the natural loopy spin network basis. We focus on the little loops attached to single vertex, say $v_{0}$, and drop the index $v$ from the little loop group elements for the sake of simplicity.
We consider the loopy states defined by basis intertwiners defined by two intertwiners, one recoupling the spins living on the edges linked to the vertex $v_{0}$ and one recoupling the little loops attached to that vertex, glued through an intermediate spin $J_{v_{0}}$, as drawn on fig.\ref{fig:intermediatespin2}:
\beq
\Phi_{\{j_{e},J_{v_{0}},i_{v},j_{\ell},\tilde{i}_{v_{0}}\}}
(\{g_{e},h_{\ell}\})
&=&
\prod_{e}\la j_{e}m_{e}^{s}|\,g_{e}\,|j_{e}m_{e}^{t}\ra\,
\prod_{\ell}\la j_{\ell}m_{\ell}^{s}|\,h_{\ell}\,|j_{\ell}m_{\ell}^{t}\ra\, \nn\\
&\times& \prod_{v\ne v_{0}}\la \otimes_{e|t(e)=v}j_{e}m_{e}^{t}|\,i_{v\,}|\otimes_{e|s(e)=v}j_{e}m_{e}^{s}\ra\,
\nn\\
&&
\la \otimes_{e|t(e)=v_{0}}j_{e}m_{e}^{t}|\,i_{v_{0}}\,|J_{v_{0}}M_{v_{0}}\otimes_{e|s(e)=v}j_{e}m_{e}^{s}\ra\,
\nn\\
%&&
&\times& \la J_{v_{0}}M_{v_{0}} \otimes_{\ell}j_{\ell}m_{\ell}^{t}|\,\tilde{i}_{v_{0}}\,|\otimes_{\ell}j_{\ell}m_{\ell}^{s}\ra
\,,
%\nn
\eeq
with an implicit sum over all magnetic moment labels.
We have assumed, as announced, that only the vertex $v_{0}$ has little loops attached to it, so all other vertices are thought as having a vanishing intermediate spin $J_{v\ne v_{0}}=0$.
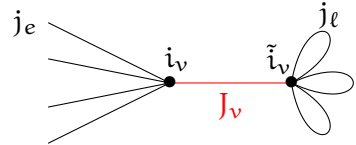
\begin{figure}
  \centering
  \begin{tikzpicture}[scale=0.8]
    %\draw[->,>=stealth,very thick] (3,0) -- (5,0);

    \coordinate(O2) at (8,0);
    \coordinate(O3) at (10,0);

    \draw (O2) -- ++(-2,1) node[left]{$j_{e}$};
    \draw (O2) -- ++(-2,0.4);
    \draw (O2) -- ++(-2,-0.4);
    \draw (O2) -- ++(-2,-1);

    \draw[in=-25,out=25,scale=3] (O3) to[loop] (O3);
    \draw[in=30,out=80,scale=3] (O3)  to[loop] node[pos=0.6,right,above]{$j_{\ell}$} (O3);
    \draw[in=-80,out=-30,scale=3] (O3) to[loop] (O3);

    \draw[red] (O2) -- (O3) node[midway,below]{$J_{v}$};
    \draw (O2) node {$\bullet$} ++(0.12,0.4) node{$i_v$};
    \draw (O3) node {$\bullet$} ++(-0.2,0.4) node{$\tilde{i}_v$};

  \end{tikzpicture}

  \caption{We introduce the intermediate spin basis for the vertices of loopy spin networks: intertwiners will decompose into two intertwiners, the first one $i_{v}$ recoupling the spins living on the base graph edges and the other $\tilde{i}_{v}$ recoupling the spins living on the little loops attached to the vertex, which are linked together by the intermediate spin $J_{v}$. When we coarse-grain by tracing over the little loops, the only remaining information is this intermediate spin $J_{v}$, which becomes the tag measuring the closure defect at the vertex. It is the remnant of the curvature fluctuations and internal geometry within the vertex.}
  \label{fig:intermediatespin2}

\end{figure}

A loopy state will decompose onto that basis,
$$|\phi\ra=\phi_{\{j_{e},J_{v_{0}},i_{v},j_{\ell},\tilde{i}_{v_{0}}\}}
\,|\Phi_{\{j_{e},J_{v_{0}},i_{v},j_{\ell},\tilde{i}_{v_{0}}\}}\ra$$
and we easily compute the resulting reduced density matrix using the orthonormality of the Wigner matrices with respect to the Haar measure on $\SU(2)$ and find that it naturally decompose onto the tagged spin network basis introduced $\Psi_{\{j_{e},J_{v},M_{v},\cI_{v}\}}$ above in \eqref{tagbasis}:
\be
\begin{array}{rcl}
%|\phi\ra=\phi_{\{j_{e},J_{v_{0}},i_{v},j_{\ell},\tilde{i}_{v_{0}}\}}
%\,|\Phi_{\{j_{e},J_{v_{0}},i_{v},j_{\ell},\tilde{i}_{v_{0}}\}}\ra
%\,,
%\quad
\rho
&=&
\tr_{\{h_{\ell}\}}|\phi\ra\la\phi| \\
&=&
\Big{(}
\sum_{\{m^{s,t}_{\ell}\}}
\big{|}
\la J_{v_{0}}M_{v_{0}} \otimes_{\ell}j_{\ell}m_{\ell}^{t}|\,\tilde{i}_{v_{0}}\,|\otimes_{\ell}j_{\ell}m_{\ell}^{s}\ra
\big{|}^{2}
\Big{)}
\, \\
&\times& 
|\Psi_{\{j_{e},J_{v_{0}},M_{v_{0}},i_{v}\}}\ra\la\Psi_{\{j_{e},J_{v_{0}},M_{v_{0}},i_{v}\}}|\,.
\end{array}
\ee
Thus the intermediate spins of the loopy spin networks, which recouple between the base graph edges and the little loop excitations, become the tags of the tagged spin network basis after tracing out the holonomies living on the little loops. This concludes the coarse-graining of the geometry of a bounded region to a single vertex plus one extra degree of freedom -the tag- registering the  excitations of geometry and curvature within that region's bulk.

\medskip

We have reviewed different coarse-graining structures. Because of their definition, they lead to natural cut-off in the context of coarse-graining. We illustrated this in particular in the case of tagged spin networks which can be implemented quite naturally with partial traces. There is no reason to think \textit{a priori} that such a coarse-graining would be exact and, in fact, as \ac{GR} is fairly non-linear, there are all the reasons to think that the coarse-graining will be approximate. But still, this illustrates several possibilities.

It should be noted here that the tagged spin network, though natural as a further coarse-graining possibility, is quite remote from the coarse-graining scheme we have been developing so far. In the case of loops, we still retain something \textit{like} the surface holonomies we long to find. In the case of tagged spin network, no such possibility seem to arise spontaneously. Still, they might at least be useful in yet another coarse-graining scheme.

Still, we can think of loopy spin network and tagged spin network as sort of dual in our thinking of coarse-graining. If the loopy spin networks might be more instrumental in the coarse-graining process itself, tagged spin network might be more useful for defining the coarse-grained structure itself. Indeed, we already hinted at the possibility of dual spin networks. But if this structure is to be trusted, curvature must be carried at the vertices. This curvature might be encoded by a tag or something similar.

Anyway, before tackling the full theory and its coarse-graining, we should try and explore the use of our coarse-graining structures with a simpler theory, namely BF theory. This is what will be tackled in the next chapter.

%\textbf{TODO:} further spaces for coarse-graining, symmetric space, operators (+ algebra), tagged spin network

%*****************************************
%*****************************************
%*****************************************
%*****************************************
%*****************************************

%*****************************************
\chapter{BF theory on the flower graph} \label{ch:BFflower}
%*****************************************

\inspiquote{There's something that doesn't make sense. Let's go and poke it with a stick.}{The Doctor}

Now that we have describe the whole kinematics of loopy spin networks, with distinguishable loops, we would like to tackle the issue of the dynamics and imposing the Hamiltonian constraints on the Hilbert space of loopy states $\cH^{\mathrm{loopy}}$. The final goal of our proposal is to write the Hamiltonian constraints of \ac{LQG} on  $\cH^{\mathrm{loopy}}$, such that it allows explicitly for local degrees of freedom,  study its renormalization group flow under the coarse-graining and extract its large scale or continuum limit. Such a program could be started in a simpler setting, like regular lattice, since these should not change much the renormalization flow.

There are however various difficulties. The first one is to find natural dynamics to write for the support graph \textit{and} for the loops. Indeed, such writing should be stable, at least in some approximation, under the coarse-graining flow. And because, \ac{GR} is highly non-linear, such a programme is not simple to start with. We expect correlations between different scales and we don't know how the dynamics of the support graph can evolve.

Therefore, we will instead describe the much simpler BF dynamics. BF theory can be considered as a consistency check for all attempts and methods to define of dynamics in (loop) quantum gravity. Moreover, once the dynamics of BF theory is properly implemented and well under control in a certain framework, one usually use it as a starting point for imposing the true gravity dynamics, with local degrees of freedom, relying on the reformulation of \ac{GR} as a BF theory with constraints. This is for instance the logic behind the construction of spinfoam models for a quantum gravity path integral \cite{Livine:2010zx,Perez:2012wv,Bianchi:2012nk}. It is therefore a good place to start studying full \ac{GR}.

But we are interesting in BF theory because of yet another feature with respect to coarse-graining: it has a trivial renormalization flow. Indeed the flatness constraint behaves very nicely under coarse-graining, as illustrated on fig.\ref{fig:Flat}~: considering a spin network graph, imposing the flatness of the connection on all small loops  guaranties that  larger loops will be flat too. This means that the dynamics for the support graph as well as the dynamics for the loops is quite simple to devise. In this chapter, we will concentrate on the loops are they are the addition of our framework, the dynamics of a fixed graph having been studied already \cite{Noui2005,Bonzom:2011hm,Bonzom:2011nv}.

\begin{figure}[h!]

  \centering

  \begin{tikzpicture}
    \coordinate(A) at (0,0);
    \coordinate(A1) at (0.2,-0.2);
    \coordinate(A2) at (0.4,-0.2);
    \coordinate(B) at (2,0);
    \coordinate(B1) at (1.8,-0.2);
    \coordinate(B2) at (2.2,-0.2);
    \coordinate(B3) at (2.4,-0.2);
    \coordinate(C) at (4,0);
    \coordinate(C1) at (3.8,-0.2);
    \coordinate(D) at (0,-2);
    \coordinate(D1) at (0.2,-1.8);
    \coordinate(D2) at (0.2,-0.6);
    \coordinate(D3) at (0.2,-0.4);
    \coordinate(E) at (2,-2);
    \coordinate(E1) at (1.8,-1.8);
    \coordinate(E2) at (2.2,-1.8);
    \coordinate(E3) at (2.2,-0.6);
    \coordinate(E4) at (2.2,-0.4);
    \coordinate(F) at (4,-2);
    \coordinate(F1) at (3.8,-1.8);
    \coordinate(O1) at (1,-1);
    \coordinate(O2) at (3,-1);

    \draw[dashed] (A) -- (B) -- (E) -- (D) -- (A);
    \draw[dashed] (B) -- (C) -- (F) -- (E);

    \draw (A) node {$\bullet$};
    \draw (B) node {$\bullet$};
    \draw (C) node {$\bullet$};
    \draw (D) node {$\bullet$};
    \draw (E) node {$\bullet$};
    \draw (F) node {$\bullet$};

    \draw[blue,rounded corners,thick] (A2) -- (B1) -- (E1) -- (D1) -- (D2);
    \draw[blue,thick,->,>=stealth] (D2) -- (D3);
    \draw[blue,rounded corners,thick] (B3) -- (C1) -- (F1) -- (E2) -- (E3);
    \draw[blue,thick,->,>=stealth] (E3) -- (E4);

    \draw[blue] (O1) node[scale=2] {\textbf{$\id$}};
    \draw[blue] (O2) node[scale=2] {\textbf{$\id$}};

    \draw[->,>=stealth,very thick] (5,-1) -- (7,-1);

    \coordinate(P) at (8,0);
    \coordinate(A0) at ($(P)+(0,0)$);
    \coordinate(A01) at ($(P)+(0.2,-0.2)$);
    \coordinate(A02) at ($(P)+(0.4,-0.2)$);
    \coordinate(B0) at ($(P)+(2,0)$);
    \coordinate(B01) at ($(P)+(1.8,-0.2)$);
    \coordinate(B02) at ($(P)+(2.2,-0.2)$);
    \coordinate(B03) at ($(P)+(2.4,-0.2)$);
    \coordinate(C0) at ($(P)+(4,0)$);
    \coordinate(C01) at ($(P)+(3.8,-0.2)$);
    \coordinate(D0) at ($(P)+(0,-2)$);
    \coordinate(D01) at ($(P)+(0.2,-1.8)$);
    \coordinate(D02) at ($(P)+(0.2,-0.6)$);
    \coordinate(D03) at ($(P)+(0.2,-0.4)$);
    \coordinate(E0) at ($(P)+(2,-2)$);
    \coordinate(E01) at ($(P)+(1.8,-1.8)$);
    \coordinate(E02) at ($(P)+(2.2,-1.8)$);
    \coordinate(E03) at ($(P)+(2.2,-0.6)$);
    \coordinate(E04) at ($(P)+(2.2,-0.4)$);
    \coordinate(F0) at ($(P)+(4,-2)$);
    \coordinate(F01) at ($(P)+(3.8,-1.8)$);
    \coordinate(O0) at ($(P)+(2,-1)$);

    \draw[dashed] (A0) -- (B0);
    \draw[dashed,gray] (B0) -- (E0);
    \draw[dashed] (E0) -- (D0) -- (A0);
    \draw[dashed] (B0) -- (C0) -- (F0) -- (E0);

    \draw (A0) node {$\bullet$};
    \draw (B0) node {$\bullet$};
    \draw (C0) node {$\bullet$};
    \draw (D0) node {$\bullet$};
    \draw (E0) node {$\bullet$};
    \draw (F0) node {$\bullet$};

    \draw[blue,rounded corners,thick] (A02) -- (C01) -- (F01) -- (D01) -- (D02);
    \draw[blue,thick,->,>=stealth] (D02) -- (D03);

    \draw[blue] (O0) node[scale=2] {\textbf{$\id$}};

  \end{tikzpicture}

  \caption{In BF theory, holonomies behave very nicely under coarse-graining. If each small loops is flat, large loops are flat too. In other words, the physical state of BF theory is a flat space, which is flat at all scales.}
  \label{fig:Flat}

\end{figure}
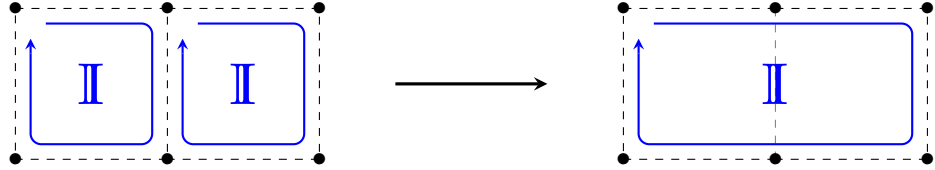

The content of this chapter is taken from our published work \cite{Charles:2016xwc}. It is organized as follows: we will first describe the natural (gauge-invariant) constraints in the classical theory and show that these constraints alone are not sufficient in the quantum theory. In the next two sections, we will introduce new constraints solving the problem in two different ways: either by insuring \textit{independent} gauge invariance or by insuring that correlations vanish. We will discuss the full constraint in another section to conclude with the interaction with the Fock structure of our space.

\section{Holonomy constraint on $\mathrm{SU}(2)$}

Considering the full space of loopy spin networks on some arbitrary graph $\Gamma$, we would like the BF Hamiltonian constraints to project onto the flat connection state(s), that is impose flatness around all the loops of the graph $\Gamma$ and also kill all the local excitations represented by the little loops at every vertex. Flatness around the loops of the background graph is the standard result for BF constraints. So here we will focus on the fate of the little loops, that we introduced. To this purpose, it suffices to focus on a single vertex, that is to work on the flower graph.

Considering the flower graph with arbitrary number of loops, as we have defined above, we introduce  the following set of constraints:
\begin{equation}
  \forall \ell \in \mathbb{N},
  \quad
  \big{(}\hat{\chi}_\ell-2\big{)} |\Psi\rangle \,=\, 0\,.
  %\hat{\chi}_\ell |\Psi\rangle = 2 |\Psi\rangle\,.
\end{equation}
We impose one constraint for every (possible) loop by imposing that the corresponding holonomy operator saturates its bound and projects on its highest eigenvalue. These constraints all commute with each other.
Let us underline the dual role of Hamiltonian constraints. As first class constraints, we need to solve them and identify their solution space, but they also generate gauge transformations and we need to gauge out their action. Here, the holonomy constraint operators both impose the flatness of the connection, but they also imply that the little loops are pure gauge, so that their action can change the number of loops to arbitrary values.
We will see below that these one-loop holonomy constraints are almost enough to fully constrain the theory to the single flat state on the flower graph.

\bigskip

Let us solve these constraints and consider a loop $\ell_{0}$ and its action of its holonomy operator $\hchi_{\ell_{0}}$ on a wave-function $\Psi\in\cH_{E}$ with support on the finite subset $E\subset\N$ of loops. A first case is when $\ell_{0}\in E$ belongs to the subset, in which case we have a simple functional equation on $\SU(2)^{E}$:
$$
(\hat{\chi}_{\ell_{0}}\Psi)\,(\{h_{\ell}\}_{\ell\in E})
\,=\,
\chi_{\frac{1}{2}}(h_{\ell_{0}})\Psi\,(\{h_{\ell}\}_{\ell\in E})
\,=\,
2\,\Psi\,(\{h_{\ell}\}_{\ell\in E})\,.
$$
The second case is when the considered loop $\ell_{0}\notin E$  doesn't belong to the subset. The holonomy operator $\hat{\chi}_{\ell_{0}}$ then creates a loop, making a transition from $\cH_{E}^{0}$ to the orthogonal space $\cH_{E\cup \{\ell_{0}\}}^{0}$. This illustrates that the flow generated by those Hamiltonian constraints can arbitrarily shift the number of loops and therefore the little loops become pure gauge at the dynamical level in BF theory. This also means that there is no solution to all holonomy constraints with support on a finite subset $E$ and a physical state must have support on all possible loops.

\medskip

To be rigorous, we need to go to the dual space  $(\cH^{\mathrm{loopy}})^*$ and solve the holonomy constraints on the space of distribution defined in the projective limit.  We are looking for a  family of distributions $\vphi_{E}$ on $\SU(2)^{E}$, that is continuous linear forms over smooth functions on  $\SU(2)^{E}$ (see appendix \ref{app:distribution} for a discussion of the definition of distributions over $\SU(2)$). The cylindrical consistency means that their evaluations on two cylindrically equivalent smooth functions must be equal:
$$
\begin{array}{l}
\forall E\subset \widetilde{E}\,,\,\,
f_{E}\sim f_{\widetilde{E}}
\\
\quad \Rightarrow
\vphi_{E}(f_{E})=\int_{\SU(2)^{E}}\vphi_{E}f_{E}
\,=\,
\int_{\SU(2)^{\widetilde{E}}}\vphi_{\widetilde{E}}f_{\widetilde{E}}=\vphi_{\widetilde{E}}(f_{\widetilde{E}})\,.
\end{array}
$$
Then the holonomy constraints read:
$$
\forall \ell\in\N\,,
\forall E \ni \ell\,,
\forall f_{E}\in{\cal C}^{\infty}_{\SU(2)^{E}}\,,
\,\,
\int_{\SU(2)^{E}}\ \vphi_{E} (\chi_{\ell}-2)f_{E}
\,=\,0\,,
$$
where we have considered by default that the loop $\ell$ belongs to the wave-function support $E$. Indeed, if $\ell$ didn't belong to $E$, then we could enlarge the subset $E$ to $E\cup \{\ell\}$ by cylindrical consistency and consider both the test function $f$ and the distribution $\vphi$ as living on that larger subset.
Our goal is to show that the unique solution to these equations is the flat state, i.e. that there exists $\lambda\in\C$ such that $\vphi_{E}=\lambda\,\delta^{\otimes E}$:
\be
\begin{array}{rcl}
\forall f_{E}\in{\cal C}^{\infty}_{\SU(2)^{E}}\,,
\vphi_{E}(f_{E})&=&\lambda \delta_{E}(f_{E}) \\
&=&
\lambda\int_{\SU(2)^{E}}\prod_{\ell\in E}\delta(h_{\ell}) f_{E}(\{h_{\ell}\}_{\ell\in E}) \\
&=&
\lambda f_{E}(\id,..,\id)\,.
\end{array}
\ee
Cylindrical consistency simply requires that the factor $\lambda$ does not depend on the subset $E$.
So we are led to solve the holonomy constrain on every finite subset $E$. Thus, let us consider the functional equation on $\SU(2)^{N}$:
\be
\forall 1\le\ell\le N\,,\,\,
\left(\hat{\chi}_\ell - 2\right)\varphi = 0\,,
\ee
where we drop the subset label $E$.

Let us start with the one-loop case and solve for distributions $\vphi$ on $\SU(2)$ the equation:
\be
\forall h\in\SU(2)\,,\,\,
\chi_{\f12}(h)\vphi(h)=2\vphi(h)\,.
\ee
Since the character $\chi_{\f12}$ is smooth and reaches its maximum value $2$ at a single point, the identity $\id$, it seems natural that the $\vphi$ must be a distribution peaked at the identity. We therefore expect that the only solution be the $\delta$-distribution on $\SU(2)$, $\vphi=\delta$. However, since the identity is actually an extremum of $\chi_{\f12}$ and that the first derivatives of the character thus vanishes at this point, this equation admit more solutions: the first derivatives of the $\delta$-distribution. This clearly came as a surprise for us.

Let us first assume that $\vphi$ is gauge-invariant, i.e. invariant under conjugation. Its Fourier decomposition on $\SU(2)$ involves only the characters in all spins:
$$
\vphi=\sum_{j\in\f\N2}\vphi_{j}\chi_{j}\,.
$$
As well known, the holonomy constraint leads to a recursion relation on the coefficients $\vphi_{j}$:
\be
\chi_{\f12}\chi_{j}=\chi_{j-\f12}+\chi_{j+\f12}
\quad\Rightarrow\qquad
2 \vphi_{0}=\vphi_{\f12}\,,\quad
2\vphi_{j\ge\f12}=\vphi_{j-\f12}+\vphi_{j+\f12}\,.
\ee
Once the initial condition $\vphi_{0}$ is fixed, these lead to a unique solution:
%\footnotemark:
\be
\vphi_{j}=(2j+1)\vphi_{0}\,,
\qquad
\vphi=\vphi_{0}\sum_{j}(2j+1)\chi_{j}=\vphi_{0}\delta\,.
\ee
%
%\footnotetext{
%Solving this type of functional equations, one must be very careful to work with well-defined distributions. For instance, let us consider solving for an arbitrary eigenvalue $\rho$ of the holonomy operator.
%The recursion relation becomes:
%$$
%\rho \vphi_{0}=\vphi_{\f12}\,,\quad
%\rho \vphi_{j\ge\f12}=\vphi_{j-\f12}+\vphi_{j+\f12}\,.
%$$
%Once we fix the initial condition $\vphi_{0}$, this recursive equation has a solution for every complex value $\rho\in\C$, but this does not systematically define a solution state, in $L^{2}$ or a distribution. The solution to the recursion is given in terms of the two solutions $\mu_{\pm}$ of the quadratic equation $\mu^{2}-\rho \mu+1=0$:
%$$
%\forall j\in\f\N2\,,\,\,
%\vphi_{j}=\f12(\mu_{+}^{2j}+\mu_{-}^{2j})\,.
%$$
%For $\rho=2$, the discriminant $(\rho^{2}-4)$ vanishes and this ansatz fails leads: instead of the power law, we get a linear growth $f_{j}=(2j+1)$, which leads back to the $\delta$-distribution peaked on the identity. For real values $|\rho|< 2$, the discriminant is negative and we get an oscillatory solution $\vphi_{j}=\cos 2j\theta$ with $\cos\theta =\rho$, which gives a $\delta$-distribution fixing the class angle of the group element $g$ to $\theta$. For other values, the resulting coefficients do not define a proper distribution. For instance, for $\rho=3$, the positive discriminant will leads to exponentially divergent coefficients $\vphi_{j}$ which are too divergent to define a distribution.
%}
%
When solving such functional equations in the Fourier basis, one should nevertheless be very careful to work with well-defined distributions. These are characterized by Fourier coefficients $\vphi^{j}$ growing at most polynomially with the spin $j$. This ensures that evaluations $\int f\vphi$ of the distribution $\vphi$ on smooth test functions $f$ are convergent series. The $\delta$-distribution is clearly a good solution. But, as an example, solving for eigenvectors of the holonomy operator associated to (real) eigenvalues (strictly) larger than 2 would lead to exponentially growing Fourier coefficients, which are too divergent to define a proper distribution. The interested reader will find more details in the appendix \ref{app:distribution}.

\medskip

On the space of  functions invariant under conjugation, everything works as expected. Let us now consider the general case dropping the requirement of gauge-invariance. The $\delta$-distribution is obviously still a solution:
\be
\forall f\in\cC^{\infty}_{\SU(2)}\,,
\,\,
\int f\, (\chi_{\f12}-2)\,\delta
=
(\chi_{\f12}(\id)-2)f(\id)
=0\,.
\ee
But, now the first derivatives of the $\delta$-distributions are also  solutions:
\be
\forall f\in\cC^{\infty}_{\SU(2)}\,,
\,\,
\int f\, (\chi_{\f12}-2)\,\pp_{x}\delta
=
\left.-(\pp_{x}\chi_{\f12})\,f-(\chi_{\f12}-2)\,f\right|_{\id}
=0\,,
\label{ppdelta}
\ee
where $x\in\R^{3}$ indicates the direction of the derivative and the derivatives of the character vanish at the identity since it is a extremum.
We remind the reader that the right-derivative $\pp_{x}^{R}$ on $\SU(2)$ is a anti-Hermitian operator ($i\pp$ is Hermitian) defined by the infinitesimal action of the $\su(2)$ generator $\vx\cdot\vJ$ (where the $\vJ$ in the fundamental spin-$\f12$ representation are simply half the Pauli matrices):
\be
\pp_{x}^{R}f(h)=
\lim_{\eps\arr 0}\f{f(h e^{i\eps \vx\cdot \vJ})-f(h)}{\eps}
=f(h\,  x)\,,\quad
\textrm{with}\quad x= \vx\cdot \vJ\,.
\ee
We usually differentiate along the three directions in $\R^{3}\sim\su(2)$ leading to the insertion of the generators $J_{a=1,2,3}$:
\be
\pp_{a}^{R}f(h)=if(h J_{a})\,,
\quad
\pp_{a}^{L}f(h)=if(J_{a}h)\,.
\ee
Acting on the $\delta$-distribution gives the following Fourier decomposition for its derivatives $\pp_{a}^{L}\delta=\pp_{a}^{R}\delta=\pp_{a}\delta$:
\be
\pp_{a}\delta(h)
=i\sum_{j} (2j+1)D^{j}_{nm}(J_{a})\,D^{j}_{mn}(h)\,,
\ee
where we use the Wigner matrices for the group element $h$ and the $\su(2)$ generators.

We can actually generate a whole tower of higher derivative solutions to the holonomy constraints. We simply need to identify the differential operators whose action on the spin-$\f12$ character vanishes at the identity. Thus, at second order, we get five new independent solutions given by the following operators:
\be
\pp_{1}\pp_{2}
\,,\,\,
\pp_{1}\pp_{3}
\,,\,\,
\pp_{2}\pp_{3}
\,,\,\,
(\pp_{1}\pp_{1}-\pp_{2}\pp_{2})
\,,\,\,
(\pp_{1}\pp_{1}-\pp_{3}\pp_{3})
\,,
\ee
%\be
%\vphi \,=\quad
%\pp_{1}\pp_{2}\delta
%\,,\,\,
%\pp_{1}\pp_{3}\delta
%\,,\,\,
%\pp_{2}\pp_{3}\delta
%\,,\,\,
%\pp_{1}\pp_{1}\delta-\pp_{2}\pp_{2}\delta
%\,,\,\,
%\pp_{1}\pp_{1}\delta-\pp_{3}\pp_{3}\delta
%\,.
%\ee
%
that is the $\pp_{a}\pp_{b}$ and $(\pp_{a}\pp_{a}-\pp_{b}\pp_{b})$ for $a\ne b$. Following this logic, we will get 7 new independent solutions at third order, and so on with $(2n+1)$ independent differential operators at order $n$, for a total of $(n+1)^{2}$ independent solutions to the holonomy constraints given by differential operators of order at most $n$ acting on the $\delta$-distribution.

Such as in the conjugation-invariant case, it is enlightening to switch to the Fourier decomposition and translate the holonomy constraint into a recursion relation on the Fourier coefficients. The difference is that we had one Fourier coefficient $\vphi^{j}$ for each spin $j$ in the gauge-invariant case while in the general case $\vphi^{j}$ is a $(2j+1)\times(2j+1)$ matrix. Implementing the recursion, we start from spin 0 and work the way up to higher spins. The problem is that the recursion relations determine only $(2j)^{2}$ matrix elements of  $\vphi^{j}$ in terms of the lower spins coefficients, leaving $(2j+1)^{2}-(2j)^{2}=(4j+1)$ matrix elements free to be specified as initial conditions. This leads to an infinite number of solutions to the recursion relations, which reproduces the tower of higher order derivative solutions.
The interested reader will find all of the details on the recursion relations in the appendix of our published work \cite{Charles:2016xwc}.

%%%%
\section{Introducing the Laplacian constraint on $\SU(2)$}
\label{derivativesolution1}
%%%%

If we work with a single loop, a single petal on the flower, then the wave-function is obviously gauge-invariant and we do not have to deal with these extra solutions to the holonomy constraint.
However, as soon as we add external legs attached to the vertex (linking the flower to other vertices in the graph) or add more loops, then we have to find a way to suppress those derivative solutions, in $\pp_{a}\delta$ and so on, which would lead to extra degrees of freedom as some kind of polarized flat states.

Since we want to ensure the full flatness of the holonomy, the most natural proposal is to constrain all the components of the group element living on the loop and not only its trace:
$$
\forall m,n=\pm\f12\,,\,\,
D^{\f12}_{mn}(h) \,\vphi(h)=\delta_{mn}\,\vphi(h)\,.
$$
One can indeed check, both from the differential calculus point of view or the recursion relations in Fourier space, that these equations admit the $\delta$-distribution as unique solutions. We can also go beyond multiplicative operators and insert some differential operators. Then supplementing the trace holonomy constraint with the other constraints $\chi_{\f12}\pp_{a}\,\vphi =2\pp_{a}\vphi$ for $a=1,2,3$ also ensures a unique flat solution. However, these constraints are not gauge-invariant: the constraint operators map wave-functions invariant under conjugation to non-invariant functions.

In order to keep gauge-invariant constraints, we go to the second derivatives and consider the Laplacian operator. Actually, we introduce the right-Laplacian $\Delta\equiv \sum_{a}\pp_{a}^{R}\pp_{a}^{R}$ and a mixed Laplacian operator $\tDelta\equiv \sum_{a}\pp_{a}^{L}\pp_{a}^{R}$ . We can see how $\Delta$ and $\tDelta$ differ through their action on the coupled character $\chi(h_{1}h_{2})$:
  $$
  \begin{array}{rcl}
  %\Delta_{h_{1}}\chi_{\f12}(h_{1}h_{2})
  \Delta_{1}\chi_{\f12}(h_{1}h_{2})&=&-\f14\chi_{\f12}(h_{1}\sigma_{a}\sigma_{a}h_{2})=-\f34\chi_{\f12}(h_{1}h_{2})
  \, \\
  \tDelta_{1}\chi_{\f12}(h_{1}h_{2})&=&-\f14\chi_{\f12}(\sigma_{a}h_{1}\sigma_{a}h_{2}) \\
  &=&-\f14\,\big{(}2\chi_{\f12}(h_{1})\chi_{\f12}(h_{2})-\chi_{\f12}(h_{1}h_{2})\big{)}\,,
  \end{array}
  $$
  which are of course equal at $h_{1}=\id$. We can now propose a new constraint:
\be
\Delta \vphi =\tDelta \vphi\,,
%\qquad
%H^{BF}
\ee
At the classical level, the differential operator $\pp_{a}$ represents the flux vector $X_{a}$: the right derivative represents the flux $\vec{X}^{s}$ at the source of the loop while the left derivative is the flux $\vec{X}^{t}$ at the target of the loop. The target flux is equal to the source flux parallely transported around the loop by the holonomy $h$. The Laplacian constraint is the equality of the scalar product $\vec{X}^{t}\cdot\vec{X}^{s}$ with the squared norm $\vec{X}^{s}\cdot\vec{X}^{s}$ and therefore means that the two flux are equal, $\vec{X}^{s}=\vec{X}^{t}$. This implies the flatness of the group element $h$ (up to the $\U(1)$ stabilizer of the flux vector).

At the quantum level, the Laplacian constraint turns out to play a different role. It implies the invariance of the wave-function by conjugation:
\be
\Delta \vphi =\tDelta \vphi
\quad\Rightarrow\quad
\forall h,g\in\SU(2)\,,\,\,
\vphi(h)=\vphi(ghg^{-1})\,.
\ee
We rigorously prove this statement in the appendix \ref{app:Laplacian} solving explicitly the recursion relations implied by the Laplacian constraint on the Fourier coefficients of $\vphi$. Another way to understand the relation of the Laplacian constraint to the invariance under conjugation is to think in terms of spin recoupling. Let us call $\vJ^{L,R}$ respectively the $\su(2)$ generators living at the two ends of the loop and defining the left and right derivations. The two Casimirs, given by the two scalar products $\vJ^{L}\cdot\vJ^{L}$ and $\vJ^{R}\cdot\vJ^{R}$, are equal and their (eigen)value is $j(j+1)$ is the loop carries the spin $j$. Then the Laplacian constraint means that their recoupling is trivial:
\be
%\vJ^{R}\cdot\vJ^{L}=\vJ^{L}\cdot\vJ^{L}
%\quad\Rightarrow\quad
%(\vJ^{R}-\vJ^{L})^{2}=0\,,
0=\vJ^{R}\cdot\vJ^{R}-\vJ^{R}\cdot\vJ^{L}=\f12(\vJ^{R}-\vJ^{L})^{2}\,,
\ee
so that the two ends of the loop recouple to the trivial representation, i.e. the spin-0. As illustrated on fig.\ref{fig:Laplacian}, this also allows to show that the Laplacian constraint operator $(\tDelta-\Delta)$ is positive and its spectrum is $k(k+1)/2$ where $k$ is an integer running from 0 to $(2j)$ if the loop carries the spin $j$.
\begin{figure}[h!]

  \centering

  \begin{tikzpicture}[scale=0.8]
    \coordinate(O2) at (-2,0);
    \coordinate(O3) at (0,0);

    \draw (O3) to[in=-45,out=+45,loop,scale=5] (O3)++(1.6,0) node {$j$};
    %\draw (O1) to[in=-30,out=+30,loop,scale=3,rotate=90] (O1) ++(0,1.2) node {$k_2$};

    \draw[red] (O2) -- (O3) node[midway,below]{$k$};
    \draw (O3) node {$\bullet$} ++(0.2,0.6) node{$\vJ^{L}$} ++(0,-1.2) node{$\vJ^{R}$};

  \end{tikzpicture}

  \caption{The left and right derivations respectively  act as graspings at the source and target of the loop, inserting $\su(2)$ generators in the wave-functions. The Laplacian operator $(\tDelta-\Delta)$ then measures the difference between the two scalar products $\vJ^{R}\cdot\vJ^{R}$ and $\vJ^{R}\cdot\vJ^{L}$, or equivalently the Casimir $(\vJ^{R}-\vJ^{L})^{2}/2$ of the recoupling of the spins at the two ends of the loop. Assuming that the loop carries the spin $j$ then recoupling $j$ with itself gives a spin $k$ running from 0 to $(2j)$.
  }
  \label{fig:Laplacian}

\end{figure}
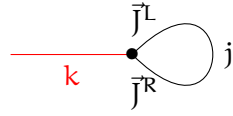
One can also see that the derivatives of the $\delta$-distribution are eigenstates of $(\tDelta-\Delta)$ with non-vanishing eigenvalues. For example, we compute:
%$$
%\int f \sum_{b} \pp_{b}^{R} \pp_{b}^{R}\pp_{a}^{R}\delta
%=
%\left.-\sum_{b} \pp_{a}^{R} \pp_{b}^{R} \pp_{b}^{R} f\right|_{\id}
%=
%-\sum_{b}f(J_{a}J_{b}J_{b})\,,
%$$
%$$
%\int f \sum_{b} \pp_{b}^{L} \pp_{b}^{R}\pp_{a}^{R}\delta
%=
%\left.-\sum_{b} \pp_{a}^{R} \pp_{b}^{R} \pp_{b}^{L} f\right|_{\id}
%=
%-\sum_{b}f(J_{b}J_{a}J_{b})
%=
%-\sum_{b}f(J_{a}J_{b}J_{b})
%+i\sum_{b}\eps^{abc}f(J_{c}J_{b})
%=-\sum_{b}f(J_{a}J_{b}J_{b})+f(J_{a})
%$$
\be
\int f(\tDelta-\Delta)\pp_{a}^{R}\delta
\,=\,
i^{3}\sum_{b}
\big{[}
  f(J_{a}J_{b}J_{b}) -f(J_{b}J_{a}J_{b})
  \big{]}
\,=\,
-if(J_{a})
\,=\,
+\int f \pp_{a}^{R}\delta
\,,
\ee
and so on with higher order differential operators. In particular, the derivative distribution $\pp_{a}\delta$ corresponds to the eigenvalue $k(k+1)/2$ for $k=1$. Higher order derivatives will explore higher eigenvalues.

To conclude, the original holonomy constraint, supplemented with the new Laplacian constraint, acting on functions on $\SU(2)$ admit the $\delta$-distribution as unique solution: the Laplacian constraint imposes invariance under conjugation while the holonomy constraint then imposes the flatness of the group element along the loop.

\begin{prop}
  There is a unique solution (up to a numerical factor) as a distribution over $\SU(2)$ to the holonomy and Laplacian constraints:
  \be
  \left|
  \begin{array}{l}
    (\hchi-2)\,\vphi=0\\
    (\Delta-\tDelta)\,\vphi=0
  \end{array}
  \right.
  \quad\Longrightarrow\quad
  \exists \lambda\in\C\,,\,\, \vphi(h)=\lambda\,\delta(h)
  \,.
  \ee
\end{prop}

Below, we look at the generic case of an arbitrary number of loops. We will show that we can supplement the holonomy constraints around each loop  either with Laplacian constraints for each loop or with  multi-loop holonomy constraints (that still act by multiplication) wrapping around several loops at once.

%%%%
\section{Holonomy constraints on $\SU(2)^{N}$ for $N\ge 2$}
\label{derivativesolution2}
%%%%

We now turn to the holonomy constraints on $\SU(2)^{N}$:
$$
\forall 1\le\ell\le N, \,\,(\hchi_{\ell}-2)\vphi=0\,,
$$
with the requirement of invariance under simultaneous conjugation of all the arguments $h_{\ell}$. Since we do not require the invariance under the individual action of conjugation on each little loop, the gauge invariance is not enough to kill the spurious solution identified above.
As proposed above, we can reach the uniqueness of the physical state by further imposing the Laplacian constraint on each loop:
\be
\forall \ell\in\N\,,\,\,
(\tDelta_{\ell}-\Delta_{\ell})\,\vphi=0\,.
\ee
This now  implies the invariance of the wave-function under the individual action of conjugation on each loop. In terms of spin recoupling, each little loop is linked to the vertex by a spin-0, as illustrated on fig.\ref{fig:0spincoupling},  this effectively trivializes the intertwiner space living at the vertex and the loops can be thought of as decoupled from one another.
The holonomy constraints then impose that the only solution state is the $\delta$-distribution.
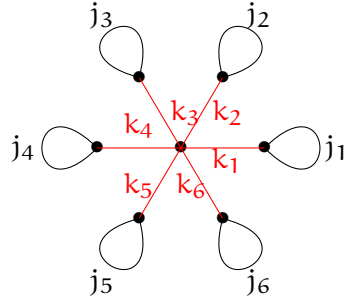
\begin{figure}[h!]
  \centering
  \begin{tikzpicture}[scale=0.55]

    \coordinate(O) at (0,0);
    \coordinate(O1) at (2,0);
    \coordinate(O2) at (1,1.7);
    \coordinate(O3) at (-1,1.7);
    \coordinate(O4) at (-2,0);
    \coordinate(O5) at (-1,-1.7);
    \coordinate(O6) at (1,-1.7);

    \draw (O) node[scale=1] {$\bullet$} ;
    \draw (O1) node[scale=1] {$\bullet$} ;
    \draw (O2) node[scale=1] {$\bullet$} ;
    \draw (O3) node[scale=1] {$\bullet$} ;
    \draw (O4) node[scale=1] {$\bullet$} ;
    \draw (O5) node[scale=1] {$\bullet$} ;
    \draw (O6) node[scale=1] {$\bullet$} ;

    \draw[red] (O) -- (O1)  ++(-0.9,-0.3) node{$k_{1}$};
    \draw[red] (O) -- (O2) node[midway,right]{$k_{2}$};
    \draw[red] (O) -- (O3) node[midway,right]{$k_{3}$};
    \draw[red] (O) -- (O4) node[midway,above]{$k_{4}$};
    \draw[red] (O) -- (O5) ++(0,0.8) node{$k_{5}$};
    %\draw[red] (O) -- (O4) node[midway,left]{$k_{4}$};
    \draw[red] (O) -- (O6) ++(-0.75,0.8) node{$k_{6}$};
    %\draw[red] (O) -- (O5) node[midway,left]{$k_{5}$};

    \draw[in=-45,out=45,scale=5] (O1)  to[loop] (O1) ++(0.35,0) node {$j_{1}$};
    \draw[in=-45,out=45,scale=5,rotate=60] (O2)  to[loop] (O2) ++(0.35,0) node {$j_{2}$};
    \draw[in=-45,out=45,scale=5,rotate=120] (O3)  to[loop] (O3) ++(0.35,0) node {$j_{3}$};
    \draw[in=-45,out=45,scale=5,rotate=180] (O4)  to[loop] (O4) ++(0.35,0) node {$j_{4}$};
    \draw[in=-45,out=45,scale=5,rotate=-120] (O5)  to[loop] (O5) ++(0.35,0) node {$j_{5}$};
    \draw[in=-45,out=45,scale=5,rotate=-60] (O6)  to[loop] (O6) ++(0.35,0) node {$j_{6}$};

  \end{tikzpicture}

  \caption{ The Laplacian constraint on a loop $\ell$ constraint the spin $j_{\ell}$ carried by the loop to recouple with itself into the trivial representation with vanishing spin $k_{\ell}=0$. Imposing this constraint on every loop, the vertex then recouples a collection of spin-0, the intertwiner is thus trivial and the loops are totally decoupled.}
  \label{fig:0spincoupling}

\end{figure}

\medskip

Instead of imposing the Laplacian constraints, another way to proceed is to introduce multi-loop holonomy constraints. To prove this, let us start by describing the gauge-invariant derivative solutions to the holonomy constraints. The general structure is as follows. One acts with arbitrary derivatives on the $\delta$-distribution $\prod_{\ell=1}^{N}\delta(h_{\ell})$. Then to ensure invariance under simultaneous conjugation, one must contract all the indices with a $\SO(3)$-invariant tensor $\cI$:
\be
\vphi^{\cI}(\{h_{\ell}\})=
\sum_{\{a_{i}^{\ell}\}_{i=1..n_{\ell}}}
\cI^{a^{1}_{1}..a^{N}_{n_{N}}}\,
\prod_{\ell=1}^{N}\pp_{a^{\ell}_{1}}..\pp_{a^{\ell}_{n_{\ell}}}\delta(h_{\ell})\,,
\ee
where $n_{\ell}$ is the order of the differential operator acting on the loop $\ell$, for an overall order $n=\sum_{\ell}n_{\ell}$, and $\cI$ is a rotational invariant tensor defining the contraction of the differential indices $a$'s, i.e. it is an intertwiner between $n$ spin-1 representations.

To be explicit, for $n=2$ differential insertions, there is a single invariant tensor: $\cI^{ab}=\delta^{ab}$. Either we act with the two derivatives on the same group elements, but then we already know that $\Delta\delta$ is not a solution to the holonomy constraint, or we act on two different loops getting the non-trivial distribution $\sum_{a}\pp_{a}\delta(h_{1})\pp_{a}\delta(h_{2})$ (here we put aside all the other loops, where no differential operator act):
\be
\la \,\sum_{a}\pp_{a}\delta_{1}\pp_{a}\delta_{2}\,|f\ra=
-f(J_{a},J_{a})\,,
\ee
which yields the evaluation $f(J_{a},J_{a})$ of the spin network state obtained by acting with the double grasping $J_{a}\otimes J_{a}$ on the test wave-function $f$.
We easily check that this provides a solution to the individual one-loop holonomy constraints:
\be
\forall f\in \cC^{\infty}_{\SU(2)^{2}}\,,\,\,
\int f (\chi_{\f12}(h_{1})-2)\,\sum_{a}^{3}\pp_{a}\delta(h_{1})\pp_{a}\delta(h_{2})
=
0\,,
\ee
The double grasping, as shown on fig.\ref{fig:doublegrasping}, couples the two loops. The goal is to suppress such coupling between the two loops in order to get as unique solution the factorized flat state $\delta^{\otimes N}$ where all the loops are entirely decoupled.
\begin{figure}[h!]
  \centering
  \begin{tikzpicture}[scale=1.5]

    \coordinate(O1) at (0,0);
    \draw (O1) node[scale=2] {$\bullet$} ;

    \draw (O1)   to[in=-45,out=+45,loop,scale=3,rotate=90] node[very near end](A){}   (O1)++(0,1) node {$h_1$};
    \draw (O1) to[in=-45,out=+45,loop,scale=3,rotate=-30] node[very near end](B){} (O1) ++(0.8,-0.7) node {$h_2$};
    \draw (O1) to[in=-45,out=+45,loop,scale=3,rotate=-150] (O1) ++(-0.8,-0.7) node {$h_3$};

    \draw[dotted,thick] (A) to[bend left] node[very near start,above,right](C){}    (B) ++(-0.1,-0.25) node{$\pp_{a}$};
    \draw (A) node[scale=1,red!50] {$\bullet$} ++(0.22,-0.05)node{$\pp_{a}$};
    \draw (B) node[scale=1,red!50] {$\bullet$} ;

  \end{tikzpicture}

  \caption{We act with derivatives $\pp_{a}$ on the group elements $h_{1}$ and $h_{2}$ and contract  the indices, which translates graphically as a double grasping linking the two loops.}
  \label{fig:doublegrasping}

\end{figure}
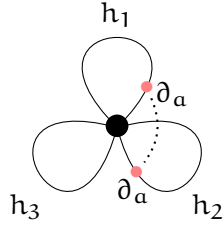
To make the system more rigid, the natural constraint to introduce is a two-loop holonomy constraint, which would kill any correlation between the two loops:
\be
(\hchi_{12}-2)\vphi (h_{1},h_{2})\equiv
(\chi_{\f12}(h_{1}h_{2})-2)\vphi (h_{1},h_{2})
=0\,.
\ee
We check that this two-loop constraint eliminates the coupled solution proposed above:
\be
\int f (\chi_{\f12}(h_{1}h_{2})-2)\,\sum_{a}^{3}\pp_{a}\delta(h_{1})\pp_{a}\delta(h_{2})
=
\left.f\Delta\chi_{\f12}\right|_{\id}
=
\f32\,f(\id,\id)\,\ne0\,.
\ee

For $n=3$ differential insertions, we still have a unique intertwiner, given by the completely antisymmetric tensor $\eps^{abc}$. This corresponds a triple grasping. The three derivatives can all act on the same loop, in which case we do not get a solution of the one-loop holonomy constraint, or they can act on two different loops, in which case it is not a solution of the two-loop holonomy constraints we have just introduced, or they can act on three different loops in which case we need to introduce a three-loop holonomy constraint to discard it:
\be
\begin{array}{rcl}
& &\int f (\chi_{\f12}(h_{1}h_{2}h_{3})-2)\eps^{abc}\pp_{a}\delta(h_{1})\pp_{b}\delta(h_{2})\pp_{c}\delta(h_{3}) \\
&=&
-\f i2^{3}f(\id)\eps^{abc}\chi_{\f12}(\sigma_{a}\sigma_{b}\sigma_{c}) \\
&=&
-\f{3i}2\,f(\id)\,\,\ne0\,.
\end{array}
\ee

For an arbitrary number $n$ of differential insertions acting on the $N$ loops, the grasping will potentially couple the $N$ loops. In order to kill all those coupled solutions, we introduce all multi-loop holonomy constraints:
\be
\forall E\subset \{1,..,N\}\,,
\,\,
\Big{[}
  \chi_{\f12}\big{(}
  \prod_{\ell\in E}h_{\ell}
  \big{)}
  -2
  \Big{]}\,\vphi=0
\,.
\ee
The ordering of the group elements is important of course for the precise definition of the multi-loop operator but is irrelevant to ensure that the action of the corresponding constraint operator on the coupled derivative distributions does not vanish.
In fact, looking deeper into the structure of $\SO(3)$-invariant tensors,
a fundamental theorem on rotational invariants states that all $\SO(3)$-invariant polynomial of $n$ 3d-vectors $\vv_{i=1..n}$ are generated by scalar products $\vv_{i}\cdot\vv_{j}$ and triple products $\vv_{i}\cdot(\vv_{j}\wedge\vv_{k})$. This means that we only need the two-loop and three-loop holonomy constraints to ensure that the flat state, defined as the $\delta$-distribution, is the only solution to the Hamiltonian constraints.

%%%%
\section{The full Hamiltonian constraints for BF theory on loopy spin networks}
\label{derivativesolution3}
%%%%

To summarize the implementation of BF theory on loopy spin networks, we have introduced individual holonomy constraints on each little loop around each vertex of the background graph. This is the usual procedure, for instance when constructing spinfoam amplitudes for BF theory from a canonical point of view. Surprisingly, these constraints are not strong enough to fully constraint the theory to the single flat state and kill all the little loop excitations. This can be backtracked to the simple fact that the identity $\id$ is an extremum  of the $\SU(2)$-character $\chi_{\f12}$ and thus the derivative of the character vanishes at that point. As a result, the $\delta$-function on $\SU(2)$ is not the unique solution to the holonomy constraints, but its first derivative are also solutions. While all the solution distributions are peaked on the identity and vanish elsewhere, we are allowed  grasping operators coupling the loops together. To forbid such such coupling and force to have a unique physical state, we have showed that we can supplement the original one-loop holonomy constraints with either one-loop Laplacian constraints or with multi-loop holonomy constraints, which leads us to two proposals for the Hamiltonian constraints for BF theory on loopy intertwiners:
\begin{itemize}
\item We  impose on each loop two gauge-invariant constraints, the holonomy constraint that acts by multiplication and the Laplacian constraint which acts by differentiation:
  \be
  \forall \ell\,,\quad
  \hchi_{\ell}\,\vphi=2\vphi\,,\quad
  \Delta_{\ell}\,\vphi=\tDelta_{\ell}\,\vphi\,.
  \ee

\item We impose all multi-loop holonomy constraints, requiring not only that the group elements $h_{\ell}$ on each loop $\ell$ is the identity $\id$ but also that all their products remain flat. This means one constraint for each finite subset $E$ of the set of all loops:
  \be
  \forall E\subset \N,\quad
  \hchi_{E}\,\vphi=2\vphi\,,
  \qquad
  \hchi_{E}\,\vphi (\{h_{\ell}\}_{\ell\in\N})=
  \chi_{\f12}\Big{(}\prod_{\ell \in E}h_{\ell}\Big{)}\,\vphi (\{h_{\ell}\}_{\ell\in\N})\,.
  \ee
  The ordering of the group elements does not matter in order to impose the flatness. These multi-loop constraints kill any correlation or entanglement between the loops. It is actually  sufficient to impose only the two-loop and three-loop holonomy constraints.

\end{itemize}

If we only impose the one-loop holonomy constraints, then the totally flat state defined by the $\delta$-distribution is not the only physical state. We get an infinite-dimensional  space of physical states, obtained by the action of first order grasping operators on the $\delta$-distribution, allowing for non-trivial coupling and correlations between the little loops. It would be interesting to understand the geometrical meaning of those states and if they play a special role in the spinfoam models for BF theory (the Ponzano-Regge and Turaev-Viro models for 3d BF theory and the Crane-Yetter model for 4d BF theory). As an example, we have in mind the recursion relation satisfied by the 6j symbol, which is understood to be the expression of the action of the holonomy operator on the flat state on the tetrahedron graph \cite{Bonzom:2009zd,Bonzom:2011hm,Bonzom:2014bua}. Our results suggest that the double and triple graspings on the 6j symbol might be other solutions to this recursion relation.That would be specially interesting since the triply grasped 6j symbols is understood to be the first order correction of the q-deformed 6j-symbol \cite{Freidel:1998ua}. On a totally different route, maybe those local excitations could provide a first extension of the topological BF theory to a field theory with local degrees of freedom.

\medskip

On the other hand, imposing the full set of Hamiltonian constraints proposed above leads to a unique physical state for BF theory: the flat state $\vphi_{\mathrm{BF}}=\delta$.
This physical state  is clearly not normalizable. But since it is unique, it is not a big problem to define the scalar product on this final one-dimension Hilbert space.
The physical scalar product on the initial Hilbert space of loopy spin networks is defined by projecting on this physical state, which amounts at the end of the day to simply evaluate the wave-functions at the identity i.e. on flat connections:
\be
\begin{array}{rcl}
\forall f,\tilde{f}\in\cH^{\mathrm{loopy}}\,,\quad
\la f|\tilde{f}\ra_{\mathrm{phys}}
&=&
\la f|\vphi_{BF}\ra\,\la\vphi_{BF}|\tilde{f}\ra \\
&=&
\overline{\la\vphi_{BF}|f\ra}\,\la\vphi_{BF}|\tilde{f}\ra \\
&=&
\overline{f(\id)}\,\tilde{f}(\id)\,.
\end{array}
\ee
As expected, we are left with a single physical state on the flower, the little loops have been projected out and all local degrees of freedom have disappeared.

\medskip

Now that we have checked that loopy spin networks allow for a correct implementation of BF theory's topological dynamics, we would like to later introduce Hamiltonian constraints allowing for local degrees of freedom. We wouldn't want to kill the little loops as happens for BF theory. The goal would be to have dynamics coupling the little loops to the spins living on the links of the background graph, in such a way that it reproduces the propagation of the local geometry excitations of \ac{GR} in a continuum limit. The strategy would be to slightly modify the BF dynamics -``constrain the BF theory''- most likely following the approaches for the dynamics of discrete/twisted geometries \cite{Bonzom:2011hm,Bonzom:2011nv,Bonzom:2013tna} or of EPRL spinfoam models \cite{Engle:2007wy,Geloun:2010vj,Bianchi:2012nk}.

\section{Revisiting the BF constraints as creation and annihilation of loops}
\label{BFsym}
%%%%%%

Let us see how to implement the flatness constraint on our Fock space of loopy spin networks with bosonic statistics for the little loops. As earlier, we do not discuss the flatness constraints around loops of the base graph $\Gamma$, which are implemented as usual by using the standard holonomy operators around those loops. Here, we will focus on the fate of the little loop excitations at every vertex of the background graph $\Gamma$. For this purpose, we can focus on a single vertex and we can restrict ourselves to the flower graph, i.e. to the Fock space of loop intertwiners around a unique vertex. As we have constructed the holonomy operator in the previous section, we propose to use it as the Hamiltonian constraints for BF theory and simply impose:
\be
H^{\mathrm{BF}}=\hchi_{\f12}-2\,.
\ee
This is a self-adjoint operator and imposing this constraint amounts to projecting onto the highest eigenvalue of the holonomy operator. Since  $\hchi_{\f12}$ creates and annihilates loops by construction, $H^{\mathrm{BF}}$ shifts the number of loops and its flow should imply that the number of loops becomes pure gauge. Let us look at the space of physical states solving this flatness constraint.
By  proposition \ref{flat-prop1}, we already know that the flat state, defined as the factorized $\delta$-distribution state, saturates the holonomy bound, $H^{\mathrm{BF}}\,|\delta\ra=0$ .
The natural question is whether the flat state is the only solution to this constraint.

We will run into the same problem as in the case of distinguishable loops of higher derivative solutions to the holonomy constraint. In order to deal with this potential infinite-dimensional space of solutions, we will introduce as before a Laplacian constraint and multi-loop holonomy operators. However we will ultimately show that we require only a finite number of constraint operators (three to be exact) to impose full flatness and the uniqueness of the physical state despite the infinite number of loop excitation modes that need to be constrained.

\medskip

More precisely, the holonomy constraint amounts to solving  functional recursion relations, relating $f_{N+1}$, $f_{N}$ and $f_{N-1}$ at each step. The problem is that this relation doesn't entirely fix $f_{N+1}$ in terms of $f_{N}$ and $f_{N-1}$, even assuming that these functions are invariant under permutations of their arguments and invariant under conjugation. Indeed it only fixes the integral $\int\mathrm{d}k \chi_{\f12}(k)\,f_{N+1}(h_{1},..,h_{N},k)$. This condition seems to fix only the spin-$\f12$ component of the function, so we face two obstacles: the non-trivial internal intertwiner structure and arbitrary higher spin excitations on each loop. We explain below how to get rid of all those modes by introducing constraints on the creation and annihilation of loops together with a Laplacian constraint.

%For instance, if we focus on the simplest component of the holonomy operator and try to solve for eigenvectors of the annihilation operator, $A\vphi =2\vphi$, then we can find both unfactorized solutions and factorized states that differ from the $\delta$-distribution\footnotemark.
%%
%\footnotetext{The eigenvector equation $A\phi =2\phi$ corresponds to the We get the functional recursion relations:
%$$
%2\phi_{N}(h_{1},..,h_{N})
%\,=\,
%\int\mathrm{d}k \,\chi_{\f12}(k)\,\phi_{N+1}(h_{1},..,h_{N},k)\,.
%$$
%First, we easily identify factorized states, $\phi_{0}=1$ and $\phi_{N}=f^{\otimes N}$, which are solutions if and only if $\int \chi_{\f12}f=2$. Assuming that $f$ is invariant under conjugation and that its integral vanishes, this only fixes the spin-$\f12$ component of $f$ and allows for generic solutions $f=2\chi_{\f12}+\sum_{j\ge 1} f^{(j)}\chi_{j}$ for arbitrary coefficients $ f^{(j)}$.
%%
%Second we also find coupled states:
%$$
%\phi_{N}=\f{\alpha^{N}}{N!}\sum_{\sigma\in S_{N}}\chi_{\f12}(h_{\sigma(1)}..h_{\sigma(N)})
%\quad\Rightarrow\quad
%A\,\phi\,=\,\f\alpha2\,\phi\,.
%$$
%}

Before treating the general case, we explore two simplified cases. First, factorized states avoid the problem of possible non-trivial intertwiner structure. It turns out that the spin-$\f12$ one-loop holonomy constraint is enough to constrain all the higher spin excitations and lead to the flat state as the unique physical state.  Second we consider the larger class of states with decoupled loops, defined mathematically as the wave-functions which are invariant under conjugation of its individual arguments (and not simply under the simultaneous conjugation of all its arguments as required by gauge invariance). In this case, the spin-$\f12$ constraint is not enough anymore and we need to explicitly introduce explicit constraints for all the higher spin excitations. We summarize these two cases in the following two propositions.

\begin{prop}
  Let us consider a factorized state $\vphi\in\cHs$, that $\vphi_{0}=1$ and $\vphi_{N}=F^{\otimes N}$ for an integrable $F$ invariant under conjugation, $F(h)=F(ghg^{-1})$. Then the constraint $\hchi_{\f12}\,\vphi=\,2\vphi$ has a unique solution, which is the flat state, $\vphi=\delta$ and $F(h)=\delta(h)-1$.
\end{prop}
\begin{proof}
  Let us look at the eigenvector equation on factorized states $\vphi$ defined as $\vphi_{0}=1$ and $\vphi_{N}=F^{\otimes N}$:
  $$
  \big{(}A+A^{\dagger}+B\big{)}\,F^{\otimes N}
  \,=\,
  4\,F^{\otimes N}\,.
  $$
  For $N=0$, this gives an integral condition on the one-loop wave-function $F$:
  $$
  \int \mathrm{d}k\,\chi_{\f12}(k)F(k)=2.
  $$
  Then, for $N=1$, we get a functional equality:
  $$
  \chi_{\f12}F+\chi_{\f12}-2=2F\,.
  $$
  Let us decompose $F$ on the spin basis. Since it is invariant under conjugation, it decomposes onto the characters $F=\sum_{j\ne 0} F_{j}\chi_{j}$. The $N=1$ equation translates into a recursion relation on the coefficients $F_{j}$ while the $N=0$ equation sets its initial condition:
  \be
  F_{\f12}=2\,,\quad
  F_{1}+1=2F_{\f12}\,,\quad
  \forall {j\ge 1}\,,\,\,
  F_{j+\f12}+F_{j-\f12}=2F_{j}\,.
  \ee
  This has a unique solution $F_{j}=(2j+1)$, which translates to $F=\delta-1$. The constraint equation for $N\ge 2$ automatically follows.
\end{proof}

The case of factorized state works because the holonomy operator couples the creation of loops and the exploration of the higher spin components of the one-loop wave-function.
Next, we move to the larger class of functions which are invariant under conjugation of its individual arguments. Then the functions $\phi_{N}$ decompose on the character basis. Imposing the one-loop holonomy constraints for all spins leads to functional recursion equations such that the flat state is solution to the holonomy constraint.
\begin{lemma}
  Considering a state invariant under conjugation of each of its arguments,
  $$
  \phi_{N}(h_{1},..,h_{N})=\phi_{N}(g_{1}h_{1}g_{1}^{-1},..,g_{N}h_{N}g_{N}^{-1})\,,\quad\forall g_{i}\in\SU(2)^{N}\,,
  $$
  it decomposes on the character basis:
  $$
  \phi_{N}(h_{1},..,h_{N})=\sum_{j_{1},..,j_{N}}\phi_{N}^{j_{1},..,j_{N}}\prod_{i}^{N}\chi_{j_{i}}(h_{i})\,.
  $$
  Assuming that the $\phi_{N}$'s are all symmetric under permutations of their arguments and that they have no 0-modes, $\int dh_{1}\phi_{N}=0$ for all $N\ge1$, then
  the only such solution to the set of holonomy constraints $\hchi_{j}\,\phi\,=\,(2j+1)\phi$ for all spins $j\in\f{\N^{*}}2$ is  the flat state $\phi_{N}=(\delta-1)^{\otimes N}$ (up to a global factor).
\end{lemma}
\begin{proof}
  The proof is straightforward by recursion. For $N=0$, the constraint gives $\phi_{1}$ in terms of the no-loop mode $\phi_{0}$:
  \be
  \forall j\ge \f12\,,\,\,
  \int dk\,\chi_{j}(k)\phi_{1}(k)=(2j+1)\phi_{0}\,,
  \ee
  which gives $\phi_{1}^{j}=(2j+1)\phi_{0}$ for all non-vanishing spins $j$ while $\phi_{1}^{0}=0$ by hypothesis. This way, if we fix the initial normalization to $\phi_{0}=1$,  we recover $\phi_{1}=(\delta-1)$.
  Then the constraint equations for $N\ge1$ reads:
  \be
  \begin{array}{lcl}
  \forall j\ge\f12 ,& & \\
  2(2j+1)\phi_{N}(h_{1},..,h_{N})
  &=&
  \int \dd k\,\chi_{j}(k)\phi_{N+1}(h_{1},..,h_{N},k) \\
  &+& \f1N\sum_{i=1}^{N}\phi_{N-1}(h_{1},..,\widehat{h_{i}},..,h_{N}) \\
  &+& \f1N\sum_{i=1}^{N}\Bigg{[}
    \chi_{j}(h_{i})\phi_{N}(h_{1},..,h_{N}) \\
  &-&\int \dd k_{i}\,\chi_{j}(k_{i})\phi_{N}(h_{1},..,{k_{i}},..,h_{N})
    \Bigg{]}
  \,.
  \end{array}
  \ee
  We can solve this equation by recursion, determining the Fourier coefficients of $\phi_{N+1}$ in terms of $\phi_{N}$ and $\phi_{N-1}$. The coefficients $\phi_{N}^{j_{1}..j_{N}}$ vanish by assumption if one of the spins $j_{i}$ is zero. When none of the spins vanishes, we show that
  \be
  \phi_{N}^{j_{1}..j_{N}}=\prod_{i=1}^{N}(2j_{i}+1)\,.
  \ee
\end{proof}
Comparing to the case of distinguishable loops, in this case where the loops are individually gauge-invariant and thus decoupled, we have traded the infinity of holonomy constraints, one for each distinguishable loop, for the infinite tower of one holonomy constraint per spin mode for indistinguishable loops. Exploiting further the Fock space structure for the bosonic little loops, we can nevertheless reduce this infinity of holonomy constraints to a pair of constraints.
Indeed, checking the details of the proof of proposition \ref{flat-prop1} on the action of holonomy operators on the $\delta$-state, we propose to use  non-Hermitian constraints and characterize the flat state as an eigenvector of the loop annihilation operator $A$ and the loop creation operator $(B+A^{\dagger})$:
\begin{lemma}
  \label{lemmaAB}
  Considering a state $\phi$ invariant under conjugation of each of its arguments, and with no 0-modes, we introduce the pair of non-Hermitian constraint operators defined by the spin-$\f12$ annihilation and creation operators acting on $\cHs$:
  \be
  A\,|\phi\ra=2\,|\phi\ra
  \,,\quad
  (B+A^{\dagger})\,|\phi\ra=2\,|\phi\ra
  \,.
  \ee
  Then the only solution to all these constraints is the flat state  $|\phi\ra=|\delta\ra$.
\end{lemma}
\begin{proof}
  Let us write explicitly the eigenvalue equations for the state $\phi$:
  \be
  \forall N\ge 0,\,\,
  \int \mathrm{d}k\,\chi_{\f12}(k)\phi_{N+1}(h_{1},..,h_{N},k)
  \,=\,
  2\phi_{N}(h_{1},..,h_{N})
  \ee
  \be
  \begin{array}{l}
  \forall N\ge 1,\,\,
  \sum_{\ell=1}^{N}
  \Bigg{[}
    \chi_{\f12}(h_{\ell})\,\Big{[}\phi_{N-1}(h_{1},..,\widehat{h_{\ell}},..,h_{N})
      +\phi_{N}(h_{1},..,h_{N})\Big{]} \\
    \quad -\,\int \mathrm{d}k_{\ell}\,\chi_{\f12}(k_{\ell})\phi_{N}(h_{1},..,{k_{\ell}},..,h_{N})
    \Bigg{]}
  \,=\,
  2N\,\phi_{N}(h_{1},..,h_{N})
  \end{array}
  \ee
  with the initial conditions equation at $N=0$ for the creation operator  $(B+A^{\dagger})$ trivially satisfied.
  We could translate these equations into recursion relations on the Fourier coefficients, but there is actually a simpler and more direct route.
  The first equation (for $A$) can be injected in the second equation turning it into a functional recursion:
  \be
  \sum_{\ell=1}^{N}
  (2-\chi_{\f12}(h_{\ell}))\,
  \Big{[}
    \phi_{N-1}(h_{1},..,\widehat{h_{\ell}},..,h_{N})
    +\phi_{N}(h_{1},..,h_{N})
    \Big{]}
  \,=\,
  0\,.
  \ee
  For $N=1$, this relates the one-loop wave-function $\phi_{1}$ to the no-loop normalization $\phi_{0}$:
  $$
  \forall h\in\SU(2),\quad
  (2-\chi_{\f12}(h))\,(\phi_{0}+\phi_{1}(h))\,=\,0\,.
  $$
  Since $\phi_{1}$ is invariant under conjugation, this holonomy constraint has a unique distributional solution up to an arbitrary factor, $(\phi_{0}+\phi_{1})\propto\delta$. The integral condition, $\int \chi_{\f12}\phi_{1}=2\phi_{0}$, fixes this factor and we recover $\phi_{1}=(\delta -1)$ as expected as we fix the normalization $\phi_{0}=1$.

  We then proceed by recursion, fixing the number of loops $N\ge 2$ and assuming that $\phi_{n}=(\delta-1)^{\otimes n}$ for all $n\le (N-1)$. Let us now prove this statement holds for $n=N$. Using the identity $(2-\chi_{\f12}(h))\delta(h)=0$ , we start by checking that:
  $$
  \sum_{\ell}^{N}
  (2-\chi_{\f12}(h_{\ell}))\,
  \Big{[}
    \prod_{i}(\delta(h_{i})-1)-\prod_{i\ne\ell}(\delta(h_{i})-1)
    \Big{]}
  \,=\,0\,.
  $$
  This implies that:
  $$
  \Bigg{(}
  \sum_{\ell}^{N} (2-\chi(h_{\ell})
  \Bigg{)}
  \,
  \Bigg{(}
  \phi_{N}(h_{1},..,h_{N})-\prod_{i}(\delta(h_{i})-1)
  \Bigg{)}
  \,=\,0\,.
  $$
  Since every holonomy operator $(2-\hchi_{\ell})$ is Hermitian positive, this means that the holonomy constraint holds for each loop individually:
  \be
  \forall \ell \le N\,,\,\,
  \sum_{\ell}^{N} (2-\chi(h_{\ell})
  \,
  \Bigg{(}
  \phi_{N}(h_{1},..,h_{N})-\prod_{i}(\delta(h_{i})-1)
  \Bigg{)}
  \,=\,0\,.
  \ee
  Since we have assumed that the wave-function is invariant under conjugation individually for each of its arguments, the only distribution solution to this equation is the product of $\delta$-function up to a global factor:
  $$
  \phi_{N}(h_{1},..,h_{N})=\prod_{i}(\delta(h_{i})-1)+\alpha \prod_{i}\delta(h_{i})\,,
  $$
  for some factor $\alpha$ to be determined. Checking this identity against the integral condition $\int \chi_{\f12}\phi_{N}=\phi_{N-1}$ yields $\alpha=0$ thus proving the proposition.
\end{proof}
We see that the requirement of the invariance under conjugation for each loop individually (stronger than gauge-invariance requiring the invariance under global conjugation) is crucial in the last step of the proof. Else we would have to deal with derivative solutions, in $\pp\delta$ and so on, as in the case of distinguishable loops.

\medskip

Our proposal amounts to adding another constraint along side the Hermitian holonomy constraint $\hchi_{\f12}=2$. Instead of taking the average of the two operators $A$ and $(B+A^{\dagger})$ and defining the holonomy operator, we subtract them and get the other constraint $B+(A^{\dagger}-A)=0$. This new constraint operator has a Hermitian part $B$ and a anti-Hermitian part $(A^{\dagger}-A)$, such that the overall structure can be interpreted as a holomorphic constraint, similar to the annihilation operator $a=\hat{x}-i\hat{p}$ for the harmonic oscillator.
From this perspective, eigenvectors of this ``holomorphic'' operator $B+(A^{\dagger}-A)$ can be considered as coherent states, which is pretty natural since we are looking into coherent superpositions of any number of loops summing over $N$, and the $\delta$-state, as a null eigenvector of that operator, can be considered as a ground state.

The trick why these two  constraint operators $A$ and $(B+A^{\dagger})$ are enough to kill all the degrees of freedom and lead to a single physical state is that they do not commute and their commutators actually generate higher spin constraints:
\begin{lemma}
  Imposing the two constraints with $A$ and $(B+A^{\dagger})$ on $\cHs$ implies a tower of constraints with all the higher spin annihilation operators:
  \be
  A\,|\phi\ra=(B+A^{\dagger})\,|\phi\ra=2\,|\phi\ra
  \qquad\Longrightarrow\qquad
  \forall j\ge \f12\,,\,\,
  A_{j}\,|\phi\ra=(2j+1)\,|\phi\ra\,.
  \ee
\end{lemma}
\begin{proof}
  Let us look at the commutator $(\hN+1)\,[A_{j},(B+A^{\dagger})]$. This commutator will vanish on solution states $|\phi\ra$. Using the commutation relations computed earlier \eqref{commAB1} and \eqref{commAB2}, we get for $j=\f12$:
  $$
  \begin{array}{l}
  (\hN+1)\,[A,(B+A^{\dagger})]=A_{1}+\id-(B+A^{\dagger})A\,, \\
  \quad
  A_{1}\,|\phi\ra=\Big{[}(B+A^{\dagger})A-\id\Big{]}\,|\phi\ra=(4-1)\,|\phi\ra=3\,|\phi\ra\,,
  \end{array}
  $$
  then  for higher spins $j\ge 1$ the commutation relation $(\hN+1)\,[A_{j},(B+A^{\dagger})]=A_{j+\f12}+A_{j-\f12}\id-(B+A^{\dagger})A_{j}$ implies:
  \be
  %(\hN+1)\,[A_{j},(B+A^{\dagger})]=A_{j+\f12}+A_{j-\f12}\id-(B+A^{\dagger})A_{j}\,,
  %\quad
  A_{j+\f12}\,|\phi\ra=\Big{[}(B+A^{\dagger})A_{j}-A_{j-\f12}\Big{]}\,|\phi\ra
  =\Big{[}2(2j+1)-2j\Big{]}\,|\phi\ra=(2j+2)\,|\phi\ra\,.
  \ee
\end{proof}

Our two (non-Hermitian) constraints do not commute and generate an infinite number of constraints killing all the higher spin excitations, leaving us at the end of the day with the single totally flat state.
In some sense, this pair of annihilation and creation constraint operators can be considered as the generators of the algebra of holonomy operators on the Fock space of loopy spin networks.

\medskip

Let us move on to the general case. Working with states invariant under conjugation for each loop individually amounts to considering states created loop by loop, only by the action of one-loop holonomy operators. This leads to decoupled loops and unfortunately  does not explore the whole space of intertwiners: we still need to reach all the states globally invariant under conjugation but not invariant under conjugation of the individual arguments, such as $\chi(h_{1}h_{2}..)$. In the spin decomposition of the wave-functions, this corresponds to the fact that the modes are not simply $\phi_{N}^{j_{1},..,j_{N}}$ but should be labelled as $\phi_{N}^{j_{1},..,j_{N},\cI}$: they do not depend only on the spins $j_{i=1..N}$ but further depend on the data of a (loopy) intertwiner $\cI$ between (two copies of) all the spins. This leads to the existence of the derivative solutions to the holonomy constraints, defined by applying differential operators (graspings) to the $\delta$-distribution.

The intertwiner structure is hard to constraint completely. One way to go is to not only use higher spin operators but introduce multi-loop holonomy constraints, as in the case of distinguishable loops. Indeed, since we would like to freeze all the spin excitations on the possible infinity of loops, it is natural to introduce one constraint operator per mode. This leads us to conjecture a set of complete holonomy constraints for BF theory. Considering all the multi-loop holonomy operators for arbitrary spins acting on the Fock space of loopy intertwiners $\cHs$:
  $$
  \forall j\in\f{\N^{*}}2\,,\quad
  \forall n\in\N^{*}\,,\quad
  \hchi^{(n)}_{j}\,|\phi\ra=\,(2j+1)\,|\phi\ra\,,
  $$
  then the only solution to all these constraints is the flat state $\phi=\delta$.
  We have checked this conjecture up to the three-loop component of the state, $N=3$, but we haven't gone further. This would require explicitly and carefully defining the multi-loop holonomy operators. We should also take special care of working with legitimate states, controlling the convergence/divergence of the series in $j$ and $N$ to ensure that the states are distributions.

We propose to take a different route in order to keep a finite number of (primary) constraints. We introduce a Laplacian constraint to project onto the space of wave-functions invariant under conjugation and use the creation and annihilation operators for loops to impose flatness:
\begin{prop}
  We consider the pair of non-Hermitian constraint operators defined by the spin-$\f12$ annihilation and creation operators acting on $\cHs$:
  \be
  A\,|\phi\ra=2\,|\phi\ra
  \,,\quad
  (B+A^{\dagger})\,|\phi\ra=2\,|\phi\ra
  \,.
  \ee
  We supplement these constraints with the Laplacian constraint:
  \be
  \begin{array}{rcl}
  (\tDelta-\Delta)\,|\phi\ra=0
  &\textrm{with}&
  (\Delta\vphi)_{N}=\f1N\sum_{\ell}^{N}\Delta_{\ell}\,\vphi_{N}\,, \\
  &\textrm{and}&
  (\tDelta\vphi)_{N}=\f1N\sum_{\ell}^{N}\tDelta_{\ell}\,\vphi_{N}\,.
  \end{array}
  \ee
  Imposing these three eigenvalue equations leads to a unique solution (up to a global factor), the flat state $|\phi\ra=|\delta\ra$ defined by $\phi_{N}=(\delta-1)^{\otimes N}$.
\end{prop}
\begin{proof}
  We start with the Laplacian constraints:
  $$
  \sum_{\ell}(\tDelta_{\ell}-\Delta_{\ell})\phi_{N}=0\,.
  $$
  Since every Laplacian constraint operator $(\tDelta_{\ell}-\Delta_{\ell})$ on each loop $\ell$ is Hermitian and positive, this imposes that each of them vanish on the wave-function, i.e. for all $\ell$ we have $(\tDelta_{\ell}-\Delta_{\ell})\phi_{N}=0$. This implies that $\phi_{N}$ is invariant under conjugation of each of its argument. Then we apply lemma \ref{lemmaAB} to prove the uniqueness of the solution state.

\end{proof}

On the one hand, the Laplacian constraint fixes how every loop is attached to the vertex, through a trivial spin-0. Each loop is invariant under conjugation on its own, states are collections of bosonic loops, each carrying a spin and with a trivial intertwiner between them. On the other hand, the constraints $A$ and $(B+A^{\dagger})$ realize explicitly the idea that BF dynamics impose that the creation and annihilation of loops are pure gauge.

\medskip

To conclude this chapter, we would like to underline the similarities and differences between the case of distinguishable loops and the Fock space of indistinguishable little loop excitations. When working with little loops endowed with bosonic statistics, one must take a special care to consistently remove the spin-0 modes on every loop to implement the cylindrical consistency of the wave-functions. This leads to a (spin-$\f12$) holonomy operator also creating and annihilating loops . We explicitly separate its components respectively creating and annihilating loops and use them as legitimate constraint operators for BF theory. This is different from distinguishable loops where holonomy operators are defined as attached to a loop: a holonomy operator acts on a given loop, exciting and shifting the spin carried by the loop.

Nevertheless, the issue of the intertwiner space living at vertex and coupling the loops is the same in both frameworks. We have identified an infinity of solutions to the holonomy constraints, constructed as differential operators acting on the $\delta$-distribution (as graspings on the spin network wave-function). These are still peaked on the identity group element, but they potentially define an infinity of gauge-invariant local degrees of freedom living at the vertex. To get rid of these ``spurious'' solutions, we have introducing a Laplacian constraint that forces each loop excitation to be invariant under conjugation, thus linking it trivially to the vertex. This allows to kill all local intertwiner excitation. Then we take as  Hamiltonian constraints for BF theory this combination on holonomy and Laplacian constraints, which lead as wanted to a unique physical state, the flat $\delta$-state.

Now, they are several directions to explore. As BF theory was chosen for its trivial renormalization flow, the study of the coarse-graining of this precise theory is not particularly interesting. It may be interesting however to look at possible renormalization flow and see if BF theory emerges as a natural fixed point. It is also possible to try and adapt techniques from spinfoams and implement simplicity conditions, in order to go toward \ac{GR}. It is also possible to try and impose a discrete dynamics on the support graph and discover the dynamics for the loops by coarse-graining. This might actually lead to \ac{LQG} as a fixed point. It might also be possible to test this framework in the context of \ac{QCD} and compare with the standard renormalization flow. In essence, the framework is more or less ready to do real field testing. But a more conceptual approach is also possible: as some degrees of freedom have been revealed by our work on hyperboloids, it would be interesting to see how we can recover them from loops. Operators corresponding to surface holonomies, for instance, might therefore be quite relevant in the writing of a coarse-grained theory.

%\textbf{TODO:} sections to refine!

%\section{Classical BF theory}

%\textbf{TODO:} first step: coarse-graining of BF theory, constraint on the graph should be the usual, what constraints on the loops? trace proposal

%\section{BF theory on the flower graph}

%\textbf{TODO:} discussion of spurious solutions, problem with correlations, new proposal with correlation, should work with only three, proposal from EEF, geom interpretation, parallel case in symmetric theory, for coarse-graining suggest a new theory: trace theory, has local degrees of freedom, should have non trivial coarse-graining

%\section{Trace action from coarse-grainio}

%*****************************************
%*****************************************
%*****************************************
%*****************************************
%*****************************************

%*****************************************
\addtocontents{toc}{\bigskip \protect\vspace{\beforebibskip} \par} % to have the bib a bit from the rest in the toc
\chapter{Conclusion} \label{ch:Conclusion}
%*****************************************

\inspiquote{Hello, I'm the Doctor}{The Doctor}

Let us now recap. \ac{LQG} is a proposal for a theory of quantum gravity. Its kinematics is well-understood: the Hilbert space is the Ashtekar-Lewandowski Hilbert space of functions over the generalized connection. A basis of this states space is given by the spin network basis which diagonalize the area and volume operators. Though some questions are still open, like the precise role of the Immirzi parameter, or the existence of other possible representations, a rigorous framework for discussing quantum geometry has been developed and is now available to write down a quantum theory of quantum \ac{GR}.

The problem is in the dynamics and in the continuum limit. The dynamics is not yet fully written down in a satisfactory manner, at least in the canonical approach. Indeed, there are very interesting proposals as Thiemann's constraint and the master constraint programme but they are not definitive yet. Simplified models, as the $\mathrm{U}(1)^3$ model, have been developed and shed light on the subject but the question remains open. This is linked to the problem of the continuum limit: even with a given dynamics, taken from the spinfoam approach for instance, it is very difficult to say if \ac{GR} is actually reproduced in some continuum/low energy limit.

Two limits can be taken \textit{a priori}: the limit of large spins on the edges of the spin network state, and the limit of a very refined graph. The first limit is sometimes called the \textit{classical} limit\graffito{The name ``classical limit'' might be somewhat misleading at we expect the classical limit of \ac{LQG} to be \ac{GR}. Here ``classical'' must be understood as ``large quantum numbers''.} and the second one the \textit{continuum} limit. There are strong indications that spinfoams indeed lead to discrete general relativity when the first limit is taken. But, for a genuine study of \ac{LQG}, we want to study the second limit. This implies graph changing dynamics, refined states and renormalization tools. Indeed, it is usually called the \textit{continuum} limit because at the limit, the spin network would describe some kind of continuum space. But for any finite graph, \ac{LQG} desperately looks discrete and may differ from \ac{GR}. Therefore, the problem is two-fold: finding the right dynamics and then study its continuum limit to compare it to \ac{GR}.

Coarse-graining is being developed as a strategy to tackle both problems at the same time. Indeed, coarse-graining is a set of techniques that have been used very successfully in the study of condensed matter systems or in lattice gauge theory in order to study phase transitions and the continuum limit. In the context of background independent theories, like \ac{GR}, the independence on the discretization should encode a continuum theory. Therefore, by studying the coarse-graining of \ac{LQG}, we might hope to find dynamics that are fixed point of the renormalization (or coarse-graining) flow. These points would correspond to continuum dynamics. Then, the critical surface around these points will determine the available parameters characterizing the dynamics. In practice, this means that the programme is very close to the \ac{AS} scenario programme but expressed in a somewhat different language. There are additional hopes in the \ac{LQG} programme: the non-perturbative language might shed light on the failure of the perturbative renormalization programme and the fixed point might be expressed more naturally in terms of geometric operators.

The problem here is that the technology to actually coarse-grain is simply not available as for now. The issue is two-sided. The first problem comes from the specifics of \ac{GR} which is, as we mentioned, a background independent theory. This means that the theory does not depend on some background metric which is expressed mathematically by the diffeomorphism invariance. The problem here comes from this independence as it makes difficult the definition of a notion of energy scale which is usually paramount in the definition of a renormalization flow. This problem is partially solved in the context of asymptotic safety \cite{Reuter:1996cp} but is difficult to solve in the \ac{LQG} context. The second issue comes from the specific structure, namely the discrete structure, of the theory. It is indeed hard to see how a continuous structure can emerge and more importantly, it is difficult to define a coarse-graining process on discrete structures.

\vspace{1em}

The goal of this Ph.D. thesis was to try and develop the technology needed to define a coarse-graining flow of \ac{LQG} and therefore determine the right dynamics as well as the means to study its continuum behavior. Because of the discrete nature of \ac{LQG}, we concentrated on two main problems. Both are related to the definition of a coarse-graining step. The first point is to search for natural variables to consider. In analogy to the Ising model coarse-graining, this would correspond to a choice of strategy like spin-decimation or spin averaging. The point is not the method itself but to isolate which macroscopic variable should be relevant. For Ising model for instance, the average spin is more natural as a macroscopic observable. Our work points toward holonomies around faces of the triangulation in the case of gravity. The second point is how to make the actual coarse-graining step. We have attempted a first go through the definition of loopy spin networks, but another route through the Immirzi parameter was also suggested. These are the two main questions that were studied in this PhD thesis.

This thesis was organized in four different parts. The first two were state of the art and concentrate on the necessary ground for our own developments. They also show how the problem might be formulated in a more precise manner. The first part concentrated on the kinematical aspects of \ac{LQG} that is on the framework itself as it is understood today. The second part concentrated on the dynamics and considered more general discussions around it to see if some insights could be gained for the coarse-graining process. The last two parts were the original work of this thesis. Each one concentrated on one of the questions we mentioned. In the first of these parts, we concentrated on the search of good coarse-grained variables. This part then exposed the work done in the following papers \cite{Charles:2015lva,Charles:2016xzi,Charles:2015rda}. The last part concentrated more on the problem of varying graphs and how to recast this problem onto fixed graphs dynamics. This presented the work done in the paper \cite{Charles:2016xwc}.

\vspace{1em}

Let us dwell a bit on our first contribution: the search for right large scale variables. Our concern was to describe large homogeneous surfaces. In the context of coarse-graining, the rationale could be understood in the following manner: we want to concentrate on the surfaces of coarse-grained volumes. This can be argued from different perspectives. First, it is the natural step  when considering homogeneous blocks. Indeed, in a non coarse-grained framework, the vertices of the spin network are naturally considered as elementary blocks defined by their surfaces. By analogy, a coarse-grained vertex would correspond to a homogeneously curved block defined by its surface. But more fundamentally, this joins other ideas in quantum gravity, like holography. The correspondence is not complete, as we expect the dynamics not to preserve homogeneity (except for very special cases) but it might be a good approximation. This is also linked to recent work by Freidel \textit{et al.} \cite{Freidel:2015gpa,Donnelly:2016auv} pointing out that a surface in quantum gravity has degrees of freedom linked to gauge invariance and these are relevant in the description of objects like black holes\graffito{Indeed, these degrees of freedom might be the reason why the firewall argument fails in \ac{LQG}.}. In particular, the degrees of freedom of a given surface are more numerous than just the degrees of freedom of a gas of puncture through a surface.% And these degrees of freedom are larger that one would expect in the standard, canonical, interpretation of loop quantum gravity.

It is convenient and usual to consider surfaces that separate vertices of a given spin network. In general however, surfaces can be more general in \ac{LQG} containing for instances edges of a spin network. In particular, surfaces can have loops embedded in them corresponding to curvature. These are the degrees of freedom pointed out by Freidel and they also seem to appear in our context of hyperbolic geometry, either with our original work with $\mathrm{SU}(2)$ or with our more recent work on $\mathrm{SL}(2,\mathbb{C})$. This suggests that the natural coarse-grained description of a state is given by a very special spin network (or suitable generalization) with both the graph dual to the discretization considered for the coarse-graining and the graph of the discretization itself. This would correspond to some kind of double graph as illustrated on fig.\ref{fig:SurfaceHolonomy}.
\begin{figure}[h!]
\centering
\includegraphics[scale=.7]{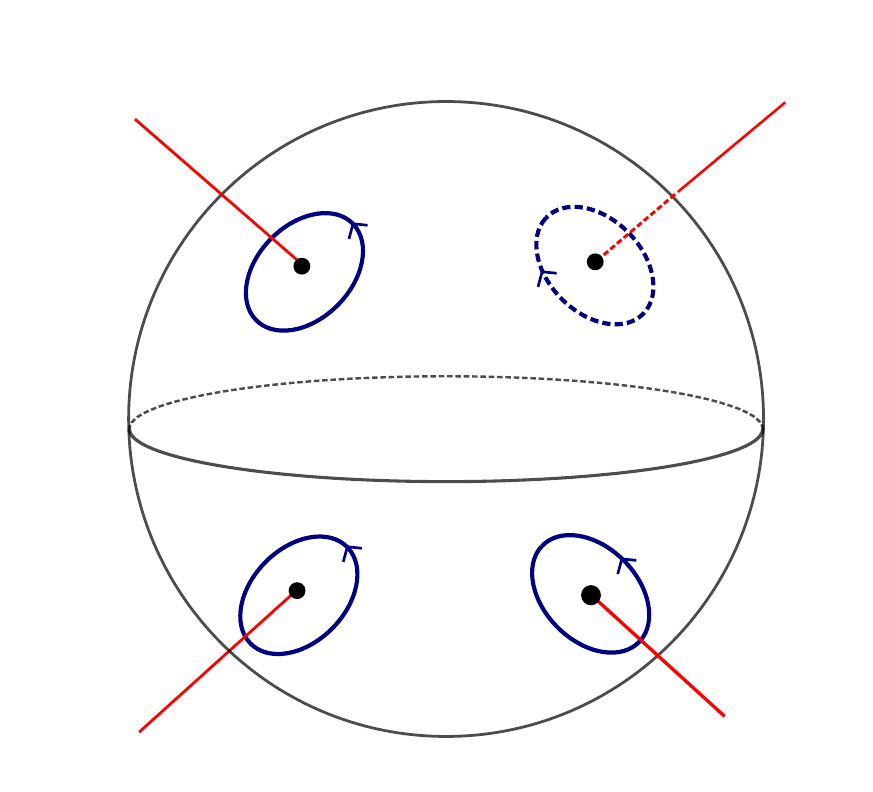}
\caption{In loop quantum gravity, spin network vertices are thought of as carrying volume excitations, of an abstract region of space bounded by a surface dual to the vertex. Quantum states of geometry are then usually defined as excitations of the holonomies of the Ashtekar-Barbero connection along the (transversal) edges puncturing the surface. Our derivation of closure constraints as discrete Bianchi identities relies on interpreting the holonomies on the dual surface as defining (non-abelian) normals to the  surface. This strongly suggests using new dual spin network structure, as a graph dressed with the data of holonomies along the edges and also around those edges. (image courtesy of Alexandre Feller)}
\label{fig:SurfaceHolonomy}
\end{figure}
Note that such a graph would not be a generalization \textit{per se}. Indeed, any usual spin network can be completed in this way provided any new edge gets a spin $0$. Conversely, any such graph is directly a spin network. The important point is the structure of the graph \textit{with respect to} the surfaces considered.

Our study of renormalization also brought out the relevance of the Immirzi parameter and in two ways. First, our latest construction used Immirzi-like parameter that were complex. The imaginary part was indispensable for the construction of the right quantities, namely the deformed normals. The link with the true Immirzi parameter should be investigated in a more general context (especially without a homogeneous surfaces). However, if the picture of the double graph which was just exposed is in anyway to be trusted, the surface holonomies, in order to carry information about the curvature, must be complex Ashtekar-Barbero holonomies. This reignites the possibility and the importance of a complex Immirzi parameter, though for our proposal, we should not restrict to self-dual connections but consider the whole class of possible complex Ashtekar-Barbero variables.

Second, we suggested to use the Immirzi parameter as a reference parameter for the coarse-graining process. At first, we suggested to use it as a cut-off. This might be a bit misguided. But the second idea, which was to use it as a renormalization parameter, might still be spot on. Let us recap a bit the thought process. At first, we deal more here with renormalization rather than coarse-graining\graffito{Note that renormalization can be thought of as coarse-graining with the cut-off removed.} and, in this context, we want to relate coupling strengths (like cross-sections) between different energy scales. This is basically the idea of renormalization: to use a more or less physical quantity to express the other physical results as a function of it. A theory is said to be renormalizable if we can express all physical results in function of a finite number of physical measurements. But because the physical results are dependant on many factors, like the energy scale, it induces a flow relating different ways of parametrizing the theory. Therefore, renormalization is really the study of how a theory is reexpressed at different energy scales. In quantum gravity however, the definition of energy is quite difficult. We do not have any background metric to define the scale of a phenomenon. This is why we proposed to use the Immirzi parameter as a reference scale.

The argument is two-fold. First, let's admit that the Immirzi parameter gets renormalized and that its behavior is monotonic (with the would-be energy scale). Then, it should be possible to parametrize the evolution of the other parameters as a function of the Immirzi parameter. This might not be the case. If the Immirzi parameter is a topological parameter for instance, which is up for debate, this renormalization will not happen. Still, in the context of usual \ac{QFT}, this was investigated \cite{Benedetti:2011nd} and it seems indeed that a renormalization flow is possible (but it might be inessential) for the Immirzi parameter. Therefore, rather than using an energy scale to label the flow of the couplings, we can use one of the parameter. There is a second point. The renormalization flow corresponds to the flow of the couplings with respect to the scale. Note also that the canonical transformation corresponding to a change of the Immirzi parameter looks like a scale transform. Indeed, in general relativity, it is difficult to relate different scales but if something very naive was to be done, it would be to scale up or down the densitized triad and change the extrinsic curvature accordingly to maintain a canonical pair of variables. This is precisely what an Immirzi transform does. This suggests a link between the Immirzi canonical transform and the renormalization programme.

Of course, we encounter a problem since the Immirzi canonical transform is not implemented unitarily in the quantum theory. There are several suggestive ways to alleviate the problem. First, it might be possible to consider larger Hilbert spaces containing the current representation as a subspace. In that case, the triad operator would have a continuum spectrum to allow for a unitary Immirzi transform. It might also be possible to find other representations of the holonomy-flux algebra. This would of course have to challenge the LOST theorem on one of its assumptions. It might also be the case that considering complex Immirzi parameters, which should allow for a continuum spectrum of the geometrical observables, enabling the definition of a unitary Immirzi transform. At this stage, all this is speculative and we just note the would-be connection between the implementation of varying Immirzi and complex Immirzi parameters.

\vspace{1em}

Our last contribution was way more technical and tried to pave the way toward an actual coarse-graining of the theory, in a way which would make use of the variables defined for large scale observables. We considered \textit{loopy spin networks} and developed the corresponding Hilbert spaces and operators in order to have a way of systematically coarse-graining the Hilbert space of loop quantum gravity. Loopy spin networks are simply spin networks with the add-on of having self-loops on vertices of the graph. These self-loops encode as vertex excitations fine details with respect to the coarse-graining scale. Once again, it is not a \textit{strict} generalization as such loops are always possible in the usual representation, but insisting on them might underline some peculiar properties. They were made into a real generalization once we considered statistics on the loops.

The introduction of these loops was natural with regard to gauge-invariance. In usual spin networks, gauge-invariance at the vertices (the closure condition) has the natural interpretation of flatness of the block. Because, we want to be able to include curvature, relaxing this condition is natural in the coarse-graining context. The possible excitations of the \textit{closure defect} are encoded in the self-loops of loopy spin networks. This also makes contact with two other remarks. First, as Rovelli pointed out \cite{Rovelli:2013fga}, gauge theories couple to gauge \textit{covariant} quantities. The closure defect or any quantity related to the self-loops might induce a natural observable, coupling small and large scales. Second, as we suggested, these loops can be interpreted as living on a surface. In that case, they might be identifiable with the degrees of freedom of a surface pointed out by Freidel \textit{et al}. These loops more naturally represent degrees of freedom of the interior of the region and the precise mapping should be investigated. Still, this is a natural framework, which can be studied by itself (as we have done) and should be investigated further for the coarse-graining of loop quantum gravity.

\vspace{1em}

To conclude, this work opens up a few possibilities, that we should recap for possible future work:
\begin{itemize}
\item One of the major points of this thesis is to reveal the importance of the structure of the double spin networks, at least in the context of homogeneously hyperbolically curved simplices. The structure should be studied further and could fuel a new interpretation of spin networks more suited for coarse-graining.
\item In the same line but on a different route, it should be possible to generalize the framework developed for hyperbolic geometries to spherical ones. Ideally a full generalization would be able to encompass both curvature signs. Then it should be possible to define \textit{curved spin networks} for which each vertex would correspond to a curved polyhedra. The curvature would be different from vertex to vertex. Such a space would be ideally suited for coarse-graining.
\item The possibility of coupling large and small scales \textit{via} the curvature defect has been suggested. This framework is actually quite hard to test. Moreover, simplified situations, like in 3d quantum gravity, are way to simple \textit{not} to work. But it might be possible to test loopy spin networks or tagged spin networks in some modified theory. For instance, it should be possible to study a modified BF theory. Only the holonomies of an independent set of loops would be fixed. This allows new (local) degrees of freedom in the quantum context and might be used to test the framework.
\item Finally, and again one of the major points, it seems that the possibility of changing the Immirzi parameter, either to any real value, or maybe to a complex value, might have a role in the coarse-graining of quantum gravity. Therefore, one very interesting pursuit of this thesis is to actually define and study the possibility of having representations of the quantum geometry capable of handling different Immirzi parameters.
\end{itemize}
Much more investigations are needed. It seems however that some ideas are starting to emerge from quite different branches of research. For instance, the full phase space of a surface in quantum gravity seems to appear in quite various context. The questions surrounding the precise role of the Immirzi parameter seem also to appear in multiple situations. It seems to us that these are very promising lines of search and more importantly that they, as we argued, are important lines of research for the coarse-graining of \ac{LQG} and the establishment of its continuum limit.

%\textbf{TODO}

%*****************************************
%*****************************************
%*****************************************
%*****************************************
%*****************************************

%\include{multiToC} % <--- just debug stuff, ignore for your documents
% ********************************************************************
% Backmatter
%*******************************************************
\appendix
\cleardoublepage
\part{Appendix}
%********************************************************************
% Appendix A
%*******************************************************
\chapter{Spinor formalism for loop quantum gravity} \label{app:spinors}

In this appendix, we recap the spinor formalism for loop quantum gravity. This formalism was mentionned in our discussion around the $\mathrm{U}(N)$ two-vertex model of cosmology and is actually quite useful to clarify some aspects of the kinematics of loop quantum gravity. What is presented here can be found in \cite{Borja:2010rc}.

\section{Schwinger representation of $\mathrm{SU}(2)$}

The spinor representation is based on the Schwinger representation of $\mathrm{SU}(2)$ and as such, we should begin there. The idea of the Schwinger representation is to put all the irreducible representations of $\mathrm{SU}(2)$ into a larger representation of a slightly larger algebra containing, of course, the algebra of $\mathrm{SU}(2)$ but also operators capable of going from one representation to the other.

Technically, it relies on the observation that the $\mathrm{SU}(2)$ can be obtained from the creation-annihilation algebra of two harmonic oscillators. More explicitly, let $a$, $b$ and $a^\dagger$, $b^\dagger$ be annihilation and creation operators with the following algebra:
\begin{equation}
[a,a^\dagger] = [b,b^\dagger] = 1
\end{equation}
all other commutators being zero. Then, the algebra of $\mathrm{SU}(2)$ can be represented as:
\begin{equation}
  \left\{
  \begin{array}{rcl}
    J_z &=& \frac{1}{2}\left(a^\dagger a - b^\dagger b\right) \\
    J_+ &=& a^\dagger b \\
    J_- &=& b^\dagger a \\
  \end{array}
  \right.
\end{equation}
It also turns out that the algebra commutes with:
\begin{equation}
  J = \frac{1}{2}\left(a^\dagger a + b^\dagger b\right)
\end{equation}

If, now, one concentrates on the usual representation for quantum harmonic oscillators, namely in our case, the Hilbert space $\mathcal{H}$ spanned by $|n_a,n_b\rangle$ which are eigenvectors of $a^\dagger a$ and $b^\dagger b$ with eigenvalues $n_a$ and $n_b$ respectively, then it can be shown that:
\begin{equation}
\mathcal{H} \simeq \bigoplus_{j\in\frac{\mathbb{N}}{2}} \mathcal{V}_j
\end{equation}
where $\mathcal{V}_j$ is the irreducible representation of $\mathrm{SU}(2)$ labelled by the half-integer $j$ corresponding to the Casimir value $j(j+1)$. It is quite easy to see that $\mathcal{H}$ shoul decompose over representations of $\mathrm{SU}(2)$ but the multiplicity must be checked. This can be done for instance by checking the dimensions.

To get back to the usual notation, let's keep $j$ to denote the representation of $\mathrm{SU}(2)$. Let $m$ denote the eigenvalue of $J_z$. The link between $n_a$, $n_b$, $j$ and $m$ is simply:
\begin{equation}
j = \frac{1}{2}(n_a + n_b),\quad m = \frac{1}{2}(n_a - n_b)
\end{equation}
Now it is obvious that indeed every (irreducible) representation of $\mathrm{SU}(2)$ is present once and only once on the Schwinger representation. The new operators $a$ and $b$ moreover are an added bonus allowing to circle among representation and transformation as a spinor.

This is suited for loop quantum gravity as al the representations are allowed on the links of a spin network. The Schwiner representation therefore allows the writing of operators that exists on the full Hilbert space.

\section{$\mathrm{U}(N)$ algebra of the intertwiner}

Now, if we consider one vertex and the corresponding intertwiner, we can devise a similar representation if the number of incoming link is fixed. Indeed, an intertwiner is an element of the following Hilbert space:
\begin{equation}
\mathrm{Inv}\left(\bigotimes_{i=1}^N \mathcal{V}_{j_i}\right)
\end{equation}
where $\mathrm{Inv}$ denotes the invariant subspace under gauge transform (common $\mathrm{SU}(2)$ action) and $j_i$ is the representation at the link $i$ which we supposed to be outgoing for simplicity. But if we want to consider generic intertwiners, it is very natural to consider the superposition of every possible spin:
\begin{equation}
\mathcal{H} \simeq \bigoplus_{\{j_i\}}\mathrm{Inv}\left(\bigotimes_{i=1}^N \mathcal{V}_{j_i}\right)
\end{equation}
where the sum is over every possible tuplets of representation. This can be rewritten as:
\begin{equation}
\mathcal{H} \simeq \mathrm{Inv}\left(\bigotimes_{i=1}^N \left[\bigoplus_{j\in\frac{\mathbb{N}}{2}}\mathcal{V}_{j}\right]\right)
\end{equation}
where we see the Schwinger representation appearing. Therefore a generic intertwiner, that is an intertwiner for a given set of links but no outgoing representation specified, is an element of the invariant subspace of the product of Schwinger representations, one for each link.

It is therefore quite natural to define the $a$, $b$ operators for each link. Let us note them $a_i$ and $b_i$ and we have:
\begin{equation}
[a_i,a_j^\dagger] = [b_i,b_j^\dagger] = \delta_{ij}
\end{equation}
all other commutators being zero. We can similarly define:
\begin{equation}
  \left\{
  \begin{array}{rcl}
    J_{i~z} &=& \frac{1}{2}\left(a_i^\dagger a_i - b_i^\dagger b_i\right) \\
    J_{i~+} &=& a_i^\dagger b_i \\
    J_{i~-} &=& b_i^\dagger a_i \\
    J_{i} &=& \frac{1}{2}\left(a_i^\dagger a_i + b_i^\dagger b_i\right)
  \end{array}
  \right.
\end{equation}
which are the $\mathrm{SU}(2)$ operators for each link. We can finally define gauge invariance as operator equations:
\begin{equation}
  \sum_i J_{i~x} |\psi\rangle = 0,\quad \sum_i J_{i~y} |\psi\rangle = 0,\quad \sum_i J_{i~z} |\psi\rangle = 0
\end{equation}
where $J_x$ and $J_y$ are defined as usual by:
\begin{equation}
J_+ = J_x + \mathrm{i}J_y,\quad J_+- = J_x - \mathrm{i}J_y
\end{equation}

Now, the interesting part: can we write down gauge-invariant observables on this space? More interestingly can we write down an algebra of such observables? This is indeed where the Schwinger representation comes in handy. Starting from $\vec{J}$, it is easy to construct gauge-invariant observables \textit{via} the dot product for instance. But such quadratic operators have no chance of forming a closed algebra. Indeed, the rank of the monomials will go one up at each commutator. But, starting from quadratic expressions in terms of the Schwinger operators $a$ and $b$, we will get once again quadratic expressions in these. Therefore it is possible to hope for an interesting algebra.

And indeed there is one, with the following operators:
\begin{equation}
  E_{ij} = a_i^\dagger a_j + b_i^\dagger b_j,\quad F_{ij} = a_i b_j - a_j b_i
\end{equation}
These operators form a closed algebra. Among these, the $E_{ij}$ form a sub-algebra with the structure of $\mathfrak{u}(n)$ and gives its name to this toolbox. These operators have a nice geometrical interpretation: the $E_{ij}$ conserve the sum $\sum_j j_i$ and can be interpreted as a discrete (area preserving) diffeomorphism group of the surface dual to the links. The $F_{ij}$ are able to create and destroy quanta of surfaces but must act on two links at a time, creating two half-integers at the same time. For this algebra of operators, the representation is cyclic, meaning that any intertwiner can be obtained by applying enough time the operators of this algebra onto the fundamental with all spins to zero. Note that this algebra seems to favor the regularly spaced discrete spectrum of area and might be an argument for using this ordering.

\section{A word on the two-vertex model}

Let us conclude this appendix with the link between this algebra and the two vertex model of cosmology. Let us consider a graph with two vertices and $N$ links between the two.

We want to define a notion of homogeneity. For this, we can define operators as previously for each vertex. For instance, we might have $E_{ij}^{(1)}$ and $E_{ij}^{(2)}$. A natural condition of homogeneity is:
\begin{equation}
  \forall i,j,~E_{ij}^{(1)} = E_{ij}^{(2)}
\end{equation}
For $i=j$ this is the matching condition, and for different $i$ and $j$ this encodes information about angles between spinors which should be the same from either of the two vertices. So basically, the surfaces and the angle between the surfaces should be the same for either point of the universe which is an idea of homogeneity.

Now, because of the $\mathfrak{u}(n)$ algebra satisfied by these operators, the constraints are actually first class and can be enforced quantum mechanically. They will also generate a kind of gauge invariance corresponding to an invariance under the action of the $\mathrm{U}(N)$ group. Because, the group acts transitively on the subspace of fixed total area, the only remaining degrees of freedom are the total area and its conjugated momentum. and this gives the $\mathrm{U}(N)$ framework for the two-vertex model. Hamiltonians respecting this symmetry can then be written and correspond to what has already been discussed in the main text.

%********************************************************************
% Appendix B
%*******************************************************
\chapter{Projective limits} \label{app:proj}

\section{General idea}

The general framework of a projective family and the projective limit can be found in \cite{Ashtekar1995}, where it is  applied to define the kinematical Hilbert space of spin network states for loop quantum gravity.
Here we apply  these definitions to loopy spin networks,  in order to define superposition states of potentially an infinite number of little loops. To this purpose, we focus on the flower graph, with a single vertex and an arbitrary number of loops attached to that central node.

In order to define precisely this idea of varying number of loops, we start with wave-functions over a finite number of loops and define a nesting, that is describe how to include a set of loops inside a larger one.
We will identity the set of all potential loops with the set of integers. Finite sets of loops are defined as finite subsets of integers.
Loops are labeled by the integers and are a priori distinguishable. For instance, a wave-function with the support on the loop number $2$ and a wave-function on the loop number $277$ are not the same though they both are one-loop states and depend on only one variable, as illustrated in fig.\ref{fig:potentialloops}.

\begin{figure}[h!]

  \centering

  \begin{tikzpicture}
    \coordinate(O1) at (0,0);
    \coordinate(O2) at (3,0);

    \draw (O1) to[loop,scale=3] (O1) ++(0,1.2) node {$\circled{2}$};
    \draw[dashed,gray] (O1) to[loop,scale=3,rotate=180] (O1) ++(0,-1.2) node {$\circled{277}$};
    \draw (O1) node[scale=1] {$\bullet$};

    \draw (1.5,0) node[scale=2] {$\neq$};

    \draw[dashed,gray,scale=3] (O2) to[loop] (O2);
    \draw[dashed,gray] (O2) ++(0,1.2) node {$\circled{2}$};
    \draw[scale=3,rotate=180] (O2) to[loop] (O2);
    \draw (O2) ++(0,-1.2) node {$\circled{277}$};
    \draw (O2) node[scale=1] {$\bullet$};

  \end{tikzpicture}

  \caption{We distinguish the different potential loops and therefore consider the resulting wave-functions as different even when they excite the same number of loops.}
  \label{fig:potentialloops}

\end{figure}
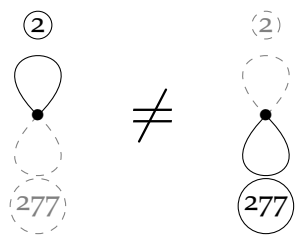

Mathematically, we consider the set of all finite subsets of integers $\mathcal{P}_{<\infty}(\mathbb{N})$. To each subset $E \in \mathcal{P}_{<\infty}(\mathbb{N})$, we associate the set $\SU(2)^{E}$ of colorings of the corresponding loops by $\SU(2)$ group elements. Then wave-functions on $E$ are gauge-invariant functions over $\SU(2)^{E}$:
\begin{equation}
  \{\Psi : \SU(2)^E \rightarrow \mathbb{C}~:~ \forall g \in \SU(2),~
  \Psi(\{gh_{\ell}g^{-1}\}_{\ell\in E}) = \Psi(\{h_{\ell}\}_{\ell\in E})\}\,.
\end{equation}
Defining the scalar product using the Haar measure over $\SU(2)^{E}$, the Hilbert space $\cH_{E}$ of quantum states on the loopy spin network defined by the subset $E$ of loops is the $L^{2}(\SU(2)^{E}/\textrm{Ad}\SU(2))$.

The space of loops is equipped with a partial directed order given by the inclusion of subsets of integers. The partial directed order encodes how different subsets are nested within one another: a wave-function over the loop number $2$ and a wave-function over the loop number $277$ are different but they are both embedded in the larger class of wave-functions which depend on both loop number $2$ and loop number $277$ as illustrated in fig.\ref{fig:Embedding}.

\begin{figure}[h!]

  \centering

  \begin{tikzpicture}
    \coordinate(O1) at (0,0);
    \coordinate(O2) at (5,0);
    \coordinate(O3) at (2.5,-4);

    \draw (O1) to[loop,scale=3] (O1) ++(0,1.2) node {$\circled{2}$};
    \draw[dashed,gray] (O1) to[loop,scale=3,rotate=180] (O1) ++(0,-1.2) node {$\circled{277}$};
    \draw (O1) node[scale=1] {$\bullet$};

    \draw[dashed,gray,scale=3] (O2) to[loop] (O2);
    \draw[dashed,gray] (O2) ++(0,1.2) node {$\circled{2}$};
    \draw[scale=3,rotate=180] (O2) to[loop] (O2);
    \draw (O2) ++(0,-1.2) node {$\circled{277}$};
    \draw (O2) node[scale=1] {$\bullet$};

    \draw[scale=3] (O3) to[loop] (O3);
    \draw (O3) ++(0,1.2) node {$\circled{2}$};
    \draw[scale=3,rotate=180] (O3) to[loop] (O3);
    \draw (O3) ++(0,-1.2) node {$\circled{277}$};
    \draw (O3) node[scale=1] {$\bullet$};

    \draw[<-,>=stealth,thick] ($(O3)+(0.5,1)$) -- ($(O2)+(-0.5,-1)$);
    \draw[<-,>=stealth,thick] ($(O3)+(-0.5,1)$) -- ($(O1)+(0.5,-1)$);

  \end{tikzpicture}

  \caption{Though different, two functions over two different loops can be embedded in a larger space of functions depending on several loops by identifying them with functions with trivial dependancy on some loops.}
  \label{fig:Embedding}

\end{figure}
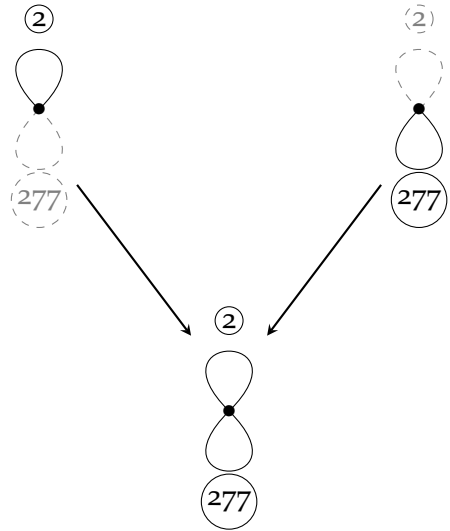

This partial ordering by inclusion of subsets induces a projective structure on the loop colorings by group elements.
We define a projector $p_{EE'}$ defined for every pair of subsets $(E,E')$ such that $E \subseteq E'$ by:
\begin{eqnarray*}
  p_{EE'} : \SU(2)^{E'} &\rightarrow& \SU(2)^{E} \\
  s &\mapsto& \restriction{s}{E}
\end{eqnarray*}
This projector is simply the canonical restriction from the larger subset $E'$ to the smaller subset $E$. These projectors satisfy a key transitivity property:
\begin{equation}
  \forall E,E',E'', \quad E\subset E'\subset E''\, \Longrightarrow p_{E'E''} \circ p_{EE'} = p_{EE''}\,,
\end{equation}
so that the couple of sets $(\SU(2)^{E},p_{EE'})_{E,E' \in \mathcal{P}_{<\infty}(\mathbb{N})}$ form what is called a \textit{projective family}.
The projective limit $\overline{\SU(2)}$ is then defined by:
\begin{equation}
  \overline{\SU(2)} = \Big{\{}
  (g_{E})_{E \in \mathcal{P}} \in \times_{E\in \mathcal{P}} \SU(2)^{E} ~:~ E \subseteq E' \Rightarrow p_{EE'} g_{E'} = g_{E}
  \Big{\}}
\end{equation}
Intuitively, this corresponds to collections of colorings on all possible subsets of loops which are compatible with each other with respect to the inclusion. Therefore, these compatibility conditions between all finite samplings of the collection, as illustrated in fig.\ref{fig:Compatibility}, is the precise implementation of the notion of a coloring of an infinite number of loops.

\begin{figure}[h!]

  \centering

  \begin{tikzpicture}
    %[thick, scale=0.6]
    \coordinate(O1) at (0,0);
    \coordinate(O2) at (-6,-4);
    \coordinate(O3) at (-2,-4);
    \coordinate(O4) at (2,-4);
    \coordinate(O5) at (6,-4);

    \draw[red,scale=3] (O1) to[loop] (O1);
    \draw[red] (O1) ++(0,1.2) node {$\circled{2}$};
    \draw[blue,scale=3,rotate=-90] (O1) to[loop] (O1);
    \draw[blue] (O1) ++(1.2,0) node {$\circled{277}$};
    \draw[green,scale=3,rotate=90] (O1) to[loop] (O1);
    \draw[green] (O1) ++(-1.2,0) node {$\circled{42}$};
    \draw (O1) node[scale=1] {$\bullet$};

    \draw[dashed,red!50,scale=3] (O2) to[loop] (O2);
    \draw[dashed,red!50] (O2) ++(0,1.2) node {$\circled{2}$};
    \draw[blue,scale=3,rotate=-90] (O2) to[loop] (O2);
    \draw[blue] (O2) ++(1.2,0) node {$\circled{277}$};
    \draw[green,scale=3,rotate=90] (O2) to[loop] (O2);
    \draw[green] (O2) ++(-1.2,0) node {$\circled{42}$};
    \draw (O2) node[scale=1] {$\bullet$};

    \draw[red,scale=3] (O3) to[loop] (O3);
    \draw[red] (O3) ++(0,1.2) node {$\circled{2}$};
    \draw[dashed,blue!50,scale=3,rotate=-90] (O3) to[loop] (O3);
    \draw[dashed,blue!50] (O3) ++(1.2,0) node {$\circled{277}$};
    \draw[green,scale=3,rotate=90] (O3) to[loop] (O3);
    \draw[green] (O3) ++(-1.2,0) node {$\circled{42}$};
    \draw (O3) node[scale=1] {$\bullet$};

    \draw[dashed,red!50,scale=3] (O4) to[loop] (O4);
    \draw[dashed,red!50] (O4) ++(0,1.2) node {$\circled{2}$};
    \draw[dashed,blue!50,scale=3,rotate=-90] (O4) to[loop] (O4);
    \draw[dashed,blue!50] (O4) ++(1.2,0) node {$\circled{277}$};
    \draw[green,scale=3,rotate=90] (O4) to[loop] (O4);
    \draw[green] (O4) ++(-1.2,0) node {$\circled{42}$};
    \draw (O4) node[scale=1] {$\bullet$};

    \draw (O5) node[scale=2] {$\cdots$};

    \draw[->,>=stealth,thick] ($(O1)+(-1.5,-1)$) -- ($(O2)+(1,1)$);
    \draw[->,>=stealth,thick] ($(O1)+(-0.5,-1.5)$) -- ($(O3)+(0.5,1)$);
    \draw[->,>=stealth,thick] ($(O1)+(0.5,-1.5)$) -- ($(O4)+(-0.5,1)$);
    \draw[->,>=stealth,thick] ($(O1)+(1.5,-1)$) -- ($(O5)+(-1,1)$);
    %\draw[<-,>=stealth,thick] ($(O3)+(-0.5,1)$) -- ($(O1)+(0.5,-1)$);

  \end{tikzpicture}

  \caption{The projective limit is made of  collections of colorings of finite subsets equipped with compatibility conditions: the coloring of a finite subset is the projection of the coloring of any larger subset.}
  \label{fig:Compatibility}

\end{figure}
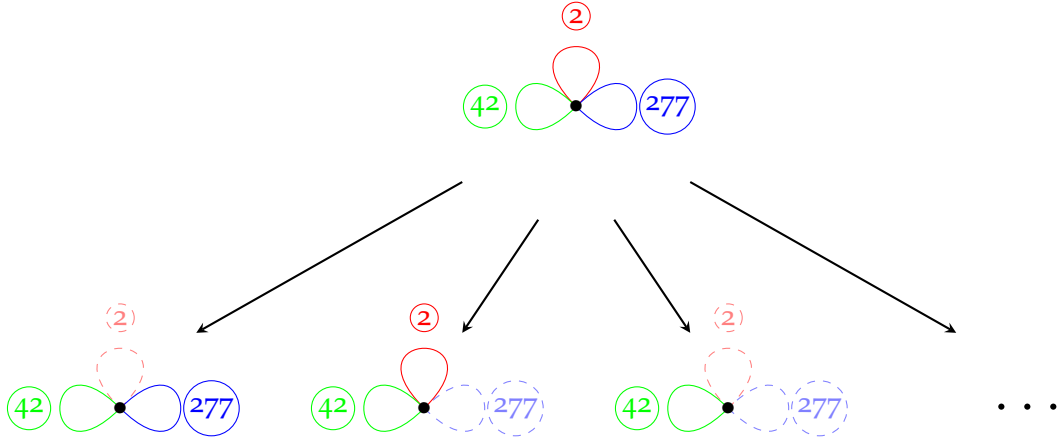

We translate the projective structure to the space of wave-functions. The projectors $p_{EE'}$ for $E\subset E'$ turn into injections $I_{EE'}\equiv p^{*}_{EE'}$ defined by their pull-backs:
\begin{eqnarray}
  I_{EE'} : \cH_{E} &\rightarrow& \cH_{E'}\\
  f &\mapsto& p^{*}_{EE'}f \,\,:\, p^{*}_{EE'}f \big{(}\{g_{\ell}\}_{\ell\in E'}\big{)}=f \big{(}\{\{g_{\ell}\}_{\ell\in E}\}\big{)}\,,\nn
\end{eqnarray}
where $p^{*}_{EE'}f$ trivially depends on group elements living on loops of $E'$ which do not belong to the smaller set $E$.
The compatibility conditions translates into an equivalence relation:
\begin{equation}
  f_{E_1} \sim f_{E_2} \Leftrightarrow \forall E_3\in\mathcal{P},~E_1 \subseteq E_3,~E_2\subseteq E_3,~p^*_{E_1 E_3} f_{E_1} = p^*_{E_2E_3} f_{E_2}
\end{equation}
This allows to define wave-functions on the projective limit of the loop colorings $\overline{\SU(2)}$ and  give a precise sense to functions over an infinite number of loops:
\begin{equation}
  \mathcal{H}=\left(\bigcup_{E\in\mathcal{P}} \mathcal{H}_{E}\right)/\sim
\end{equation}

\section{Second representation}

In order to make this projective limit less abstract and easier to handle, we use another representation. For every equivalence class of wave-functions in the projective limit $\cH$, let us remove all the trivial dependency and pick its representant based on the smallest subset of loops. So, in practice, we define spaces of ``proper states'', i.e. wave-functions that have no trivial dependency:
\begin{eqnarray}
  \mathcal{H}_{E}^0 =
  \{
  f \in \mathcal{H}_{E}\,\, :\, \forall \ell\in E\,,\,\int_{\SU(2)}\dd h_{\ell}\,f=0
  \}\,.
\end{eqnarray}
This is the space of functions really defined on the subset $E$, with an actual dependance on each loop  and no constant term. The integral condition removes all the spin-0 components of the wave-functions.
First, we show that an arbitrary wave-function over the subset $E$ of loops can be fully decomposed into proper states with support on all the subsets of $E$:
%%%
\begin{lemma}
  The following  isomorphism holds as a pre-Hilbertian space isomorphism:
  \begin{equation}
    \forall E \in \mathcal{P}_{<\infty}(\mathbb{N}),
    \quad
    \mathcal{H}_{E} \simeq \bigoplus_{F \in \mathcal{P}(E)} \mathcal{H}_{F}^0\,,
  \end{equation}
  where the direct sum is over all subsets $F\subset E$. This isomorphism is realized  through the projections $f_{F}=P_{E,F}f$ acting on wave-functions $f\in\cH_{E}$:
  \be
  f_{F}\big{(}
  \{h_{\ell}\}_{\ell\in F}
  \big{)}
  \,=\,
  \sum_{\tF\subset  F}
  (-1)^{\#\tF}
  \int \prod_{\ell\in E\setminus F}\mathrm{d}g_{\ell}
  \prod_{\ell\in\tF}\mathrm{d}k_{\ell}\,
  f\big{(}
  \{h_{\ell}\}_{\ell\in F\setminus \tF},
  \{k_{\ell}\}_{\ell\in\tF},
  \{g_{\ell}\}_{\ell\in E\setminus F}
  \big{)}\,.
  \ee
  Its inverse is the re-summation of the projections:
  \be
  f=\sum_{F\subset E} f_{F}\,.
  \ee

\end{lemma}
%%%
\begin{proof}
  We proceed in two steps. First we check that each projection $f_{F}$ is a proper state,
  $$
  \forall\ell\in F\,,\quad\int \mathrm{d}h_{\ell}\,f_{F}=0\,,
  $$
  and that re-summing these projections $\sum_{F\subset E} f_{F}$ yields $f$.
  Second, we check that the integral condition, ensuring that there is no spin-0 mode, also implies that the subspaces $\mathcal{H}_{F}^0$ are pairwise orthogonal, which concludes the proof.
\end{proof}

This decomposition generalizes to the projective limit:
\begin{prop}
  The following  isomorphism holds as a pre-Hilbertian space isomorphism:
  \begin{equation}
    \mathcal{H} \simeq \bigoplus_{E \in \mathcal{P}_{<\infty}(\mathbb{N})} \mathcal{H}_{E}^0\,.
  \end{equation}
\end{prop}
\begin{proof}
  If $(f_E)_{E \in \mathcal{P}_{<\infty}(\mathbb{N})}$ is in $\bigoplus_{E \in \mathcal{P}_{<\infty}(\mathbb{N})} \mathcal{H}_{E}^0$, we  define the set of subsets on which the state $f$ does not vanish:
  \begin{equation}
    C_f = \{ E \in \mathcal{P}_{<\infty}(\mathbb{N}) : f_E \neq 0 \}\,.
  \end{equation}
  By definition of the direct sum, $C_f$ is finite, so we can define the finite subset $F = \cup_{E \in C_f} E$ and the re-summation map:
  \begin{eqnarray}
    \phi : \bigoplus_{E \in \mathcal{P}_{<\infty}(\mathbb{N})} \mathcal{H}_{E}^0 &\rightarrow& \mathcal{H} \\
    (f_E)_{E \in \mathcal{P}_{<\infty}(\mathbb{N})} &\mapsto& \left[\sum_{E \in C_f} f_E\right]
    \nn
  \end{eqnarray}
  where the brackets refer to the equivalence class of the function.
  This map is obviously  linear and we now look for a definition of its inverse.
  So let us consider a state in  $\mathcal{H}$, that is an equivalence class $s$. We define the set of subsets of loops on which it has support:
  \begin{equation}
    D_s = \{ E \in \mathcal{P}_{<\infty}(\mathbb{N}) : \exists f \in s,~f \in \mathcal{H}_{E} \}
  \end{equation}
  Then we consider the smallest set in $D_s$ , which can be defined\graffito{This is the point where we choose not to use the completion and just have an isomorphism of pre-Hilbertian spaces in order to have the existence of $F_f$.} as the intersection $F_s = \bigcap_{E \in D_s} E$.
  In a sense, this is the minimal support of the state $s$.
  We choose a representative $f^{s}$ of the equivalence class $s$ in $F_{s}$. It is actually unique by definition of the equivalence relation.
  Then we consider the decomposition in proper states of $f^{s}$ over all subsets $F$ of $F_{s}$ and define:
  \begin{eqnarray}
    \psi : \mathcal{H} &\rightarrow& \bigoplus_{E \in \mathcal{P}_{<\infty}(\mathbb{N})} \mathcal{H}_{E}^0 \\
    ~s &\mapsto& \sum_{F\subset F_{s}}P_{F_{s},F} f^{s} \nn
  \end{eqnarray}
  It is direct to check that it is indeed the inverse of $\phi$.

\end{proof}

This decomposition into proper states is very useful to visualize the space: each wave-function can be decomposed into a sum of wave-functions over a finite number of loops but with no trivial dependancy. This gives a precise meaning to superpositions of numberss of loops.

%********************************************************************
% Appendix C
%*******************************************************
\chapter{Distributions on $\mathrm{SU}(2)$} \label{app:distribution}

\section{Dual space definition}

We would like to impose the holonomy constraints for BF theory which read for a single group element:
\begin{equation}
  \forall g \in \mathrm{SU}(2),\,\,
  \hchi\vphi\,(g)=
  \chi_\frac{1}{2}(g) \vphi(g) = 2\,\ ,\vphi(g)\,.
\end{equation}
If we stay in the strict framework of the Hilbert space $L^{2}(\SU(2))$, no square integrable function actually provides such an eigenvector for $\hchi$ and we should solve this equation in the dual space. As is standard in  quantum mechanics, the natural framework for solving the equation is a rigged-Hilbert space (or Gelfand triple), that is a triplet: $\mathcal{S} \subset \mathcal{H} \subset \mathcal{S}^*$. The space $\mathcal{H}$ is the Hilbert space. The smaller space $\mathcal{S}$ is provided with a stronger topology than the induced one and can thought of as the test function space, while its dual $\mathcal{S}^{*}$ is the  space of continuous linear forms on $\cS$ and defines the distribution space.
The major property of $\mathcal{S}$ is to be small enough for the algebra of observables to be defined over it. Then the operator algebra can be naturally extended on $\mathcal{S}^{*}$ and thus on $\mathcal{H}$. For instance, an operator $A$ defined on $\mathcal{S}$ acts on a (dual) state $\varphi$ be in $\mathcal{S}^*$ as:
\begin{equation}
  \forall f \in \mathcal{S},~A\varphi(f) = \varphi(A^\dagger f)\,.
\end{equation}

So let us be explicit for functions over $\SU(2)$. The Hilbert space $\mathcal{H}$ is the space of square-integrable functions.  The space $\mathcal{S}$  is usually chosen to be the Schwarz space so that canonical position and momentum operator can be defined. Here, the rapid fall-off condition is not needed since we are dealing with a compact group, but we keep the smoothness requirement:
\begin{equation}
  \mathcal{S} = \{\psi \in \mathcal{H} / \psi \in \mathcal{C}^\infty\}
\end{equation}
%This space is just the space of smooth functions over $\mathrm{SU}(2)$ that are invariant under conjugation.
%
Regarding the topology, the space $\mathcal{S}$ is naturally endowed with the convergence on every $\mathcal{C}^k$ space.
More precisely  the space of $\mathcal{C}^k$ functions is equipped with the following norm:
\begin{equation}
  \|f\|_{\mathcal{C}^k} = \sup_{0 \le i \le k} \|\partial_{a_1,...,a_i} f\|_\infty
\end{equation}
This norm has two nice properties. First, differentiation is continuous from $\mathcal{C}^k$ to $\mathcal{C}^{k-1}$. Second, the topology induced by the norms are finer as $k$ goes to infinity. So the limit topology on $\mathcal{S} = \bigcap_{k \in \mathbb{N}} \mathcal{C}^k$goes as follows: a sequence of functions $f_{n\in\N}$ in $\mathcal{S}$ admits  $0$ as its limit if the sup-norm of all its derivatives $\|\partial_\alpha f_{n}\|_{\infty}$ go to $0$ for arbitrary multi-index $\alpha$. This is topology is naturally finer than all the $\mathcal{C}^k$ topologies and the differentiation is still continuous.  Provided with this topology, $\mathcal{S}$ is a Frechet space: it is complete and metrizable (though no norm is defined). Note that, although all the $\mathcal{C}^k$ are Banach spaces, their descending intersection $\bigcap_{k \in \mathbb{N}} \mathcal{C}^k$ is not.

\section{Fourier point of view}

Things are usually clearer and more explicit in the Fourier decomposition. Let us consider the Fourier decomposition of a function over $\SU(2)$ on the Wigner matrices:
$$
f(g)=\sum_{j,a,b}f^{j}_{ab}D^{j}_{ab}\,.
$$
By the Fourier convergence theorem, smoothness actually translates into a rapid fall-off of the  Fourier coefficients:
\begin{equation}
  f\in\cS
  \quad\Longleftrightarrow\quad
  \forall K\in\N\,,\,\,\sum_{j} |f^{j}| d_j^K < +\infty\,,
\end{equation}
where $d_{j}=(2j+1)$ is the dimension of the spin-$j$ representation and $|f^{j}|$ can equally be the sup-norm or the square-norm of the matrix $f^{j}_{ab}$.
This also means that the Fourier coefficients of a distribution  cannot diverge faster than polynomially:
\be
\vphi\in\cS^{*}
\quad\Longleftrightarrow\quad
\exists K\,,\,\,\sum_{j} |\vphi^{j}| d_j^{-K} < +\infty\,.
\ee
The strong topology on $\cS$ means that a sequence of smooth functions $f_{n}$ converges to 0 in $\cS$ if and only if all the $K$-power sums go to 0:
\be
\lim_{n\arr\infty}f_{n}=0
\quad\Longleftrightarrow\quad
\forall K\in\N\,,\,\,\sum_{j} |f^{j}_{n}| d_j^K \,\,\underset{n\arr\infty}\longrightarrow 0\,.
\ee
This ensures that the evaluations of a distribution $\vphi$ will also converge $\vphi(f_{n})\arr0$.

\section{Holonomy constraint}

Let us apply this to the holonomy constraints for functions invariant under conjugation on $\SU(2)$. In this case, all functions decompose on the characters,
$$
\vphi(g)=\sum_{j\in\f\N2}\vphi_{j}\chi_{j}(g)\,,
$$
and the eigenvalue problem $\chi(g)\vphi(g)=2\vphi(g)$ translates into a recursion relation on the Fourier coefficients:
\be
\rho \vphi_{0}=\vphi_{\f12}\,,\quad
\rho \vphi_{j\ge\f12}=\vphi_{j-\f12}+\vphi_{j+\f12}\,.
\ee
Once we fix the initial condition $\vphi_{0}$, this recursive equation has a solution for every complex value $\rho\in\C$, but this does not systematically defines a solution state, in $L^{2}$ or a distribution. The solution to the recursion is given in terms of the two solutions $\mu_{\pm}$ of the quadratic equation $\mu^{2}-\rho \mu+1=0$:
\be
\forall j\in\f\N2\,,\,\,
\vphi_{j}=\f12(\mu_{+}^{2j}+\mu_{-}^{2j})\,.
\ee
For $\rho=2$, the discriminant $(\rho^{2}-4)$ vanishes and this ansatz fails leads: instead of the power law, we get a linear growth $f_{j}=(2j+1)$, which leads back to the $\delta$-distribution peaked on the identity. For real values $|\rho|< 2$, the discriminant is negative and we get an oscillatory solution $\vphi_{j}=\cos(2j\theta)$ with $\cos\theta =\rho$, which gives a $\delta$-distribution fixing the class angle of the group element $g$ to $\theta$. For $|\rho|>2$, the positive discriminant will leads to exponentially divergent coefficients $\vphi_{j}$ and do not define a proper distribution.

%********************************************************************
% Appendix D
%*******************************************************
\chapter{Laplacian constraint, double recursion and flatness equations} \label{app:Laplacian}

We introduce another constraint supplementing the holonomy constraint in order to truly impose flatness and get the $\delta$-distribution as unique solution: we impose the Laplacian constraint $(\tDelta-\Delta)=0$, where $\Delta=\pp^{L}_{a}\pp^{L}_{a}=\pp^{R}_{a}\pp^{R}_{a}$ is the usual Laplacian operator and $\tDelta=\pp^{L}_{a}\pp^{R}_{a}$ mixes the right and left derivations. These two operators do not change the spin $j$ and act rather simply on the Wigner matrices:
\be
\begin{array}{l}
\Delta D^{j}_{mn}(h)=-D^{j}_{mn}(hJ_{a}J_{a})=-j(j+1)D^{j}_{mn}(h), \\
\Delta D^{j}_{mn}(h)=-D^{j}_{mn}(J_{a}hJ_{a})\,.
\end{array}
\ee
Using the explicit action of the three $\su(2)$ generators, and applying the Cauchy-Schwarz inequality to bound the sums, we can check that the operator $(\tDelta-\Delta)$ is positive.
We can also translate the Laplacian constraint $\Delta \vphi =\tDelta\vphi$ into equations on the Fourier coefficient matrices $\vphi^{j}$:
\beq
j(j+1)\vphi^{j}_{nm}&=&nm\,\vphi^{j}_{nm} \nn\\
&+&\,\,\f12\vphi^{j}_{n-1,m-1}\sqrt{(j+n)(j-n+1)(j+m)(j-m+1)}\nn\\
&+&\,\,\f12\vphi^{j}_{n+1,m+1}\sqrt{(j-n)(j+n+1)(j-m)(j+m+1)}\,.\nn\\
&&
\label{Lrecursion}
\eeq
This is a recursion at fixed spin $j$ on the matrix elements of each $\vphi^{j}$ independently. It works at fixed $(n-m)$, that is along the diagonals of the matrix, determining the matrix elements from, say, the highest weight components:
$$
\begin{array}{ll}
\vphi^{j,j}&\arr \,\vphi^{j-1,j-1}\arr\vphi^{j-2,j-2}\arr\dots \\
\vphi^{j,j-1}&\arr\, \vphi^{j-1,j-2}\arr\vphi^{j-2,j-3}\arr\dots \\
\vphi^{j,j-2}&\arr \,\vphi^{j-1,j-3}\arr\vphi^{j-2,j-4}\arr\dots 
\end{array}
$$
Putting this constraint with the holonomy constraint, we get a double recursion structure. The holonomy constraint realizes a recursion on the spin $j$, determining the matrix $\vphi^{j}$  from the lower spin matrices, while the Laplacian constraint implements a recursion on the magnetic moment $m$ within each matrix $\vphi^{j}$:
$$
\vphi^{0} \arr
\mat{cc}{\searrow & \searrow \\
\searrow & \searrow }
\arr
\mat{ccc}{\searrow & \searrow & \searrow \\
\searrow & \searrow & \searrow \\ \searrow & \searrow & \searrow  }
\arr
\mat{cccc}{\searrow & \searrow & \searrow & \searrow \\
\searrow & \searrow & \searrow & \searrow \\ \searrow & \searrow & \searrow & \searrow \\
 \searrow & \searrow & \searrow & \searrow}
 \arr\dots
$$
This allows to solve the problem of the infinite initial conditions needed for the holonomy constraint.
We easily check that $\vphi_{nm}=\delta_{nm}$ is a solution:
$$
j(j+1)=m^{2}+\f12(j+m)(j-m+1)+\f12(j-m)(j+m+1)\,.
$$
But then, we completely solve the constraint and show that it implies that the function is invariant under conjugation:
\begin{prop}
Let us consider the Laplacian constraint $(\Delta-\tDelta)\,\vphi=0$ translated to the Fourier decomposition $\vphi(h)=\sum_{j,m,n}(2j+1)\tr\,\vphi^{j}D^{j}(h)$. Then the each of the Fourier coefficient matrix $\vphi^{j}$ at fixed spin $j$ is proportional to the identity. This means that $\vphi$ is invariant under conjugation.
\end{prop}
\begin{proof}
Let us fix $j\ge \f12$. The spin-0 component $\vphi^{0}$ is unconstrained and left free.
The recursion relation \eqref{Lrecursion} allows to start with an element $\vphi_{j,j-M}$ with $0\le M \le 2j$ and to determine all of the following components along the corresponding diagonal, $\vphi_{j-N,j-M-N}$ for $0\le N\le (2j-M)$.
One actually gets a relation in terms of combinatorial factors:
\be
\vphi_{j-N,j-M-N}=
\vphi_{j,j-M}\, \sqrt{\f{(M+N)!}{M!N!}\,\f{(2j)!(2j-M-N)!}{(2j-M)!(2j-N)!}}
\,.
\ee
In particular, one obtains for the other end of the diagonal:
\be
\vphi_{-j+M,-j}
\,=\,
\vphi_{j,j-M}\,\f{(2j)!}{M!(2j-M)!}\,.
\ee
The trick is that the recursion relation \eqref{Lrecursion} is symmetric under the exchange $(n,m)\leftrightarrow (-m,-n)$: we start now from the other end of the same diagonal and work our way back to the initial top element. Therefore the previous equality holds but in the opposite  way:
\be
\vphi_{j,j-M}
\,=\,
\vphi_{-j+M,-j}\,\f{(2j)!}{M!(2j-M)!}
\,=\,
\vphi_{j,j-M}\left(\f{(2j)!}{M!(2j-M)!}\right)^{2}
\ee
Thus either:
\be
\f{(2j)!}{M!(2j-M)!}=1
\ee
or:
\be
\vphi_{j,j-M}=0\,.
\ee
In the special case $M=0$, along the principal diagonal, the recursion relation simplifies and reads $\vphi^{j,j}=\vphi^{j-1,j-1}=\vphi^{j-2,j-2}=\dots$. For all the other cases $M\ne 0$, the matrix elements must vanish. This concludes the proof that the matrix $\vphi^{j}$ must be proportional to the identity.

\end{proof}

%********************************************************************
% Other Stuff in the Back
%*******************************************************
\cleardoublepage%********************************************************************
% Bibliography
%*******************************************************
% work-around to have small caps also here in the headline
\manualmark
\markboth{\spacedlowsmallcaps{\bibname}}{\spacedlowsmallcaps{\bibname}} % work-around to have small caps also
%\phantomsection 
\refstepcounter{dummy}
\addtocontents{toc}{\protect\vspace{\beforebibskip}} % to have the bib a bit from the rest in the toc
\addcontentsline{toc}{chapter}{\tocEntry{\bibname}}
\label{app:bibliography}
\printbibliography

\cleardoublepage%*******************************************************
% Declaration
%*******************************************************
\refstepcounter{dummy}
\pdfbookmark[0]{Declaration}{declaration}
\chapter*{Declaration}
\thispagestyle{empty}
This thesis is a presentation of my original research work.
Wherever contributions of others are involved, every effort is
made to indicate this clearly, with due reference to the literature,
and acknowledgement of collaborative research and discussions.
The work was done under the guidance of Doctor Etera Livine, at the physics laboratory of the École Normale Supérieure de Lyon.
\bigskip
 
\noindent\textit{\myLocation, \myTime}

\smallskip

\begin{flushright}
    \begin{tabular}{m{5cm}}
        \\ \hline
        \centering\myName \\
    \end{tabular}
\end{flushright}

\cleardoublepage\pagestyle{empty}

\hfill

\vfill

\pdfbookmark[0]{Colophon}{colophon}
\section*{Colophon}
This document was typeset using the typographical look-and-feel \texttt{classicthesis} developed by Andr\'e Miede. 
The style was inspired by Robert Bringhurst's seminal book on typography ``\emph{The Elements of Typographic Style}''. 
\texttt{classicthesis} is available for both \LaTeX\ and \mLyX: 
\begin{center}
\url{https://bitbucket.org/amiede/classicthesis/}
\end{center}
%Happy users of \texttt{classicthesis} usually send a real postcard to the author, a collection of postcards received so far is featured here: 
%\begin{center}
%\url{http://postcards.miede.de/}
%\end{center}
 
\bigskip

\noindent\finalVersionString

\end{document}